%% file: Article.tex
\documentclass[]{imsart}

\usepackage[]{graphicx}\usepackage[]{color}

\makeatletter
\def\maxwidth{ %
  \ifdim\Gin@nat@width>\linewidth
    \linewidth
  \else
    \Gin@nat@width
  \fi
}
\makeatother

\definecolor{fgcolor}{rgb}{0.345, 0.345, 0.345}

\usepackage{framed}
\makeatletter
 {\par\unskip\endMakeFramed%
 \at@end@of@kframe}
\makeatother

\definecolor{shadecolor}{rgb}{.97, .97, .97}
\definecolor{messagecolor}{rgb}{0, 0, 0}
\definecolor{warningcolor}{rgb}{1, 0, 1}
\definecolor{errorcolor}{rgb}{1, 0, 0}

\usepackage{alltt} 
\usepackage[utf8]{inputenc}
\usepackage{todonotes}
\usepackage{latexsym}
\usepackage{amssymb}
\usepackage{epsfig}
\usepackage{amsmath}
\usepackage{xcolor}
\usepackage{float}
\usepackage{hyperref}
\usepackage{subfig}
\usepackage{multirow}
\usepackage{booktabs}
\usepackage[title]{appendix}
\usepackage{algorithm} 
\usepackage{algpseudocode} 
\makeatletter
\newcommand{\algmargin}{\the\ALG@thistlm}
\makeatother
\newlength{\whilewidth}
\algnewcommand{\parState}[1]{\State%
	\parbox[t]{\dimexpr\linewidth-\algmargin}{\strut #1\strut}}

\IfFileExists{upquote.sty}{\usepackage{upquote}}{}

\usepackage{enumitem}
\usepackage{todonotes}
\RequirePackage[OT1]{fontenc}
\RequirePackage{amsmath,amssymb,amsthm}
\RequirePackage[round]{natbib}

\hypersetup{
  citecolor=black}
	
\RequirePackage[colorlinks,citecolor=blue,urlcolor=blue]{hyperref}
\usepackage{graphicx}
\usepackage{mathtools}

\startlocaldefs
\numberwithin{equation}{section}
\theoremstyle{plain}

\endlocaldefs

\theoremstyle{plain}
\long\def\comment#1{}
\newtheorem{theorem}{Theorem}

\newtheorem{assumptiona}{Assumption A.\!}
\newtheorem{assumptionb}{Assumption B.\!}
\newtheorem{assumptionc}{Assumption C.\!}
\theoremstyle{definition}

\numberwithin{definition}{section}

\newtheorem{remark}{Comment}[section]
\numberwithin{remark}{section}

\DeclareMathOperator{\diag}{diag}

\DeclareMathOperator{\Var}{Var}

\newcommand{\Pb}{\mathbb{P}}

\newcommand{\R}{\mathbb{R}}

\newcommand{\F}{\mathcal{F}}

\renewcommand{\qed}{\hfill{\tiny \ensuremath{\blacksquare} }}%

\newcommand{\Ep}{{\mathrm{E}}}

\newcommand{\g}{g}

\newcommand{\M}{\mathrm{M}}

\newcommand{\bP}{\mathbb{P}}
\newcommand{\E}{\mathbb{E}}

\makeatletter
\newcommand{\setword}[2]{%
  \phantomsection
  #1\def\@currentlabel{\unexpanded{#1}}\label{#2}%
}
\makeatother

\begin{document}

\begin{frontmatter}
\title{Estimation and Uniform Inference in Sparse High-Dimensional Additive Models}
\runtitle{Inference in High-Dimensional Additive Models}
\runauthor{Bach, Klaassen, Kueck, Spindler}
\thankstext{T1}{Version: April, 2024.}

\begin{aug}
\author{\fnms{\textit{Philipp}} \snm{\textit{Bach}}\ead[label=e1]{}},
\author{\fnms{\textit{Sven}} \snm{\textit{Klaassen}}\ead[label=e2]{}},\\
\author{\fnms{\textit{Jannis}} \snm{\textit{Kueck}}\ead[label=e3]{}\thanksref{T2}}
\and
\author{\fnms{\textit{Martin}} \snm{\textit{Spindler}}\ead[label=e4]{}}

\thankstext{T2}{Corresponding author: \textit{kueck@dice.hhu.de}.}


\address{University of Hamburg, Heinrich Heine University Düsseldorf}


%
%
\end{aug}

\begin{abstract}
We develop a novel method to construct uniformly valid confidence bands for a nonparametric component $f_1$ in the sparse additive model $Y=f_1(X_1)+\ldots + f_p(X_p) + \varepsilon$ in a high-dimensional setting. Our method integrates sieve estimation into a high-dimensional Z-estimation framework, facilitating the construction of uniformly valid confidence bands for the target component $f_1$. To form these confidence bands, we employ a multiplier bootstrap procedure. Additionally, we provide rates for the uniform lasso estimation in high dimensions, which may be of independent interest. 
Through simulation studies, we demonstrate that our proposed method delivers reliable results in terms of estimation and coverage, even in small samples.
\end{abstract}


\end{frontmatter}

\input{Introduction.tex}
\input{Notation.tex}

\input{Setting.tex}

\input{Estimation.tex}

\input{Main_Results.tex}
\input{Simulation.tex}

\clearpage
\input{Proofs.tex}
\begin{appendices}
\input{Appendix.tex}

\input{Uniform_lasso.tex}

\input{Computational_Details.tex}
\clearpage

\end{appendices}

\newpage

\footnotesize
\pagebreak
\bibliographystyle{plainnat}
\bibliography{Literatur_GAM_Kopie}

\end{document}

%% file: Introduction.tex
\section{Introduction}
Nonparametric regression offers a way to estimate  the relationship $f$ between a target variable $Y$ and input variables $X=(X_1,\ldots,X_p)^T$ without imposing restrictive functional assumptions. This relationship is expressed as
\[
Y=f(X_1,\ldots,X_p) + \varepsilon,
\]
where $\varepsilon$ denotes the random error term satisfying $\Ep[\varepsilon|X]=0$.  However, in settings with a large number of regressors $p$, especially in cases where  $p$ exceeds the number of observations $n$, the well-known curse of dimensionality makes it practically impossible to estimate the regression function $f(X_1, \ldots, X_p)$. A popular approach to overcome this limitation of nonparametric estimation in practice is to impose an additive structure of the regression function, leading to the following additive models
\begin{align}\label{GAM}
Y= f_1(X_1) + \ldots + f_p(X_p) + \varepsilon,
\end{align}
where $f_j(\cdot)$, $j=1,\ldots,p$, are smooth univariate functions. However, even within this framework, challenges can arise when incorporating a large number of regressors or components, or numerous sieve terms to model the component functions $f_j$. While in many studies, the dimension of the problem is fixed, we follow the recent literature and allow the number of parameters to grow asymptotically with and exceed the sample size. This approach provides a reasonable approximation for various empirical settings involving complex data. 
The concept of additive models has its origins in the works of \citet{friedman1981}, \citet{stone1985} and \citet{hastie1990}. Readers seeking  an introduction to additive models will find useful textbook treatments in \citet{hastie1990} and \citet{wood2017}. A  seminal contribution to the analysis of semi- and nonparametric models, including additive models, was made by \cite{Robinson1988}.  Since then, numerous studies focusing on nonparametric estimation of these models have  been published, with contributions from \cite{LintonNielsen}, \cite{Andrews1990}, \cite{FanJiang}, \cite{Horowitz2004},  \cite{HorowitzMammen2006},  \cite{Mammen2006}, \cite{Horowitz2012} and many more. 

Traditionally, the literature has concentrated on estimation and inference in low-dimensional settings where $p$ is fixed. Recent years, however, have witnessed considerable progress in the understanding and analysis of high-dimensional additive models that allow the number of components to grow with the sample size. For example, the theoretical literature has provided insights into estimation rates in high-dimensional additive models, as in the work of \cite{sardy2004}, \cite{lin2006} and many others \citep{ravikumar2009, meier2009, huang2010, koltchinskii2010, kato2012two, petersen2016, lou2016}. 

The theoretical results in high-dimensional settings rely on a sparsity assumption, which requires that only a small number $s$ of coefficients or components are relevant, meaning they are non-zero. While the number of parameters or covariates exceed the sample size, or the number of parameters can increase along with the sample size, sparsity requires that the number of relevant parameters remains smaller than the sample size. However, it may still increase as the sample size grows, thus introducing additional structure.
While high-dimensional additive models have been studied extensively, there has been limited focus on valid inference within them, particularly in terms of constructing valid hypothesis tests or confidence regions. Approaches to construct confidence bands have been proposed by \cite{haerdle1989}, \cite{sun1994}, \cite{fan2000}, \cite{claeskens2003} and \cite{zhang2010}, but only in the widely studied fixed-dimensional setting. Only recently have new results been derived regarding valid inference in high-dimensional additive models. In the remainder of the introduction section, we will review these recent advancements and highlight our contribution to this emerging literature.\\

The primary aim of our paper is to provide a method for constructing uniformly valid inference and confidence bands in sparse high-dimensional models in the sieve framework. In doing so, we contribute to the growing literature on high-dimensional inference in additive models, especially that on debiased/double machine learning. The double machine learning approach \citep{belloni2014, Chernozhukov2018} offers a general framework for uniformly valid inference in high-dimensional settings. 
Similar methods, such as those proposed by \cite{vandegeer2014} and \cite{zhang2014}, have also produced valid confidence intervals for low-dimensional parameters in high-dimensional linear models. These studies are based on the so-called debiasing approach, which provides an alternative framework for valid inference. The framework entails a one-step correction of the lasso estimator, resulting in an asymptotically normally distributed estimator of the low-dimensional target parameter. For a survey on post-selection inference in high-dimensional settings and its generalizations, we refer to \cite{spindler2015}.\\
In research closely related to ours, \cite{Kozbur2021} proposes a post-nonparametric double selection approach for a scalar functional of a single component.
We consider the same setting as \citet{Kozbur2021}, which encompasses a more general additively separable model
\[
Y= f_{1}(X_1)+ f_{-1}(X_2,\dots,X_p) + \varepsilon,
\]
which includes the additive model
\[
Y= f_1(X_1) + \ldots + f_p(X_p) + \varepsilon.
\]
\citet{Kozbur2021} focuses on inference on functionals of the form $\theta=a(f_1)$ resulting in pointwise confidence intervals based on a penalized series estimator. Although many of these functionals are of interest, \citet{Kozbur2021} approach cannot be easily extended to uniform inference because it considers only differentiable functionals. In contrast, our framework allows us to extend these results and clarify the underlying assumptions. First, building on recent results for  inference on high-dimensional target parameters by \cite{belloni2018uniformly} and \cite{Belloni2014b}, we are able to construct uniformly valid confidence bands for the whole function $f_1$. Second, the post-nonparametric double selection approach in \citet{Kozbur2021} relies on very technical assumptions concerning model selection properties of lasso in Assumption 7 that are difficult to verify (see also the discussion in \cite{lu2019} and Example 2 in \cite{Kozbur2021}). We have much simpler and less technical assumptions since we do not need to rely on model selection properties of lasso. We only make use of the rate of convergence of our lasso estimate to orthogonalize our estimation procedure. In this context, we provide tailored results for uniform lasso estimation that are necessary for establishing valid inference.\\

In a recent study based on the previously mentioned debiasing approach by \cite{zhang2014}, \cite{gregory2016} propose an estimator for the first component $f_1$ in a high-dimensional additive model in which the number of additive components $p$ may increase with the sample size. The estimator is constructed in two steps. The first step entails constructing an undersmoothed estimator based on near-orthogonal projections with a group lasso bias correction. A debiased version of the first step estimator is used to generate pseudo responses $\hat{Y}$. These pseudo responses are then used in the second step, which involves applying a smoothing method to a nonparametric regression problem with $\hat{Y}$ and covariates $X_1$. Under sparsity assumptions concerning the number of nonzero additive components, the so-called oracle property is shown. Accordingly, the proposed estimator in \cite{gregory2016} is asymptotically equivalent to the oracle estimator, which is based on the true functions $f_2,\dots,f_p$. The asymptotic properties of the oracle estimator are well understood and carry over to the proposed debiasing estimate, including the methodology for constructing uniformly valid confidence intervals for $f_1$.
Importantly, \cite{gregory2016} do not explicitly focus on inference, and their analysis requires much stronger assumptions to obtain the oracle property. For example, these assumptions include normally distributed errors independent of $X$, as well as  a bounded support of $X$. Similar to our framework, a large set of basis functions is chosen in their work, such as polynomials or splines, to approximate the components $f_1$ and $f_{-1}$. One distinctive feature of our work compared to \cite{gregory2016} is that we allow the degree of approximating functions to grow to infinity with increasing sample size.\\

A procedure explicitly addressing the construction of uniformly valid confidence bands for the components in high-dimensional additive models has been developed by \cite{lu2019}. The authors emphasize that achieving uniformly valid inference in these models is challenging due to the difficulty of directly generalizing the ideas from the fixed-dimensional case. Whereas confidence bands in the low-dimensional case are mostly built using kernel methods, the estimators for high-dimensional sparse additive models typically rely on sieve estimators based on dictionaries. To derive their results, \cite{lu2019} combine both kernel and sieve methods to draw upon the advantages of each, resulting in a kernel-sieve hybrid estimator.  This is a two-step estimator with tuning parameters for kernel estimation and sieves estimation, such as the bandwidth and penalization levels, which must be chosen by cross-validation. Because of the local structure of the hybrid estimator, the framework 
of  \cite{lu2019} differs from ours in that they consider an additive local approximation model with sparsity (ATLAS), in which they need to impose a local sparsity structure.

The advantage of our proposed estimator is that we do not have to leave the sieves framework while establishing the uniform validity of the resulting confidence bands. Interestingly, \cite{lu2019} conclude that "it is challenging to study the uniform confidence
band through pure sieve-type approaches" (p.~4). However, this is exactly our main contribution: achieving uniform inference within a sieve framework. We accomplish this by framing the problem as one of  high-dimensional Z-estimation, utilizing recent results from \cite{belloni2018uniformly}. Additionally, we offer a theory-driven approach for choosing the penalization level for the lasso estimation steps, eliminating the need for computationally intense cross-validation procedures. Similar to \cite{gregory2016}, \cite{lu2019} assume normally distributed errors that are independent of $X$. In contrast, our model framework allows us to dispense with the normality assumption and only requires the existence of the fourth moment of the error term. Moreover, our main results are also compatible with a heteroskedastic error. Additionally, we can overcome the requirement outlined in \cite{lu2019} that the number of relevant regressors is bounded. Instead, our model allows the number of relevant regressors to grow to infinity with increasing sample size.

\subsection{Organization of the Paper}
The paper is organized as follows. Section \ref{Setting} introduces and motivates the main regression problem in a high-dimensional additive model.  Section \ref{Estimation} presents the estimation method. In Section \ref{main_results}, the main result is provided. Section \ref{Simulation} presents a simulation study, highlighting the small sample properties and implementation of our proposed method. Section \ref{Proofs} provides the proof of the main theorem. The Appendix includes additional technical material. Appendix \ref{AppendixA} presents a general result for uniform inference on a high-dimensional linear functional. Appendix \ref{uniformestimation} provides results in terms of uniform lasso estimation rates in high-dimensions which might be of independent interest. Computational details and additional simulation results are presented in Appendix \ref{compdetails} and Appendix \ref{addsimresults}. 

%% file: Notation.tex
\subsection{Notation}\label{Notation}
Throughout the paper, we consider a random element $W$ taken from a common probability space $(\Omega,\mathcal{A},P)$.
We denote by $P\in \mathcal{P}_n$ a probability measure out of a large class of probability measures, which may vary with the sample size (because the model is allowed to change with $n$). By $\mathbb{P}_n$ we denote the empirical probability measure. Additionally, let $\E$ and $\E_n$ be the expectations with regard to $P$, and $\mathbb{P}_n$, respectively. $\mathbb{G}_n(\cdot)$ denotes the empirical process
\begin{align*}
\mathbb{G}_n(f):=\sqrt{n}\left(\frac{1}{n}\sum\limits_{i=1}^n f(W_i)-\E[f(W_i)]\right)
\end{align*}
for a class of suitably measurable functions $\mathcal{F}:\mathcal{W}\to \mathbb{R}$. $\|\cdot\|_{P,q}$ denotes the $L^q(P)$-norm. 
Furthermore, $\|v\|_1=\sum_{l=1}^p |v_l|$ denotes the $\ell_1$-norm, $\|v\|_2=\sqrt{v^Tv}$ the $\ell_2$-norm and $\|v\|_0$ equals the number of non-zero components of a vector $v\in\mathbb{R}^p$. We define $v_{-l}:=(v_1,\dots,v_{l-1},v_{l+1}\dots,v_p)^T\in\R^{p-1}$ for any $1\le l\le p$. $\|v\|_\infty=\sup_{l=1,\dots,p }|v_l|$ denotes the $\sup$-norm. Let $c$ and $C$ denote positive constants that are independent of $n$ with values that may change at each appearance. The notation $a_n\lesssim b_n$ means $a_n\le Cb_n$ for all $n$ and some $C$. Furthermore, $a_n=o(1)$ denotes that there exists a sequence $(b_n)_{\ge 1}$ of positive numbers such that $|a_n|\le b_n$ for all $n$, where $b_n$ is independent of $P\in\mathcal{P}_n$ for all $n$, and $b_n$ converges to zero. Finally, $a_n=O_P(b_n)$ means that for any $\epsilon>0$ there exists a $C$ such that $P(a_n>Cb_n)\le\epsilon$ for all $n$.  

%% file: Setting.tex
\section{Setting}\label{Setting}
\subsection{Motivation and Illustration}\label{motivation}
Before introducing the formal setting in Section \ref{formalsetting}, we would like to motivate the basic ideas with a simplified example. We consider an additive model with two covariables
\begin{align}
Y = f_1(X_1) + f_2(X_2) + \varepsilon.
\label{2comp}
\end{align}
Our goal is to perform valid inference on the object of interest $f_1$. In other words, we would like to provide a uniform confidence band for $f_1$ as illustrated in Figure \ref{illustrationexample}. Hence, we treat $f_2$ in \eqref{2comp} as a nuisance function.
Next, we assume that it is possible to represent the two components through an approximately linear representation. For the first component, the representation is given by
$$f_1(X_1) = \theta_0^T g(X_1) + b_1(X_1).$$
Here, $g(X_1)$ is a basis (e.g., a spline basis, sieve terms or a polynomial series) consisting of $d_1$ terms, $\theta_0$ is the corresponding coefficient vector and $b_1(X_1)$ is an approximation error.
For the second component, we assume an analogous representation 
\begin{align}
\label{f2representation}
f_2(X_2)= \beta_0^T h(X_2) + b_2(X_2),
\end{align}
where the basis $h(X_2)$ consists of $d_2$ terms. 
We allow the dimensions $d_1$ and $d_2$ to grow with the sample size in order to establish the asymptotic results in high-dimensional settings. Accordingly, the number of components $p$ can also grow with the sample size.\footnote{However, for the ease of notation we will later subsume them in the component $f_{-1}(X_2,\dots,X_p)$.}
To derive a uniformly valid confidence band, estimation and inference of 
\begin{align}
f_1(\cdot) \approx \theta_0^T g(\cdot)=\sum\limits_{l=1}^{d_1}\theta_{0,l}g_l(\cdot) 
\label{estim}
\end{align}
is required. In a naive approach to estimate the first component $\theta_{0,1}$ of the vector $\theta_0$, one might employ lasso or other machine learning methods to select the relevant components of the basis expansion terms in $\theta_0$ and $\beta_0$ in the regression model (\ref{2comp}). A possible second step could involve estimating a regression of the dependent variable on all components that have been selected in the lasso estimation step. The final estimator for the first component $\theta_{0,1}$ could then be obtained from this regression, and this procedure could be repeated iteratively for all other components in $f_1(x_1)$. 
\begin{figure}[ht]
\begin{center}
\includegraphics[scale=0.35]{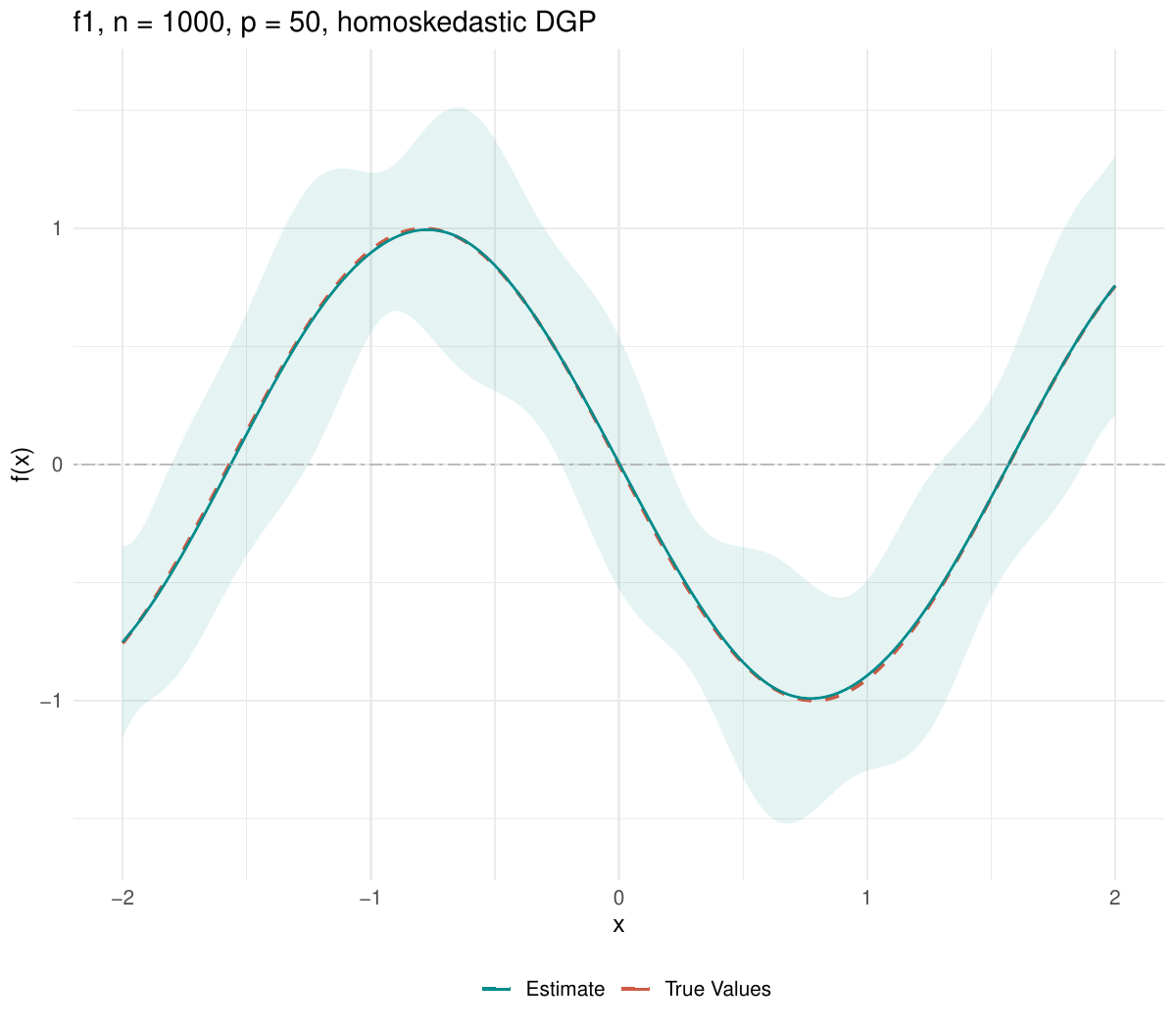} 
\caption{Illustration of the estimator and confidence bands in a simulated example.}
\label{illustrationexample}
\end{center}
Predicted component $\hat{f}_1(x_1)$ with $95\%$-confidence bands (green shaded area) obtained by our proposed estimator for a component $f_1(x_1)$ with $f_1(x_1)= - \sin(2 \cdot x)$. The dashed red curve corresponds to the true function $f_1(x_1)$. The predicted component $\hat{f}_1(x_1)$, which is obtained from our proposed estimator, is illustrated by the green curve. The underlying data are generated according to the homoskedastic DGP in the simulation study in Section \ref{Simulation} with $n=1000$ observations, $p=50$ regressors and $\rho = 0$. For more information on the DGP, we refer to the description of the DGP in Section \ref{Simulation}.
\end{figure}
\noindent
However, this approach is problematic because it generally fails to yield valid results. In other words, using the lasso estimator necessitates accounting for variable selection, leading to the non-standard problem of post-selection inference. Intuitively, the naive procedure based on a single lasso estimation step could result in selection mistakes that invalidate the resulting confidence interval for $\theta_{0,1}$. These critical selection mistakes will not occur for variables with high predictive power for the dependent variable $Y$, but rather refer to variables that are potentially highly correlated with the basis expansion term $g_1()$. This could lead to an omitted variable bias, preventing the estimator from asymptotically converging to a normal distribution. To overcome this limitation of post-selection inference, the statistical literature has developed the double machine learning and the debiasing approach, as  mentioned in the introduction.
To address the potential bias introduced by lasso estimation, \cite{belloni2014} propose including an auxiliary regression for the corresponding covariate of the target parameter. Here, we consider
\begin{align}
g_1(X_1) = \gamma_0^TZ_{-1} + \nu
\label{auxreg}
\end{align}
where $\nu$ is an error term and $Z_{-1}$ is defined as
$$Z_{-1}=(g_2(X_1),\dots,g_{d_1}(X_1),h_1(X_{-1}),\dots,h_{d_2}(X_{-1}))^T.$$
Later, we will also allow for an approximation error in this equation. \cite{belloni2014} propose including in the final regression not only the covariates selected in the first step of the naive approach but also augmenting this set of variables with Lasso-selected regressors from the auxiliary regression. This procedure is equivalent to constructing a so-called Neyman orthogonal moment function with respect to the nuisance part. This is essential for ensuring valid post-selection inference for the first component of the vector $\theta_0$. In Section \ref{formalsetting}, we will provide more details about this property. Heuristically, the additional regression step in Equation \eqref{auxreg} will lead to robustness against moderate selection mistakes. 
It can be shown formally, that this procedure implements an orthogonal moment equation
$$\E\left[\psi_1(W,\theta_{0,1},\eta_{0,1})\right]=0,$$ 
where the first component of $\theta_0$ is our target parameter and all other involved parameters are considered nuisance parameters, denoted $\eta_{0,1}$.
\cite{belloni2014} developed an approach for valid inference for one parameter. In high-dimensional additive models, the major technical challenge arises from the need to conduct inference for the potentially high-dimensional vector $\theta_0$. In other words, the number of elements in $\theta_0$ for which we would like to construct a valid confidence region is allowed to grow with the sample size. Each component of $\theta_0$, $\theta_{0,l}$ with $\l=1, \ldots, d_1$, is determined by an orthogonal moment condition and we will demonstrate how uniformly valid confidence bands can be constructed by embedding into a high-dimensional Z-estimation framework. Finally, we illustrate how the estimation of $\theta_0$ can be translated into uniformly valid confidence bands for the target function $f_1\approx \theta_0^T g(\cdot)$ using a multiplier bootstrap procedure.\\ 
\noindent
Figure \ref{illustrationexample} illustrates the use of our estimation procedure, which is described in Section \ref{Estimation}, by providing a preview of our simulation studies in Section \ref{Simulation}. Our estimation methods yield an unbiased estimate for the target component $f_1$, as well as corresponding $(1-\alpha)$ confidence interval that is valid over a compact interval of $X_1$. In the next section, we consider a more general additively separable model and introduce a general formulation of the underlying problem.

\subsection{Formal Setting} \label{formalsetting}
Consider the following nonparametric additively separable model
$$Y=f(X)+\varepsilon=f_1(X_1)+f_{-1}(X_{-1}) + \varepsilon$$
with $\E[\varepsilon|X]=0$ and $c\le \Var(\varepsilon|X)\le C$. Let the scalar response $Y$ and features $X=(X_1,\dots,X_p)$ take values in $\mathcal{Y}$ and in $\mathcal{X}=\mathcal{X}_1\times\dots\times\mathcal{X}_p$, respectively. We assume the observation of $n$ i.i.d. copies $(W^{(i)})_{i=1}^n= (Y^{(i)},X^{(i)})_{i=1}^n$ of $W=(Y,X)$, with the number of covariates $p$ allowed to grow with the sample size $n$. To ensure identifiability, we assume $\E[f_{1}(X_{1})]=0$. Our aim is to construct uniformly valid confidence regions for the first nonparametric component of the regression function. In other words, we want to find functions $\hat{l}(x)$ and $\hat{u}(x)$ fulfill
\begin{align*}
P\left(\hat{l}(x)\le f_1(x)\le \hat{u}(x), \forall x\in I\right)\to 1-\alpha.
\end{align*}
Here, $I\subseteq\mathcal{X}_1$ is a bounded interval of interest in which we want to conduct inference. We approximate $f_1$ and $f_{-1}$ using linear combinations of approximating functions $g_1,\dots,g_{d_1}$ and $h_1,\dots,h_{d_2}$, respectively. Define
$$g(x_1):=(g_1(x_1),\dots,g_{d_1}(x_1))^T$$
for $x_1\in\mathbb{R}^{}$ and
$$h(x_{-1}):=(h_1(x_{-1}),\dots,h_{d_2}(x_{-1}))^T$$
for $x_{-1}\in\mathbb{R}^{p-1}$.
It is important to note that we allow the number of approximating functions $d_1$ and $d_2$ to increase with sample size.
Assume that the approximations are given by
\begin{align}\label{linmod1}
f_1(X_1)= \theta_0^T g(X_1) + b_1(X_1),
\end{align}
where $\theta_{0,l}\in\Theta_l$ and analogously
\begin{align}\label{linmod2}
f_{-1}(X_{-1}):= \beta_0^T h(X_{-1}) + b_2(X_{-1}),
\end{align}
where $b_1$ and $b_2$ denote the error terms. Additionally, it is convenient to define the union of approximation functions
$$z(x):=(g_1(x_1),\dots,g_{d_1}(x_1),h_1(x_{-1}),\dots,h_{d_2}(x_{-1}))^T$$
for $x \in\mathbb{R}^p$, where we abbreviate $Z:=z(X)$.
As a common assumption in high-dimensional statistics, we will impose sparsity in the linear approximation of $f$, meaning that only a relatively small number of the coefficients in $(\theta_0,\beta_0)$, denoted $s_1$, are non-zero which implies that only a subset of approximation functions in $Z$ is needed to approximate $f_1$ and $f_{-1}$ sufficiently well. Of course, the identity of these relevant approximation functions is not known. Sparsity can occur because of two different reasons: Either we have many irrelevant covariables $X_2,\ldots,X_p$ in our model or we have sparse representations of the regression functions.
For each element $g_l$ of $g$, we also consider
\begin{align}\label{second_stage}
g_l(X_1)=(\gamma^{(l)}_0)^TZ_{-l}+\nu^{(l)} =(\gamma^{(l,1)}_0)^TZ_{-l}+ (\gamma^{(l,2)}_0)^TZ_{-l} +\nu^{(l)}
\end{align}
with 
$\E[\nu^{(l)}Z_{-l}] = 0$, such that $(\gamma^{(l)}_0)^TZ_{-l}$ is the linear projection of $g_l(X_1)$ onto $Z_{-l}$. Here, $\gamma^{(l,1)}_0$ denotes the sparse part of the coefficient vector $\gamma^{(l)}_0$.
The auxiliary regression (\ref{second_stage}) is used to construct an orthogonal score function for valid inference in a high-dimensional setting, as described in Section \ref{motivation}.
Estimating $$f_1(\cdot)\approx \theta_0^Tg(\cdot)$$
can be recast into a general Z-estimation problem of the form
$$\E\left[\psi_l(W,\theta_{0,l},\eta_{0,l})\right]=0,\quad l\in1,\dots, d_1$$
with target parameter $\theta_0$ where the score functions are defined by
\begin{align*}
\psi_l(W,\theta,\eta)=&\left(Y-\theta g_l(X_1)-(\eta^{(1)})^TZ_{-l}-\eta^{(3)}(X)\right)\left(g_l(X_1)-(\eta^{(2)})^TZ_{-l}\right).
\end{align*}
Here,  
$\eta=(\eta^{(1)},\eta^{(2)},\eta^{(3)})^T$
with
$\eta^{(1)}\in \mathbb{R}^{d_1+d_2-1},\eta^{(2)}\in \mathbb{R}^{d_1+d_2-1}$ and $\eta^{(3)}\in \ell^{\infty}(\mathbb{R}^p)$ are nuisance functions.
The true nuisance parameter $\eta_{0,l}$ is given by
\begin{align*}
\eta^{(1)}_{0,l}&:=\beta_0^{(l)}\\
\eta^{(2)}_{0,l}&:=\gamma_0^{(l)}\\
\eta^{(3)}_{0,l}(X)&:=b_1(X_1) + b_2(X_{-1}),
\end{align*}
where $\beta_0^{(l)}$ is defined as 
$$\beta_0^{(l)}:=(\theta_{0,1},\dots,\theta_{0,l-1},\theta_{0,l+1},\dots\theta_{0,d_1},\beta_{0,1},\dots,\beta_{0,d_2})^T.$$
Essentially, the index $l$ determines which coefficient is not contained in $\beta_0^{(l)}$. Formally, we will assume sparsity of $\beta_0^{(l)}$, i.e. $\|\beta_0^{(l)}\|_0\le s_1$ and we also assume sparsity in the auxiliary regression, i.e. $\|\gamma_0^{(l,1)}\|_0\le s_2$  for all $l=1,\ldots,d_1$. These assumptions are discussed in detail in Comment \ref{splines} and Comment \ref{commentApprox}.
The third part of the nuisance functions captures the error made by the approximation of $f_1$ and $f_{-1}$, which is independent from $l$. Therefore, we sometimes omit $l$. 
\begin{remark}
The score $\psi$ is linear in $\theta$, meaning
\begin{align*}
\psi_l(W,\theta,\eta)=\psi^{a}_l(X,\eta^{(2)})\theta+\psi^b_l(X,\eta)
\end{align*}
with
$$\psi^{a}_l(X,\eta^{(2)})=-g_l(X_1)(g_l(X_1)-(\eta^{(2)})^TZ_{-l})$$
and 
$$\psi^b_l(W,\eta)=(Y-(\eta^{(1)})^T Z_{-l}-\eta^{(3)}(X))(g_l(X_1)-(\eta^{(2)})^TZ_{-l})$$
for all $l=1,\dots,d_1$.
\end{remark}
\begin{remark}
The score function $\psi$ satisfies the \textit{moment condition}, namely
\begin{align*}
\E\left[\psi_l(W,\theta_{0,l},\eta_{0,l})\right]&=0
\end{align*}
for all $l=1,\dots,d_1$, and, given further conditions mentioned in Section \ref{main_results}, the near \textit{Neyman orthogonality} condition
\begin{align*}
D_{l,0}[\eta,\eta_{0,l}]&:=\partial_t\big\{\E[\psi_l(W,\theta_{0,l},\eta_{0,l}+t(\eta-\eta_{0,l}))]\big\}\big|_{t=0}\approx 0,
\end{align*}
where $\partial_t$ denotes the derivative with respect to $t$.
\end{remark}

\begin{remark}
Mostly for ease of notation, we consider the nonparametric additively separable model $Y=f(X)+\varepsilon=f_1(X_1)+f_{-1}(X_{-1}) + \varepsilon$. In the statistical literature the most studied special case is the fully additive separable model
$Y=f(X)+\varepsilon=f_1(X_1)+\ldots + f_{p}(X_{p}) + \varepsilon$. This model has more structure that can be exploited, but what matters from a theoretical perspective is the number of parameters $d_2$, respectively the sparsity index $s_1$, for the nuisance parts $f_{-1}(X_{-1})$ or $f_2(X_2)+\ldots + f_{p}(X_{p})$, which can grow with the sample size in our asymptotic framework.
\end{remark}

%% file: Estimation.tex
\section{Estimation}\label{Estimation}
In this section, we describe our estimation method and how the uniform valid confidence bands are constructed. First, nuisance functions are estimated using lasso regressions. Subsequently, they are plugged into the moment conditions and solved for the target parameters, which yield an estimate $\hat{f}_1$ for the first component in the additive regression model. The lower and upper curve of the confidence bands are then determined using the estimated covariance matrix and a critical value, which is obtained through a multiplier bootstrap procedure. The technical details for the estimation are given in this section.\\

\subsection{Estimation and Algorithm}

\noindent
Let
$
g(x)=\left(g_1(x),\dots,g_{d_1}(x)\right)^T\in\R^{d_1\times 1},
$
and
\begin{align*}
\psi(W,\theta,\eta)=\left(\psi_1(W,\theta_1,\eta_1),\dots,\psi_{d_1}(W,\theta_{d_1},\eta_{d_1})\right)^T\in\R^{d_1\times 1}
\end{align*}
for some vectors
$\theta=(\theta_1,\dots,\theta_{d_1})^T$
and $\eta=(\eta_1,\dots,\eta_{d_1})^T.$
For each $l=1,\dots,d_1$, let $\hat{\eta}_l=\left(\hat{\eta}_l^{(1)},\hat{\eta}_l^{(2)},\hat{\eta}_l^{(3)}\right)$ be an (lasso) estimator of the nuisance function. The estimator $\hat{\theta}_0$ of the target parameter
$\theta_0=(\theta_{0,1},\dots,\theta_{0,d_1})^T$
is defined as the solution of
\begin{align}\label{estimator}
\sup\limits_{l=1,\dots,d_1}\left\{\left|\mathbb{E}_n^{}\Big[\psi_{l}\big(W,\hat{\theta}_{l},\hat{\eta}_{l}\big)\Big]\right|-\inf_{\theta\in\Theta_{l}}\left|\mathbb{E}_n^{}\Big[\psi_{l}\big(W,\theta,\hat{\eta}_{l}\big)\Big]\right|\right\}\le\epsilon_{n},
\end{align}
where $\epsilon_{n}=O\left(\delta_n\varsigma_n^{-1/2}n^{-1/2}\right)$ is the numerical tolerance. The sequences of positive constants, $\delta_n$ and $\varsigma_n$, are defined in Assumption A.\ref{A1} and Assumption A.\ref{A2}. Since the score is linear, the explicit solution is given by
$$\hat{\theta}_{0,l}=-\frac{\E_n[\psi^b_l(W,\hat{\eta}_l)]}{\E_n[\psi^{a}_l(W,\hat{\eta}^{(2)}_l)]}.$$
Therefore, \eqref{estimator} obviously holds if we use the estimator $\hat{\theta}_{0,l}$. We consider the representation in \eqref{estimator} because we rely on more general results derived in Appendix \ref{AppendixA} to analyze the theoretical properties of our new estimation procedure.
Lastly, the target function $f_1(\cdot)$ can be estimated by
\begin{align}
\hat{f}_1(\cdot):=\hat{\theta}_0^T g(\cdot).
\end{align}
Define the Jacobian matrix
\begin{align*}
J_0:=\frac{\partial}{\partial\theta}\E[\psi(W,\theta,\eta_0)]\bigg|_{\theta=\theta_0}=\diag\left(J_{0,1},\dots,J_{0,d_1}\right)\in \mathbb{R}^{d_1\times d_1}
\end{align*}
with
\begin{align*}
J_{0,l}&=\E\left[\psi_l^{a}(W,\eta^{(2)}_{0,l})\right]=-\E[(\nu^{(l)})^2]
\end{align*}
for all $l=1,\dots,d_1$. Observe that $\Sigma_{\varepsilon\nu}:=\E\big[\psi(W,\theta_0,\eta_0)\psi(W,\theta_0,\eta_0)^T\big]$ is the covariance matrix of $\varepsilon\nu:=(\varepsilon\nu^{(1)},\dots,\varepsilon\nu^{(d_1)})$. Define the approximate covariance matrix
\begin{align*}
\Sigma_n:&=J_0^{-1}\E\big[\psi(W,\theta_0,\eta_0)\psi(W,\theta_0,\eta_0)^T\big](J_0^{-1})^T\\
&=J_0^{-1}\Sigma_{\varepsilon\nu}(J_0^{-1})^T \in \mathbb{R}^{d_1\times d_1}
\end{align*}
with
\begin{align*}
\Sigma_n:&=\begin{pmatrix}
\frac{\E[(\varepsilon\nu^{(1)})^2]}{\E[(\nu^{(1)})^2]^2} & \frac{\E\big[\varepsilon\nu^{(1)}\varepsilon\nu^{(2)}\big]}{\E[(\nu^{(1)})^2]\E[(\nu^{(2)})^2]} & \dots & \frac{\E\big[\varepsilon\nu^{(1)}\varepsilon\nu^{(d_1)}\big]}{\E[(\nu^{(1)})^2]\E[(\nu^{(d_1)})^2]} \\
\frac{\E\big[\varepsilon\nu^{(2)}\varepsilon\nu^{(1)}\big]}{\E[(\nu^{(2)})^2]\E[(\nu^{(1)})^2]} & \frac{\E[(\varepsilon\nu^{(2)})^2]}{\E[(\nu^{(2)})^2]^2}  & \dots & \frac{\E\big[\varepsilon\nu^{(2)}\varepsilon\nu^{(d_1)}\big]}{\E[(\nu^{(2)})^2]\E[(\nu^{(d_1)})^2]} \\
\vdots & \vdots & \ddots & \vdots \\
 \frac{\E\big[\varepsilon\nu^{(d_1)}\varepsilon\nu^{(1)}\big]}{\E[(\nu^{(d_1)})^2]\E[(\nu^{(1)})^2]} & \frac{\E\big[\varepsilon\nu^{(d_1)}\varepsilon\nu^{(2)}\big]}{\E[(\nu^{(d_1)})^2]\E[(\nu^{(2)})^2]} & \dots & \frac{\E[(\varepsilon\nu^{(d_1)})^2]}{\E[(\nu^{(d_1)})^2]^2}
\end{pmatrix}.
\end{align*}
The approximate covariance matrix can be estimated by replacing every expectation with its empirical analog and plugging in the estimated parameters
\begin{align*}
\hat{\Sigma}_n:&=\hat{J}^{-1}\E_n\big[\psi(W,\hat{\theta}_0,\hat{\eta})\psi(W,\hat{\theta}_0,\hat{\eta})^T\big](\hat{J}^{-1})^T
=\hat{J}^{-1}\hat{\Sigma}_{\varepsilon\nu}(\hat{J}^{-1})^T.
\end{align*}
This estimated covariance matrix can be used to construct the confidence bands
\begin{align*}
\hat{u}(x)&:= \hat{f}_1(x)+\frac{(g(x)^T\hat{\Sigma}_n g(x))^{1/2}c_\alpha}{\sqrt{n}}\\
\hat{l}(x)&:= \hat{f}_1(x)-\frac{(g(x)^T\hat{\Sigma}_n g(x))^{1/2}c_\alpha}{\sqrt{n}},
\end{align*}
where $c_{\alpha}$ is a critical value determined by the following standard multiplier bootstrap method introduced in \cite{chernozhukov2013gaussian}. Define
$$\hat{\psi}_x(\cdot):=(g(x)^T\hat{\Sigma}_n g(x))^{-1/2}g(x)^T\hat{J}^{-1}\psi(\cdot,\hat{\theta}_{0},\hat{\eta}_{0})$$
and let
\begin{align*}
\hat{\mathcal{G}}=\left(\hat{\mathcal{G}}_x\right)_{x\in I}=\left(\frac{1}{\sqrt{n}}\sum\limits_{i=1}^n \xi^{(i)}\hat{\psi}_x\left(W^{(i)}\right) \right)_{x\in I},
\end{align*}
where $(\xi^{(i)})_{i=1}^n$ are independent standard normal random variables (especially independent from the data $(W^{(i)})_{i=1}^n$). The multiplier bootstrap critical value $c_\alpha$ is given by the $(1-\alpha)$-quantile of the conditional distribution of $\sup_{x\in I}|\hat{\mathcal{G}}_x|$ given $(W^{(i)})_{i=1}^n$. This estimation procedure can be summarized in Algorithm \ref{algest}.
\begin{algorithm}[H]
	\caption{HDAM} 
	\label{algest}
	\flushleft{
	Input: $n$ training examples of the form $W^{(i)}=(Y^{(i)},X_1^{(i)},X_{-1}^{(i)})$, where $Y^{(i)}$ is the response, $X_1^{(i)}$ the covariate of interest and $X_{-1}^{(i)}$ are additional covariates. Dictionaries of the approximating functions $g_1,\dots,g_{d_1}$ for $f_1$ and $h_1,\dots,h_{d_2}$ for $f_{-1}$, a significance level $\alpha$, an interval $I$ for inference and a number of bootstrap repetitions $B$.}
	\begin{algorithmic}[1]
	\State Use the dictionary to construct the matrix $Z:=(g_1(X_1),\dots,g_{d_1}(X_1),h_1(X_{-1}),\dots,h_{d_2}(X_{-1}))$.
	\State Fit a lasso/post-lasso regression of the vector $Y$ onto $Z$ and save the estimated coefficients $(\tilde{\theta}_{1},\dots,\tilde{\theta}_{d_1},\hat{\beta}_{1},\dots,\hat{\beta}_{d_2})$ and the corresponding residuals $\hat{\varepsilon}$.
		\For {$l = 1,\dots,d_1$}
		\parState{Fit a lasso/post-lasso regression of the vector $g_l(X_1)$ onto $Z_{-l}$ and save the estimated coefficients $(\hat{\gamma}^{(l)}_{1},\dots\hat{\gamma}^{(l)}_{d_1+d_2-1})$ and the corresponding residuals $\hat{\nu}^{(l)}$.}
		\parState{Plug in the estimated coefficients as nuisance parameters into the score function $\psi_l(W,\cdot,\hat{\eta}_l) ) $ to solve \eqref{estimator}. Save the resulting estimate $\hat{\theta}_{l}$ and scores $\psi_l(W,\hat{\theta}_l,\hat{\eta}_l)$ into the corresponding vector $\hat{\theta}_0$ and matrix $\psi(W,\hat{\theta}_0,\hat{\eta}_0)$, respectively.}
		\EndFor
		\parState{Use the estimated residuals $\hat{\varepsilon}$ and $\hat{\nu}$ to construct the estimates $\hat{\Sigma}_n$ and $\hat{J}$.}	
		\For {$x\in I$}
		\State Calculate the vector $\hat{\psi}_x(W):=(g(x)^T\hat{\Sigma}_n g(x))^{-1/2}g(x)^T\hat{J}^{-1}\psi(W,\hat{\theta}_{0},\hat{\eta}_{0})$	
		\For {$b = 1,\dots, B$}
			\State{Draw $(\xi_i^{(b)})_{i=1}^n$ independent standard normal random variables.}
			\State{Calculate $\hat{\mathcal{G}}_x^{(b)}=\frac{1}{\sqrt{n}}\sum_{i=1}^n \xi_i^{(b)}\hat{\psi}_x(W^{(i)}).$}
		\EndFor
		\EndFor
		\State Calculate the critical value $c_\alpha:=$ $(1-\alpha)$-quantile of $\sup_{x\in I}|\hat{\mathcal{G}}_x^{(b)}|$ with respect to the bootstrap repetitions.
		\For {$x\in I$}
		\State Construct the confidence band as
	 $\hat{\theta}_0^T g(x) \pm \frac{(g(x)^T\hat{\Sigma}_n g(x))^{1/2}c_\alpha}{\sqrt{n}}$
		\EndFor
	\end{algorithmic} 
\end{algorithm}

\subsection{Practical Considerations: Estimation of Sparse High-Dimensional Additive Models}

In addition to the significance level $\alpha$, the number of bootstrap repetitions and an interval $I$ for inference, it is necessary to specify a dictionary of approximating functions $g_1, \ldots, g_{d_1}$ for $f_1$  and $h_1, \ldots, h_{d_2}$ for $f_{-1}$ as inputs to Algorithm \ref{algest}. In our theoretical results and our simulation experiments, we focus on B-splines. From a practical point of view, B-splines are specified by parameters for a polynomial degree and the number of knots. In our simulation study, we use cubic B-splines. The number of knots plays a role for the quality of the approximation for $f_1$ as well as the width of the confidence bands. For example, a higher number of knots leads to a higher number of basis functions and corresponding number of fitted coefficients. In our simulation study, we observed that our results are relatively robust with regard to the number of knots and only change moderately with changes of the specification. We base the choice of the number of knots on preliminary evidence in terms of the approximation quality for $f_1$. We also investigated the use of a cross-validated choice of the number of knots, which would not be backed by our theoretical framework, but could serve as a practical guidance. Our results, which are available on request, indicate that a cross-validated choice of the B-spline specification tends to lead to conservative results, i.e., wide confidence intervals with higher coverage than the nominal level.

Based on the provided input, the algorithm proceeds with various nuisance estimation steps. Due to the large number of coefficients, we base estimation on the lasso or post-lasso, i.e., a lasso estimator that is followed by an ordinary least squares estimation step \citep{belloni:2013}. The performance of regularized estimators depends on the choice of the corresponding penalty parameters. We base the choice of the lasso penalty on theoretical arguments as provided in Appendix \ref{uniformla}. In our software implementation, we use the rigorous lasso learner as provided by the R package \texttt{hdm} \citep{hdm}. We also experimented with a cross-validated choice of the lasso penalty as provided by the \texttt{glmnet} package \citet{glmnet:2010}. However, we find that the cross-validated penalty choice was more costly under computational consideration and also less stable. In our simulation experiments, we observe that the estimation performance of the lasso learner has important implications on the accuracy of our point estimator and, accordingly, on the coverage results of the confidence bands. High-dimensional additive models are characterized by a large number of parameters, which have to be estimated with a possibly small number of observations. In our simulation results in Section \ref{Simulation}, we find that the theory-based lasso learner exhibits a good performance in selecting the sparse parameters even in highly correlated settings.

Once, the lasso estimation has been performed, the corresponding residuals are plugged into the variance-covariance matrix. This, in turn, is used to construct the simultaneous confidence bands via a multiplier bootstrap procedure. The latter is based on a random perturbation of the score function, for example, by an i.i.d. standard normal distributed random variable. This procedure is very appealing from a computational point of view as it does not require resampling and reestimation of the parameters, as for example in classical bootstrap procedures. It is generally recommended to use a large number of bootstrap repetitions, $B\ge 500$.

%% file: Main_Results.tex
\section{Main Results}\label{main_results}
Now, we proceed to specifying the conditions required to create uniformly valid confidence bands using Algorithm \ref{algest}. To represent $f_1$ and $f_{-1}$ by their approximations in (\ref{linmod1}) and (\ref{linmod2}), we need to choose an appropriate set of approximating functions $g=(g_1,\dots,g_{d_1})$ and $h=(h_1,\dots,h_{d_2})$, respectively. In this context, let $\bar{d}_n:=\max(d_1,d_2,n,e)$ and $C$ be a strictly positive constant independent of $n$ and $l$, where $e$ in $\bar{d}_n$ denotes the Euler's number. $(\tilde{A}_n)_{n \geq 1}$ denotes a sequence of positive constants, possibly growing to infinity with $\tilde{A}_n \geq n$ for all $n$. For more details, we refer to Appendix \ref{AppendixA}. The number of non-zero coefficients in Equation (\ref{linmod2}) and (\ref{second_stage}), is given by the sparsity index $s=\max(s_1,s_2)$. Additionally, we set $t_1:=\sup_{x\in I} \|g(x)\|_0\le d_1$. The definition of $t_1$ is helpful if the functions $g_l$, $l=1,\dots,d_1$, are local, meaning that for any point $x$ in $I$ there are at most $t_1<<d_1$ non-zero functions. Furthermore, $g(I)$ denotes the image of the approximation functions concerning the interval of interest $I$.\\ \\
The following assumptions are valid uniformly in $n\ge n_0$ and $P\in\mathcal{P}_n$:
\begin{assumptiona}\label{A1}\ \\
\begin{itemize}
\item[(i)] It holds $\frac{1}{\sqrt{t_1}}\lesssim \inf\limits_{x\in I} \|g(x)\|_2\le C<\infty$, $\quad\sup_{x\in I} \sup_{l=1,\dots,d_1}|g_l(x)|\le C$ and for all $\varepsilon>0$
\begin{align*}
\log N(\varepsilon,g(I),\|\cdot\|_{2})\le Ct_1\log\left(\frac{\tilde{A}_n}{\varepsilon}\right)
\end{align*}
where $N (\epsilon, g(I), || \cdot ||_{2})$ denotes the covering number of $g(I)$ with radius $\epsilon$ with respect to the $|| \cdot ||_{2}$-norm.
\item[(ii)]  
The approximation errors obey
\begin{align*}
\left(b_1(X_1)+b_2(X_{-1})\right)^2\le Cs_1\log(\bar{d}_n)/n\ \text{(a.s.)}.
\end{align*}
\item[(iii)] It holds $$\E\Big[\nu^{(l)}\big(b_1(X_1)+b_2(X_{-1})\big)\Big]\le C\delta_n \varsigma_n^{-1/2}n^{-1/2}$$
with $\varsigma_n$ defined in Assumption A.\ref{A2} (iv) and $\delta_n=\sqrt{\frac{n^{2/q}s^2\varsigma_n\log^2(\bar{d}_n)}{n}}=o(1)$ where $q$ is defined in Assumption A.\ref{A2} (iii) . 
\end{itemize}
\end{assumptiona}
\noindent
Assumption A.\ref{A1}$(i)$ contains regularity conditions on $g$. We assume that the supremem of the $\ell_2$-norm of $g(x)$ is bounded, but the infimum is allowed to decrease with sample size, affecting the growth conditions in A.\ref{A2}$(v)$. Assumption A.\ref{A1}$(ii)$  is a condition on the approximation error and is discussed further in Comment \ref{splines}. This assumption is mild because the number of approximating functions may increase with sample size. Lastly, Assumption A.\ref{A1}$(iii)$ ensures that any violation of the  exact Neyman Orthogonality due to approximation errors is negligible. It is worth noting that if $b_1(X_{1})$ and $b_2(X_{-1})$ are measurable with respect to $Z_{-l}$ (for example, in a linear approximate sparse setting for the conditional expectation) the exact Neyman Orthogonality holds. 
\begin{assumptiona}\label{A2}\ \\
\begin{itemize}
\item[(i)] For all $l=1,\dots,d_1$, $\Theta_l$ contains a ball of radius $$\log(\log(n))n^{-1/2}\log^{1/2}(d_1\vee e)\varsigma_n\log(n)$$ centered at $\theta_{0,l}$ with
$$\sup\limits_{l=1,\dots,d_1}\sup\limits_{\theta_l\in\Theta_l} |\theta_l|\le C.$$
\item[(ii)] It holds
\begin{align*}
\|\beta_0^{(l)}\|_0\le s_1,\quad \|\beta_0^{(l)}\|_2\le C
\end{align*}
for all $l=1,\dots,d_1$ and 
\begin{align*}
\max\limits_{l=1,\dots,d_1}\|\gamma_0^{(l)}\|_1\le C
\end{align*}
with
\begin{align*}
    \max\limits_{l=1,\dots,d_1}\|\gamma_0^{(l,1)}\|_0\le s_2
\end{align*}
and 
\begin{align*}
\max\limits_{l=1,\dots,d_1}\|\gamma_0^{(l,2)}\|_{1}^2&\le C\sqrt{\varsigma_n^{-1/2}s_2^2\log(\bar{d}_n)/n},\\ \max\limits_{l=1,\dots,d_1}\|(\gamma_0^{(l,2)})^TZ_{-l}\|_{P,2}^2&\le C\varsigma_n^{-1}s_2\log(\bar{d}_n)/n.
\end{align*}
\item[(iii)] The constructed approximation functions are almost surely bounded,
$$\|Z\|_{\infty}\le C\quad \text{(a.s.)},$$
and 
$$\|\varepsilon\|_{P,q}\le C$$
for $q\ge 4$ corresponding to the growth rates in $(v)$. 


\item[(iv)] It holds
$$c \varsigma_n^{-1} \le \inf\limits_{\|\xi\|_2=1} \E[(\xi^T Z)^2] \le \sup\limits_{\|\xi\|_2=1} \E[(\xi^T Z)^2] \le C \varsigma_n^{-1}$$
and 
$$ c \varsigma_n^{-1}\le \inf\limits_{\|\xi\|_2=1} \xi^T\Sigma_\nu\xi\le \sup\limits_{\|\xi\|_2=1} \xi^T\Sigma_\nu\xi \le C \varsigma_n^{-1}.$$
Further,  
$$c\varsigma_n^{-2}\le\E\left[(\nu^{(l)})^2 Z_{-l,j}^2\right]\le C\varsigma_n^{-2},$$
uniformly for all  $l=1,\dots,d_1$, $j = 1,\dots, (d_1 -1)$ and
\begin{align*}
\max\limits_{l=1,\dots,d_1}\max\limits_{j=1,\dots,d_1 - 1}\|\nu^{(l)} Z_{-l,j}\|_{P,3}&\le C\varsigma_n^{-1}K_{n},\\
    \max\limits_{l=1,\dots,d_1}\max\limits_{j=1,\dots,d_1 -1} \E\left[(\nu^{(l)})^4 Z_{-l,j}^4\right]&\le C\varsigma_n^{-4} L_n.
\end{align*}

\item[(v)] There exists a positive number $q\ge 0$ such that the following growth condition is fulfilled:
\begin{itemize}
\item[(a)] $n^{\frac{2}{q}}t_1^3\varsigma_n\log(\tilde{A}_n\varsigma_n^{-1/2})\left(\frac{s^2\log^2(\bar{d}_n) }{n}\vee \frac{t_1^9\varsigma_n^2\log^2(\tilde{A}_n\varsigma_n^{-1/2})}{n}\right)=o(1)$
\item[(b)] $t_1^7\varsigma_n\log(\tilde{A}_n\varsigma_n^{-1/2})\left(\frac{s\log(\bar{d}_n)\log(\tilde{A}_n)}{n}\vee\frac{t_1^{9}\varsigma_n^2\log^6(\tilde{A}_n\varsigma_n^{-1/2})}{n}\right)=o(1)$
\end{itemize}



\end{itemize}
\end{assumptiona}
\noindent
Assumptions A.\ref{A2}$(i)$ and $(ii)$ are regularity and sparsity conditions, permitting the number of nonzero regression coefficients $s_1$ and $s_2$ to grow to infinity with increasing sample size. A detailed comment on the sparsity condition is given in Comment \ref{commentApprox}. Assumption A.\ref{A2}$(iii)$ contains tail conditions on the approximating functions
(and therefore on the original variables) as well as for the error term. Assumption A.\ref{A2}$(iv)$ is an eigenvalue condition that restricts the correlation between the basis elements (and therefore between the original variables), as well as the correlation between the error term $\nu_l$ from the auxiliary regressions. The key innovation of A.\ref{A2}$(iv)$ is that we allow the eigenvalues to vanish as $n$ increases. This is essential if we use an increasing number of local approximation functions $g_l$, $l=1,\dots,d_1$, which implies that the support and hence the variance of each $g_l$ is vanishing with increasing sample size. Typically, we have $\varsigma_n\approx d_1$ as we will illustrate in Comment \ref{splines} for B-Splines. This is in line with the Assumption 5 in \cite{kozbur2015} which tackles the same problem of local estimators where $\zeta_0(d_1)=d_1^{1/2}$ is the rate to orthogonalize B-Splines. A related assumption on the restricted eigenvalue can also be found in Assumption A3 in \cite{lu2019} where the minimum eigenvalue is of order $O(d_1^{-1})$ as well.
Lastly, Assumption A.\ref{A2}$(v)$ provides the growth conditions. These are given in general terms and depend on the choice of the approximation functions. Choosing B-Splines simplifies the growth conditions significantly as we will also outline in Comment \ref{splines}.\\
\begin{theorem}\label{confinbands}
Under Assumptions A.\ref{A1} and A.\ref{A2}, it holds that
\begin{align*}
P\left(\hat{l}(x)\le f_1(x)\le \hat{u}(x), \forall x\in I\right)\to 1-\alpha
\end{align*}
uniformly over $P\in\mathcal{P}_n$ where $c_\alpha$ is a critical value determined through the multiplier bootstrap method.
\end{theorem}
Theorem \ref{confinbands} provides uniform confidence bands for the whole target component $f_1$. Particularly, with probability $1-\alpha$, we have
$$\sup_x|\hat{f}_1(x)-f_1(x)|^2=O\left(\frac{sup_x (g(x)^T\hat{\Sigma}_ng(x))c_\alpha^2}{n}\right)=O\left(\frac{\varsigma_n c_\alpha^2}{n}\right)$$
by Assumption A.\ref{A2}$(iv)$. Hence, the rate of convergence/the width of the confidence band strongly depends on the parameter $\varsigma_n$ and therefore on the number of approximating functions $d_1$ where typically $\varsigma_n\approx d_1$ for local approximating functions.
\begin{remark}\label{splines}[\textbf{B-Splines}]
An appropriate and common choice in series estimation is B-Splines. B-Splines are positive and local in the sense that $g(x)\ge 0$ and $\sup_{x\in I}\|g(x)\|_0\le t_1$ for every $x$, where $t_1$ is the order of the spline. The $l_1$-norm of B-Splines is bounded by 1, meaning
\begin{align*}
\|g(x)\|_1=\sum\limits_{j=1}^{d_1}g_j(x)\le 1
\end{align*} 
for every $x$ (due to partition of unity). Hence, Assumption A.\ref{A1}$(i)$ is met with
$$\frac{1}{\sqrt{t_1}}\le\inf\limits_{x\in I} \|g(x)\|_2\le\sup\limits_{x\in I} \|g(x)\|_2\le 1 \quad\text{and}\quad \sup\limits_{x\in I} \sup\limits_{l=1,\dots,d_1}|g_l(x)|\le 1.$$
Now, we go into greater detail regarding the condition on the covering number of the image of $g$. Especially if $t_1<d_1$, the complexity of the approximating functions is reduced significantly. One obtains
\begin{align*}
g(I)\subseteq \bigcup\limits_{j=1}^{\binom{d_1}{t_1}} g^{(j)}(I),
\end{align*}
where each $g^{(j)}(I)$ is only dependent on $t_1$ nonzero components. It is straightforward to see that for each $g^{(j)}(I)$ the covering numbers satisfy
\begin{align*}
N(\varepsilon,g^{(j)}(I),\|\cdot\|_{2})\le \left(\frac{6\sup_{x\in I}\|g(x)\|_2}{\varepsilon}\right)^{t_1}
\end{align*}
(cf. \cite{vanweak}), implying
\begin{align*}
\log N(\varepsilon,g(I),\|\cdot\|_{2})&\le \log\left(\sum\limits_{j=1}^{\binom{d_1}{t_1}}N(\varepsilon,g^{(j)}(I),\|\cdot\|_{2})\right)\\
&\le \log\left(\left(\frac{e\cdot d_1}{t_1}\right)^{t_1}\left(\frac{6\sup_{x\in I}\|g(x)\|_2}{\varepsilon}\right)^{t_1}\right)\\
&\le t_1\log\left(\left(\frac{6ed_1\sup_{x\in I}\|g(x)\|_2}{t_1}\right)\frac{1}{\varepsilon}\right)\\
&\lesssim t_1\log\left(\frac{d_1}{\varepsilon}\right).
\end{align*}
It is worth noting that the second moment of a univariate B-spline (with a fixed order) is given by
\begin{align*}
    \E[g(X)^2]\approx d_1^{-1},
\end{align*}
see e.g. \cite{shen1998local} or \cite{de1978practical}. In this case it is reasonable to assume that Assumption
 A.\ref{A2}$(iv)$ holds with $\varsigma_n=d_1$. Let us now discuss Assumption A.\ref{A1}$(ii)$. Define the following class of additive, smooth functions:
$$\F_p^m(\tilde{s}):=\{f:X\mapsto\sum_{j\in \tilde{S}}f_j(X_j)\mid|\tilde{S}|\le \tilde{s}, f_j\in C^m[0,1], \E[f_j(X_j)]=0\ \forall j \in \tilde{S}\}$$
and assume $f_{-1}\in \F_p^m(\tilde{s})$. Using B-Splines we are able to approximate $f_{-1}$ sufficiently well. Let $h_1,\dots,h_{\tilde{d}}$ be B-Splines of order $\tilde{t}$ with equidistant knots and $\mathcal{B}(\tilde{d},\tilde{t})$ be the corresponding linear space. Relying on \cite{shen1998local}, for each $j\in\tilde{S}$, we can find a $f_j^*\in \mathcal{B}(\tilde{d},\tilde{t})$ such that
\begin{align*}
    \|f_j-f_j^*\|_\infty = O\left((\tilde{d}-\tilde{t})^{-\tilde{t}}\right)
\end{align*}
for all $\tilde{t}\le m$. This directly implies
\begin{align*}
    \sup_{x_{-1}} b_2(x_{-1}) \le C \tilde{s} (\tilde{d}-\tilde{t})^{-\tilde{t}}.
\end{align*}
With the same argument for $f_1\in C^m[0,1]$ and choosing $\tilde{d}=d_1$, we obtain 
\begin{align*}
    \sup_{x} (b_1(x_1) + b_2(x_{-1})) \le C (\tilde{s}+1) (\tilde{d}-\tilde{t})^{-\tilde{t}}.
\end{align*}
For $\tilde{t} = 2$ and $\tilde{s}= O(1)$, we have to ensure that
\begin{align*}
    (sup_{x} (b_1(x_1) + b_2(x_{-1})))^2 \le C \tilde{d}^{-4} \le \frac{s_1\log(\bar{d}_n)}{n}.
\end{align*}
Here, for simplicity, we assume that only a finite number of regressors in our model are relevant ($\tilde{s}=O(1)$) and hence $s_1\approx d_1$ and $\bar{d}_n\approx pd_1$. This implies
\begin{align*}
    s_1 = O(n^{1/5}\log^{-1/5}(\bar{d}_n)).
\end{align*}
which can be aligned with our growth rates in Assumption A.\ref{A2}$(v)$ which we discuss next. It is worth noting that in this case 
$$ sup_{x} (b_1(x_1) + b_2(x_{-1}))\lesssim d_1^{-2}$$
which is in line with Assumption 3 in \cite{kozbur2015} with  $K=d_1$ and $\alpha_0=2$. 
Choosing the order of B-Splines $t_1 = \log(n)$ and keeping in mind that $\varsigma_n\approx d_1$ and $\tilde{A}_n\approx d_1$, the growth rates in Assumption A.\ref{A2}$(v)$ simplify to
\begin{itemize}
\item[(a)] $n^{\frac{1}{\tilde{q}}}d_1\log(d_1)\left(\frac{s^2\log^2(\bar{d}_n) }{n}\vee \frac{d_1^2\log^2(d_1)}{n}\right)\lesssim n^{\frac{1}{\tilde{q}}}\frac{d_1^3\log^3(\bar{d}_n)}{n} =o(1)$
\item[(b)] $d_1\log(d_1)\left(\frac{d_1\log(\bar{d}_n)\log(d_1)}{n}\vee\frac{d_1^2\log^6(d_1)}{n}\right)\lesssim \frac{d_1^3\log^7(d_1)}{n}=o(1)$
\end{itemize}
for a suitable $\tilde{q}$. This is in line with $s_1\approx d_1\gtrsim n^{1/5}$ which is needed to ensure a sufficient small approximation error in Assumption A.\ref{A1}$(ii)$. We observe that the total number of regressors $p$ and hence the total number of approximation functions $\bar{d}_n\approx d_1p$ can grow at an exponential rate with sample size. This means that the set of approximating functions can be larger than the sample size. This situation is common for lasso-based estimators. Our growth condition is in line with other results in the
literature, e.g., \cite{belloni2018uniformly}, \cite{Belloni2014b} and many others. 
 Obviously, this is only one special case in which the assumptions A.\ref{A2}$(v)$ and A.\ref{A1}$(ii)$ are
satisfied. It is worth noting that the Assumption A.\ref{A2}$(v)$ can still hold true if we allow the number of relevant regressors to increase with sample size.

This discussion points precisely to the following trade-off: 
For a given model, such as an additive model with a finite number of relevant regressors, the number of approximating functions $d_1$ and $d_2$ are allowed to grow with sample size to ensure a sufficiently small approximation error in Assumption A.\ref{A1}$(ii)$. On the other hand, more parameters in the linear model increases the estimation error in the nuisance parameters, which is controlled by the growth assumptions in Assumption A.\ref{A2}$(v)$. Therefore, both assumptions need to be balanced. This trade-off can also be observed in Assumption 11.2 and 11.3 in \cite{kozbur2015}. In contrast to these growth rates, our weaker growth conditions are simultaneously feasible if the target function is twice-differentiable as outlines above. This can be explained by the fact that we are using improvements in series estimators in \cite{belloni2018uniformly} to relax the assumptions in \cite{kozbur2015} who uses results in \cite{NEWEY1997147}. Particularly, \cite{kozbur2015} assumes $n^{-1}d_1^4=o(1)$ in Assumption 4, while we roughly assume $n^{-1}d_1^3=o(1)$ above. Hence, choosing $d_1\approx n^{1/5}$, $\bar{d}_n=d_2$ can be as large as $\exp(o(n^{2/15}))$ which can be larger than $n$ to satisfy the growth condition in a). It is worth noting that the growth condition on $d_2$ in Assumption 11.2 in \cite{Kozbur2021} is much more restrictive than in our paper.

 The second growth condition in b) guarantees the validity of multiplier bootstrap and allows us to construct uniformly valid confidence regions. It is in line with \cite{chernozhukov2013gaussian}, but we allow vanishing eigenvalues.
\end{remark}
\begin{remark}\label{commentApprox}
The sparsity condition in A.\ref{A2}$(ii)$ restricts the number of nonzero regression coefficients $s_1$ and $s_2$ in the Equations (\ref{linmod1}), (\ref{linmod2}) and (\ref{second_stage}). Through this, we especially assume that the regression function $f$ can be adequately approximated using only $s_1$ relevant basis functions. Note that we do not directly control the number of relevant covariables, but rather the total number of approximating functions. This sparsity condition differs from that used by \cite{gregory2016} and \cite{lu2019} who restrict the number of relevant additive components in the model (\ref{GAM}). Our model also includes the approximate sparse setting through the error terms $b_1$ and $b_2$ in (\ref{linmod1}) and (\ref{linmod2}), which offers greater flexibility and is more realistic for many applications.\\
The sparsity condition in A.\ref{A2}$(ii)$ which restricts the number of nonzero regression coefficients $s_1$ in the Equations (\ref{linmod1}) and (\ref{linmod2}) is standard in high-dimensional regression models. The condition on the number of nonzero regression coefficients $s_2$ in the auxiliary regression is also standard but \eqref{second_stage} needs further discussion. 
It assumes that each approximating function itself may be approximated sufficiently
well by $s_2$ or fewer additive components which is comparable to Assumption (B4) in \cite{gregory2016} or Assumption 6 in \cite{Kozbur2021}.
Lemma 1 in \cite{gregory2016} shows that (B4) holds under a mixing type condition on the covariates if one uses piecewise polynomials for approximation. Analog to the sparsity assumption in the main regression, sparsity in \eqref{second_stage} ensures fast estimation rates of the nuisance parameters. The auxiliary regression in \eqref{second_stage} is basically used to orthogonalize our estimation procedure. Heuristically, for any M-estimator, the orthogonality property essentially requires the inverse of the Hessian matrix of the population loss function to have sparse columns (cf. Assumption 6 in \cite{lu2019}). Therefore such a sparsity assumption or related assumption is unavoidable. 
\end{remark}


%% file: Simulation.tex
\section{Simulation Results}\label{Simulation}

\begin{table}[t]
\begin{center}
\begin{small}
\begin{tabular}{c l}
 \hline  \hline \\[-1.8ex]
 Component & Function  \\[0.9ex] \hline &  \\[-1.8ex]  
1 &$f_1(x_1)= - \sin(2 \cdot x)$ \\
 2 & $f_2(x_2) = x^2 - \frac{25}{12}$\\ 
 3 & $f_3(x_3) = x$ \\
4 & $f_4(x_4) = \exp(-x) - \frac{2}{5} \cdot \sinh(\frac{5}{2})$\\
$5, \ldots, p$ & $f_j(x_j) = 0$.  \\ \hline \hline \\[-1.8ex]
\end{tabular}
\end{small} 
\caption{Function definitions in the data generating processes, simulation study.} 
\label{dgps}
\end{center}
Data generating processes are based on settings in \cite{gregory2016} and \cite{meier2009}.
\end{table}

We assess the empirical performance of our estimator through a series of simulation experiments based on the settings outlined in \citet{gregory2016}, which in turn build on the work of \citet{meier2009}. In our first set of simulations, we evaluate the empirical coverage of our simultaneous confidence bands replicating the data generating process from \citet{gregory2016} and \citet{meier2009}, which incorporates a homoskedastic error term. Additionally, because our estimation framework is also compatible with heteroskedastic error terms, we extend the simulation study to include a heteroskedastic version of the aforementioned design. 

Furthermore, we investigate the performance of our estimator and confidence bands in scenarios involving interactions among the nuisance functions $f_{-j}(x_{-j})$. These scenarios are challenging in two regards: First, modeling interactions requires a large number of interaction terms, quickly leading to very high-dimensional settings. Second, specifying interactions among various functions in empirical applications is challenging because of the risk that the chosen model will not precisely match the true underlying data generating process. By exploiting a sparse structure, our theoretical framework can tolerate a moderate degree of misspecification in the nuisance functions. We introduce two interacted designs -- one correctly specified and one misspecified--providing a comprehensive assessment of our methodology.
We consider the finite sample performance of our estimator in a high-dimensional additive model of the form
\begin{align*}
y_i = \sum_{j=1}^{p} f_j(x_{i,j}) + \epsilon_{i,j}, 
\end{align*}
with $i = 1, \ldots, n$ and $j = 1, \ldots, p$. The functions $f_j(x_j)$ are nonlinear for $j=1, 2, 4$ and linear for $j=3$, as can be recognized from their definition in Table \ref{dgps}. The functions $f_5(x_5), \ldots, f_p(x_p)$ are null components, thus implementing a sparse setting. In all of our considered settings, the regressors $X$ are marginally uniformly distributed on an interval, $I = [-2.5, 2.5]$ with correlation matrix $\Sigma$ with $\Sigma_{k,l}=\rho^{|k-l|}$, $1 \le k,l \le p$. We consider the values $\rho = 0$ and $\rho = 0.5$ because they implement the weakest and strongest correlation structures from the simulations in \citet{gregory2016}. We generate data sets for scenarios with  $n \in \{100, 1000\}$ observations and $p \in \{50, 150\}$ explanatory variables $X_1, \ldots, X_p$.

First, we consider a homoskedastic setting with a standard normally distributed error term $\epsilon \sim N\left(0, 1\right)$, which is identical to the design in \citet{gregory2016}. Second, we define a comparable heteroskedastic setting by specifying an error term $\epsilon_j\sim N\left(0, \sigma_j(x_j)\right)$ with $\sigma_j(x_j) = \underline{\sigma} \cdot (1 + |x_j|)$ and $\underline{\sigma} = \sqrt{\frac{12}{67}}$.  This value of $\underline{\sigma}$ ensures a signal-to-noise ratio that is comparable to the homoskedastic setting. Third, we implement two different interacted scenarios, in which we focus on the estimation of the function $f_1(x_1)$. In the correctly specified interacted scenario, which we will later refer to as Scenario $I$, the specification of the estimated sparse additive model is identical to that of the underlying data generating process. In this model, the function $f_2(x_2)$ is interacted with all other components except for $f_1(x_1)$, making the scenario suitable for our estimation approach. Including the interactions of all $p-1$ nuisance functions leads to a high-dimensional setting. For example, if one specifies cubic B-splines with 3 knots in a setting with $p=150$ covariates, the overall dimension of the estimated model is $d_1+d_2 = 6,228$. In the data generating process of the misspecified interacted Scenario $II$, all non-sparse functions are interacted with each other, except for $f_1(x_1)$. In contrast, the model that is estimated empirically only specifies a basic additive model without any interaction terms. Hence, the interacted Scenario $II$ implements a situation with considerable model misspecification.


\noindent
In the simulation, we use the estimator and the multiplier bootstrap procedure we proposed in Section \ref{Estimation} to generate predictions $\hat{f}_j(x_j)$ for the function $f_j(x_j)$ and construct simultaneous confidence bands that are defined in terms of  $\hat{l}_j(x_j)$ and $\hat{u}_j(x_j)$. We approximate the functions $f_j(x_j)$ in the additive model by cubic B-splines. Estimation is performed using post-lasso with a theory-based choice of the penalty level as implemented in the R package \texttt{hdm} \citep{hdm}. Further details related to the implementation in the simulation study can be found in Appendix \ref{compdetails}.\\
Table \ref{basicsettings} presents the empirical coverage achieved by the estimated simultaneous $95\%$-confidence bands in the homoskedastic and heteroskedastic scenario. The results are obtained in $R = 500$ repetitions and the confidence bands are constructed over an interval of values of $x_j$, $I = [-2,2]$. A confidence band is considered to cover the function $f_j(x_j)$ if it entirely contains the true function, or, stated more formally, if for all values of $x_j \in I$ it holds that $\hat{l}_j(x_j)\le f_j(x_j) \le \hat{u}_j(x_j)$.\\

The results can be interpreted as evidence supporting the validity of our inference method in high-dimensional additive models in both homoskedastic and heteroskedastic settings. The coverage of the confidence bands is close to the nominal level across all settings, in particular when the sample size is large ($n=1000$) or if the correlation structure among the covariates is weak ($\rho=0$). We find that a precise approximation of the target function $f_j(x_j)$ through an appropriate specification of the employed B-Splines is essential for the coverage results of the confidence bands.
\begin{table}[t]
\centering
\begin{tabular}{l l l c c c c c}
 \hline  \hline \\[-1.8ex]
$n$    & $p$   & $\rho$ & $f_1$   & $f_2$   & $f_3$   & $f_4$   & $f_5$ \\[0.9ex] \hline &  \\[-1.8ex]  
\multicolumn{8}{c}{\textit{Gregory et al. (2021), homoskedastic}} \\ \\
$100$  & $50$  & $0.0$  & $91.4$  & $93.0$  & $91.4$  & $94.0$  & $93.2$ \\
       &       & $0.5$  & $90.2$  & $91.6$  & $84.0$  & $85.6$  & $86.2$ \\ \\
$100$  & $150$ & $0.0$  & $89.6$  & $91.8$  & $88.2$  & $94.0$  & $94.0$ \\
       &       & $0.5$  & $92.8$  & $92.8$  & $86.8$  & $87.2$  & $87.6$ \\ \\
$1000$ & $50$  & $0.0$  & $95.2$  & $94.8$  & $93.8$  & $95.0$  & $94.4$ \\
       &       & $0.5$  & $92.8$  & $90.8$  & $84.4$  & $94.2$  & $88.8$ \\ \\
$1000$ & $150$ & $0.0$  & $94.0$  & $94.6$  & $93.8$  & $94.6$  & $94.4$ \\
       &       & $0.5$  & $92.0$  & $87.8$  & $80.4$  & $89.4$  & $89.0$ \\[0.9ex]  \\[-1.8ex]  
\multicolumn{8}{c}{\textit{Gregory et al. (2021), heteroskedastic}} \\ \\
$100$  & $50$  & $0.0$  & $92.0$  & $92.6$  & $90.0$  & $93.6$  & $91.6$ \\
       &       & $0.5$  & $88.8$  & $91.8$  & $84.8$  & $89.2$  & $86.6$ \\ \\
$100$  & $150$ & $0.0$  & $90.4$  & $91.6$  & $93.2$  & $94.2$  & $94.8$ \\
       &       & $0.5$  & $92.4$  & $89.4$  & $87.2$  & $88.8$  & $87.0$ \\ \\
$1000$ & $50$  & $0.0$  & $94.8$  & $95.0$  & $93.0$  & $95.4$  & $95.6$ \\
       &       & $0.5$  & $94.4$  & $90.8$  & $87.2$  & $94.0$  & $90.2$ \\ \\
$1000$ & $150$ & $0.0$  & $95.0$  & $95.6$  & $94.2$  & $94.6$  & $95.2$ \\
       &       & $0.5$  & $93.2$  & $86.4$  & $80.2$  & $88.0$  & $90.0$ \\  \hline  \hline \\[-1.8ex]
\end{tabular}
\caption{Coverage results, homoskedastic and heteroskedastic simulation settings.}
Coverage achieved by simultaneous 95\%-confidence bands in $R=500$ repetitions as generated over a range of values of $x_j$, $I=[-2, 2]$.
\label{basicsettings}
\end{table}
Figures \ref{homoskf1} to \ref{homoskf5} illustrate the average results for the estimators of $f_1(x_1), \ldots, f_5(x_5)$, along with the corresponding joint confidence bands. The green dashed line and the green shaded areas represent the results for the settings with $\rho = 0$, and the blue curve and shaded areas refer to the setting with $\rho = 0.5$. The red solid line shows the true value of the function $f_j(x_j)$. The corresponding results for the heteroskedastic settings are very similar and, hence, made available in Appendix \ref{addsimresults}. The plots depicting the average results provide additional insights into the quality of estimation across different settings. First, they illustrate that the estimation benefits substantially from larger samples. In samples with $n=100$ the confidence bands exhibit noticeable variability and width, whereas in larger samples they become considerably narrower. It is important to note that this observation is expected. For example, in the high-dimensional settings with $n=100$ and $p=150$, cubic B-splines with 3 knots would lead to an overall dimension of $d_1 + d_2 = 900$, which is challenging in terms of estimation. Second, comparing the  results across different values of $\rho$ helps elucidate the challenge in terms of approximation quality. For example, Figure \ref{homoskf4} shows the results for $f_4(x_4)$ in the homoskedastic design. Whereas the approximation appears to be very accurate for $\rho = 0$ in the setting with $n=100$, it is considerably worse for $\rho=0.5$. However, in larger samples, the quality of the approximation improves, which is in line with the results of \citet{gregory2016}. Similar observations can be made for $f_3(x_3)$ and $f_5(x_5)$. As in the previously considered settings, the strength of the correlation of the regressors $X$ plays an important role among the performance of the estimator $\hat{f}_1(x_1)$.
\begin{table}[t]
\centering
\begin{tabular}{l l l c }
 \hline  \hline \\[-1.8ex]
   $n$    & $p$   & $\rho$    & $f_1$ \\[0.9ex] \hline &  \\[-1.8ex]  

$100$  & $50$  & $0.0$  & $92.2$   \\
       &       & $0.5$  & $90.4$  \\ \\
$100$  & $150$ & $0.0$  & $93.0$  \\
       &       & $0.5$  & $95.0$  \\ \\
$1000$ & $50$  & $0.0$  & $94.2$   \\
       &       & $0.5$  & $92.8$ \\ \\
$1000$ & $150$ & $0.0$  & $94.0$  \\
       &       & $0.5$  & $92.4$  \\  \hline  \hline \\[-1.8ex]
\end{tabular}
\caption{Coverage results for $f_1$ in the interacted Scenarios $I$.}
Coverage achieved by simultaneous $95\%$ confidence bands in $R=500$ repetitions. Scenario $I$: Correctly specified interacted setting.
\label{interactedsettingsI}
\end{table}
\begin{figure}[t]
\begin{center}
\includegraphics[scale=0.25]{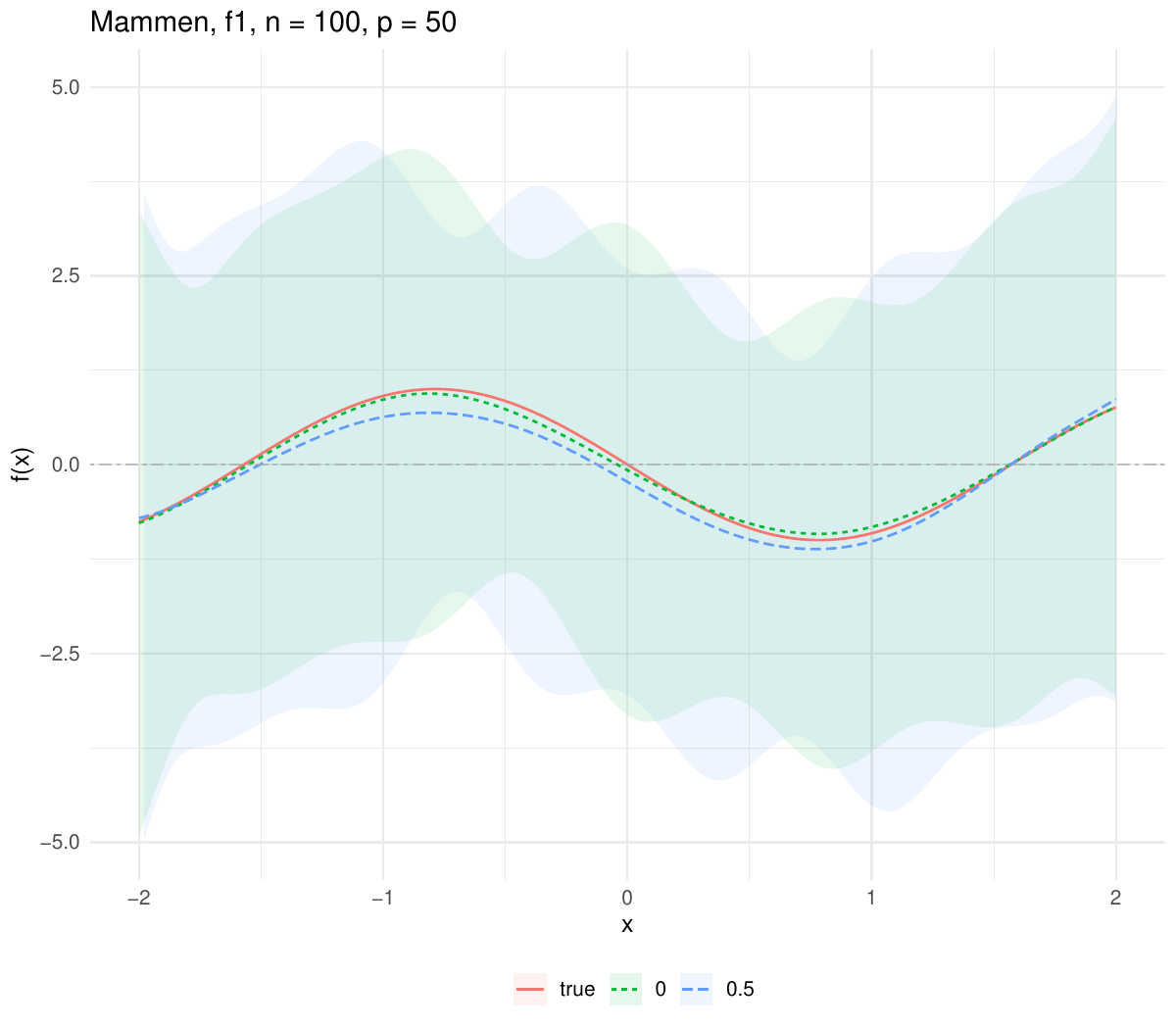} 
\includegraphics[scale=0.25]{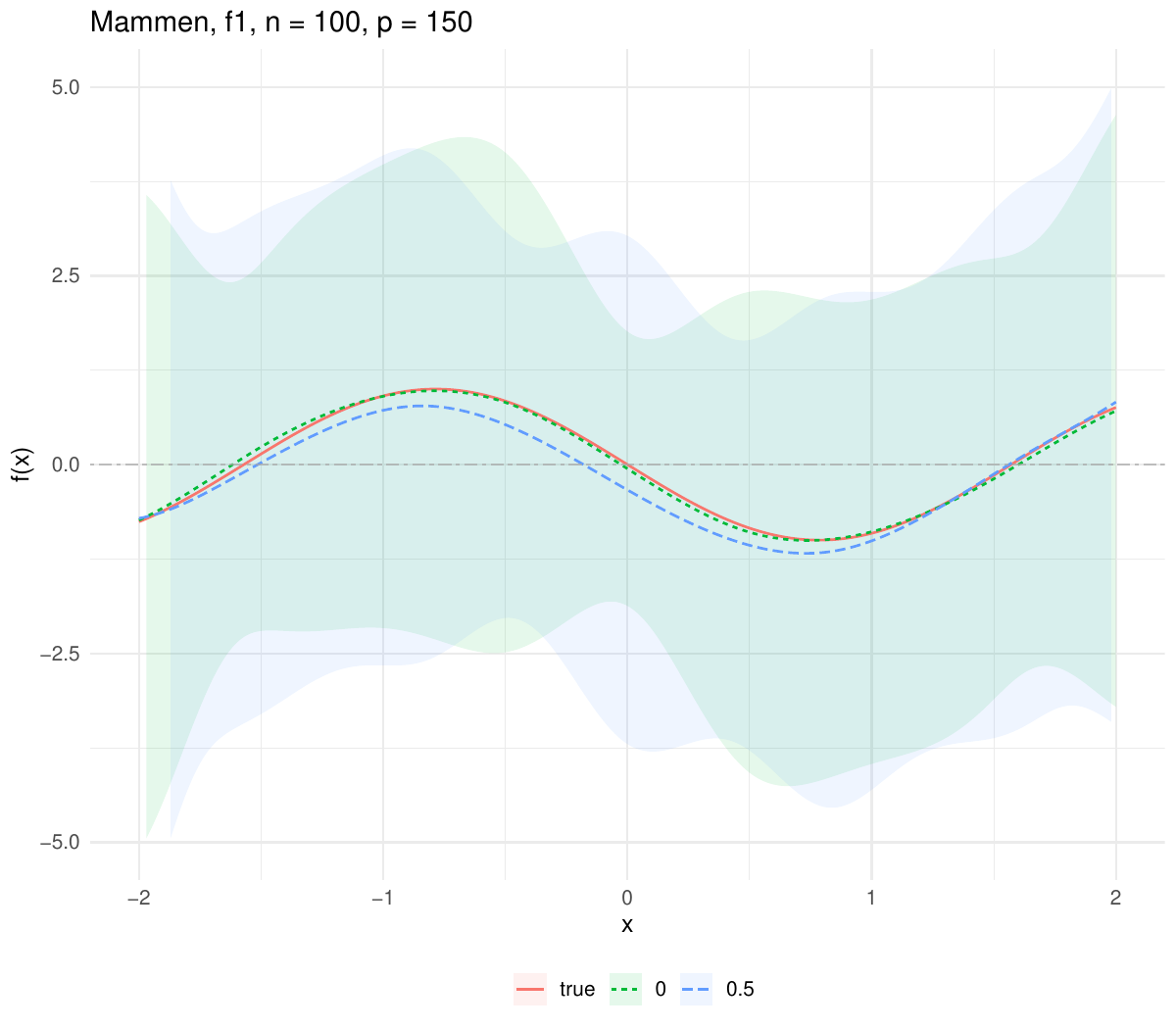} \\
\includegraphics[scale=0.25]{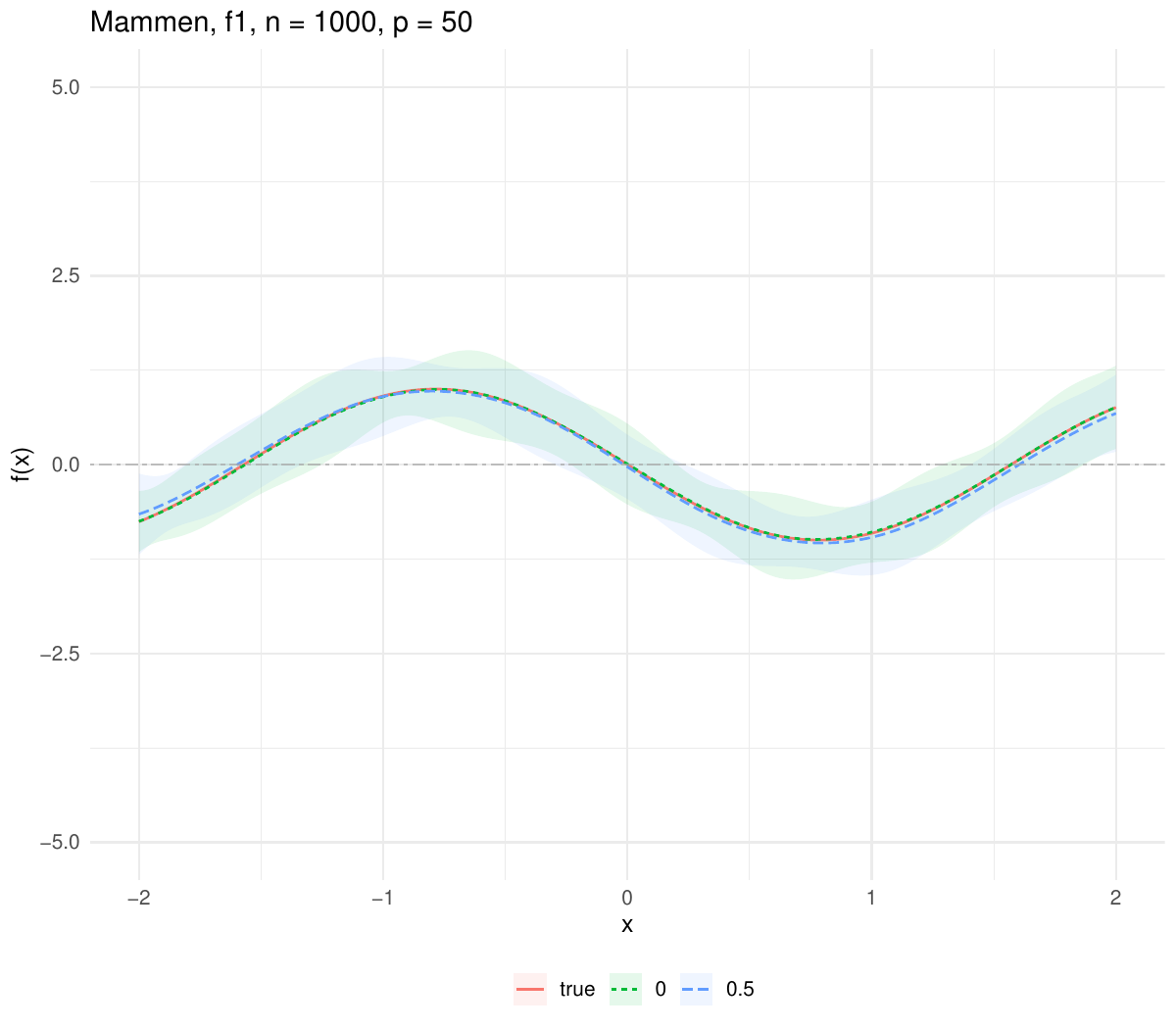}  
\includegraphics[scale=0.25]{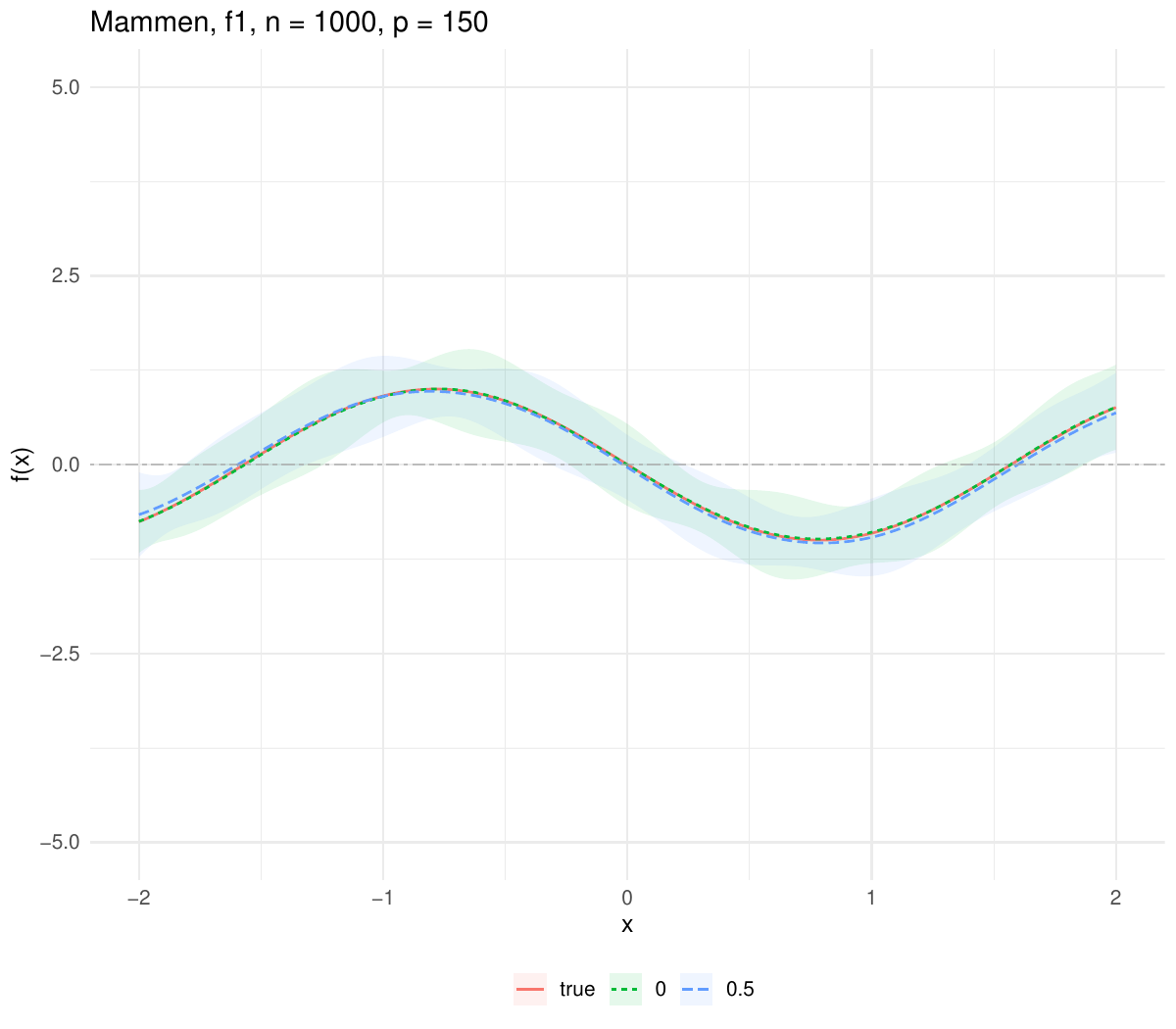} 
\caption{Average confidence bands, $f_1(x_1)$, homoskedastic setting.}
\label{homoskf1}
\end{center}
The green dashed curve illustrates averaged estimated functions $\hat{f}_1(x_1)$  as obtained in $R=500$ repetitions in the scenario with $\rho = 0$. The corresponding averaged $95\%$-confidence bands are shaded green. The blue long-dashed line illustrates the results for the setting with $\rho = 0.5$ together with corresponding averaged confidence bands (shaded blue). The true function $f_1(x_1)$ is illustrated by the red solid curve.
\end{figure}
\begin{table}[t]
\centering
\begin{tabular}{l l l c }
 \hline  \hline \\[-1.8ex]
   $n$    & $p$   & $\rho$    & $f_1$  \\[0.9ex] \hline &  \\[-1.8ex] 
$100$  & $50$  & $0.0$   & $91.0$   \\
       &       & $0.5$   & $93.4$ \\ \\
$100$  & $150$ & $0.0$   & $92.8$  \\
       &       & $0.5$  & $90.8$  \\ \\
$1000$ & $50$  & $0.0$  & $92.2$   \\
       &       & $0.5$  & $93.8$ \\ \\
$1000$ & $150$ & $0.0$  & $92.6$  \\
       &       & $0.5$  & $92.6$   \\  \hline  \hline \\[-1.8ex]
\end{tabular}
\caption{Coverage results for $f_1$ in the interacted Scenario $II$.}
Coverage achieved by simultaneous $95\%$ confidence bands in $R=500$ repetitions. Scenario $II$: Misspecified interacted setting.
\label{interactedsettingsII}
\end{table}
\noindent
\begin{figure}[t]
\begin{center}
\includegraphics[scale=0.25]{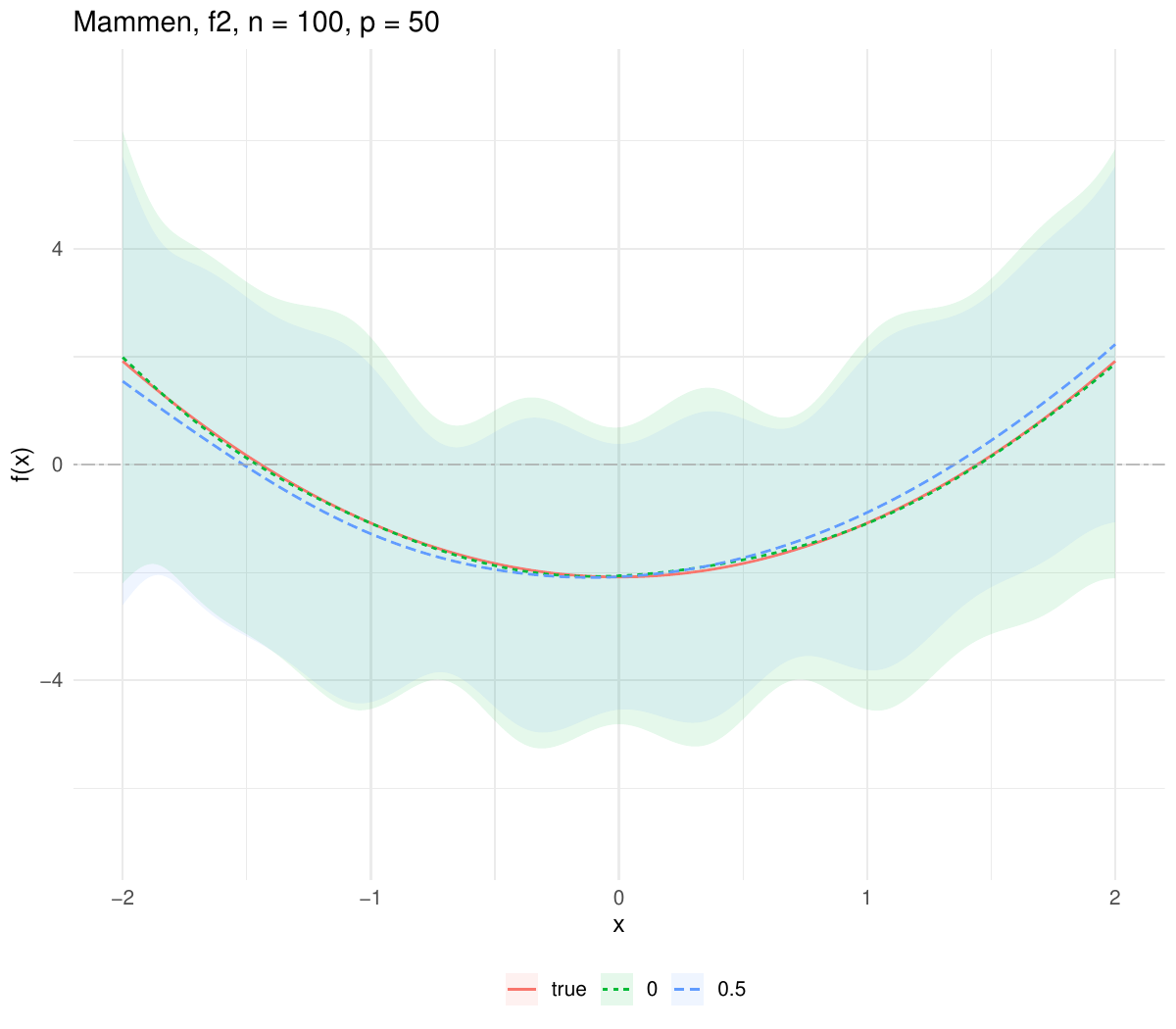} 
\includegraphics[scale=0.25]{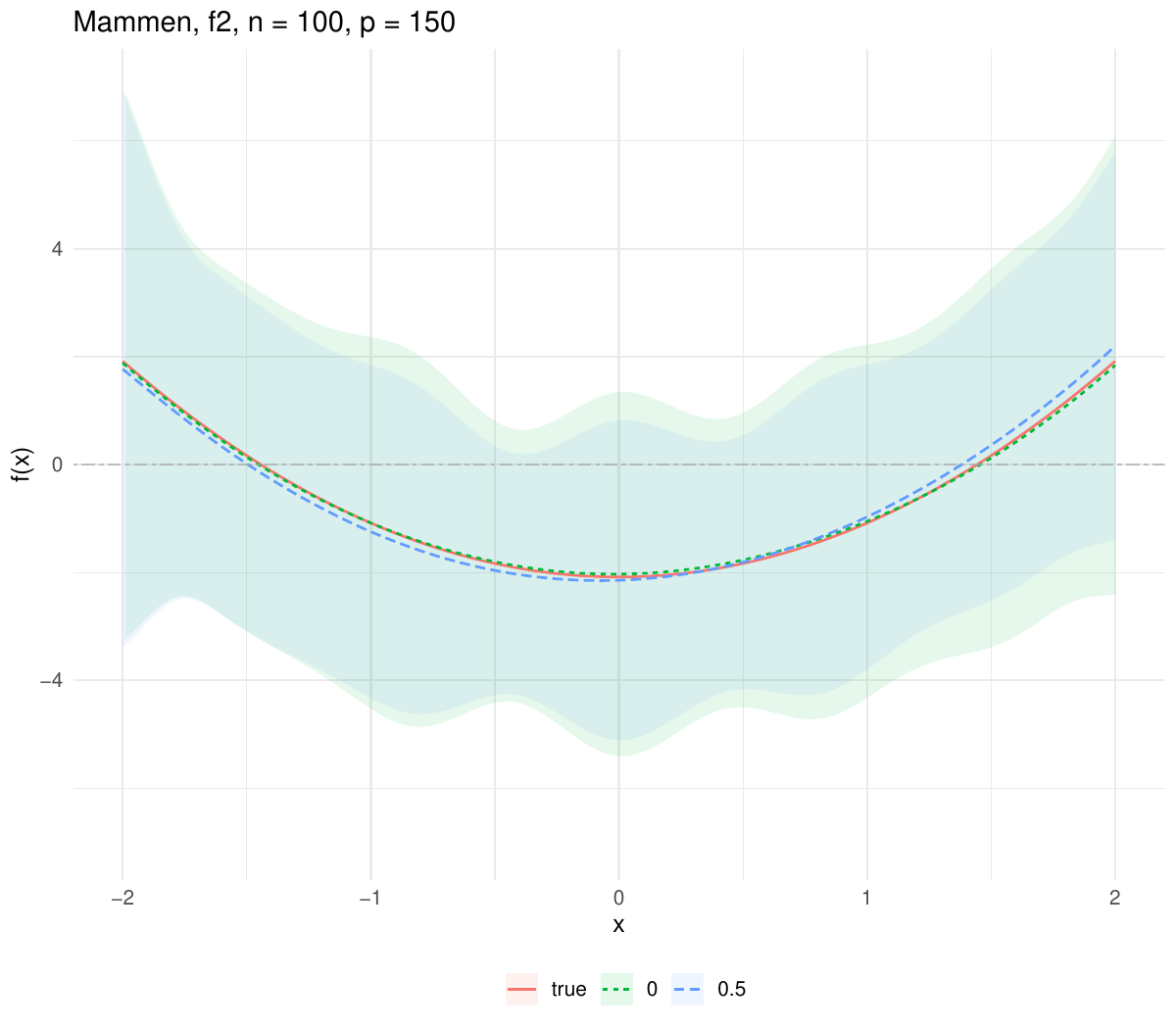} \\
\includegraphics[scale=0.25]{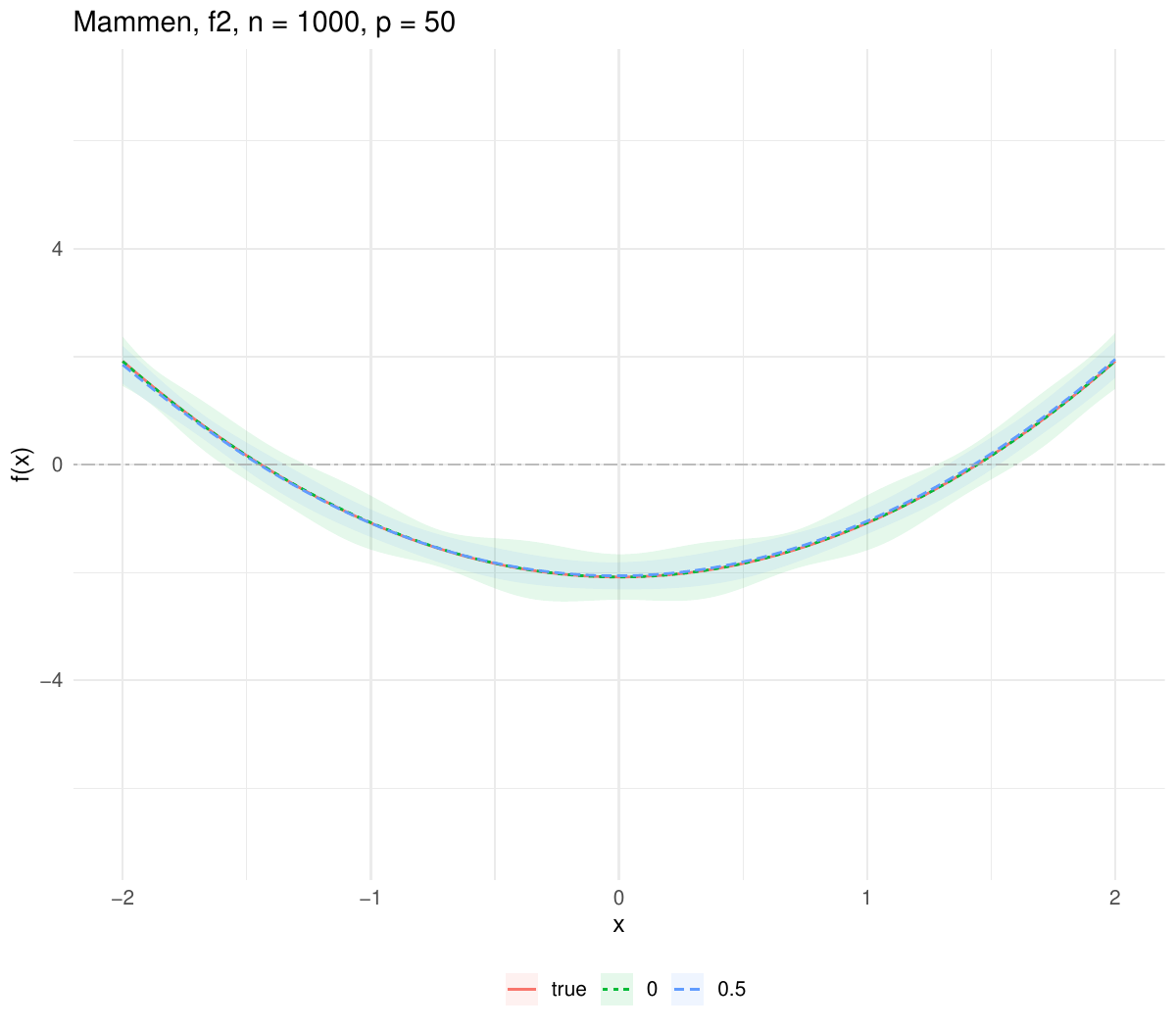}  
\includegraphics[scale=0.25]{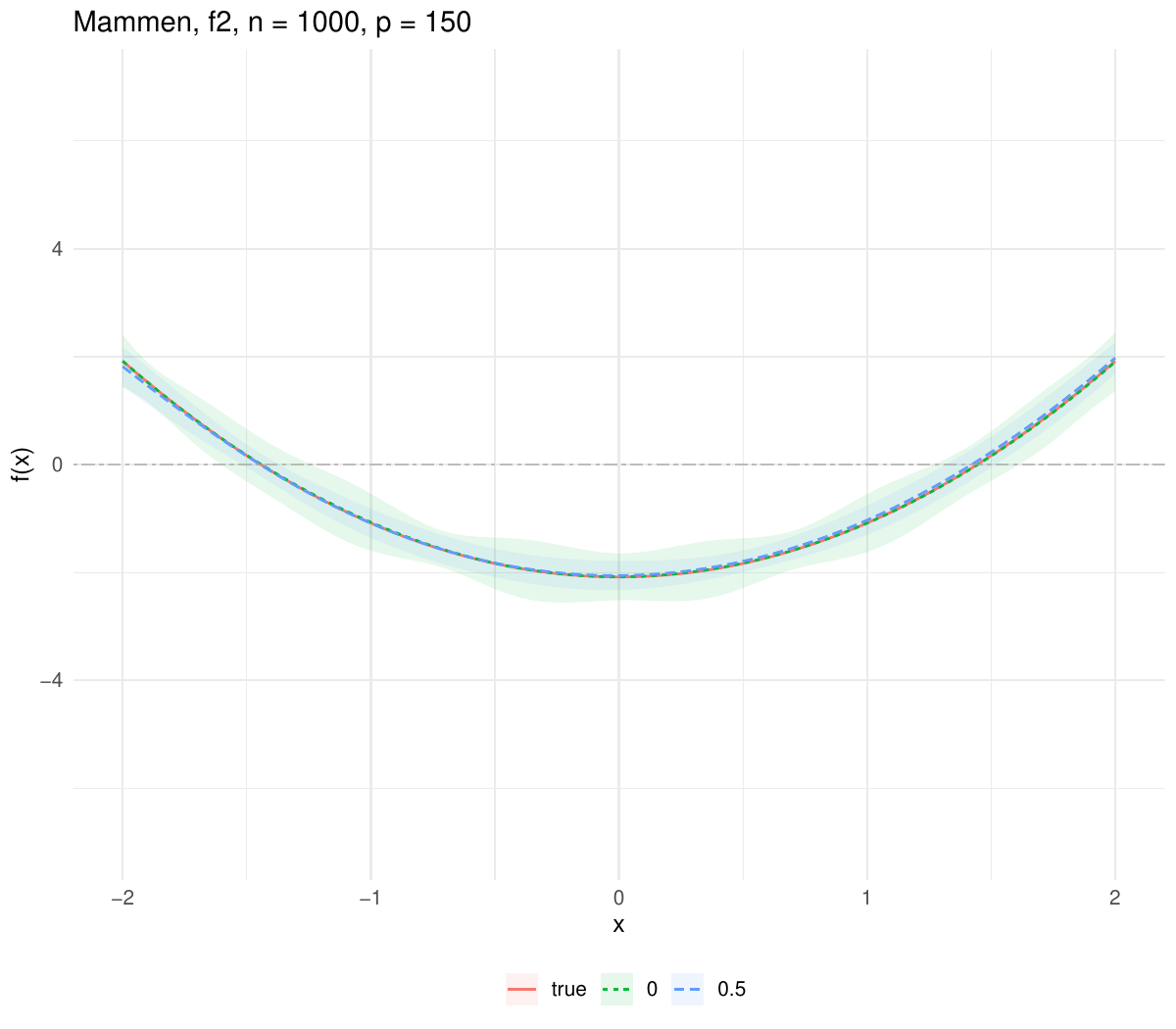} 
\caption{Average confidence bands, $f_2(x_2)$, homoskedastic setting.}
\label{homoskf2}
\end{center}
The green dashed curve illustrates averaged estimated functions $\hat{f}_2(x_2)$  as obtained in $R=500$ repetitions in the scenario with $\rho = 0$. The corresponding averaged $95\%$ confidence bands are shaded green. The blue long-dashed line illustrates the results for the setting with $\rho = 0.5$ together with corresponding averaged confidence bands (shaded blue). The true function $f_2(x_2)$ is illustrated by the red solid curve.
\end{figure}
\begin{figure}[t]
\begin{center}
\includegraphics[scale=0.25]{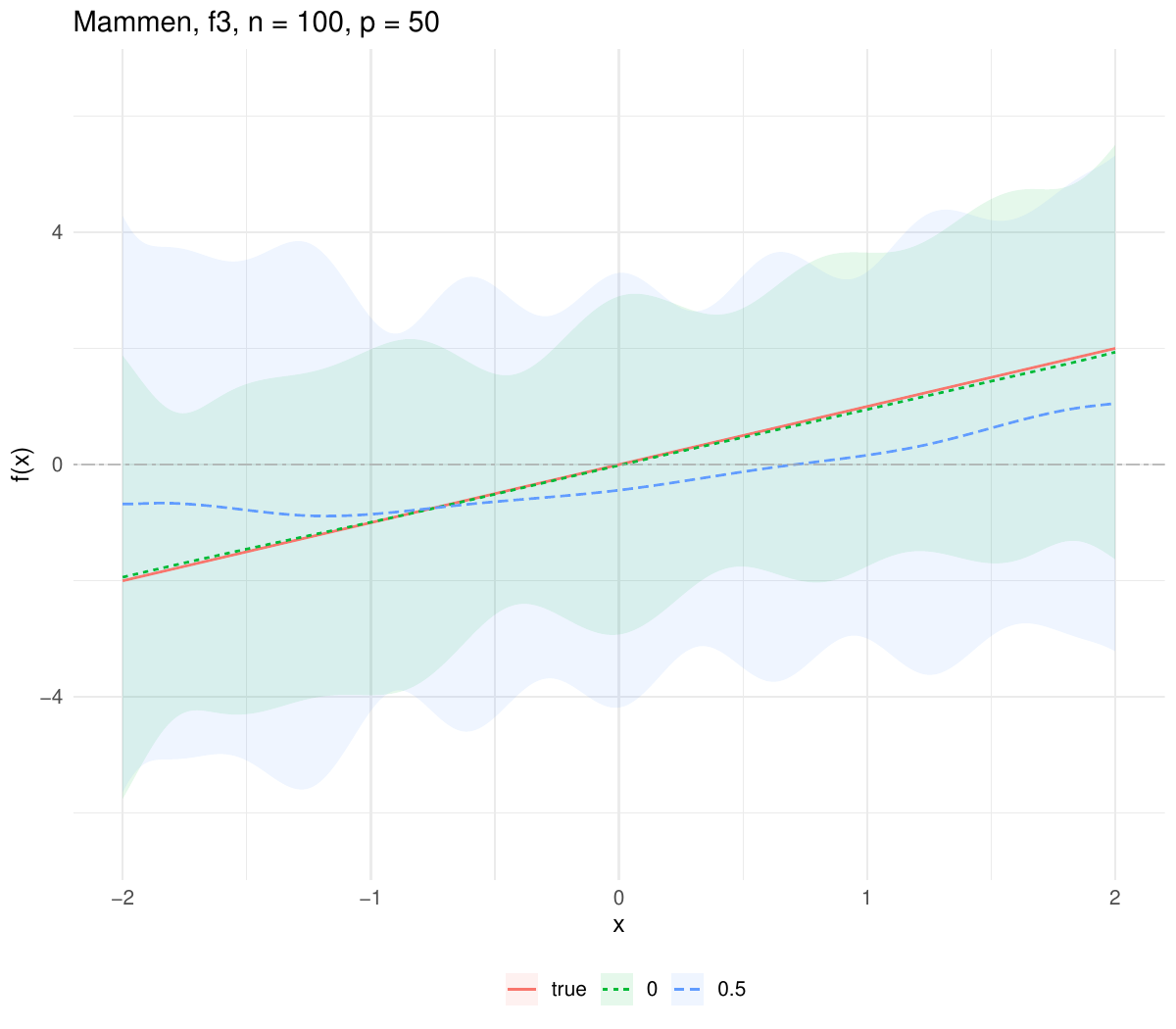} 
\includegraphics[scale=0.25]{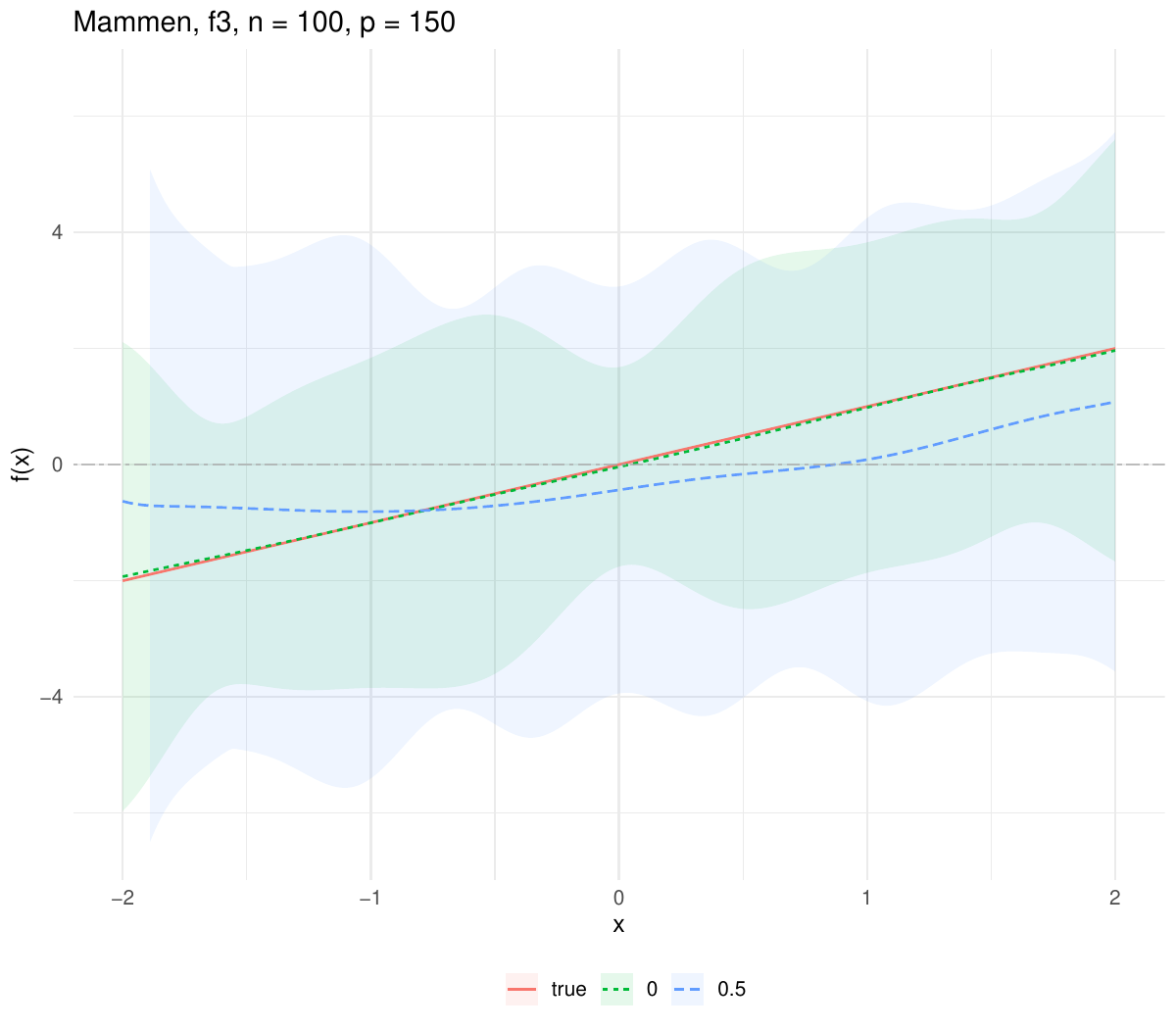} \\
\includegraphics[scale=0.25]{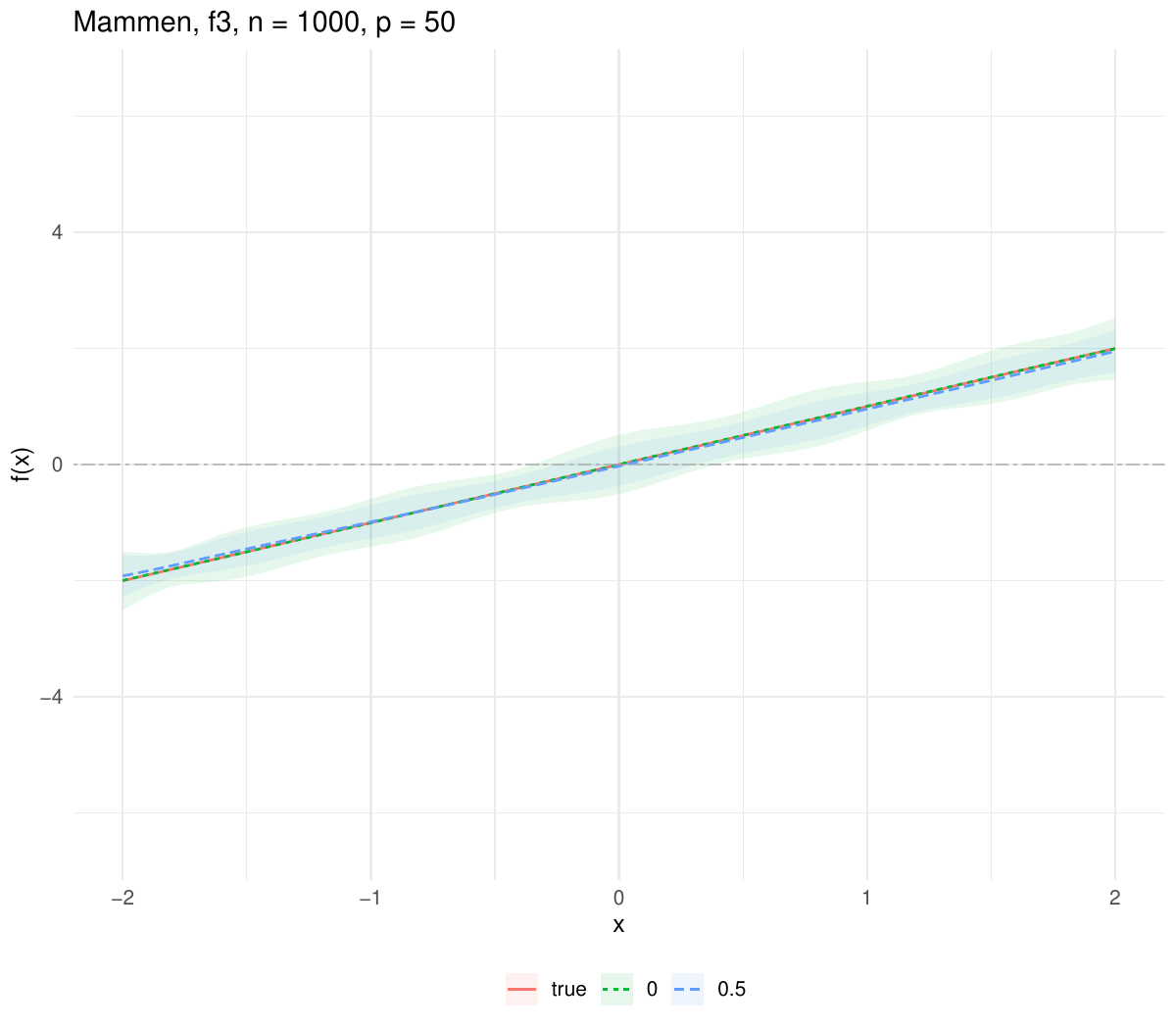}  
\includegraphics[scale=0.25]{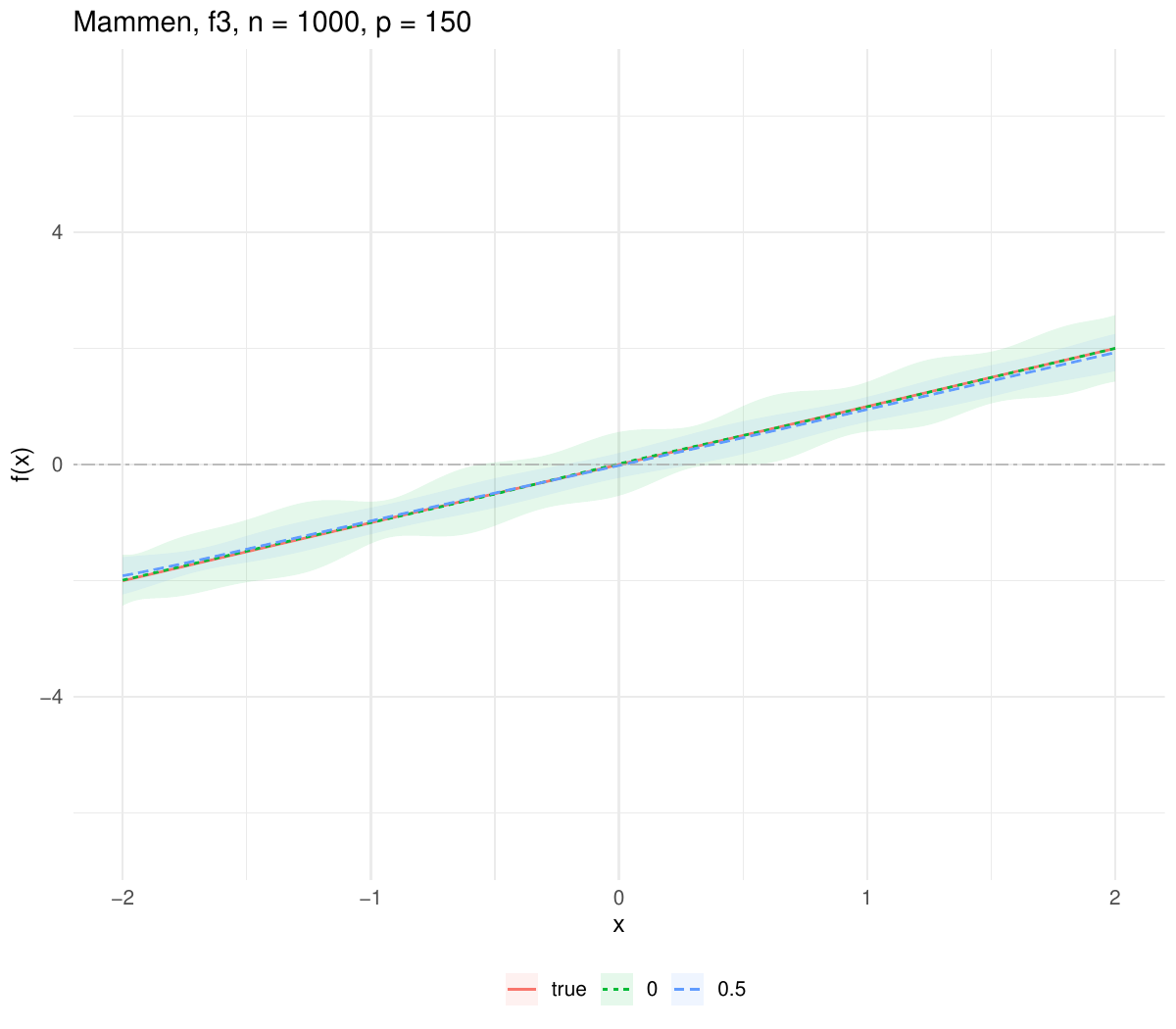} 
\caption{Average confidence bands, $f_3(x_3)$, homoskedastic setting.}
\label{homoskf3}
\end{center}
The green dashed curve illustrates averaged estimated functions $\hat{f}_3(x_3)$  as obtained in $R=500$ repetitions in the scenario with $\rho = 0$. The corresponding averaged $95\%$ confidence bands are shaded green. The blue long-dashed line illustrates the results for the setting with $\rho = 0.5$ together with corresponding averaged confidence bands (shaded blue). The true function $f_3(x_3)$ is illustrated by the red solid curve.
\end{figure}
\begin{figure}[t]
\begin{center}
\includegraphics[scale=0.25]{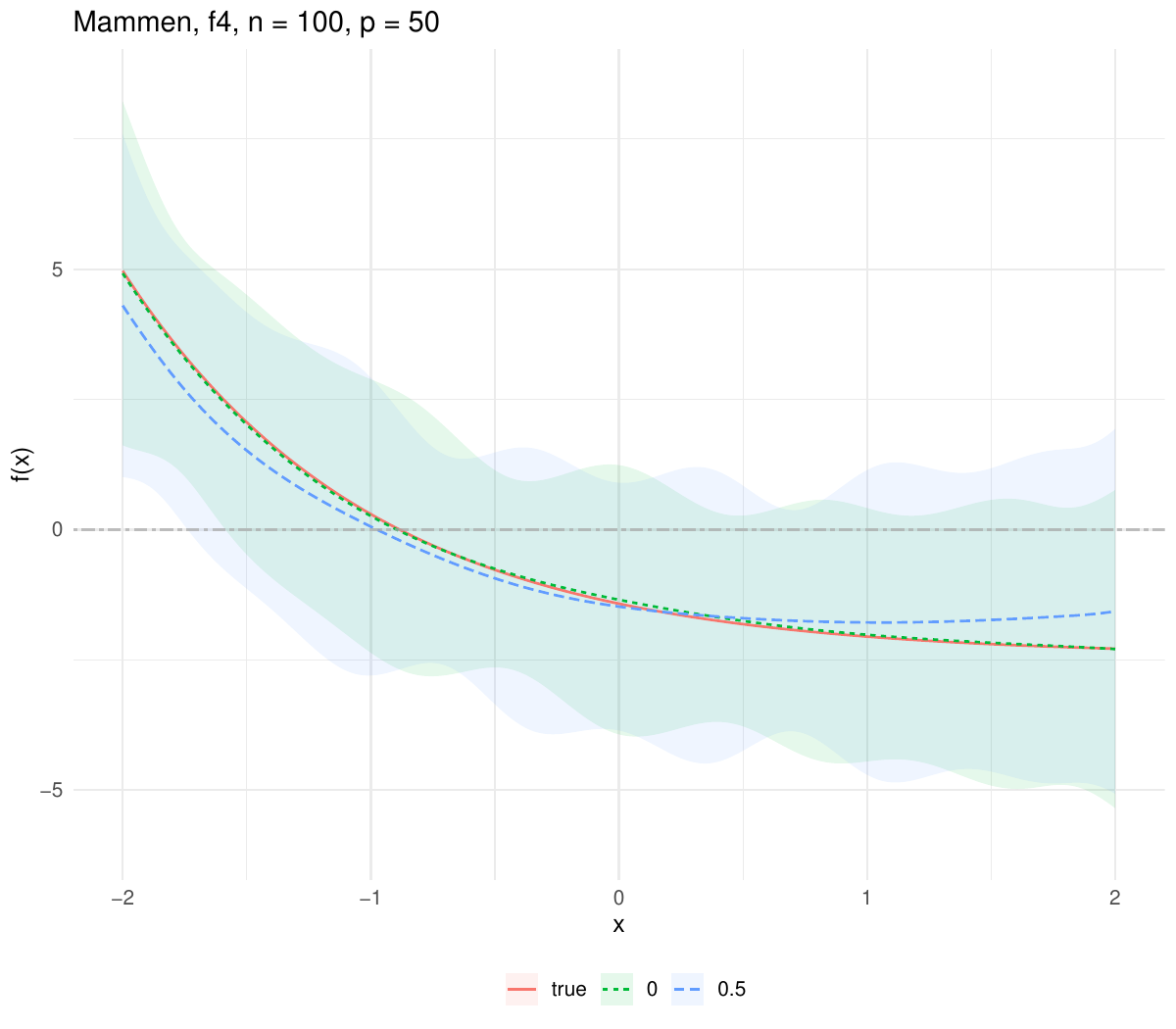} 
\includegraphics[scale=0.25]{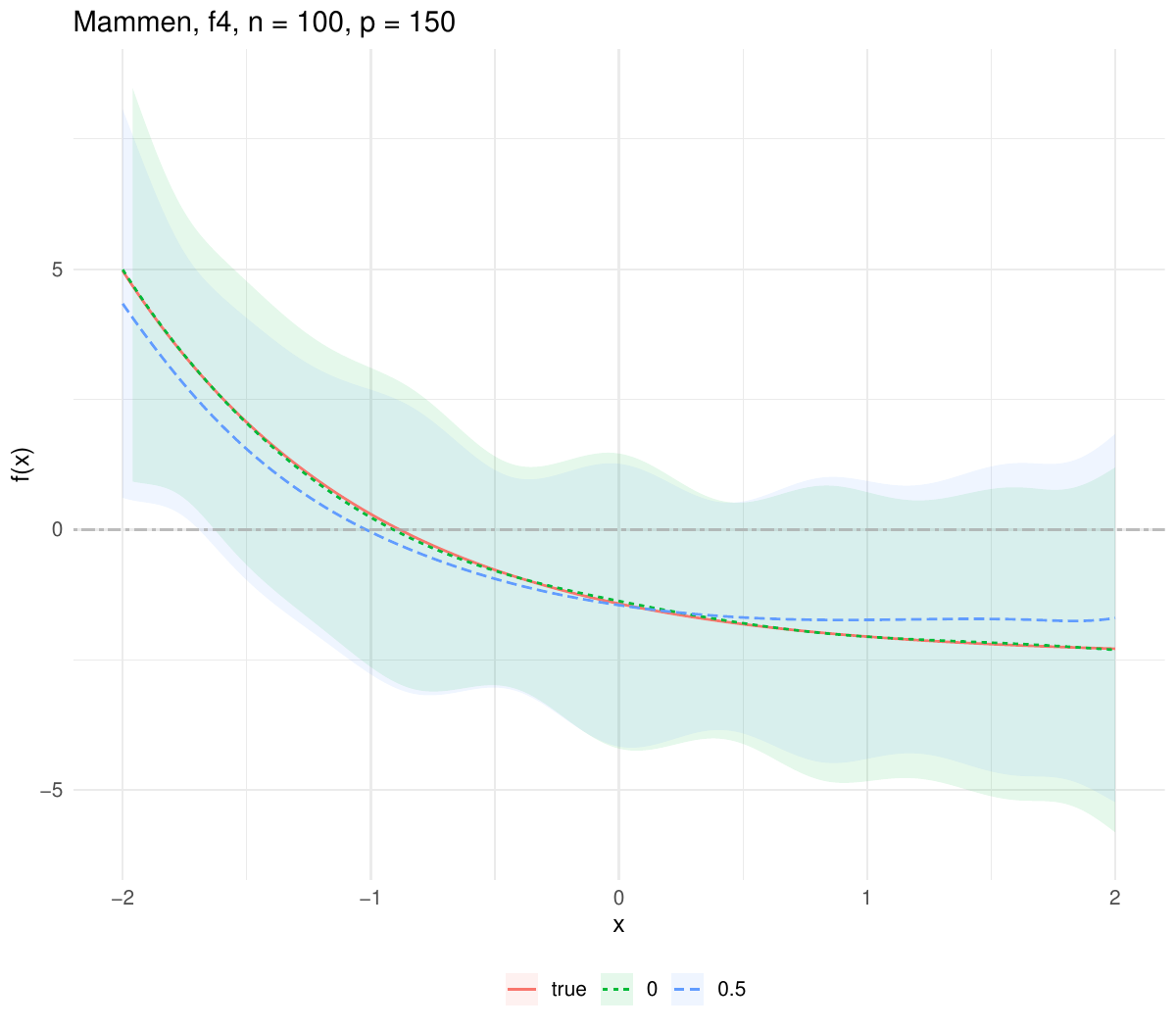} \\
\includegraphics[scale=0.25]{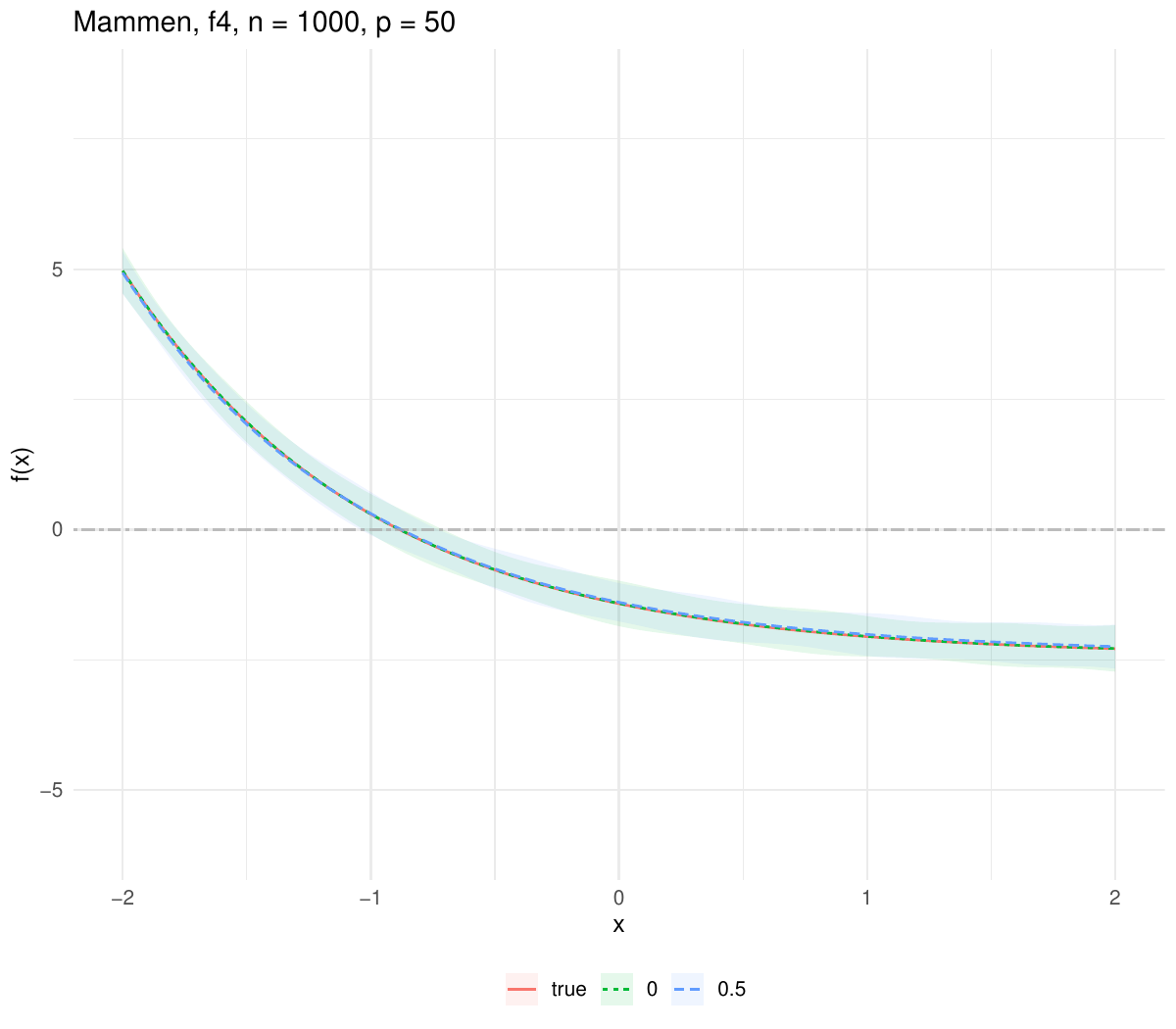}  
\includegraphics[scale=0.25]{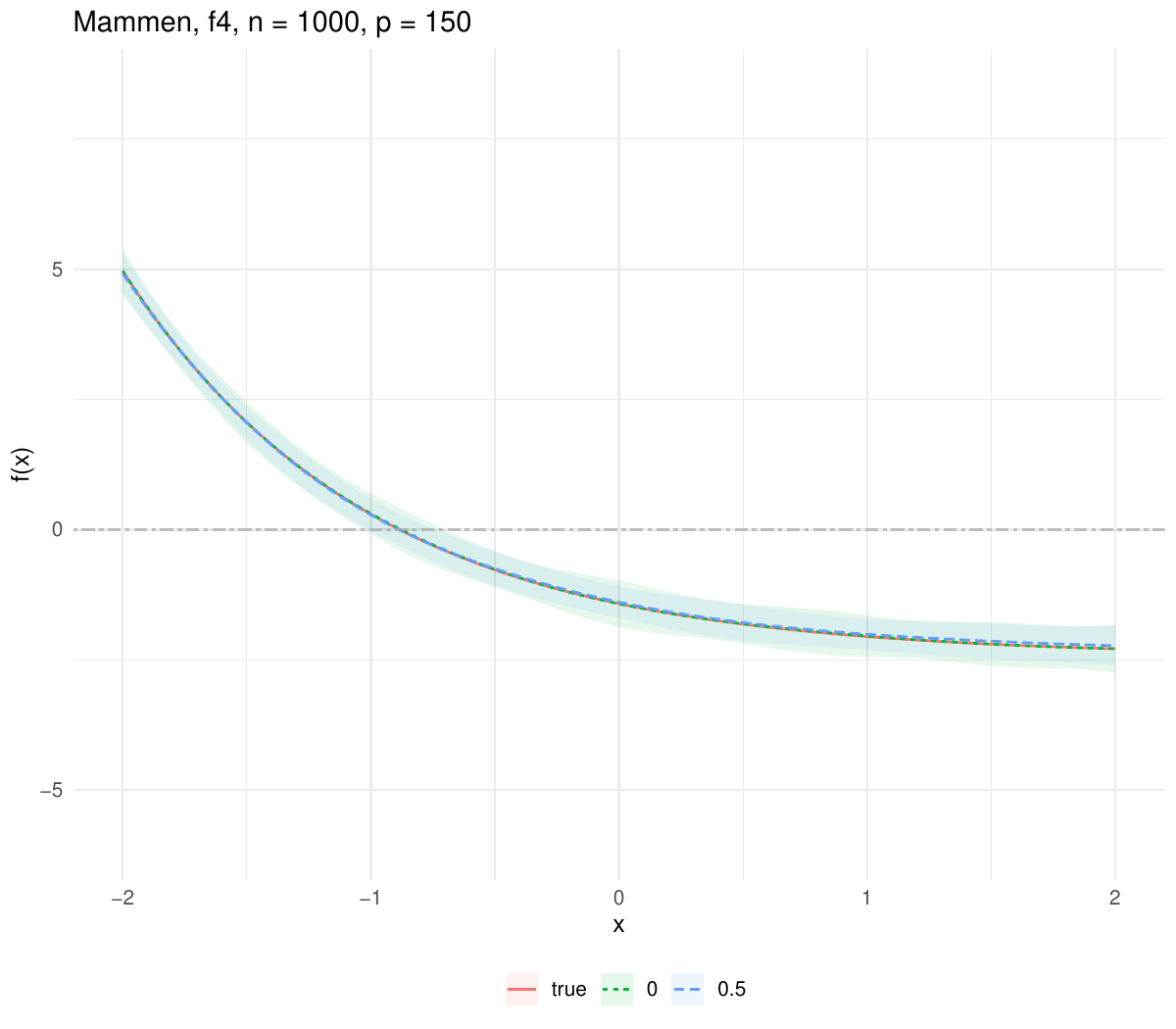} 
\caption{Average confidence bands, $f_4(x_4)$, homoskedastic setting.}
\label{homoskf4}
\end{center}
The green dashed curve illustrates averaged estimated functions $\hat{f}_4(x_4)$  as obtained in $R=500$ repetitions in the scenario with $\rho = 0$. The corresponding averaged $95\%$ confidence bands are shaded green. The blue long-dashed line illustrates the results for the setting with $\rho = 0.5$ together with corresponding averaged confidence bands (shaded blue). The true function $f_4(x_4)$ is illustrated by the red solid curve.
\end{figure}
\begin{figure}[t]
\begin{center}
\includegraphics[scale=0.25]{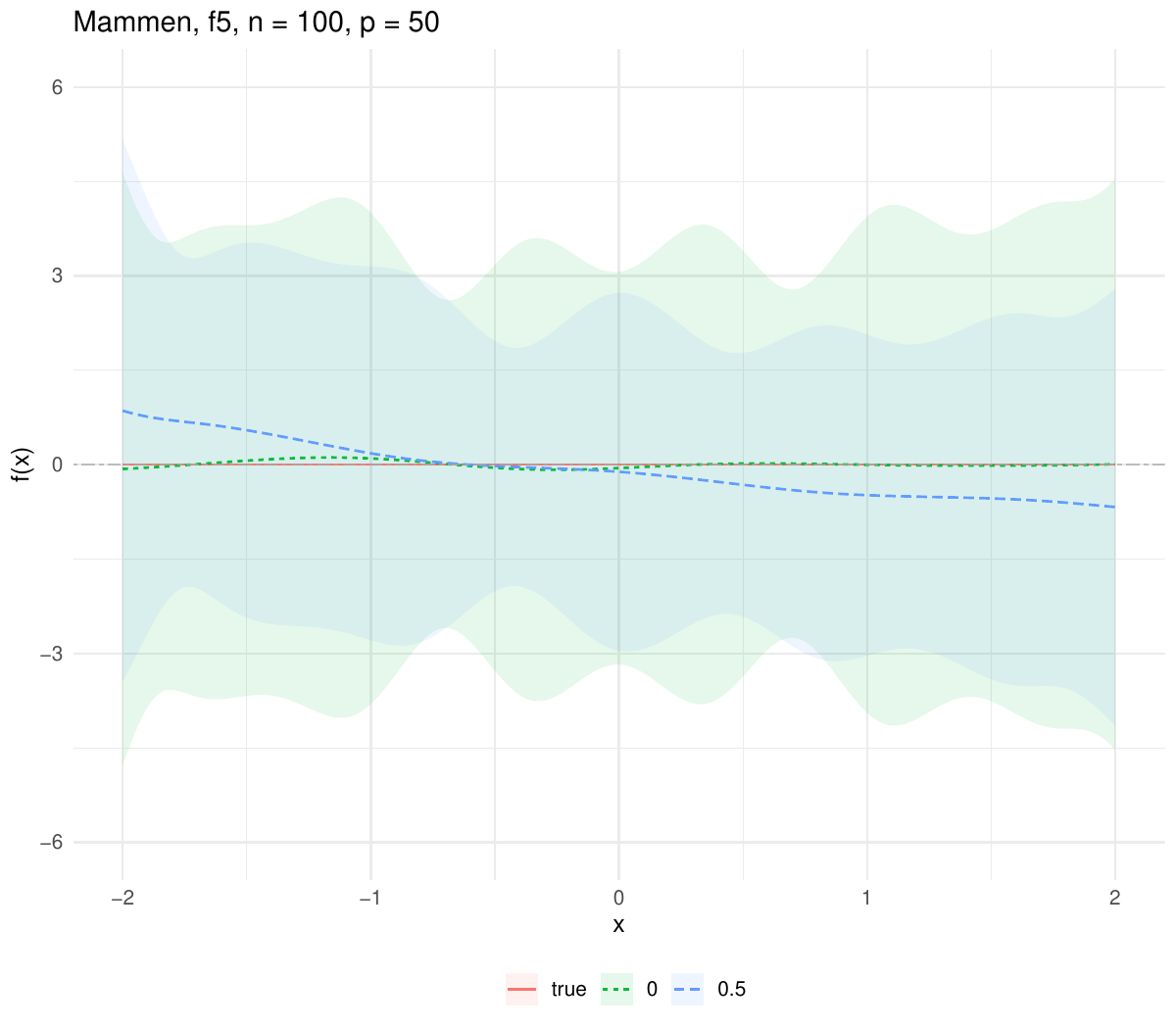} 
\includegraphics[scale=0.25]{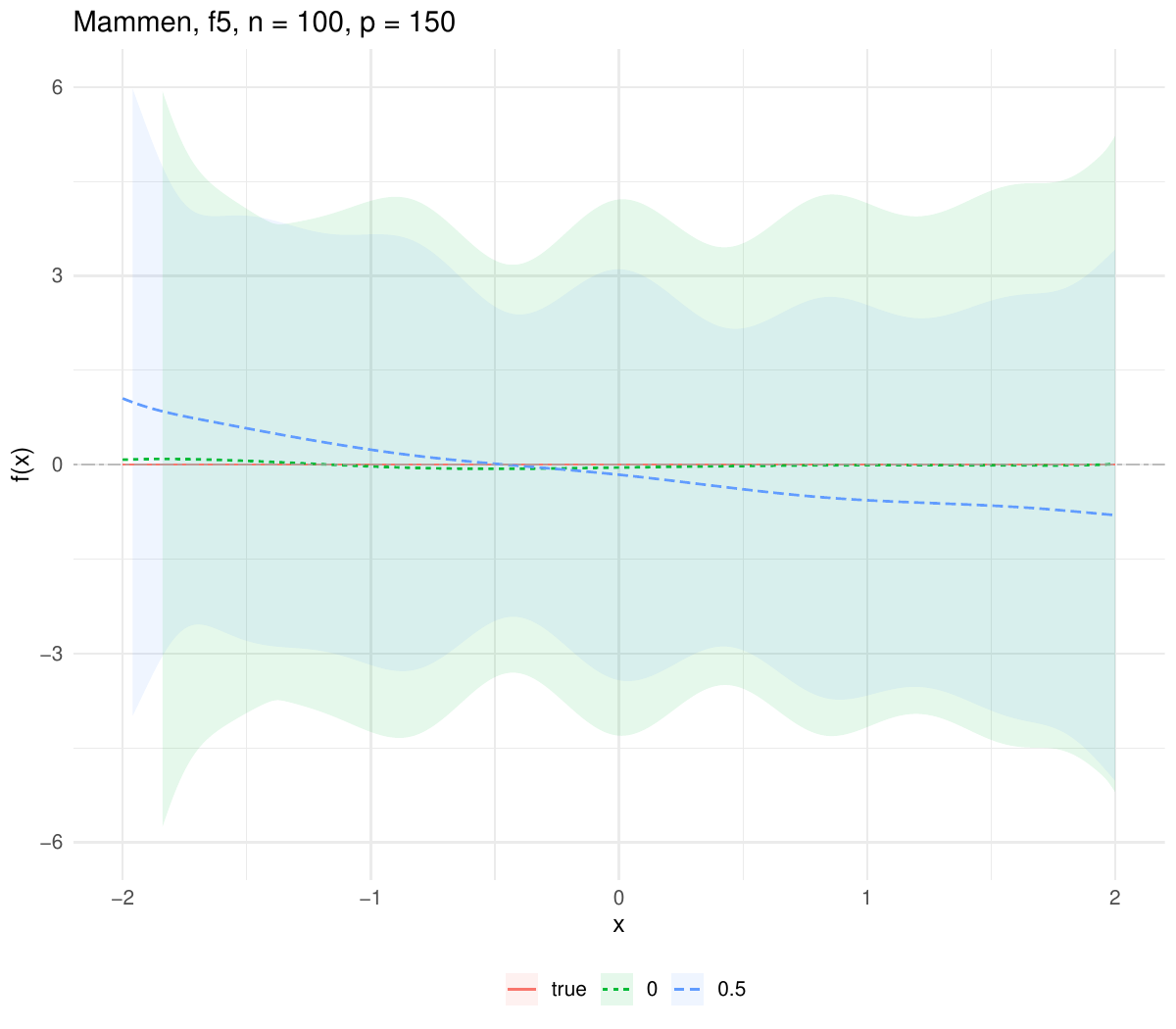} \\
\includegraphics[scale=0.25]{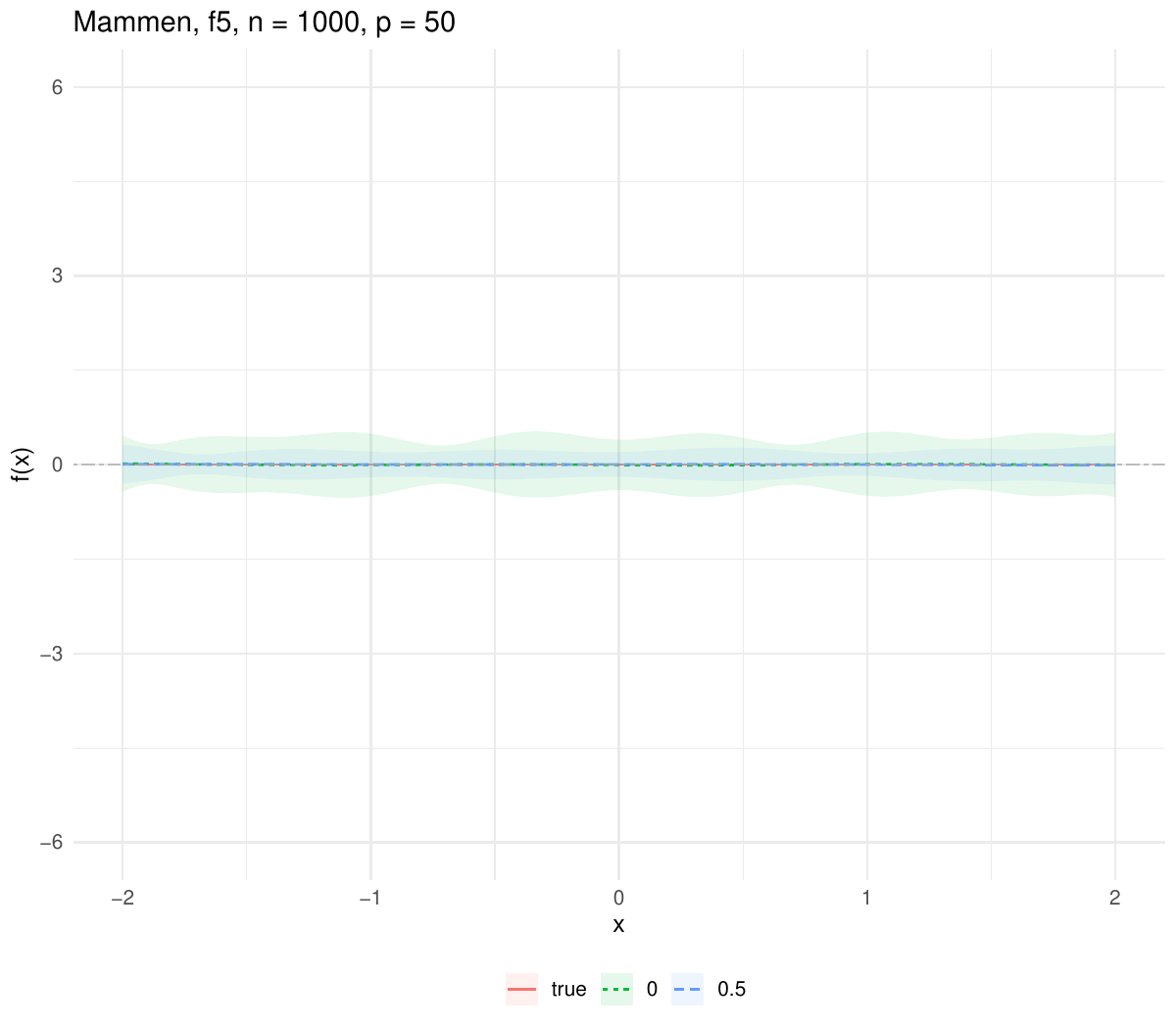}  
\includegraphics[scale=0.25]{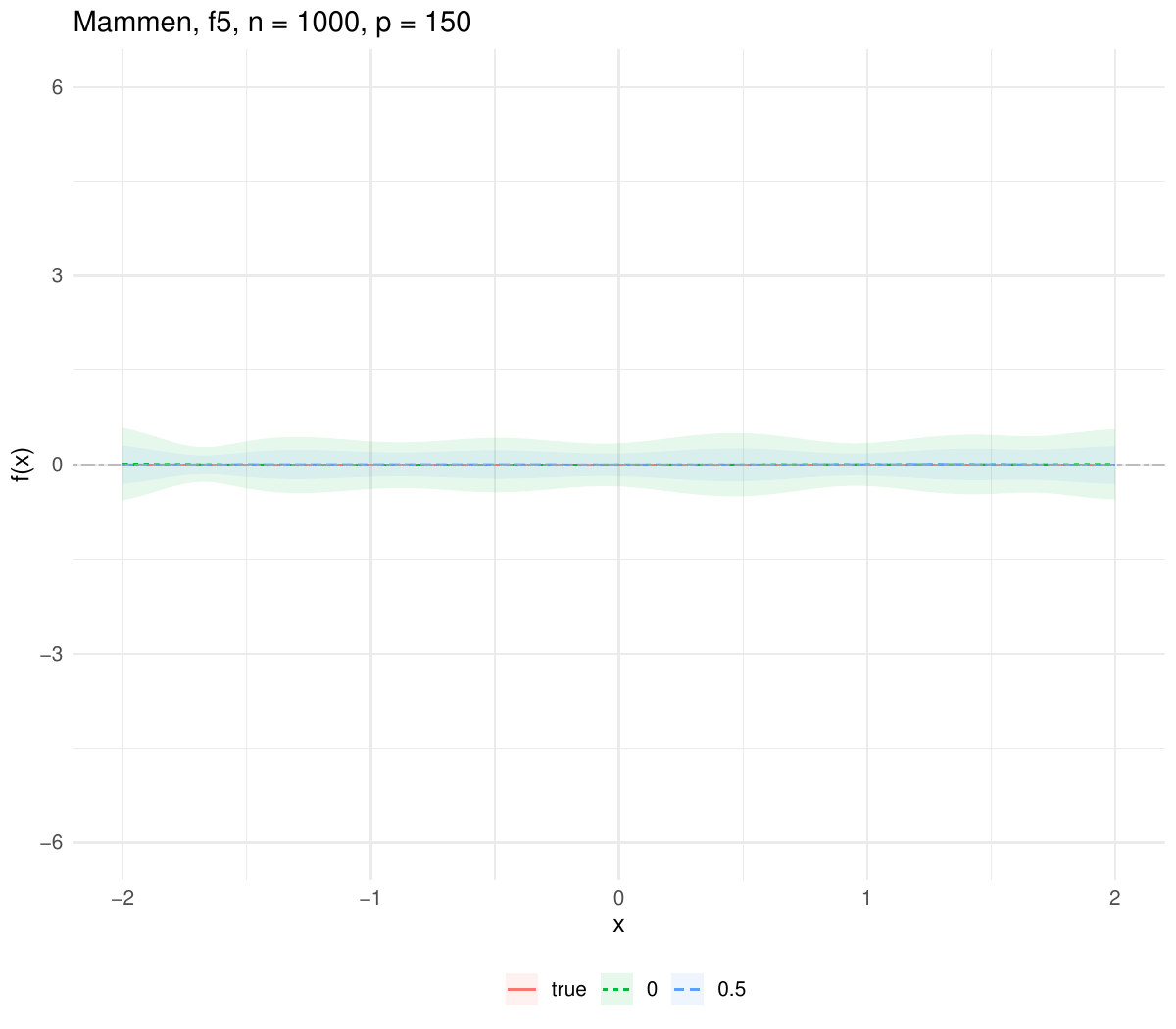} 
\caption{Average confidence bands, $f_5(x_5)$, homoskedastic setting.}
\label{homoskf5}
\end{center}
The green dashed curve illustrates averaged estimated functions $\hat{f}_5(x_5)$  as obtained in $R=500$ repetitions in the scenario with $\rho = 0$. The corresponding averaged $95\%$ confidence bands are shaded green. The blue long-dashed line illustrates the results for the setting with $\rho = 0.5$ together with corresponding averaged confidence bands (shaded blue). The true function $f_5(x_5)$ is illustrated by the red solid curve.
\end{figure}
\noindent
In addition to the homoskedastic and heteroskedastic designs, we investigate the performance of our estimator and confidence bands in interacted simulation settings. The results in Tables \ref{interactedsettingsI} and \ref{interactedsettingsII} and Figures \ref{interactedI} and \ref{interactedII} provide simulation evidence in very high-dimensional settings -- both with a correctly (Scenario $I$) and an incorrectly (Scenario $II$) specified additive model. As a consequence of the higher dimensionality and the increased noise in the estimation problem, the confidence intervals are wider than in the basic settings considered previously. In all cases and irrespective of the model misspecification in Scenario $II$, the empirical coverage is close to the nominal level. As in the settings considered before, the strength of the correlation among the regressors $X$ makes estimation of $f_1(x_1)$ more challenging (see for example the top panels in Figure \ref{interactedI}).

\noindent
In general, the performance of the estimator and the confidence bands depends on the specification of the cubic B-splines used to approximate the target functions (in terms of the number of knots). In our simulations, we observe that the quality of estimation and width of the confidence bands change only moderately when varying the number of knots. Overall, we conclude that our results are relatively robust to the specifications of the B-splines. We find that the performance of the point estimator for the target component $f_j(x_j)$ relies on variable selection through the lasso estimator. In line with our expectations, the lasso estimator is better at selecting the sparse components of the B-splines in larger samples and settings with weaker correlation structure, i.e., $\rho = 0$. In our simulations, we chose the specification of the B-splines estimator based on preliminary evidence. We also investigated the performance of a cross-validated choice (results available upon request). While the latter lacks theoretical justification, we found it to be relatively conservative, resulting in wide confidence bands. During our simulation experiments, we also investigated the performance of an alternative lasso learner that is based on a cross-validated choice of the penalty term. We found that the cross-validated penalty choice was computationally more expensive and inferior in terms of the estimation performance, as indicated by selecting very few of the sparse components.
Overall, we interpret the simulation results as supportive evidence for our estimation approach. The approximation and coverage results for our estimator and confidence bands across different simulation settings are comparable to those provided in \citet{gregory2016}. As in \citet{gregory2016}, the approximation quality of our estimator depends on the underlying correlation structure of the covariates $X$, the parametrization of the B-splines, and the variable selection performance of lasso. In their simulation study, \citet{gregory2016} report pointwise coverage results, whereas our results refer to the coverage of the true function $f_j(x_j)$ over all points of $x_j$ in an interval $I$ by a simultaneous confidence band.
\begin{figure}[t]\begin{center}
\includegraphics[scale=0.25]{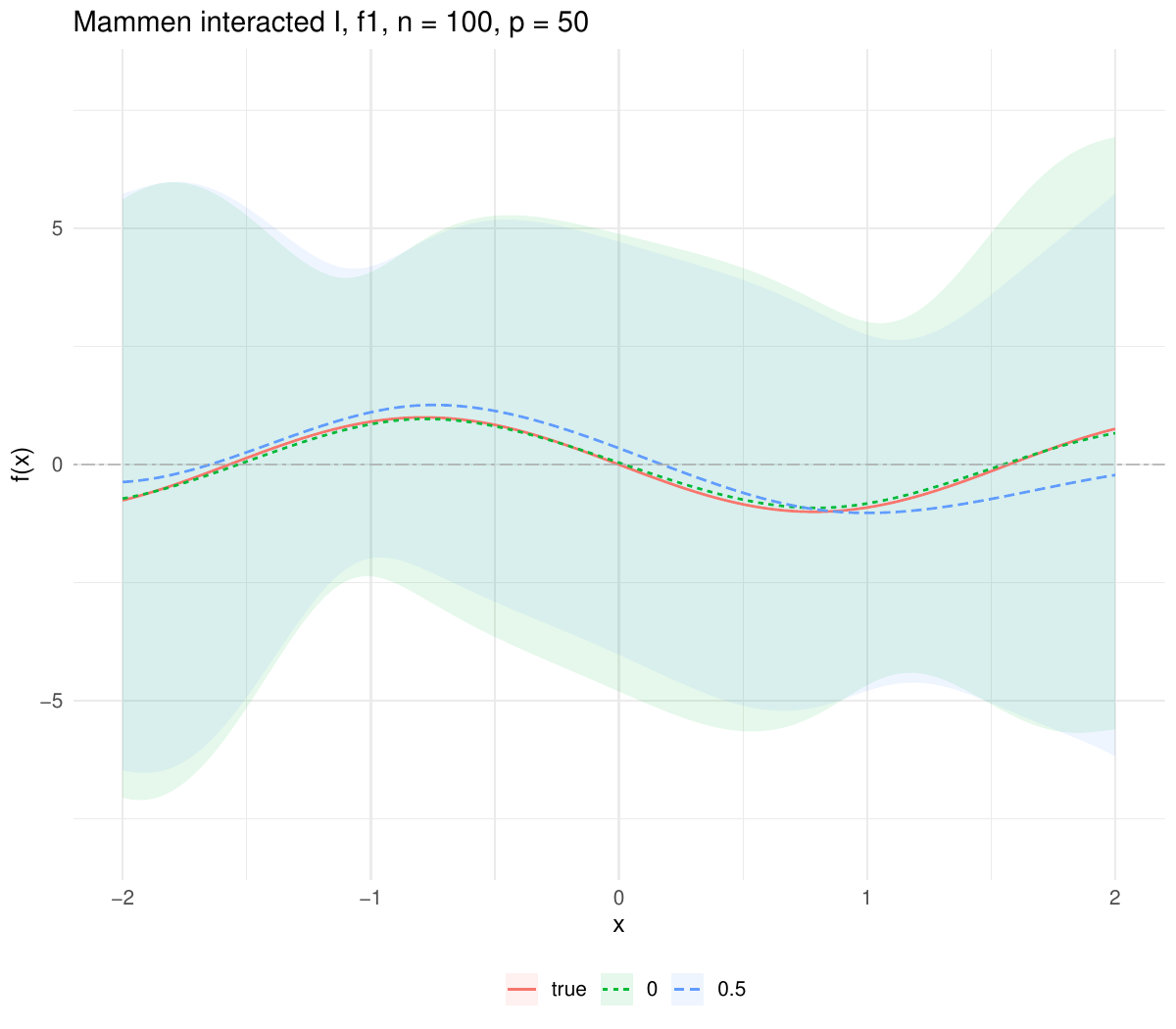} 
\includegraphics[scale=0.25]{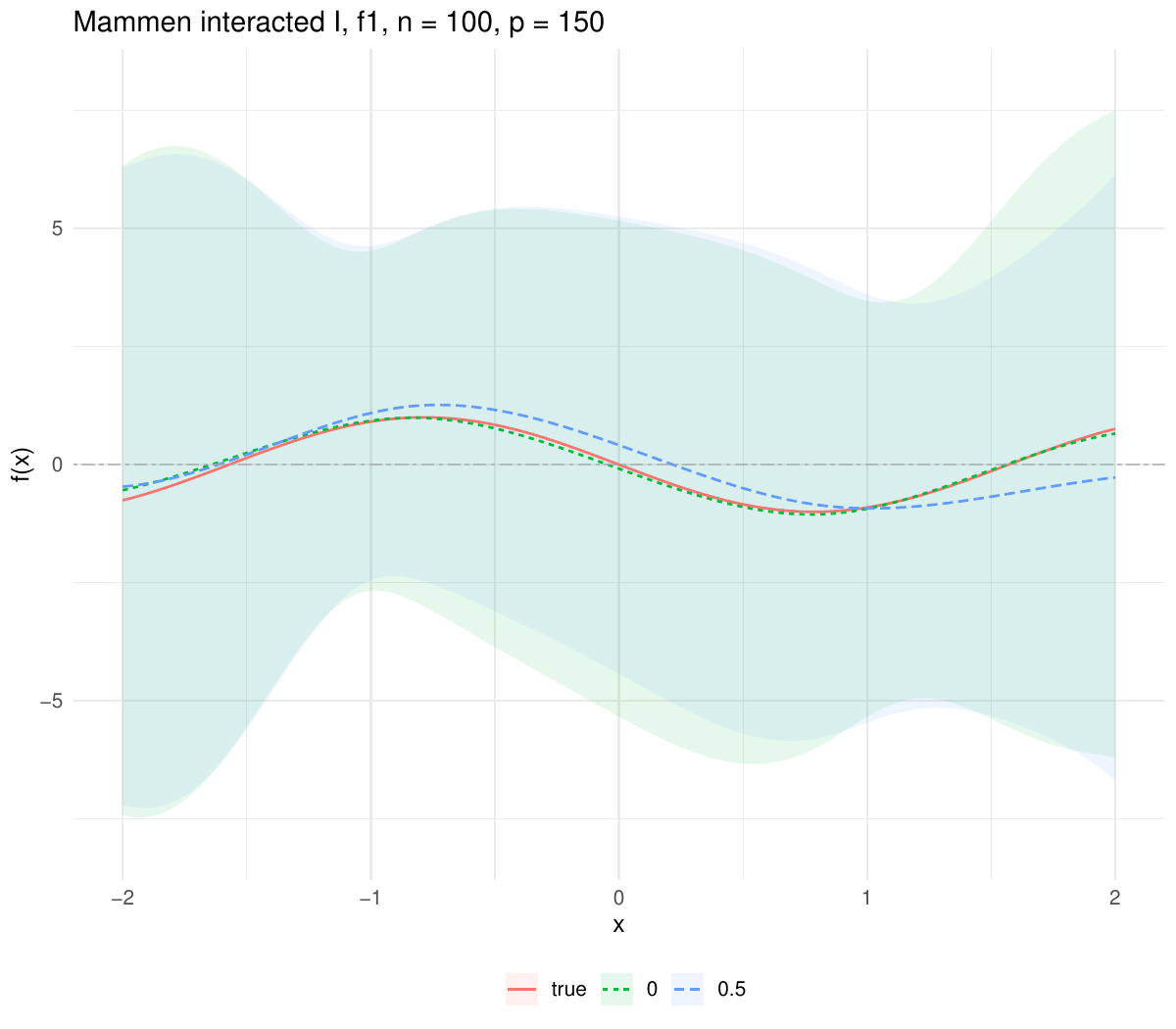} \\
\includegraphics[scale=0.25]{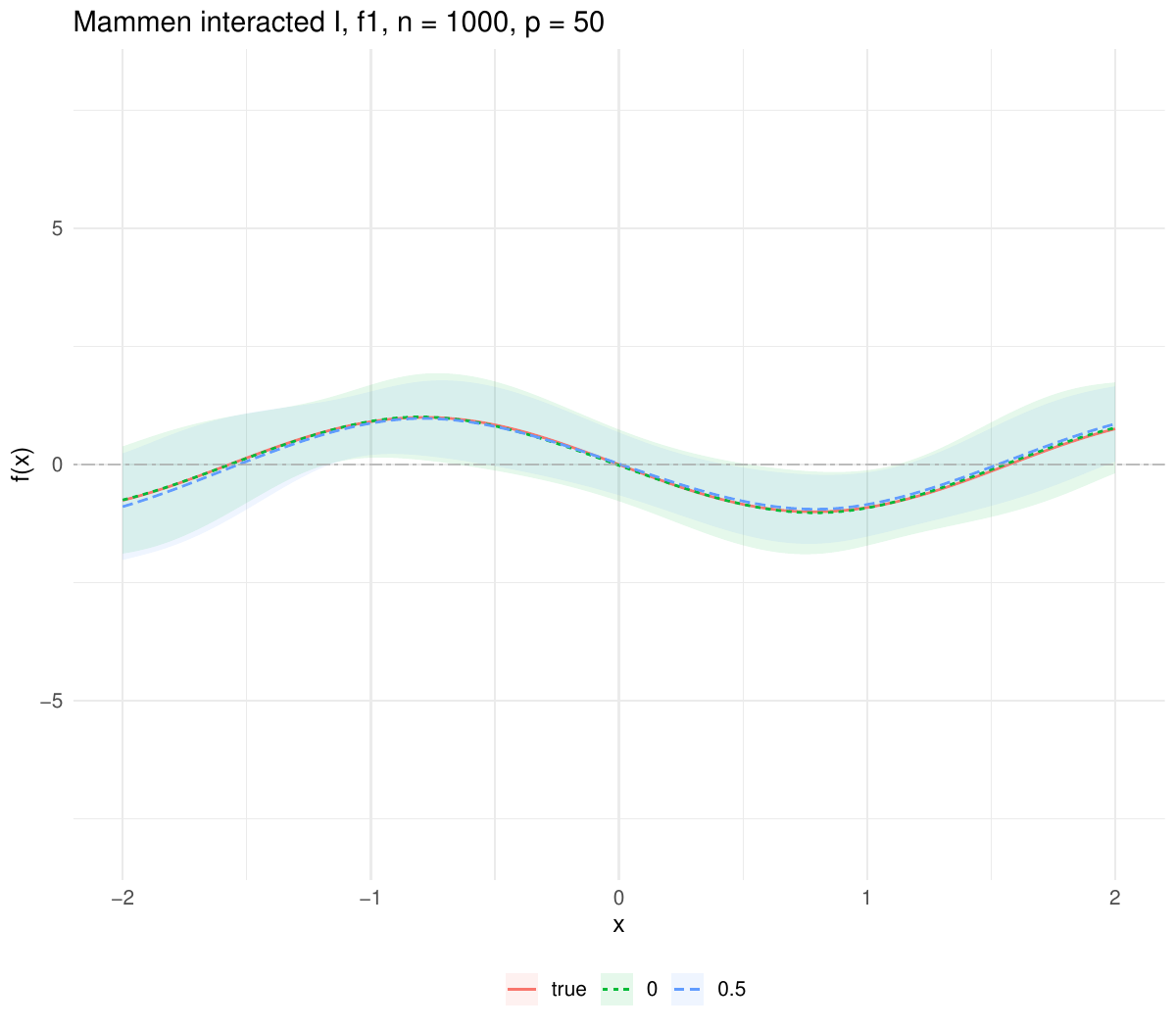}  
\includegraphics[scale=0.25]{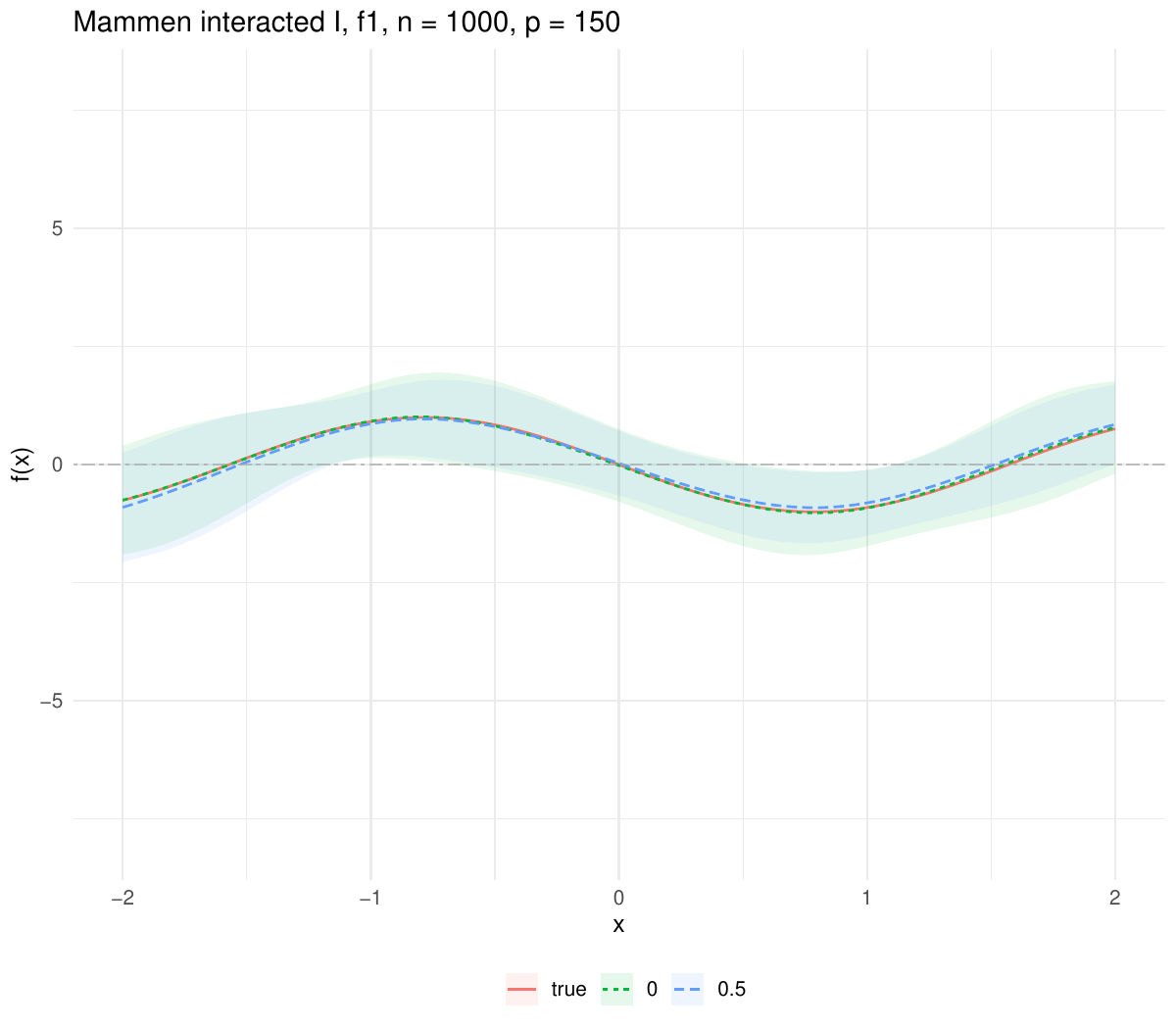} 
\caption{Average confidence bands, $f_1(x_1)$, interacted Scenario $I$ (correctly specified).}
\label{interactedI}
\end{center}
The green dashed curve illustrates averaged estimated functions $\hat{f}_1(x_1)$  as obtained in $R=500$ repetitions in the scenario with $\rho = 0$. The corresponding averaged $95\%$-confidence bands are shaded green. The blue long-dashed line illustrates the results for the setting with $\rho = 0.5$ together with corresponding averaged confidence bands (shaded blue). The true function $f_1(x_1)$ is illustrated by the red solid curve.
\end{figure}

\begin{figure}[t]
\begin{center}
\includegraphics[scale=0.25]{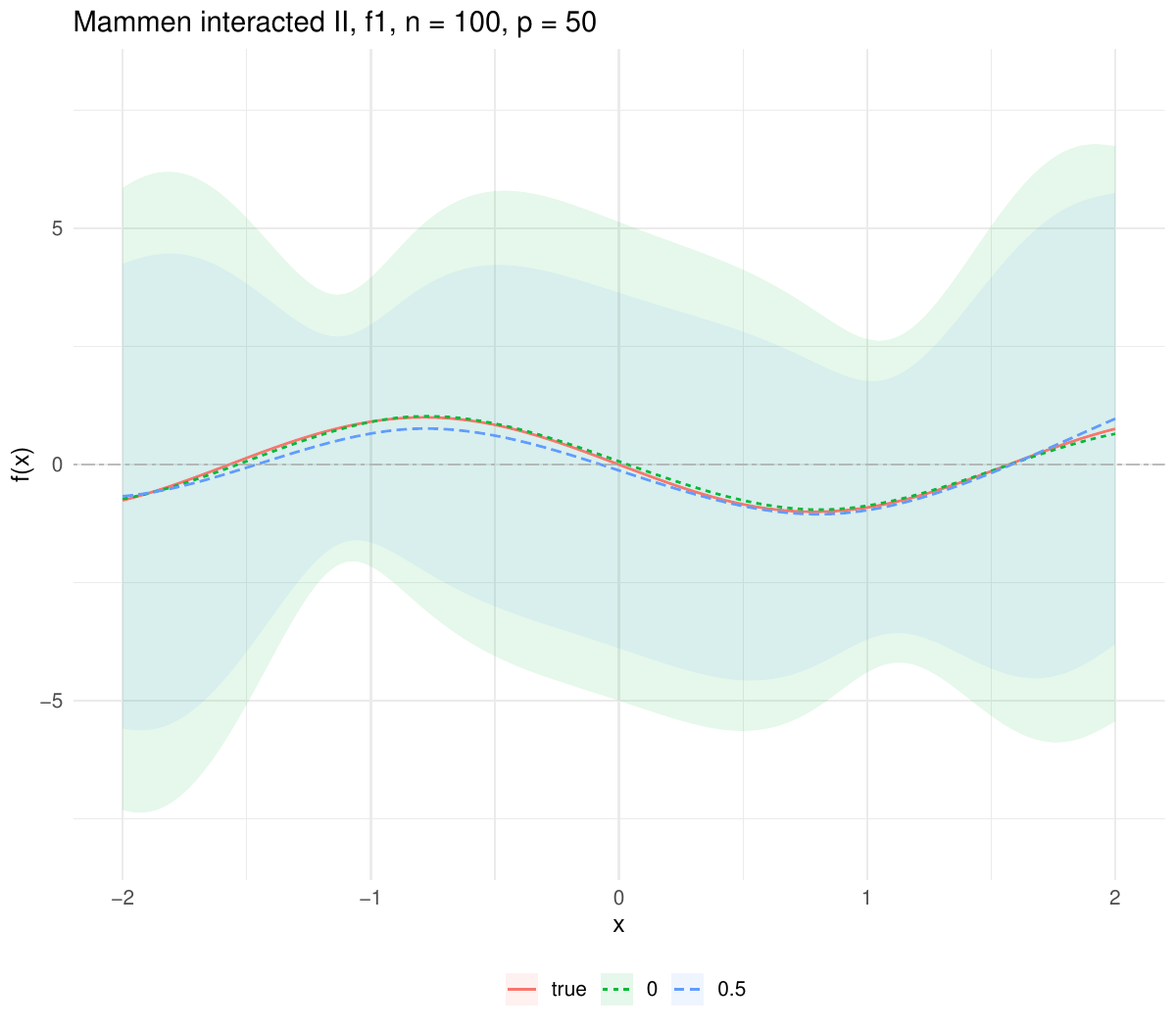} 
\includegraphics[scale=0.25]{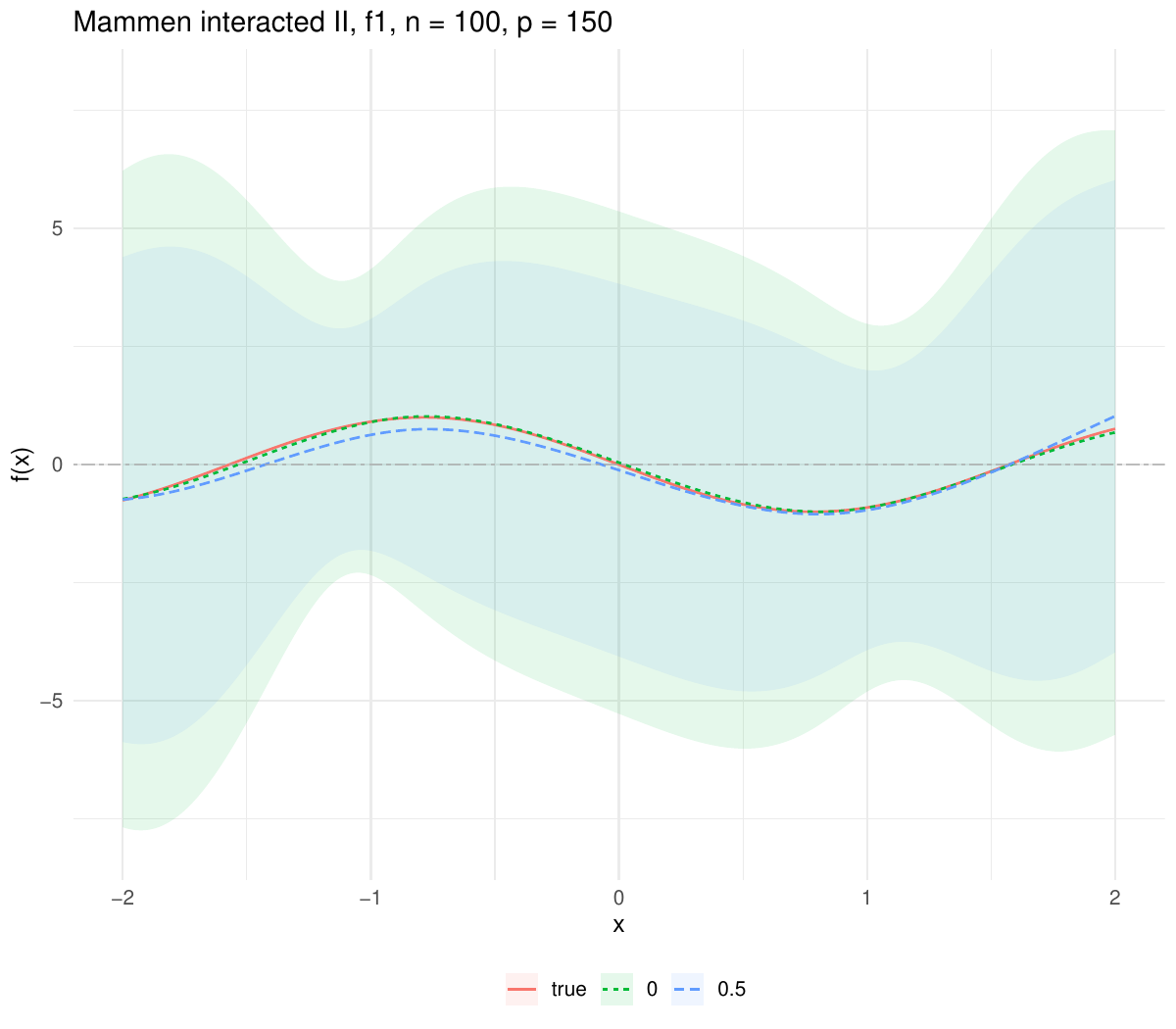} \\
\includegraphics[scale=0.25]{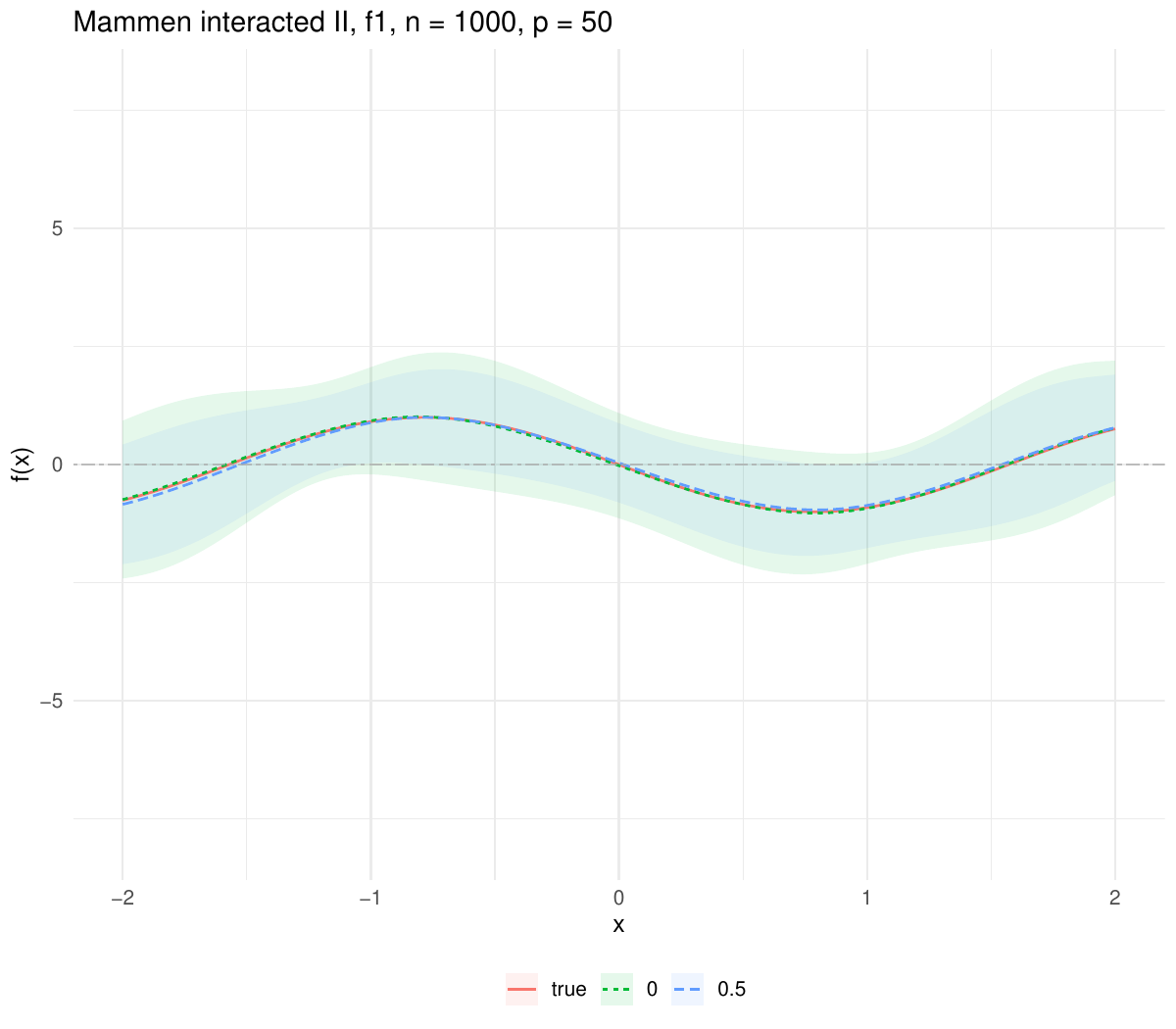}  
\includegraphics[scale=0.25]{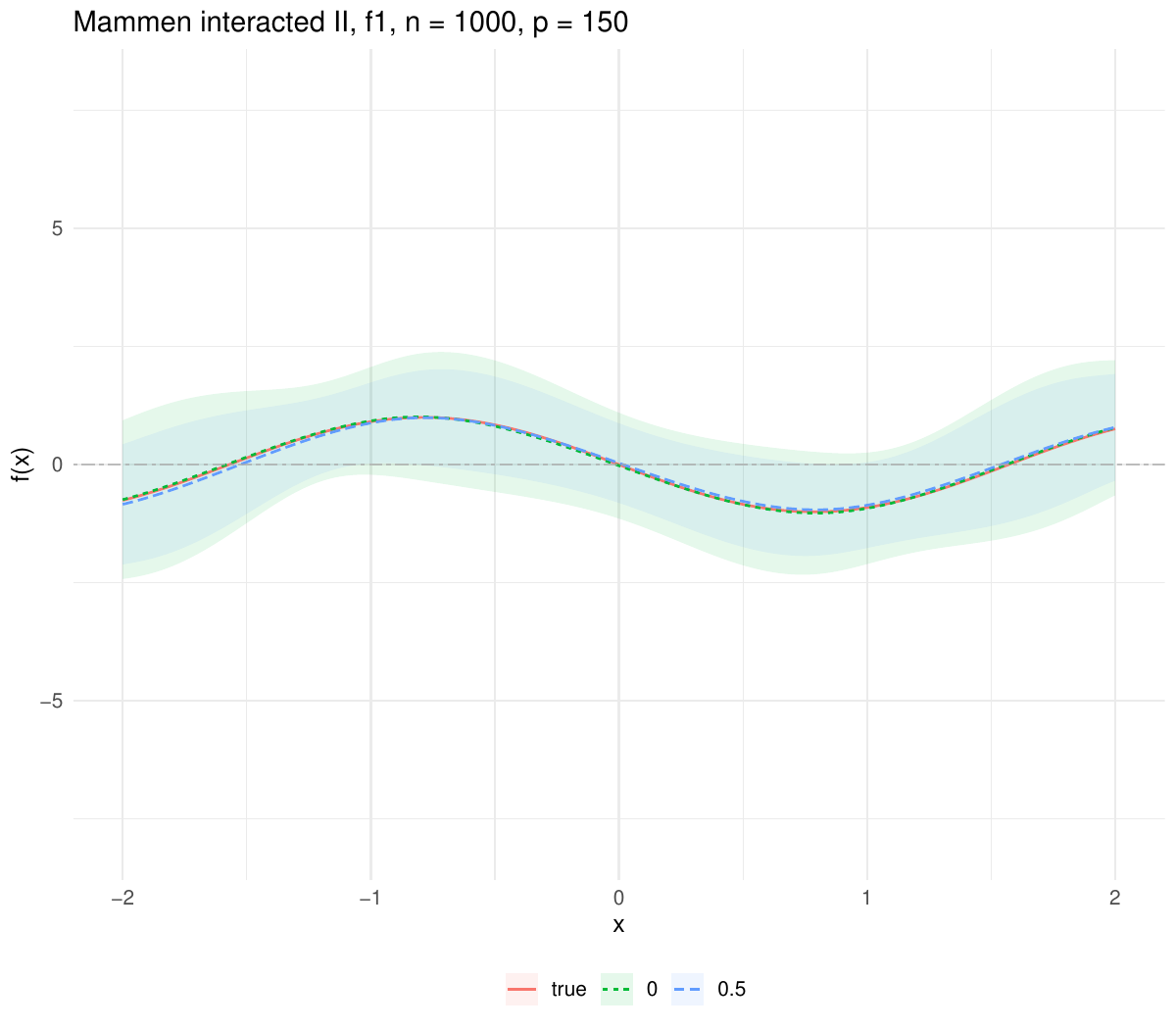} 
\caption{Average confidence bands, $f_1(x_1)$, interacted Scenario $II$ (misspecified).}
\label{interactedII}
\end{center}
The green dashed curve illustrates averaged estimated functions $\hat{f}_1(x_1)$  as obtained in $R=500$ repetitions in the scenario with $\rho = 0$. The corresponding averaged $95\%$-confidence bands are shaded green. The blue long-dashed line illustrates the results for the setting with $\rho = 0.5$ together with corresponding averaged confidence bands (shaded blue). The true function $f_1(x_1)$ is illustrated by the red solid curve.
\end{figure}

\newpage

%% file: Proofs.tex
\section{Proofs}\label{Proofs}
\begin{proof}[Proof of Theorem \ref{confinbands}]\ \\
We will prove that the Assumptions  A.\ref{A1} and  A.\ref{A2} imply the Assumptions B.\ref{assumption2.0}-B.\ref{assumption2.4} stated in Appendix \ref{AppendixA} and then the claim follows by applying Theorem \ref{uniformappendix}. Without loss of generality, we assume $\min(d_1,n)\ge e$ to simplify notation. \\ \\
\textbf{Assumption B.\ref{assumption2.0}}\\
Conditions $(i)$ is directly assumed in A.\ref{A1}$(i)$.
Since the eigenvalues of $\Sigma_{\nu}$ are of order $O(\varsigma_n^{-1})$ due to A.\ref{A2}$(iv)$, it holds
\begin{align*}
    c \varsigma_n^{-1}\le | J_{0,l}| \le C \varsigma_n^{-1}
\end{align*}
uniformly over $l=1,\dots,d_1.$
Further, since $c\le \Var(\varepsilon|X)\le C$ a.s., the eigenvalues of $\Sigma_{\varepsilon\nu}$ are also of order $O(\varsigma_n^{-1})$. Hence,
\begin{align*}
\Sigma_n=J_0^{-1}\Sigma_{\varepsilon\nu}(J_0^{-1})^T
\end{align*}
directly implies B.\ref{assumption2.0}$(ii)$.\\
\textbf{Assumption B.\ref{assumption2.1}}\\
For each $l=1,\dots,d_1$, the moment condition holds
\begin{align*}
\E\left[\psi_l(W,\theta_{0,l},\eta_{0,l})\right]&=\E\big[\varepsilon\nu^{(l)}\big]\\
&=\E\big[\nu^{(l)}\underbrace{\E\left[\varepsilon|X\right]}_{=0}\big]\\
&=0.
\end{align*}
For all $l=1,\dots,d_1$, define the convex set
\begin{align*}
T_l:=\Big\{\eta=(\eta^{(1)},\eta^{(2)},\eta^{(3)})^T:&\eta^{(1)},\eta^{(2)} \in \R^{d_1+d_2-1},\\
&\eta^{(3)}\in\ell^{\infty}(\R^p)\Big\}
\end{align*}
and endow $T_l$ with the norm 
\begin{align*}
\|\eta\|_e:=\max\left\{\varsigma_n^{-1/2}\|\eta^{(1)}\|_2,\|\eta^{(2)}\|_2,\|\eta^{(3)}(X)\|_{P,2})^TZ_{-l}\|_{P,2}\right\}.
\end{align*}
Further, let $\tau_n:=\sqrt{\frac{s\log(\bar{d}_n)}{n}}$ with $s=\max(s_1,s_2)$ and define the corresponding nuisance realization set
\begin{align*}
\mathcal{T}_l:=\bigg\{\eta\in T_l:&\eta^{(3)}\equiv 0, \|\eta^{(1)}\|_0\vee\|\eta^{(2)}\|_0\le Cs, \\
&\varsigma_n^{-1/2}\|\eta^{(1)}-\beta_0^{(l)}\|_2\vee\|\eta^{(2)}-\gamma_0^{(l,1)}\|_2\le C\tau_n,\\
&\varsigma_n^{-1/2}\|\eta^{(1)}-\beta_0^{(l)}\|_1\vee\|\eta^{(2)}-\gamma_0^{(l,1)}\|_1\le C\sqrt{s}\tau_n\bigg\}\cup \{\eta_{0,l}\}
\end{align*}
for a sufficiently large constant $C$. 
Due to A.\ref{A2}$(ii)$ and $(iii)$, it holds
\begin{align*}
    |\nu^{(l)}| &= \left|g_l(X_1) - (\gamma_0^{(l)})^T Z_{-l}\right|\\ 
    &\le |g_l(X_1)| + \|\gamma_0^{(l)}\|_1 \|Z_{-l}\|_{\infty}\\
    &\le C
\end{align*}
almost surely.
For $\mathcal{F}:=\{\varepsilon\nu^{(l)}:l=1,\dots,d_1\}$, it holds
\begin{align*}
\mathcal{S}_n:&=\E\left[\sup\limits_{l=1,\dots,d_1}\left|\sqrt{n}\E_n\left[\psi_l(W,\theta_{0,l},\eta_{0,l})\right]\right|\right]\\
&=\E\left[\sup\limits_{f\in\mathcal{F}}\mathbb{G}_n(f)\right]
\end{align*}
and the envelope $\sup_{f\in\mathcal{F}}|f|$ satisfies
\begin{align*}
\|\max_{l=1,\dots,d_1}\varepsilon\nu^{(l)}\|_{P,q}&\le C\|\varepsilon\|_{P,q}\le C
\end{align*}
for $q$ defined in Assumption A.\ref{A2} $(iii)$.
We can apply Lemma P.2 from \cite{belloni2018uniformly} with $|\mathcal{F}|=d_1$ to obtain
\begin{align*}
\mathcal{S}_n\le C\log^{\frac{1}{2}}(d_1)+ C\log^{\frac{1}{2}}(d_1)\left(n^{\frac{2}{q}}\frac{\log^{}(d_1)}{n}\right)^{1/2}\lesssim \log^{\frac{1}{2}}(d_1),
\end{align*}
due to A.\ref{A2}$(v)$. Finally, Assumption A.\ref{A2}$(i)$ implies B.\ref{assumption2.1}$(i)$. Assumption B.\ref{assumption2.1}$(ii)$ holds since for all $l=1,\dots,d_1$, the map $(\theta_l,\eta_l)\mapsto\psi_l(X,\theta_l,\eta_l)$ is twice continuously Gateaux-differentiable on $\Theta_l\times \mathcal{T}_l$, which directly implies the differentiability  of the map  $(\theta_l,\eta_l)\mapsto\E[\psi_l(X,\theta_l,\eta_l)]$. Additionally, for every $\eta \in\mathcal{T}_l\setminus\{\eta_{0,l}\}$, we have
{\allowdisplaybreaks
\begin{align*}
D_{l,0}[\eta,\eta_{0,l}]&:=\partial_t\big\{\E[\psi_l(W,\theta_{0,l},\eta_{0,l}+t(\eta-\eta_{0,l}))]\big\}\big|_{t=0}\\
&=\E\big[\partial_t\big\{\psi_l(W,\theta_{0,l},\eta_{0,l}+t(\eta-\eta_{0,l}))\big\}\big]\big|_{t=0}\\
&=\E\bigg[\partial_t\bigg\{\Big(Y-\theta_{0,l} g_l(X_1)-\big(\eta^{(1)}_{0,l}+t(\eta^{(1)}-\eta^{(1)}_{0,l})\big)^T Z_{-l}\\
&\quad -\big(\eta^{(3)}_{0,l}(X)+t(\eta^{(3)}(X)-\eta^{(3)}_{0,l}(X))\big)\Big)\\
&\quad\Big(g_l(X_1)-\big(\eta^{(2)}_{0,l}+t(\eta^{(2)}-\eta^{(2)}_{0,l})\big)^T Z_{-l}\Big)\bigg\}\bigg]\bigg|_{t=0}\\
&=\E\left[\varepsilon(\eta^{(2)}_{0,l}-\eta^{(2)})^T Z_{-l}\right]+\E\left[\nu^{(l)}(\eta^{(1)}_{0,l}-\eta^{(1)})^T Z_{-l}\right]\\
&\quad +\E\left[\nu^{(l)}\left(\eta^{(3)}_{0,l}(X)-\eta^{(3)}(X)\right)\right] 
\end{align*}
}
with
\begin{align*}
\E\left[\varepsilon(\eta^{(2)}_{0,l}-\eta^{(2)})^T Z_{-l}\right]=\E\left[(\eta^{(2)}_{0,l}-\eta^{(2)})^T Z_{-l}\E[\varepsilon|X]\right]=0,
\end{align*}
\begin{align*}
\E\left[\nu^{(l)}(\eta^{(1)}_{0,l}-\eta^{(1)})^T Z_{-l}\right]&=(\eta^{(1)}_{0,l}-\eta^{(1)})^T\E\left[ Z_{-l}\nu^{(l)}\right]=0,
\end{align*}

and
\begin{align*}
\E\left[\nu^{(l)}\left(\eta^{(3)}_{0,l}(X)-\eta^{(3)}(X)\right)\right]=\ &\E\left[\nu^{(l)}\big(b_1(X_1)+b_2(X_{-1})\big)\right]\le C\delta_n\varsigma_n^{-1/2} n^{-1/2}
\end{align*}
due to Assumption A.\ref{A1}$(iii)$ with $\delta_n=\sqrt{\frac{n^{2/q}s^2\varsigma_n\log^2(\bar{d}_n)}{n}}$. Due to the linearity of the score and the moment condition, it holds
\begin{align*}
\E[\psi_l(W,\theta_l,\eta_{0,l})]=J_{0,l}(\theta_l-\theta_{0,l})
\end{align*}
and by Assumption A.\ref{A2}$(iv)$
$$c\varsigma_n^{-1}\le |J_{0,l}|=\E\left[(\nu^{(l)})^2\right] \le C\varsigma_n^{-1}$$
Assumption B.\ref{assumption2.1}$(iv)$ is satisfied.\\
For all $t\in[0,1)$, $l=1,\dots,d_1$, $\theta_l\in\Theta_l$ and $\eta_l\in\mathcal{T}_l\setminus\{\eta_{0,l}\}$, we have
\begin{align*}
&\E\left[\left(\psi_l(W,\theta_l,\eta_l)-\psi_l(W,\theta_{0,l},\eta_{0,l})\right)^2\right]\\
=\ &\E\left[\left(\psi_l(W,\theta_l,\eta_l)-\psi_l(W,\theta_{0,l},\eta_l)+\psi_l(W,\theta_{0,l},\eta_l)-\psi_l(W,\theta_{0,l},\eta_{0,l})\right)^2\right]\\
\le\ &C\bigg(\E\left[\left(\psi_l(W,\theta_l,\eta_l)-\psi_l(W,\theta_{0,l},\eta_l)\right)^2\right]\\
&\quad\vee \E\left[\left(\psi_l(W,\theta_{0,l},\eta_l)-\psi_l(W,\theta_{0,l},\eta_{0,l})\right)^2\right]\bigg)
\end{align*}
with 
\begin{align*}
&\E\left[\left(\psi_l(W,\theta_l,\eta_l)-\psi_l(W,\theta_{0,l},\eta_l)\right)^2\right]\\
=\ &|\theta_l-\theta_{0,l}|^2\E\left[\left(g_l(X_1)(g_l(X_1)-(\eta_l^{(2)})^T Z_{-l})\right)^2\right]\\
\le\ &C |\theta_l-\theta_{0,l}|^2 \E\left[\left(\nu^{(l)}+(\eta_{0,l}^{(2)}-\eta_l^{(2)})^T Z_{-l} \right)^2\right]\\
\le\ &C |\theta_l-\theta_{0,l}|^2\left(\|\nu^{(l)}\|^2_{P,2}\vee\|(\eta_{0,l}^{(2)}-\eta_l^{(2)})^T Z_{-l}\|^2_{P,2}\right)\\
\le\ &C |\theta_l-\theta_{0,l}|^2\left(\varsigma_n^{-1}\vee\|\eta_{0,l}^{(2)}-\eta_l^{(2)}\|^2_2\varsigma_n^{-1}\right)\\
\le\ &C \varsigma_n^{-1}|\theta_l-\theta_{0,l}|^2
\end{align*}
due to Assumption A.\ref{A2}$(ii)$, $(iv)$ and the definition of $\mathcal{T}_l$. With similar arguments, we obtain
\begin{align*}
&\E\left[\left(\psi_l(W,\theta_{0,l},\eta_l)-\psi_l(W,\theta_{0,l},\eta_{0,l})\right)^2\right]\\
=\ &\E\Bigg[\bigg(\Big(Y-\theta_{0,l}\g_l(X_1)-(\eta_l^{(1)})^T Z_{-l}-\eta_l^{(3)}(X)\Big)\Big(g_l(X_1)-(\eta_l^{(2)})^T Z_{-l}\Big)\\
&\quad - \Big(Y-\theta_{0,l}\g_l(X_1)-(\eta_{0,l}^{(1)})^T Z_{-l}-\eta_{0,l}^{(3)}(X)\Big)\Big(g_l(X_1)-(\eta_{0,l}^{(2)})^T Z_{-l}\Big)\bigg)^2\Bigg]\\
=\ &\E\Bigg[\bigg(\Big(Y-\theta_{0,l}\g_l(X_1)-(\eta_l^{(1)})^T Z_{-l}-\eta_l^{(3)}(X)\Big)\\
&\quad\cdot\Big((\eta_{0,l}^{(2)}-\eta_l^{(2)})^T Z_{-l}\Big)\\
&\quad + \Big(g_l(X_1)-(\eta_{0,l}^{(2)})^T Z_{-l}\Big)\\
&\quad \cdot\Big((\eta_{0,l}^{(1)}-\eta_l^{(1)})^T Z_{-l}+\eta_{0,l}^{(3)}(X)-\eta_l^{(3)}(X)\Big)\bigg)^2\Bigg]\\
\le\ &C \bigg(\varsigma_n^{-1/2}\|\eta_{0,l}^{(2)}-\eta_l^{(2)}\|_2\vee \varsigma_n^{-1/2}\|\eta_{0,l}^{(1)}-\eta_l^{(1)}\|_2\vee\|\eta^{(3)}_{0,l}(X)\|_{P,2}\bigg)^2\\
\le\ &C\|\eta_{0,l}-\eta_l\|_e^2,
\end{align*}
where we used the definition of $\mathcal{T}_l$, A.\ref{A2}$(ii)$ and
\begin{align*}
\sup_{\|\xi\|_2=1}\E[(\xi^T Z)^2]\le C\varsigma_n^{-1}.
\end{align*}
Therefore, Assumption B.\ref{assumption2.1}$(v)(a)$ holds since it is straightforward to show Assumption B.\ref{assumption2.1}$(v)(a)$ for $\eta_l = \eta_{0,l}$. It holds
{\allowdisplaybreaks
\begin{align*}
&\bigg|\partial_t\E\Big[\psi_l(W,\theta_l,\eta_{0,l}+t(\eta_l-\eta_{0,l}))\Big]\bigg|\\
=\ &\bigg|\E\bigg[\partial_t\bigg\{\Big(Y-\theta_{0,l} g_l(X_1)-\big(\eta^{(1)}_{0,l}+t(\eta_l^{(1)}-\eta^{(1)}_{0,l})\big)^T Z_{-l}\\
&\quad -\big(\eta^{(3)}_{0,l}(X)+t(\eta_l^{(3)}(X)-\eta^{(3)}_{0,l}(X))\big)\Big)\\
&\quad\cdot\Big(g_l(X_1)-\big(\eta^{(2)}_{0,l}+t(\eta_l^{(2)}-\eta^{(2)}_{0,l})\big)^T Z_{-l}\Big)\bigg\}\bigg]\bigg|\\
=\ &\bigg|\E\bigg[\Big(Y-\theta_{0,l} g_l(X_1)-(\eta^{(1)}_{0,l}+t(\eta_l^{(1)}-\eta^{(1)}_{0,l}))^T Z_{-l}\\
&\quad -(\eta^{(3)}_{0,l}(X)+t(\eta_l^{(3)}(X)-\eta^{(3)}_{0,l}(X)))\Big)\\
&\quad \cdot\Big((\eta^{(2)}_{0,l}-\eta_l^{(2)})^T Z_{-l} \Big)\\
&\quad +\Big(g_l(X_1)-(\eta^{(2)}_{0,l}+t(\eta_l^{(2)}-\eta^{(2)}_{0,l}))^T Z_{-l}\Big)\\
&\quad\cdot\Big((\eta^{(1)}_{0,l}-\eta_l^{(1)})^T Z_{-l}+\eta^{(3)}_{0,l}(X)-\eta_l^{(3)}(X)\Big)\bigg]\bigg|\\
=\ &| I_{1,1} + I_{1,2} + I_{1,3}+ I_{1,4}|
\end{align*}
}
with
\begin{align*}
I_{1,1}&=\E\bigg[\Big(Y-\theta_{0,l} g_l(X_1)-(\eta^{(1)}_{0,l}+t(\eta_l^{(1)}-\eta^{(1)}_{0,l}))^T Z_{-l}\\
&\quad -(\eta^{(3)}_{0,l}(X)+t(\eta_l^{(3)}(X)-\eta^{(3)}_{0,l}(X)))\Big)\Big((\eta^{(2)}_{0,l}-\eta_l^{(2)}))^T Z_{-l}\Big)\bigg]\\
&\le C \varsigma_n^{-1}\|\eta^{(2)}_{0,l}-\eta_l^{(2)}\|_2,\\
I_{1,2}&=\E\bigg[\Big(Y-\theta_{0,l} g_l(X_1)-(\eta^{(1)}_{0,l}+t(\eta_l^{(1)}-\eta^{(1)}_{0,l}))^T Z_{-l}\\
&\quad -(\eta^{(3)}_{0,l}(X)+t(\eta_l^{(3)}(X)-\eta^{(3)}_{0,l}(X)))\Big)\bigg]\\
&=-t\E\bigg[\Big((\eta_l^{(1)}-\eta^{(1)}_{0,l})^T Z_{-l}+(\eta_l^{(3)}(X)-\eta^{(3)}_{0,l}(X))\Big)\bigg]\\
&\le C\varsigma_n^{-1/2},\\
I_{1,3}&=\E\bigg[\Big(g_l(X_1)-(\eta^{(2)}_{0,l}+t(\eta_l^{(2)}-\eta^{(2)}_{0,l}))^T Z_{-l}\Big)\Big((\eta^{(1)}_{0,l}-\eta_l^{(1)})^T Z_{-l}\Big)\bigg]\\
&\le C \varsigma_n^{-1}\|\eta^{(1)}_{0,l}-\eta_l^{(1)}\|_2
\end{align*}
since $\E[\varepsilon|X]=0$ and $\E[\nu^{(l)}Z_{-l}]=0$. Further, it holds
\begin{align*}
I_{1,4}&=\E\bigg[\Big(g_l(X_1)-(\eta^{(2)}_{0,l}+t(\eta_l^{(2)}-\eta^{(2)}_{0,l}))^T Z_{-l}\Big)\Big(\eta^{(3)}_{0,l}(X)\Big)\bigg]\\
&\le C\varsigma_n^{-1/2}\|\eta^{(3)}_{0,l}(X)\|_{P,2}
\end{align*}
since
\begin{align*}
    &\quad\E\bigg[\Big(g_l(X_1)-(\eta^{(2)}_{0,l}+t(\eta_l^{(2)}-\eta^{(2)}_{0,l}))^T Z_{-l}\Big)^2\bigg]^{1/2}\\
    &\le \E[(\nu^{(l)})^2]^{1/2}+\E\left[\left((\eta_l^{(2)}-\eta^{(2)}_{0,l})^T Z_{-l}\right)^2\right]^{1/2}\\
    &\le C\varsigma_n^{-1/2}\left(1+\|\eta_l^{(2)}-\eta^{(2)}_{0,l}\|_2^2\right)^{1/2}.
\end{align*}
This implies Assumption B.\ref{assumption2.1}$(v)(b)$. Finally, to obtain Assumption B.\ref{assumption2.1}$(v)(c)$, we note that
{\allowdisplaybreaks
\begin{align*}
&\partial_t^2\E\left[\psi_l(W,\theta_{0,l}+t(\theta_l-\theta_{0,l}),\eta_{0,l}+t(\eta_l-\eta_{0,l}))\right]\\
=\ &\partial_t\E\bigg[\Big(Y-\big(\theta_{0,l}+t(\theta_l-\theta_{0,l}) \big)g_l(X_1)-\big(\eta^{(1)}_{0,l}+t(\eta_l^{(1)}-\eta^{(1)}_{0,l})\big)^T Z_{-l}\\
&\quad -\big(\eta^{(3)}_{0,l}(X)+t(\eta_l^{(3)}(X)-\eta^{(3)}_{0,l}(X))\big)\Big)\\
&\quad \cdot\Big((\eta^{(2)}_{0,l}-\eta_l^{(2)}))^T Z_{-l}\Big)\\
&\quad +\Big(g_l(X_1)-\big(\eta^{(2)}_{0,l}+t(\eta_l^{(2)}-\eta^{(2)}_{0,l})\big)^T Z_{-l}\Big)\\
&\quad\cdot\Big((\theta_{0,l}-\theta_{l}) )g_l(X_1)+(\eta^{(1)}_{0,l}-\eta_l^{(1)})^T Z_{-l}+\eta^{(3)}_{0,l}(X)\Big)\bigg]\\
=\ &2\E\bigg[\Big((\theta_{0,l}-\theta_{l}) g_l(X_1)+(\eta_{0,l}^{(1)}-\eta^{(1)}_{l})^T Z_{-l} + \eta_{0,l}^{(3)}(X)\Big)\\
&\quad\cdot\Big((\eta^{(2)}_{0,l}-\eta_l^{(2)}))^T Z_{-l}\Big)\bigg]\\
\le &2\E\bigg[\Big((\theta_{0,l}-\theta_{l}) g_l(X_1)+(\eta_{0,l}^{(1)}-\eta^{(1)}_{l})^T Z_{-l} + \eta_{0,l}^{(3)}(X)\Big)^2\bigg]\\
&\quad\vee 2\E\bigg[\Big((\eta^{(2)}_{0,l}-\eta_l^{(2)}))^T Z_{-l}\Big)^2\bigg]\\
\le &C \left(|\theta_{0,l}-\theta_l|^2\varsigma_n^{-1}\vee \|\eta_{0,l}-\eta_l\|_e^2\right)
\end{align*}
}
using the same arguments as above.\\ \\
\textbf{Assumption B.\ref{assumption2.2}}\\
Note that the Assumptions B.\ref{assumption2.2}$(ii)$ and $(iii)$ both hold by the construction of $\mathcal{T}_l$ and the Assumptions A.\ref{A1}$(ii)$ and A.\ref{A2}$(ii)$. The main part to verify Assumption B.\ref{assumption2.2} is to show that the estimates of the nuisance function are contained in the nuisance realization set with high probability. We will rely on uniform lasso estimation results stated in Appendix \ref{uniformestimation}. Consider a scaled version of the auxiliary regression
\begin{align*}
    \tilde{Z}_l = (\gamma_0^{(l)})^T \tilde{Z}_{-l} + \tilde{\nu}^{(l)}
\end{align*}
with $\tilde{Z}_l:= \varsigma_n^{1/2} Z_l$ and $ \tilde{\nu}^{(l)}:=  \varsigma_n^{1/2} \nu^{(l)}$. First, we want to proof that $\hat{\eta}^{(2)}_{l}\in\mathcal{T}_l$ for all $l=1,\dots,d_1$. We have to check the Assumptions C.\ref{condC}$(i)$ to $(iv)$. Due to Assumption A.\ref{A2}$(iii)$, the first part of Assumption C.\ref{condC}$(i)$ is satisfied with $M_n = \varsigma_n^{1/2},$ since we have already shown $|\nu^{(l)}|\le C$ almost surely uniformly over $l = 1,\dots, d_1$. Further, 
$$c\le\E\left[(\tilde{\nu}^{(l)})^2 \tilde{Z}_{-l,j}^2\right]\le C,$$
uniformly for all  $l=1,\dots,d_1$, $j = 1,\dots, d_1 -1$ and
\begin{align*}
\max\limits_{l=1,\dots,d_1}\max\limits_{j=1,\dots,d_1 - 1}\|\tilde{\nu}^{(l)} \tilde{Z}_{-l,j}\|_{P,3}&\le CK_{n},\\
    \max\limits_{j=1,\dots,d_1}\max\limits_{j=1,\dots,d_1 -1} \E\left[(\tilde{\nu}^{(l)})^4 \tilde{Z}_{-l,j}^4\right]&\le C L_n
\end{align*}
due to Assumption A.\ref{A2}$(iv)$. Assumption A.\ref{A2}$(iv)$ also directly implies Assumption C.\ref{condC}$(ii)$. Combining Assumption C.\ref{condC}$(ii)$ with Assumption A.\ref{A2}$(ii)$ implies Assumption C.\ref{condC}$(iii)$, since
\begin{align*}
\max\limits_{l=1,\dots,d_1}\|\gamma_0^{(l,2)}\|_{2}^2\le C\varsigma_n \max\limits_{l=1,\dots,d_1}\|(\gamma_0^{(l,2)})^T Z_{-l}\|_{P,2}^2 \le C s_2\log(\bar{d}_n)/n.
\end{align*}
The growth conditions C.\ref{condC}$(iv)$ are included in Assumption A.\ref{A2}$(v)$.
Therefore,
$$\hat{\eta}^{(2)}_{l}\in \mathcal{T}_l\quad\text{for all } l=1,\dots,d_1$$ with probability $1-o(1)$. To estimate $\eta^{(1)}_{0,l}$, we run a lasso regression of $Y$ on $\tilde{Z}$:
\begin{align*}
    Y = (\varsigma_n^{-1/2}\theta_0)^T (\varsigma_n^{1/2} g(X_1)) + (\varsigma_n^{-1/2}\beta_0)^T (\varsigma_n^{1/2} h(X_{-1})) + b_1(X_1) + b_{-1}(X_{-1}) + \varepsilon
\end{align*}
Define the corresponding nuisance parameter
$$\tilde{\beta}_0^{(l)}:=\varsigma_n^{-1/2}\beta_0^{(l)} = \varsigma_n^{-1/2}(\theta_{0,1},\dots,\theta_{0,l-1},\theta_{0,l+1},\dots\theta_{0,d_1},\beta_{0,1},\dots,\beta_{0,d_2})^T.$$
With analogous arguments as in the proof of Theorem \ref{uniformlasso}, it holds
\begin{align*}
\|\tilde{\beta}_0^{(l)}-\hat{\tilde{\beta}}^{(l)}\|_0&\le Cs_1,\\
\|\tilde{\beta}_0^{(l)}-\hat{\tilde{\beta}}^{(l)}\|_2&\le C\sqrt{\frac{s_1\log(\bar{d}_n)}{n}},\\
\|\tilde{\beta}_0^{(l)}-\hat{\tilde{\beta}}^{(l)}\|_1&\le C\sqrt{\frac{s_1^2\log(\bar{d}_n)}{n}}
\end{align*}
with probability $1-o(1)$ using Assumptions A.\ref{A1}$(ii)$, A.\ref{A2}$(ii)$-$(v)$ and
\begin{align*}
c\le\E\big[\varepsilon^2 \tilde{Z}_{l}^2\big]&= \E\big[\tilde{Z}_{l}^2\underbrace{\E[\varepsilon^2|X]}_{=\Var(\varepsilon|X)} \big]\le C.
\end{align*}
This directly implies that with probability $1-o(1)$ the nuisance realization set $\mathcal{T}_l$ contains $\hat{\eta}^{(1)}_{l}$ for all $l=1,\dots,d_1$.\\
Combining the results above with $\hat{\eta}^{(3)}\equiv 0$, we obtain Assumption B.\ref{assumption2.2}$(i)$. Define
\begin{align*}
\mathcal{F}_1:=\big\{\psi_l(\cdot,\theta_l,\eta_l):l=1,\dots,d_1,\theta_l\in\Theta_l,\eta_l\in\mathcal{T}_l\big\}.
\end{align*}
To bound the complexity of $\mathcal{F}_1$, we exclude the true nuisance function (the true nuisance function is the only element of $\mathcal{T}_l$ with a nonzero approximation error):
\begin{align*}
\mathcal{F}_{1,1}:=\big\{\psi_l(\cdot,\theta_l,\eta_l):l=1,\dots,d_1,\theta_l\in\Theta_l,\eta_l\in\mathcal{T}_l\setminus\{\eta_0^{(l)}\}\big\}\subseteq \mathcal{F}_{1,1}^{(1)}\mathcal{F}_{1,1}^{(2)}
\end{align*}
with 
\begin{align*}
\mathcal{F}_{1,1}^{(1)}&:=\big\{W\mapsto Y-\theta_l g_l(X_1)-(\eta^{(1)}_l)^T Z_{-l}:l=1,\dots,d_1,\theta_l\in\Theta_l,\eta_l\in\mathcal{T}_l\setminus\{\eta_0^{(l)}\}\big\}\\
\mathcal{F}_{1,1}^{(2)}&:=\big\{W\mapsto g_l(X_1)-(\eta^{(2)}_l)^T Z_{-l}:l=1,\dots,d_1,\theta_l\in\Theta_l,\eta_l\in\mathcal{T}_l\setminus\{\eta_0^{(l)}\}\big\}.
\end{align*}
Note that the envelope $F_{1,1}^{(1)}$ of $\mathcal{F}_{1,1}^{(1)}$ satisfies
\begin{align*}
\|F_{1,1}^{(1)}\|_{P,2q}&\le \bigg\|\sup_{l=1,\dots,d_1}\sup_{\theta_l\in\Theta_l,\|\eta_{0,l}^{(1)}-\eta_l^{(l)}\|_1\le C\sqrt{s}\tau_n}\Big(|\varepsilon| + |\eta_{0}^{(3)}(X)|\\
&\quad + |(\theta_{0,l}-\theta_l)g_l(X_1)|+|(\eta_{0,l}^{(1)}-\eta_l^{(1)})^T Z_{-l}|\Big)\bigg\|_{P,2q}\\
&\lesssim \|\varepsilon\|_{P,2q}+\| \eta_{0}^{(3)}(X)\|_{P,2q}+ \|\sup_{l=1,\dots,d_1}g_l(X_1)\|_{P,2q}\\
&\quad + \varsigma_n^{-1/2}\sqrt{s_1}\tau_n\|\sup_{j=1,\dots,d_1+d_2} Z_j\|_{P,2q}\\
&\lesssim C + \varsigma_n^{-1/2}\sqrt{s_1}\tau_n
\end{align*}
due to  A.\ref{A1}$(ii)$, A.\ref{A2}$(iii)$ and analogously 
\begin{align*}
\|F_{1,1}^{(2)}\|_{P,2q}\lesssim  C + \sqrt{s_2}\tau_n.
\end{align*}
Next, note that due to Lemma 2.6.15 from \cite{vanweak} the set
\begin{align*}
\mathcal{G}_{1,1}:=\big\{Z\mapsto \xi^T Z: \xi\in \R^{d_1+d_2+1},\|\xi\|_0\le Cs,\|\xi\|_2\le C\big\}
\end{align*}
is a union over $\binom{d_1+d_2+1}{Cs}$ VC-subgraph classes $\mathcal{G}_{1,1,k}$ with VC indices less or equal to $Cs+2$. Therefore, $\mathcal{F}_{1,1}^{(1)}$ and $\mathcal{F}_{1,1}^{(2)}$ are unions over $\binom{d_1+d_2+1}{Cs}$ respectively  $\binom{d_1+d_2}{Cs}$ VC-subgraph classes, which combined with Theorem 2.6.7 from \cite{vanweak} implies
\begin{align*}
\sup_Q\log N(\varepsilon\|F_{1,1}^{(1)}\|_{Q,2},\mathcal{F}_{1,1}^{(1)},\|\cdot\|_{Q,2})\lesssim s_1\log\left(\frac{d_1+d_2}{\varepsilon}\right)
\end{align*}
and 
\begin{align*}
\sup_Q\log N(\varepsilon\|F_{1,1}^{(2)}\|_{Q,2},\mathcal{F}_{1,1}^{(2)},\|\cdot\|_{Q,2})\lesssim s_2\log\left(\frac{d_1+d_2}{\varepsilon}\right).
\end{align*}
Using basic calculations, we obtain
\begin{align*}
\sup_Q\log N(\varepsilon\|F_{1,1}\|_{Q,2},\mathcal{F}_{1,1},\|\cdot\|_{Q,2})\lesssim s\log\left(\frac{d_1+d_2}{\varepsilon}\right),
\end{align*}
where $F_{1,1}:=F_{1,1}^{(1)}F_{1,1}^{(2)}$ is an envelope for $\mathcal{F}_{1,1}$ with 
\begin{align*}
\|F_{1,1}\|_{P,q}\le \|F_{1,1}^{(1)}\|_{P,2q}\|F_{1,1}^{(2)}\|_{P,2q}\lesssim (C+\varsigma_n^{-1/2}\sqrt{s_1\tau_n})(C+\sqrt{s_2\tau_n})\lesssim C.
\end{align*}
Define
$$\mathcal{F}_{1,2}:=\big\{\psi_l(\cdot,\theta_l,\eta_{0,l}):l=1,\dots,d_1,\theta_l\in\Theta_l\big\}$$
and, with an analogous argument, we obtain
\begin{align*}
\sup_Q\log N(\varepsilon\|F_{1,2}\|_{Q,2},\mathcal{F}_{1,2},\|\cdot\|_{Q,2})\lesssim \log\left(\frac{d_1}{\varepsilon}\right),
\end{align*}
where the envelope $F_{1,2}$ of $\mathcal{F}_{1,2}$ obeys
\begin{align*}
\|F_{1,2}\|_{P,q}\lesssim  C.
\end{align*}
Combining the results above, we obtain
\begin{align*}
\sup_Q\log N(\varepsilon\|F_{1}\|_{Q,2},\mathcal{F}_{1},\|\cdot\|_{Q,2})\lesssim s\log\left(\frac{d_1+d_2}{\varepsilon}\right),
\end{align*}
where the envelope $F_1:=F_{1,1}^{(1)}F_{1,1}^{(2)}\vee F_{1,2}$ of $\mathcal{F}_1$ satisfies
\begin{align*}
\|F_{1}\|_{P,q}\lesssim  C.
\end{align*}
Therefore, Assumption B.\ref{assumption2.2}$(iv)$ holds with $\upsilon_n\lesssim s$ and $a_n = d_1\vee d_2$.\\
At first remark that for each $l=1,\dots,d_1$, we have
\begin{align*}
    c&\le \Var(\varepsilon|X)\\
    &\le \E\left[\big(Y-\theta_lg_l(X_1)-(\eta_l^{(1)})^T Z_{-l} -\eta^{(3)}(X)\big)^2|X\right]\\
    &= \E\left[\big(\varepsilon+ (\theta_{0,l}-\theta_l)g_l(X_1)+(\eta_{0,l}^{(1)} -\eta_l^{(1)})^T Z_{-l} +\eta_{0}^{(3)}(X)\big)^2|X\right]\\
    &= \underbrace{\E\left[\varepsilon^2|X\right]}_{=\Var(\varepsilon|X)\le C}+\E\bigg[\big(\underbrace{(\theta_{0,l}-\theta_l)g_l(X_1)}_{\le C}+\underbrace{(\eta_{0,l}^{(1)} -\eta_l^{(1)})^T Z_{-l}}_{\le \|\eta_{0,l}^{(1)} -\eta_l^{(1)}\|_1 C\le C} + \underbrace{\eta_{0}^{(3)}(X)}_{\le C}\big)^2|X\bigg]\\
    &\le C,
\end{align*}
due to $c\le Var(\epsilon|X)\le C$ and the growth rates.
For all $f\in\mathcal{F}_1$, we have
\begin{align*}
&\E\left[f^2\right]^{\frac{1}{2}}\\
=&\ \E\left[\big(Y-\theta_lg_l(X_1)-(\eta^{(1)})^T Z_{-l} -\eta^{(3)}(X)\big)^2\big(g_l(X_1)-(\eta^{(2)})^T Z_{-l}\big)^2\right]^{\frac{1}{2}}\\
=&\ \E\Big[\big(g_l(X_1)-(\eta^{(2)})^T Z_{-l}\big)^2\\
&\quad\cdot\E\left[\big(Y-\theta_lg_l(X_1)-(\eta^{(1)})^T Z_{-l} -\eta^{(3)}(X)\big)^2|X\right]\Big]^{\frac{1}{2}}\\
\end{align*}
such that
\begin{align*}
    c\varsigma_n^{-\frac{1}{2}}\le \E[f^2]^{\frac{1}{2}}\le C\varsigma_n^{-\frac{1}{2}},
\end{align*}
due to Assumption A.\ref{A2}$(iv)$. This corresponds to Assumption B.\ref{assumption2.2}$(v)$. Assumption B.\ref{assumption2.2}$(vi)(a)$ holds by the definition of $\tau_n$ and $\upsilon_n\lesssim s$. To verify the next growth condition, we note
\begin{align*}
&(\tau_n+\mathcal{S}_n\log(n)/\sqrt{n})(\upsilon_n\log(\bar{d}_n))^{1/2}+n^{-1/2+1/q}\upsilon_n \log(\bar{d}_n)\\
\lesssim\ &(\tau_n+\log^\frac{1}{2}(d_1)\log(n)/\sqrt{n})(s\log(\bar{d}_n))^{1/2}+n^{-1/2+1/q}s\log(\bar{d}_n)\\
\lesssim\ &\left(n^{\frac{2}{q}}\frac{s^2\log^{2}(\bar{d}_n)}{n}\right)^{\frac{1}{2}}\\
\lesssim\ &\delta_n\varsigma_n^{-1/2}
\end{align*}
for a given $q\ge 4$ with $$\delta_n =\sqrt{\frac{n^{2/q}s^2\varsigma_n\log^2(\bar{d}_n)}{n}}=o(1)$$ due to Assumption A.\ref{A2}$(v)$ and analogously
\begin{align*}n^{1/2}\tau_n^2=\frac{s\log(\bar{d}_n)}{\sqrt{n}}\lesssim\delta_n\varsigma_n^{-1/2}.
\end{align*}
\\ \\
\textbf{Assumption B.\ref{assumption2.3}$(i)-(ii)$ }\\
Define
\begin{align*}
\mathcal{F}_0:=\{\psi_x(\cdot):x\in I\},
\end{align*}
where $\psi_x(\cdot):=(g(x)^T\Sigma_n g(x))^{-1/2}g(x)^TJ_0^{-1}\psi(\cdot,\theta_{0},\eta_{0})$. It is easy to verify that B.\ref{assumption2.3}$(ii)$ holds with
$$L_n= t_1^{9/2}\varsigma_n^{3/2}.$$  
Using the same argument, we can conclude that the envelope $F_0$ of $\F_0$ satisfies
\begin{align*}
\|F_0\|_{P,q}&=\E\left[\sup_{x\in I}\left|(g(x)^T\Sigma_n g(x))^{-1/2}g(x)^TJ_0^{-1}\psi(W,\theta_{0},\eta_{0})\right|^q\right]^{\frac{1}{q}}\\
&\lesssim t_1^{1/2}\varsigma_n^{-1/2} \E\left[\sup_{x\in I}\left|g(x)^TJ_0^{-1}\psi(W,\theta_{0},\eta_{0})\right|^q\right]^{\frac{1}{q}}\\
&= t_1^{1/2}\varsigma_n^{-1/2}\E\left[\sup_{x\in I}\left|\sum_{l=1}^{d_1} g_l(x)J_{0,l}^{-1}\psi_l(W,\theta_{0,l},\eta_{0,l})\right|^q\right]^{\frac{1}{q}}\\
&\lesssim t_1^{3/2}\varsigma_n^{1/2} \E\left[\sup_{l=1,\dots,d_1}\left|\varepsilon \nu^{(l)}\right|^q\right]^{\frac{1}{q}}\\
&\lesssim t_1^{3/2}\varsigma_n^{1/2}
\end{align*}
with $t_1^{3/2}\varsigma_n^{1/2}
\lesssim t_1^{9/2}\varsigma_n^{3/2}=L_n$.
It is worth noting that this can be relaxed to $\|F_0\|_{P,q}\lesssim t_1^{1/2}\varsigma_n^{1/2}$ if B-Splines are used for approximation.
To bound the entropy of $\F_0$, we note that
\begin{align*}
&\quad\big\|\psi_x(W)-\psi_{\tilde{x}}(W)\big\|_{P,2}\\
&=\Big\|(g(x)^T\Sigma_n g(x))^{-1/2}\sum_{l=1}^{d_1} g_l(x) \E[(\nu^{(l)})^2]^{-1}\psi_l(W,\theta_{0,l},\eta_{0,l})\\
&\quad -(g(\tilde{x})^T\Sigma_n g(\tilde{x}))^{-1/2}\sum_{l=1}^{d_1} g_l(\tilde{x}) \E[(\nu^{(l)})^2]^{-1}\psi_l(W,\theta_{0,l},\eta_{0,l})\Big\|_{P,2}\\
&\le | (g(x)^T\Sigma_n g(x))^{-1/2}-(g(\tilde{x})^T\Sigma_n g(\tilde{x}))^{-1/2}|\Big\| g(x)^T J_0^{-1}\psi(W,\theta_{0,l},\eta_{0,l})\Big\|_{P,2}\\
&\quad +(g(\tilde{x})^T\Sigma_n g(\tilde{x}))^{-1/2}\Big\|\big(g(x) -g(\tilde{x})\big)^TJ_0^{-1}\psi(W,\theta_{0,l},\eta_{0,l})\Big\|_{P,2}\\
&\lesssim | (g(x)^T\Sigma_n g(x))^{-1/2}-(g(\tilde{x})^T\Sigma_n g(\tilde{x}))^{-1/2}|\sup_{x\in I}\|g(x)\|_2\varsigma_n^{1/2}\\
&\quad +t_1^{1/2}\|g(x)-g(\tilde{x})\|_2
\end{align*}
due to the eigenvalues of $\Sigma_n$.
Additionally, it holds
\begin{align*}
&|(g(x)^T\Sigma_n g(x))^{-1/2}-(g(\tilde{x})^T\Sigma_n g(\tilde{x}))^{-1/2}|\\
\lesssim\ &\left|\left(\frac{g(\tilde{x})^T\Sigma_n g(\tilde{x})}{g(x)^T\Sigma_n g(x)}\right)^{1/2}-1\right|\varsigma_n^{-1/2} t_1^{1/2}\\
\lesssim\ &|g(\tilde{x})^T\Sigma_n g(\tilde{x})-g(x)^T\Sigma_n g(x)|\varsigma_n^{-3/2} t_1^{3/2}\\
=\  &|(g(x)-g(\tilde{x}))^T\Sigma_n (g(x)+g(\tilde{x}))|\varsigma_n^{-3/2} t_1^{3/2}\\
\lesssim\ &\|g(x)-g(\tilde{x})\|_2\sup_x\|g(x)\|_2\varsigma_n^{-1/2} t_1^{3/2}
\end{align*}
which implies
\begin{align*}
\big\|\psi_x(W)-\psi_{\tilde{x}}(W)\big\|_{P,2}\lesssim  \|g(x)-g(\tilde{x})\|_2 t_1^{3/2}.
\end{align*}
Using the same argument as in Theorem 2.7.11 from \cite{vanweak}, we obtain
\begin{align*}
&\quad\sup_Q\log N(\varepsilon\|F_{0}\|_{Q,2},\F_0,\|\cdot\|_{Q,2})\\
&\lesssim \sup_Q\log N\left(\left(\frac{\varepsilon  t_1^{3/2}\varsigma_n^{1/2}}{t_1^{3/2}}\right) t_1^{3/2},\F_0,\|\cdot\|_{Q,2}\right)\\
&\le \log N\left(\left(\varepsilon \varsigma_n^{1/2}\right),g(I),\|\cdot\|_{2}\right)\\
&\lesssim t_1\log\left(\frac{\tilde{A}_n \varsigma_n^{-1/2}}{\varepsilon}\right)
\end{align*}
by Assumption A.\ref{A1}$(i)$. Therefore, Assumption B.\ref{assumption2.3}$(i)$ is satisfied with $\varrho_n=t_1$ and $A_n = \tilde{A}_n\varsigma_n^{-1/2}.$\\ \\
\textbf{Assumption B.\ref{assumption2.4}}\\
Next, we want to prove that with probability $1-o(1)$ it holds
\begin{align}\label{Jconvergence}
\sup_{l=1,\dots,d_1}|\hat{J}_l-J_{0,l}|=\varsigma_n^{-1}\tilde{\epsilon}_n
\end{align}
with $\tilde{\epsilon}_n=o(1)$ where $\hat{J}_l=\E_n[-g_l(X_1)(g_l(X_1)-(\hat{\eta}_l^{(2)})^TZ_{-l})]$. It holds
\begin{align*}
|\hat{J}_l-J_{0,l}|&\le |\hat{J}_l-\E[-g_l(X_1)(g_l(X_1)-(\hat{\eta}_l^{(2)})^TZ_{-l})]|\\
&\quad+|\E[-g_l(X_1)(g_l(X_1)-(\hat{\eta}_l^{(2)})^TZ_{-l})]+J_{0,l}|
\end{align*}
with 
\begin{align*}
&|\E[-g_l(X_1)(g_l(X_1)-(\hat{\eta}_l^{(2)})^TZ_{-l})]+J_{0,l}|\\
\le&|\E[g_l(X_1)(\hat{\eta}_l^{(2)}-\eta_{0,l}^{(2)})^T Z_{-l})]|\\
\lesssim &\ \tau_n\varsigma_n^{-1}.
\end{align*}
Let 
\begin{align*}
\tilde{\mathcal{G}}_1:=\bigg\{&X\mapsto -g_l(X_1)(g_l(X_1)-(\eta_l^{(2)})^TZ_{-l}):l=1,\dots,d_1,\|\eta_l^{(2)}\|_0\le Cs_2,\\
&\|\eta^{(2)}_l-\eta^{(2)}_{0,l}\|_2\le C\tau_n,\|\eta^{(2)}-\eta^{(2)}_{0,l}\|_1\le C\sqrt{s_2}\tau_n\bigg\}.
\end{align*}
For any $q\ge 2$, the envelope $\tilde{G_1}$ of $\tilde{\mathcal{G}}_1$ satisfies
\begin{align*}
\E[\tilde{G}_1^q]^{\frac{1}{q}}&\le\E \left[\sup_{l=1,\dots,d_1}\sup_{\eta^{(2)}:\|\eta^{(2)}_l-\eta^{(2)}_{0,l}\|_1\le C\sqrt{s}\tau_n} |g_l(X_1)|^q|(g_l(X_1)-(\eta_l^{(2)})^TZ_{-l})|^q\right]^{\frac{1}{q}}\\
&\lesssim \|\sup_{l=1,\dots,d_1}\nu^{(l)}\|_{P,q}\vee
\E\bigg[\sup_{l=1,\dots,d_1}\sup_{\eta^{(2)}:\|\eta^{(2)}_l-\eta^{(2)}_{0,l}\|_1\le C\sqrt{s}\tau_n}(\eta_{0,l}^{(2)}-\eta_l^{(2)})^TZ_{-l})^{q}\bigg]^{\frac{1}{q}}\\
&\lesssim C \vee \sqrt{s_2}\tau_n\\
&\lesssim C
\end{align*}
and, with similar arguments as above, we obtain
\begin{align*}
\sup_{g\in \tilde{\mathcal{G}}_1}\E[g^2]^{\frac{1}{2}}
&\lesssim \varsigma_n^{-1/2} \vee\tau_n\varsigma_n^{-1/2}
\lesssim \varsigma_n^{-1/2}.
\end{align*}
Further, we have
\begin{align*}
\sup_Q\log N(\varepsilon\|\tilde{G}_1\|_{Q,2},\tilde{\mathcal{G}}_1,\|\cdot\|_{Q,2})\lesssim s_2\log\left(\frac{d_1+d_2}{\varepsilon}\right).
\end{align*}
Therefore, by using Lemma P.2 from \cite{belloni2018uniformly} with $\sigma=\varsigma_n^{-1/2}$, it holds
\begin{align*}
\sup_{l=1,\dots,d_1}|\hat{J}_l-J_{0,l}|&\lesssim \sup_{f\in\tilde{\mathcal{G}}_1}|\E_n[f(X)]-\E[f(X)]|+\tau_n\varsigma_n^{-1}\\
&\lesssim K\left(\varsigma_n^{-1/2}\sqrt{\frac{s_2\log(\bar{d}_n\varsigma_n^{1/2})}{n}}+n^{\frac{1}{q}}\frac{s_2\log^{}(\bar{d}_n\varsigma_n^{1/2})}{n}\right) +\tau_n\varsigma_n^{-1}\\
&\lesssim\varsigma_n^{-1}\tilde{\epsilon}_n
\end{align*}
with probability $1-o(1)$ and
$$\tau_n\lesssim\tilde{\epsilon}_n\lesssim\left(\sqrt{\frac{s_2\varsigma_n\log(\bar{d}_n\varsigma_n^{1/2})}{n}}+n^{\frac{1}{q}}\frac{s_2\varsigma_n\log^{}(\bar{d}_n\varsigma_n^{1/2})}{n}\right)=o(1).$$
Next, we want to bound the restricted eigenvalues of $\hat{\Sigma}_{\varepsilon\nu}$ with high probability by showing
\begin{align}\label{sparse_ev2}
\sup\limits_{\|v\|_2=1,\|v\|_0\le t_1}|v^T\big(\hat{\Sigma}_{\varepsilon\nu}-\Sigma_{\varepsilon\nu}\big)v|\lesssim u_n
\end{align}
with 
\begin{align*}
u_n&\lesssim \left(n^{\frac{3}{q}}\frac{t_1\log(d_1)}{\varsigma_nn}\right)^{1/2}\vee \left(\frac{t_1^2s\log(\bar{d}_n)}{\varsigma_nn}\right)^{1/2}\vee \frac{t_1^2s\log(\bar{d}_n)}{n}\\
& \lesssim  \left(n^{\frac{3}{q}}\frac{t_1\log(d_1)}{\varsigma_nn}\right)^{1/2}\vee \left(\frac{t_1^2s\log(\bar{d}_n)}{\varsigma_nn}\right)^{1/2}= o(1)
\end{align*}
for a suitable $q\ge 4$ defined in Assumption A\ref{A2}.. Define $\xi_i:=\varepsilon_i\nu_i$, $\hat{\xi}_i:=\hat{\varepsilon}_i\hat{\nu}_i$ and observe that
\begin{align*}
&\hat{\Sigma}_{\varepsilon\nu}-\Sigma_{\varepsilon\nu}\\
=\ &\frac{1}{n}\sum\limits_{i=1}^n \hat{\xi}_i\hat{\xi}_i^T -\E[\xi_i\xi_i^T]\\
=\ &\frac{1}{n}\sum\limits_{i=1}^n \xi_i\xi_i^T -\E[\xi_i\xi_i^T]\\
&\ +\frac{1}{n}\sum\limits_{i=1}^n \xi_i\big(\hat{\xi}_i-\xi_i\big)^T+\frac{1}{n}\sum\limits_{i=1}^n \big(\hat{\xi}_i-\xi_i\big)\xi_i^T+\frac{1}{n}\sum\limits_{i=1}^n \big(\hat{\xi}_i-\xi_i\big)\big(\hat{\xi}_i-\xi_i\big)^T.
\end{align*}
Using the Lemma Q.1 from \cite{belloni2018uniformly}, we can bound the first part.\\
Due to the tail conditions on $\varepsilon$ and boundedness of $\nu$, we obtain 
\begin{align*}
\left(\E\left[\max_{1\le i\le n}\|\varepsilon_i\nu_i\|_\infty^2\right]\right)^{1/2}&\lesssim \E\left[\max_{1\le i\le n}|\varepsilon_i|^2\right]^{1/2}
\\
&\lesssim n^{1/q}
\end{align*}
for $q$ defined in Assumption A.\ref{A2}(iii). Then, Lemma Q.1 implies
\begin{align*}
&\ \E\left[\sup\limits_{\|v\|_2=1,\|v\|_0\le t_1}\Big|v^T\Big(\frac{1}{n}\sum\limits_{i=1}^n\xi_i\xi_i^T-\E[\xi_i\xi_i^T]\Big)v\Big|\right]\\
=&\ \E\left[\sup\limits_{\|v\|_2=1,\|v\|_0\le t_1}\Big|\E_n\Big[\big(v^T\xi_i\big)^2-\E\big[\big(v^T\xi_i\big)^2\big]\Big]\Big|\right]\\
\lesssim&\ \tilde{\delta}_n^2+\tilde{\delta}_n \varsigma_n^{-1/2}\lesssim \left(n^{\frac{3}{q}}\frac{t_1\log(d_1)}{\varsigma_nn}\right)^{1/2}
\end{align*}
where we used
\begin{align*}
\tilde{\delta}_n &\lesssim \left(n^{\frac{2}{q}-1}t_1\log^2(t_1)\log(d_1)\log(n)\right)^{\frac{1}{2}}\\
&\lesssim \left(n^{\frac{3}{q}}\frac{t_1\log^{}(d_1)}{n}\right)^{\frac{1}{2}}\lesssim \varsigma_n^{-1/2}
\end{align*}
by the growth conditions in A.\ref{A2}$(v)$.
Using Markov's inequality, we directly obtain
\begin{align*}
\sup\limits_{\|v\|_2=1,\|v\|_0\le t_1}\Big|v^T\Big(\frac{1}{n}\sum\limits_{i=1}^n\xi_i\xi_i^T-\E[\xi_i\xi_i^T]\Big)v\Big|\lesssim u_n
\end{align*}
with probability $1-o(1)$. Note that by applying the results on covariance estimation from \cite{chen2012masked} instead would lead to comparable growth rates.\\
With probability $1-o(1)$, it holds
\begin{align*}
\sup_{l=1,\dots,d_1} |\hat{\theta}_{l}-\theta_{0,l}|\lesssim \varsigma_n^{1/2}\tau_n
\end{align*}
analog to step 1 in the proof of Theorem 2.1 in \cite{belloni2018uniformly}. Define
\begin{align*}
\tilde{\mathcal{G}}^2_2:=\big\{(\psi_l(\cdot,\theta_l,\eta_l)-\psi_l(\cdot,\theta_{0,l},\eta_{0,l}))^2:\ & l=1,\dots,d_1,|\theta_l-\theta_{0,l}|\le C\varsigma_n^{1/2}\tau_n,\\
& \eta_l\in\mathcal{T}_l\setminus\{\eta_{0,l}\}\big\},
\end{align*}
with
\begin{align*}
\sup_Q\log N(\varepsilon\|\tilde{G}_2^2\|_{Q,2},\tilde{\mathcal{G}}_2^2,\|\cdot\|_{Q,2})\lesssim s\log\left(\frac{d_1+d_2}{\varepsilon}\right).
\end{align*}
Here, $\tilde{G}_2^2$ is a measurable envelope of $\tilde{\mathcal{G}}_2^2$ with 
\begin{align*}
\tilde{G}_2^2=\sup_{l=1,\dots,d_1}\sup_{\theta_l:|\theta_l-\theta_{0,l}|\le C\varsigma_n^{1/2}\tau_n,\eta_l\in\mathcal{T}_l}\big(\psi_l(W,\theta_l,\eta_l)-\psi_l(W,\theta_{0,l},\eta_{0,l})\big)^2
\end{align*}
and 
\begin{align*}
 &\|\tilde{G}_2^2\|_{P,\tilde{q}}\\
\lesssim&\  \Big\|\sup_{l,\theta_l,\eta_l^{(2)}}\Big((\theta_{0,l}-\theta_l)g_l(X_1)\big(g_l(X_1)-(\eta_l^{(2)})^TZ_{-l}\big)\Big)^2\Big\|_{P,\tilde{q}}\\
&\ +\Big\|\sup_{l,\eta_l}\Big(\big(Y-\theta_{0,l}g_l(X_1)-(\eta_l^{(1)})^TZ_{-l}-\eta_l^{(3)}(X)\big)\\
&\quad\quad\big((\eta_{0,l}^{(2)}-\eta_l^{(2)})^TZ_{-l}\big)\Big)^2\Big\|_{P,\tilde{q}}\\
&\ +\Big\|\sup_{l,\eta_l^{(1)},\eta_l^{(3)}}\Big(\big(g_l(X_1)-(\eta_{0,l}^{(2)})^TZ_{-l}\big)\\
&\quad\quad\big((\eta_{0,l}^{(1)}-\eta_l^{(1)})^TZ_{-l}+\eta_{0,l}^{(3)}(X)-\eta_{l}^{(3)}(X)\big)\Big)^2\Big\|_{P,\tilde{q}}\\
&=:T_1+T_2+T_3.
\end{align*}
For any $\tilde{q}\le q/2$ in Assumption A.\ref{A2}, it holds
\begin{align*}
T_1&\lesssim \tau_n^2\varsigma_n^{}\Big\|\sup_{l,\eta_l^{(2)}}\Big(g_l(X_1)\big(g_l(X_1)-(\eta_l^{(2)})^TZ_{-l}\big)\Big)^2\Big\|_{P,\tilde{q}}\\
&\lesssim \tau_n^2\varsigma_n^{}\Big\|\sup_{l,\eta_l^{(2)}}\Big(g_l(X_1)-(\eta_l^{(2)})^TZ_{-l}\Big)^2\Big\|_{P,\tilde{q}}\\
&\lesssim \tau_n^2\varsigma_n^{},
\end{align*}
and
\begin{align*}
T_2&\lesssim s_2\tau_n^2
\end{align*}
where we used that $\|\varepsilon\|_{P,q}\le C$, $\|Z\|_\infty\le C$ (a.s.) and  that $\eta^{(2)}_l\in\mathcal{T}_l$. Further, we have
\begin{align*}
T_3&\lesssim \Big\|\sup_{l,\eta_l^{(1)},\eta_l^{(3)}}\Big((\eta_{0,l}^{(1)}-\eta_l^{(1)})^TZ_{-l}+\eta_{0,l}^{(3)}(X)\Big)^2\Big\|_{P,q}\\
&\lesssim s_1\tau_n^2 + \tau_n^2
\end{align*}
due to Assumption A.\ref{A1}$(ii)$. By using an analogous argument as above, we obtain
\begin{align*}
\tilde{\sigma}:&=\sup\limits_{f\in\tilde{\mathcal{G}}_2^2}\E\left[f(X)^2\right]^\frac{1}{2}\\
&=\sup\limits_{l=1,\dots,d_1}\sup\limits_{\theta_l:|\theta_l-\theta_{0,l}|\le C\varsigma^{1/2}\tau_n,\eta_l\in\mathcal{T}_l}\E\left[\left(\psi_l(W,\theta_l,\eta_l)-\psi_l(W,\theta_{0,l},\eta_{0,l})\right)^4\right]^{\frac{1}{2}}\\
&\lesssim \tau_n^2(s\vee\varsigma_n).
\end{align*}
Again, we can apply Lemma P.2 from \cite{belloni2018uniformly} to obtain
\begin{align*}
\sup\limits_{f\in\tilde{\mathcal{G}}_2^2}|\E_n[f(X)]-\E[f(X)]|\le K&\bigg(\tilde{\sigma}\sqrt{\frac{s\log(\bar{d}_n)}{n}}+n^{\frac{1}{q}}\|\tilde{G}_2^2\|_{P,\tilde{q}}\frac{s\log(\bar{d}_n)}{n}\bigg)\\
&\lesssim (s\vee \varsigma_n)\tau_n^3\vee n^{\frac{1}{\tilde{q}}}(s\vee \varsigma_n)\tau_n^4
\end{align*}
with probability $1-o(1)$. Note that we have already shown Assumption B.\ref{assumption2.1}$(v)(a)$ which implies
\begin{align*}
\sup\limits_{f\in\tilde{\mathcal{G}}_2^2} \E[f(X)]&\le C\left(|\theta_l-\theta_{0,l}|^2\varsigma_n^{-1}\vee \|\eta_{0,l}-\eta_l\|_e^2\right)\\
&\lesssim \tau_n^2.
\end{align*}
Combined, this implies
\begin{align}\label{growthstrong}
\sup_{l=1,\dots,d_1}\E_n\left[\left(\hat{\varepsilon}_i\hat{\nu}^{(l)}_i-\varepsilon_i\nu_i^{(l)}\right)^2\right]&\le \sup\limits_{f\in\tilde{\mathcal{G}}_2^2} \E_n[f(X)] \nonumber\\
&\lesssim \tau_n^2+\left((s\vee \varsigma_n)\tau_n^3\vee n^{\frac{1}{\tilde{q}}}(s\vee \varsigma_n)\tau_n^4\right)\lesssim \tau_n^2
\end{align}
by growth the growths assumptions in A.\ref{A2} (v) and, with an analogous argument, we obtain 
\begin{align*}
\sup_{l=1,\dots,d_1}\E_n\left[\left(\varepsilon_i\nu_i^{(l)}\right)^2\right]\lesssim \varsigma_n^{-1}.
\end{align*}
Therefore, it holds
\begin{align*}
& \sup\limits_{\|v\|_2=1,\|v\|_0\le t_1}|v^T\frac{1}{n}\sum\limits_{i=1}^n \xi_i\big(\hat{\xi}_i-\xi_i\big)^Tv|\\
=\  &\sup\limits_{\|v\|_2=1,\|v\|_0\le t_1}|\E_n\left[v^T \xi_i\big(\hat{\xi}_i-\xi_i\big)^Tv\right]|\\
\le\ & \sup\limits_{\|v\|_2=1,\|v\|_0\le t_1}\left|\left(\E_n\left[\left(v^T \xi_i\right)^2\right]\E_n\left[\left(v^T\big(\hat{\xi}_i-\xi_i\big)\right)^2\right]\right)^{\frac{1}{2}}\right|\\
\lesssim\ &\varsigma_n^{-1/2} \sup\limits_{\|v\|_2=1,\|v\|_0\le t_1}\left|\left(\E_n\left[\left(v^T\big(\hat{\xi}_i-\xi_i\big)\right)^2\right]\right)^{\frac{1}{2}}\right|\\
=\ & \varsigma_n^{-1/2}\sup\limits_{\|v\|_2=1,\|v\|_0\le t_1}\left(\sum_{k=1}^{d_1}\sum_{l=1}^{d_1} v_kv_l\E_n\left[(\hat{\varepsilon}_i\hat{\nu}_i^{(k)}-\varepsilon_i\nu_i^{(k)})(\hat{\varepsilon}_i\hat{\nu}_i^{(l)}-\varepsilon_i\nu_i^{(l)})\right]\right)^{\frac{1}{2}}\\
\lesssim\ &\varsigma_n^{-1/2}t_1\sup_{l=1,\dots,d_1}\E_n\left[(\hat{\varepsilon}_i\hat{\nu}_i^{(l)}-\varepsilon_i\nu_i^{(l)})^2\right]^{\frac{1}{2}}\\
\lesssim\ & t_1\varsigma_n^{-1/2}\tau_n
\end{align*}
and 
\begin{align*}
\sup\limits_{\|v\|_2=1,\|v\|_0\le t_1}|v^T\frac{1}{n}\sum\limits_{i=1}^n \big(\hat{\xi}_i-\xi_i\big)\big(\hat{\xi}_i-\xi_i\big)^Tv|\lesssim t_1^2\tau_n^2
\end{align*}
with probability $1-o(1)$. Combining the steps above, implies (\ref{sparse_ev2}) if $u_n=o(1)$ which is ensured by the growth conditions. Next, note that for every sparse vector $w\in \R^{d_1}$ ($\|w\|_0\le t_1$) there exists a corresponding matrix $M_w$ 
$$M_w\in \R^{d_1\times d_1}: (M_w)_{k,l}=\begin{cases}1 \text{ if } w_k\neq 0 \wedge w_l\neq 0\\
0 \text{ else,} \end{cases}$$
such that 
\begin{align*}
w^T(\hat{\Sigma}_{\varepsilon\nu}-\Sigma_{\varepsilon\nu})w=w^T\left(M_w \odot(\Sigma_{\varepsilon\nu}-\hat{\Sigma}_{\varepsilon\nu})\right)w.
\end{align*}
Due to (\ref{sparse_ev2}), it holds
\begin{align*}
\sup_{\|w\|_0\le t_1}\sup_{\|v\|_2=1}\left|v^T\left(M_w \odot(\hat{\Sigma}_{\varepsilon\nu}-\Sigma_{\varepsilon\nu})\right)v\right|&\le \sup_{\|v\|_2=1,\|v\|_0\le t_1}\left|v^T(\hat{\Sigma}_{\varepsilon\nu}-\Sigma_{\varepsilon\nu})v\right|\lesssim u_n,
\end{align*}
which implies 
\begin{align*}
\sup_{\|w\|_0\le t_1}\|M_w \odot(\hat{\Sigma}_{\varepsilon\nu}-\Sigma_{\varepsilon\nu})\|_2\lesssim u_n
\end{align*}
and 
\begin{align*}
\sup_{\|w\|_0\le t_1}\|M_w \odot\hat{\Sigma}_{\varepsilon\nu}\|_2\lesssim \varsigma_n^{-1}
\end{align*}
due to Assumption A.\ref{A2}$(iv)$. This can be used to show that
\begin{align}\label{sparse_ev}
\sup\limits_{\|v\|_2=1,\|v\|_0\le t_1}|v^T\big(\hat{\Sigma}_n-\Sigma_n\big)v|\lesssim \varsigma_n\left(\tilde{u}_n\vee\tilde{\epsilon}_n\right)
\end{align}
where $\tilde{u}_n:=\varsigma_n u_n=o(1)$
with probability $1-o(1)$ which can be interpreted as an upper bound for the sparse eigenvalues of $\hat{\Sigma}_n-\Sigma_n$. It holds
\begin{align*}
\hat{\Sigma}_n-\Sigma_n&= \hat{J}^{-1}\hat{\Sigma}_{\varepsilon\nu}(\hat{J}^{-1})^T-J_0^{-1}\Sigma_{\varepsilon\nu}(J_0^{-1})^T\\
&=(\hat{J}^{-2}-J_0^{-2})\hat{\Sigma}_{\varepsilon\nu}+J_0^{-2}\big(\hat{\Sigma}_{\varepsilon\nu}-\Sigma_{\varepsilon\nu}\big).
\end{align*}
Note that
\begin{align*}
&\sup\limits_{\|v\|_2=1,\|v\|_0\le t_1}|v^T\big(\hat{J}^{-2}-J_0^{-2})\hat{\Sigma}_{\varepsilon\nu}v|\\
=\ &\sup\limits_{\|v\|_2=1,\|v\|_0\le t_1}|v^T\big(\hat{J}^{-2}-J_0^{-2})\left(M_v\odot\hat{\Sigma}_{\varepsilon\nu}\right)v|\\
\le\ &\sup\limits_{\|w\|_0\le t_1}\left\|\left(M_w\odot\hat{\Sigma}_{\varepsilon\nu}\right)\right\|_2\left\|\big(\hat{J}^{-2}-J_0^{-2})\right\|_2\\
\lesssim\ & \varsigma_n^{-1}\left\|\big(\hat{J}^{-2}-J_0^{-2})\right\|_2\\
\lesssim & \varsigma_n\tilde{\epsilon}_n
\end{align*}
due to the sub-multiplicative spectral norm, $|J_{0,l}^{-1}|=O(\varsigma_n)$ and \eqref{Jconvergence} which implies
\begin{align}\label{J2Rate}
\left\|\hat{J}^{-2}-J_0^{-2}\right\|&\le \left\|\hat{J}^{-1}(\hat{J}^{-1}-J_0^{-1})\right\|+\left\|J_0^{-1}(\hat{J}^{-1}-J_0^{-1})\right\|\nonumber \\
&\lesssim\varsigma_n \left\|\hat{J}^{-1}-J_0^{-1}\right\|\lesssim\varsigma_n^3 \left\|\hat{J}-J_0\right\|=\varsigma_n^2\tilde{\epsilon}_n
\end{align}
with $\left\|\hat{J}-J_0\right\|=\sup_{l=1,\dots,d_1}|\hat{J}_l-J_{0,l}|=\varsigma_n^{-1}\tilde{\epsilon}_n$. The second term can be bounded by
\begin{align*}
&\sup\limits_{\|v\|_2=1,\|v\|_0\le t_1}|v^TJ_0^{-2}\big(\hat{\Sigma}_{\varepsilon\nu}-\Sigma_{\varepsilon\nu})v|\lesssim \varsigma_n^2 u_n=\varsigma_n\tilde{u}_n.
\end{align*}
This implies (\ref{sparse_ev}). We finally obtain
\begin{align*}
\sup_{x\in I}\left|\frac{(g(x)^T\hat{\Sigma}_n g(x))^{1/2}}{(g(x)^T\Sigma_n g(x))^{1/2}}-1\right|&\le \sup_{x\in I}\left|\frac{(g(x)^T\hat{\Sigma}_n g(x))}{(g(x)^T\Sigma_n g(x))}-1\right|\\
&\lesssim \varsigma_n^{-1}t_1\sup_{x\in I}\left|g(x)^T\big(\hat{\Sigma}_n-\Sigma_n\big) g(x)\right|\\
&\le  \varsigma_n^{-1}t_1 \sup\limits_{\|v\|_2=1,\|v\|_0\le t_1}|v^T\big(\hat{\Sigma}_n-\Sigma_n\big)v|\\
&=t_1(\tilde{u}_n\vee\tilde{\epsilon}_n)\lesssim t_1\tilde{u}_n
\end{align*}
with probability $1-o(1)$ which is the first part of Assumption B.\ref{assumption2.4} with $\epsilon_n\lesssim t_1 \tilde{u}_n$ and $\epsilon_n t_1 \log(\tilde{A}_n\varsigma_n^{-1/2})=o(1)$. \\ \\
\textbf{Assumption B.\ref{assumption2.3}$(iii)-(iv)$ }\\
Define
\begin{align*}
\sigma_x:&=(g(x)^T\Sigma_n g(x))^{1/2},\\
\hat{\sigma}_x:&=(g(x)^T\hat{\Sigma}_n g(x))^{1/2}
\end{align*}
and
$$\hat{\F}_0:=\{\psi_x(\cdot)-\hat{\psi}_x(\cdot):x\in I\}$$
with  $\hat{\psi}_x(\cdot):=\hat{\sigma}_x^{-1}g(x)^T\hat{J}_0^{-1}\psi(\cdot,\hat{\theta}_{},\hat{\eta}_{})$. For every $x$ and $\tilde{x}$, it holds
\begin{align*}
&\|\psi_x(W)-\hat{\psi}_x(W)-(\psi_{\tilde{x}}(W)-\hat{\psi}_{\tilde{x}}(W))\|_{\Pb_n,2}\\
 = &\Big\|\sigma_x^{-1}g(x)^TJ_0^{-1}\psi(W,\theta_{0},\eta_{0})-\sigma_{\tilde{x}}^{-1}g({\tilde{x}})^TJ_0^{-1}\psi(W,\theta_{0},\eta_{0}) \\
 & -\big(\hat{\sigma}_x^{-1}g(x)^T\hat{J}^{-1}\psi(W,\hat{\theta}_{},\hat{\eta}_{})-\hat{\sigma}_{\tilde{x}}^{-1}g({\tilde{x}})^T\hat{J}^{-1}\psi(W,\hat{\theta}_{},\hat{\eta}_{}) \big)\Big\|_{\Pb_n,2}\\
  = &\Big\|\sum_{l=1}^{d_1} (\sigma_x^{-1} g_l(x)-\sigma_{\tilde{x}}^{-1}g_l(\tilde{x}))J_{0,l}^{-1}\psi_l(W,\theta_{0,l},\eta_{0,l}) \\
 & -\sum_{l=1}^{d_1} (\hat{\sigma}_x^{-1} g_l(x)-\hat{\sigma}_{\tilde{x}}^{-1}g_l(\tilde{x}))\hat{J}_l^{-1}\psi_l(W,\hat{\theta}_{l},\hat{\eta}_{l})\Big\|_{\Pb_n,2}\\
   \le &\Big\|\sum_{l=1}^{d_1} (\sigma_x^{-1} g_l(x)-\sigma_{\tilde{x}}^{-1}g_l(\tilde{x}))\Big(J_{0,l}^{-1}-\hat{J}_{l}^{-1}\Big)\psi_l(W,\theta_{0,l},\eta_{0,l})\Big\|_{\Pb_n,2} \\
     &+\Big\|\sum_{l=1}^{d_1} (\sigma_x^{-1} g_l(x)-\sigma_{\tilde{x}}^{-1}g_l(\tilde{x}))\hat{J}_{l}^{-1}\Big(\psi_l(W,\theta_{0,l},\eta_{0,l})-\psi_l(W,\hat{\theta}_{l},\hat{\eta}_{l})\Big)\Big\|_{\Pb_n,2}\\
 & +\Big\|\sum_{l=1}^{d_1}\Big( (\sigma_x^{-1} g_l(x)-\sigma_{\tilde{x}}^{-1}g_l(\tilde{x}))- (\hat{\sigma}_x^{-1} g_l(x)-\hat{\sigma}_{\tilde{x}}^{-1}g_l(\tilde{x}))\Big)\hat{J}_{l}^{-1}\psi_l(W,\hat{\theta}_{l},\hat{\eta}_{l})\Big\|_{\Pb_n,2}\\
 =:& I_{4,1}+I_{4,2}+I_{4,3}.
\end{align*} 
We obtain
\begin{align*}
I_{4,1}&=\Big\|\sum_{l=1}^{d_1} (\sigma_x^{-1} g_l(x)-\sigma_{\tilde{x}}^{-1}g_l(\tilde{x}))\Big(J_{0,l}^{-1}-\hat{J}_{l}^{-1}\Big)\psi_l(W,\theta_{0,l},\eta_{0,l})\Big\|_{\Pb_n,2} \\
     &\le\sigma_x^{-1}\Big\|(g(x)-g(\tilde{x}))^T\Big(J_{0}^{-1}-\hat{J}_{}^{-1}\Big)\psi(W,\theta_{0},\eta_{0})\Big\|_{\Pb_n,2}\\
     &\quad+|\sigma_x^{-1}-\sigma_{\tilde{x}}^{-1}|\Big\|g(\tilde{x})^T\Big(J_{0}^{-1}-\hat{J}_{}^{-1}\Big)\psi(W,\theta_{0},\eta_{0})\Big\|_{\Pb_n,2}\\
     &\lesssim \varsigma_n^{-1/2}\sqrt{t_1} \|g(x)-g(\tilde{x})\|_2\sup\limits_{\|v\|_2=1,\|v\|_0\le 2t_1}\Big\|v^T\Big(J_{0}^{-1}-\hat{J}_{}^{-1}\Big)\psi(W,\theta_{0},\eta_{0})\Big\|_{\Pb_n,2}\\
     &\quad + \varsigma_n^{-1/2}t_1^{3/2}\|g(x)-g(\tilde{x})\|_2\sup\limits_{x\in I}\|g(x)\|_2^{2}\sup\limits_{\|v\|_2=1,\|v\|_0\le t_1}\Big\|v^T\Big(J_{0}^{-1}-\hat{J}_{}^{-1}\Big)\psi(W,\theta_{0},\eta_{0})\Big\|_{\Pb_n,2}\\
     &\lesssim  \varsigma_n^{-1/2}t_1^{3/2}\|g(x)-g(\tilde{x})\|_2\varsigma_n^{1/2}\tilde{\epsilon}_n\lesssim t_1^{3/2}\tilde{\epsilon}_n\|g(x)-g(\tilde{x})\|_2,
\end{align*}
with $t_1^{3/2}\tilde{\epsilon}_n=o(1)$ where we used that
\begin{align*}
&\ \sup\limits_{\|v\|_2=1,\|v\|_0\le t_1}\Big\|v^T\Big(J_{0}^{-1}-\hat{J}_{}^{-1}\Big)\psi(W,\theta_{0},\eta_{0})\Big\|^2_{\Pb_n,2}\\
=&\ \sup\limits_{\|v\|_2=1,\|v\|_0\le t_1}\Big|v^T\Big(J_{0}^{-1}-\hat{J}_{}^{-1}\Big)\frac{1}{n}\sum\limits_{i=1}^n \xi_i\xi_i^T \Big(J_{0}^{-1}-\hat{J}_{}^{-1}\Big)^Tv\Big|\\
\le&\ \left\|J_{0}^{-1}-\hat{J}_{}^{-1}\right\|^2_2\sup\limits_{\|v\|_0\le t_1}\Big\|M_v\odot\left(\frac{1}{n}\sum\limits_{i=1}^n \xi_i\xi_i^T\right)\Big\|_2\\
\lesssim& (\varsigma_n\tilde{\epsilon}_n)^2\varsigma_n^{-1}\lesssim \varsigma_n\tilde{\epsilon}_n^2
\end{align*}
by \eqref{J2Rate}. Analogously, we obtain 
\begin{align*}
I_{4,2}&=\Big\|\sum_{l=1}^{d_1} (\sigma_x^{-1} g_l(x)-\sigma_{\tilde{x}}^{-1}g_l(\tilde{x}))\hat{J}_{l}^{-1}\Big(\psi_l(W,\theta_{0,l},\eta_{0,l})-\psi_l(W,\hat{\theta}_{l},\hat{\eta}_{l})\Big)\Big\|_{\Pb_n,2}\\
     &\le\sigma_x^{-1}\Big\|(g(x)-g(\tilde{x}))^T\hat{J}_{}^{-1}\Big(\psi(W,\theta_{0},\eta_{0})-\psi(W,\hat{\theta}_{},\hat{\eta}_{})\Big)\Big\|_{\Pb_n,2}\\
     &\quad+|\sigma_x^{-1}-\sigma_{\tilde{x}}^{-1}|\Big\|g(\tilde{x})^T\hat{J}_{}^{-1}\Big(\psi(W,\theta_{0},\eta_{0})-\psi(W,\hat{\theta}_{},\hat{\eta}_{})\Big)\Big\|_{\Pb_n,2}\\
     &\lesssim \varsigma_n^{-1/2}\sqrt{t_1}\|g(x)-g(\tilde{x})\|_2\sup\limits_{\|v\|_2=1,\|v\|_0\le 2t_1}\Big\|v^T\hat{J}_{}^{-1}\Big(\psi(W,\theta_{0},\eta_{0})-\psi(W,\hat{\theta}_{},\hat{\eta}_{})\Big)\Big\|_{\Pb_n,2}\\
     &\quad + \varsigma_n^{-1/2}t_1^{3/2}\|g(x)-g(\tilde{x})\|_2\sup\limits_{x\in I}\|g(x)\|_2^2\sup\limits_{\|v\|_2=1,\|v\|_0\le t_1}\Big\|v^T\hat{J}_{}^{-1}\Big(\psi(W,\theta_{0},\eta_{0})-\psi(W,\hat{\theta}_{},\hat{\eta}_{})\Big)\Big\|_{\Pb_n,2}\\
     &\lesssim \varsigma_n^{-1/2}t_1^{3/2}\|g(x)-g(\tilde{x})\|_2\sup\limits_{\|v\|_2=1,\|v\|_0\le t_1}\Big\|v^T\hat{J}_{}^{-1}\Big(\psi(W,\theta_{0},\eta_{0})-\psi(W,\hat{\theta}_{},\hat{\eta}_{})\Big)\Big\|_{\Pb_n,2}\\
     &\lesssim \varsigma_n^{-1/2}t_1^{3/2} \|g(x)-g(\tilde{x})\|_2 \varsigma_n t_1\tau_n\\
     &\lesssim \sqrt{\frac{t_1^5\varsigma_n s\log(\bar{d}_n)}{n}}\|g(x)-g(\tilde{x})\|_2=o(\|g(x)-g(\tilde{x})\|_2),
\end{align*}
by growth conditions where we used that 
$$\sup\limits_{\|v\|_2=1,\|v\|_0\le t_1}\Big\|v^T\hat{J}_{}^{-1}\Big(\psi(W,\theta_{0},\eta_{0})-\psi(W,\hat{\theta}_{},\hat{\eta}_{})\Big)\Big\|_{\Pb_n,2}\lesssim \varsigma_n t_1\tau_n$$
as shown in \eqref{growthstrong}. Further, it holds
\begin{align*}
I_{4,3}&= \Big\|\sum_{l=1}^{d_1}\Big( (\sigma_x^{-1} g_l(x)-\sigma_{\tilde{x}}^{-1}g_l(\tilde{x}))- (\hat{\sigma}_x^{-1} g_l(x)-\hat{\sigma}_{\tilde{x}}^{-1}g_l(\tilde{x}))\Big)\hat{J}_{l}^{-1}\psi_l(W,\hat{\theta}_{l},\hat{\eta}_{l})\Big\|_{\Pb_n,2}\\
&\le \big| \sigma_x^{-1}-\hat{\sigma}_x^{-1} \big|\Big\|(g(x)-g(\tilde{x}))^T\hat{J}_{}^{-1}\psi(W,\hat{\theta}_{},\hat{\eta}_{})\Big\|_{\Pb_n,2}\\
&\quad +\big| (\sigma_x^{-1}-\hat{\sigma}_x^{-1})-(\sigma_{\tilde{x}}^{-1}-\hat{\sigma}_{\tilde{x}}^{-1})\big|\Big\|g(\tilde{x})^T\hat{J}_{}^{-1}\psi(W,\hat{\theta}_{},\hat{\eta}_{})\Big\|_{\Pb_n,2}.
\end{align*}
Note that
\begin{align*}
&\big| (\sigma_x^{-1}-\hat{\sigma}_x^{-1})-(\sigma_{\tilde{x}}^{-1}-\hat{\sigma}_{\tilde{x}}^{-1})\big|\\
=\ &\Big| \frac{1}{\sigma_x \sigma_{\tilde{x}}}(\sigma_{\tilde{x}}-\sigma_x)-\frac{1}{\hat{\sigma}_x \hat{\sigma}_{\tilde{x}}}(\hat{\sigma}_{\tilde{x}}-\hat{\sigma}_{x})\Big|\\
= \ &\frac{1}{\hat{\sigma}_x \hat{\sigma}_{\tilde{x}}}\Big| \frac{\hat{\sigma}_x \hat{\sigma}_{\tilde{x}}}{\sigma_x \sigma_{\tilde{x}}}(\sigma_{\tilde{x}}-\sigma_x)-(\hat{\sigma}_{\tilde{x}}-\hat{\sigma}_{x})\Big|\\
\lesssim\ & t_1\varsigma_n^{-1}\left(\big| (\sigma_{\tilde{x}}-\sigma_x)-(\hat{\sigma}_{\tilde{x}}-\hat{\sigma}_{x})\big|+\Big|\frac{\hat{\sigma}_x \hat{\sigma}_{\tilde{x}}}{\sigma_x \sigma_{\tilde{x}}}-1\Big|\big|\sigma_{\tilde{x}}-\sigma_x\big|\right)
\end{align*}
with 
\begin{align*}
\Big|\frac{\hat{\sigma}_x \hat{\sigma}_{\tilde{x}}}{\sigma_x \sigma_{\tilde{x}}}-1\Big|\big|\sigma_{\tilde{x}}-\sigma_x\big|
&\le\Big(\Big|\frac{\hat{\sigma}_x }{\sigma_x }-1\Big|\frac{ \hat{\sigma}_{\tilde{x}}}{ \sigma_{\tilde{x}}}+\Big|\frac{\hat{\sigma}_{\tilde{x} }}{\sigma_{\tilde{x}} }-1\Big|\Big)\big|\sigma_{\tilde{x}}-\sigma_x\big|\\
&\lesssim\epsilon_n\frac{1}{\sigma_x}\big|\sigma^2_{\tilde{x}}-\sigma^2_x\big|\\
&\lesssim\epsilon_n\sqrt{t_1}\varsigma_n^{1/2}\|g(x)-g(\tilde{x})\|_2\sup_x\|g(x)\|_2
\end{align*}
uniformly over $x\in I$ with probability $1-o(1)$ and
\begin{align*}
&\big| (\sigma_{\tilde{x}}-\sigma_x)-(\hat{\sigma}_{\tilde{x}}-\hat{\sigma}_{x})\big|\\
\le\ &\frac{1}{(\hat{\sigma}_{\tilde{x}}+\hat{\sigma}_{x})} \big| (\sigma_{\tilde{x}}^2-\sigma_x^2)-(\hat{\sigma}_{\tilde{x}}^2-\hat{\sigma}_{x}^2)\big|+\Big|\left(\frac{1}{(\sigma_{\tilde{x}}+\sigma_{x})}-\frac{1}{(\hat{\sigma}_{\tilde{x}}+\hat{\sigma}_{x})}\right)(\sigma_{\tilde{x}}^2-\sigma_x^2)\Big|\\
\lesssim\ & \sqrt{t_1}\varsigma_n^{-1/2}\left(\big|(\sigma_{\tilde{x}}^2-\sigma_x^2)-(\hat{\sigma}_{\tilde{x}}^2-\hat{\sigma}_{x}^2)\big|+\Big|\frac{(\hat{\sigma}_{\tilde{x}}+\hat{\sigma}_{x})}{(\sigma_{\tilde{x}}+\sigma_x)}-1\Big| \big|\sigma^2_{\tilde{x}}-\sigma^2_x\big|\right).
\end{align*}
Using an analogous argument as in the verification of Assumption B.\ref{assumption2.4}, we obtain
\begin{align*}
|(\sigma_x^{2}-\hat{\sigma}_x^{2})-(\sigma_{\tilde{x}}^{2}-\hat{\sigma}_{\tilde{x}}^{2})|&=|(g(x)-g(\tilde{x}))^T(\Sigma_n-\hat{\Sigma}_n) (g(x)+g(\tilde{x}))|\\
&\le \|(\Sigma_n-\hat{\Sigma}_n)(g(x)-g(\tilde{x}))\|_2\sup_{x\in I}\|g(x)\|_2\\
&\lesssim \varsigma_n\tilde{u}_n\|g(x)-g(\tilde{x})\|_2 
\end{align*}
with probability $1-o(1)$ where the last inequality holds due the order of the sparse eigenvalues in (\ref{sparse_ev}). Additionally,
\begin{align*}
\Big|\frac{(\hat{\sigma}_{\tilde{x}}+\hat{\sigma}_{x})}{(\sigma_{\tilde{x}}+\sigma_x)}-1\Big| \big|\sigma^2_{\tilde{x}}-\sigma^2_x\big|&\le \sup_{x\in I} \Big|\frac{\hat{\sigma}_x }{\sigma_x }-1\Big| \big|\sigma^2_{\tilde{x}}-\sigma^2_x\big|\\
&\lesssim\epsilon_n\varsigma_n\|g(x)-g(\tilde{x})\|_2
\end{align*}
with probability $1-o(1)$. Therefore, we obtain
\begin{align*}
I_{4,3}&\lesssim \sqrt{t_1}\varsigma_n^{-1/2}\epsilon_n\|g(x)-g(\tilde{x})\|_2\sup\limits_{\|v\|_2=1,\|v\|_0\le t_1}\Big\|v^T\hat{J}_{}^{-1}\psi(W,\hat{\theta}_{},\hat{\eta}_{})\Big\|_{\Pb_n,2}\\
&\quad +t_1^{3/2}\varsigma_n^{-1}\left(\epsilon_n\varsigma_n^{1/2}\vee \varsigma_n^{1/2}\tilde{u}_n\right)\|g(x)-g(\tilde{x})\|_2\sup\limits_{\|v\|_2=1,\|v\|_0\le t_1}\Big\|v^T\hat{J}_{}^{-1}\psi(W,\hat{\theta}_{},\hat{\eta}_{})\Big\|_{\Pb_n,2}\\
&\lesssim t_1^{3/2}\left(\epsilon_n\vee \tilde{u}_n\right) \|g(x)-g(\tilde{x})\|_2.
\end{align*}
Combining the steps above, we obtain
\begin{align*}
\|\psi_x(W)-\hat{\psi}_x(W)-(\psi_{\tilde{x}}(W)-\hat{\psi}_{\tilde{x}}(W))\|_{\Pb_n,2}\le \|g(x)-g(\tilde{x})\|_2\|\hat{F}_0\|_{\Pb_n,2}
\end{align*}
with 
\begin{align*}
\|\hat{F}_0\|_{\Pb_n,2}=o(1)
\end{align*}
due to the growth condition in Assumption A.\ref{A2}$(v)(b)$. Using the same argument as Theorem 2.7.11 from \cite{vanweak}, we obtain with probability $1-o(1)$
\begin{align*}
\log N(\varepsilon,\hat{\F}_0,\|\cdot\|_{\Pb_n,2})&\le\log N(\varepsilon\|\hat{F}_0\|_{\Pb_n,2},\hat{\F}_0,\|\cdot\|_{\Pb_n,2})\\
&\le \log N(\varepsilon,g(I),\|\cdot\|_{2})\\
&\le\bar{\varrho}_n\log\left(\frac{\bar{A}_n}{\varepsilon}\right)
\end{align*}
with $\bar{\varrho}_n=t_1$ and $\bar{A}_n\lesssim \tilde{A}_n$ by Assumption A.\ref{A1}$(i)$. Additionally, it holds
\begin{align*}
&\|\psi_x(W)-\hat{\psi}_x(W)\|_{\Pb_n,2}\\
=\ &\Big\|\sigma_x^{-1}g(x)^TJ_0^{-1}\psi(W,\theta_{0},\eta_{0})-\hat{\sigma}_x^{-1}g(x)^T\hat{J}^{-1}\psi(W,\hat{\theta}_{},\hat{\eta}_{})\Big\|_{\Pb_n,2}\\
\le\ &\sigma_x^{-1}\Big\|g(x)^T\Big(J_0^{-1}-\hat{J}^{-1}\Big)\psi(W,\theta_{0},\eta_{0})\Big\|_{\Pb_n,2}\\
&+\sigma_x^{-1}\Big\|g(x)^T\hat{J}^{-1}\Big(\psi(W,\theta_{0},\eta_{0})-\psi(W,\hat{\theta}_{},\hat{\eta}_{})\Big)\Big\|_{\Pb_n,2}\\
&+|\sigma_x^{-1}-\hat{\sigma}_x^{-1}|\Big\|g(x)^T\hat{J}^{-1}\psi(W,\hat{\theta}_{},\hat{\eta}_{})\Big\|_{\Pb_n,2}\\
\lesssim\ &\sqrt{t_1}\tilde{\epsilon}_n+\sqrt{\frac{t_1^3\varsigma_ns\log(\bar{d}_n)}{n}}+\sqrt{t_1}\epsilon_n\\
\lesssim\ &\Bigg(\frac{t_1s_2\varsigma_n\log(\bar{d}_n\varsigma_n^{1/2})}{n}+\left(n^{\frac{1}{q}}\frac{t_1s_2\varsigma_n\log(\bar{d}_n\varsigma_n^{1/2})}{n}\right)^2+\frac{t_1^3\varsigma_ns\log(\bar{d}_n)}{n}\\
&+n^{\frac{3}{q}}\frac{t_1^4\varsigma_n\log(d_1)}{n}+\frac{t_1^5s\varsigma_n\log(\bar{d}_n)}{n}\Bigg)^{1/2}\\
\lesssim & \ \sqrt{\frac{t_1^5s\varsigma_n\log(\bar{d}_n)}{n}}
\end{align*}
for q large enough by \eqref{J2Rate}, \eqref{Jconvergence} and using analogous arguments as above. Therefore, B.\ref{assumption2.3}$(iii)$ holds with
$$\bar{\delta}_n\lesssim\sqrt{\frac{t_1^5s\varsigma_n\log(\bar{d}_n)}{n}}$$

To complete the proof, we verify all growth conditions from Assumptions B.\ref{assumption2.3} and B.\ref{assumption2.4}. As shown in the verification of B.\ref{assumption2.2}$(vi)$, it holds
\begin{align*}
t_1^2\delta_n^2\varrho_n\log(A_n)=t_1^3n^{2/q-1}s^2\varsigma_n\log^2(\bar{d}_n) \log(\tilde{A}_n\varsigma_n^{-1/2})=o(1).
\end{align*}
Additionally,
\begin{align*}
n^{-\frac{1}{7}}L_n^{\frac{2}{7}}\varrho_n\log(A_n)=\left(\frac{t_1^{16}\varsigma_n^3\log^7(\tilde{A}_n\varsigma_n^{-1/2})}{n}\right)^{1/7}=o(1)
\end{align*}
and 
\begin{align*}
n^{\frac{2}{3q}-\frac{1}{3}}L_n^{\frac{2}{3}}\varrho_n\log(A_n)=\left(n^{\frac{2}{q}}\frac{t_1^{12}\varsigma_n^3\log^3(\tilde{A}_n\varsigma_n^{-1/2})}{n}\right)^{1/3}=o(1)
\end{align*}
for $q$ large enough due to growth condition in Assumption A.\ref{A2}$(v)(c)$. Note that 
\begin{align*}
\varepsilon_n\varrho_n\log(A_n)&=\varepsilon_n t_1\log(\tilde{A}_n\varsigma_n^{-1/2})=t_1^2\tilde{u}_n\log(\tilde{A}_n\varsigma_n^{-1/2})\\
&=t_1^2\sqrt{\frac{t_1^2s\varsigma_n\log(\bar{d}_n)}{n}}\log(\tilde{A}_n\varsigma_n^{-1/2})\\
&=\sqrt{\frac{t_1^6s\varsigma_n\log(\bar{d}_n)\log^2(\tilde{A}_n\varsigma_n^{-1/2})}{n}}=o(1). 
\end{align*}
Hence, to conclude we need to show that
\begin{align*}
\bar{\delta}^2_n\bar{\varrho}_n\varrho_n\log(\bar{A}_n)\log(A_n)=\bar{\delta}_n^2 t_1^2\log(\tilde{A}_n)\log(\tilde{A}_n\varsigma_n^{-1/2})=o(1).
\end{align*}
which holds true since 
\begin{align*}
\bar{\delta}_n^2 t_1^2\log(\tilde{A}_n\varsigma_n^{-1/2})\log(\tilde{A}_n)
\lesssim \frac{t_1^7s\varsigma_n\log(\bar{d}_n)\log(\tilde{A}_n\varsigma_n^{-1/2})\log(\tilde{A}_n)}{n}
=o(1)
\end{align*}
due to Assumption A.\ref{A2}$(v)(b)$.

\end{proof}

%% file: Appendix.tex
\section{Uniformly valid confidence bands}\label{AppendixA}

As in \cite{belloni2018uniformly}, we consider the problem of estimating the set of parameters $\theta_{0,l}$ for $l=1,\dots,d_1$ in the moment condition model,
\begin{align}\label{momentcondition}
\E[\psi_{l}(W,\theta_{0,l},\eta_{0,l})]=0, \qquad l=1,\dots,d_1,
\end{align} 
where $W$ is a random variable, $\psi_l$ a known score function, $\theta_{0,l}\in \Theta_l$ a scalar of interest, and $\eta_{0,l}\in T_l$ a high-dimensional nuisance parameter where $T_l$ is a convex set in a normed space equipped with a norm $\|\cdot\|_e$. Let $\mathcal{T}_l$ be some subset of $T_l$, which contains the nuisance estimate $\hat{\eta_l}$ with high probability. Belloni et al. (2018) provide an appropriate estimator $\hat{\theta}_{l}$ and are able to construct simultaneous confidence bands for $(\theta_{0,l})_{l=1,\dots,d_1}$ where $d_1$ may increase with sample size $n$. In this section, we are particularly interested in the linear functional
\begin{align*}
G(x)=\sum\limits_{l=1}^{d_1}\theta_{0,l}g_l(x),
\end{align*}
where $(g_l)_{l=1,\dots,d_1}$ is a given set of functions with
\begin{align*}
g_l: I\subseteq\R\rightarrow \R, \qquad l=1,\dots,d_1. 
\end{align*}
We assume that the score functions $\psi_l$ are constructed to satisfy the near-orthogonality condition, namely
\begin{align}\label{nearorthogonality}
D_{l,0}[\eta,\eta_{0,l}]&:=\partial_t\big\{\E[\psi_l(W,\theta_{0,l},\eta_{0,l}+t(\eta-\eta_{0,l}))]\big\}\big|_{t=0}=O\left(\delta_n \varsigma_n^{-1/2} n^{-1/2}\right),
\end{align}
where $\partial_t$ denotes the derivative with respect to $t$. The sequences of positive constants $\varsigma_n^{-1}$  and $(\delta_n)_{n\ge 1}$ converge to zero and are defined in Assumption $B$. We aim to construct uniform valid confidence bands for the target function $G(x)$, namely 
\begin{align*}
P(\hat{l}(x)\le G(x)\le\hat{u}(x), \forall x\in I)\rightarrow 1-\alpha.
\end{align*}
Let $\hat{\eta}_l=\left(\hat{\eta}_l^{(1)},\hat{\eta}_l^{(2)}\right)$ be an estimator of the nuisance function. The estimator $\hat{\theta}_0$ of the target parameter
$$\theta_0=(\theta_{0,1},\dots,\theta_{0,d_1})^T$$
is defined as the solution of
\begin{align}\label{estimator_app}
\sup\limits_{l=1,\dots,d_1}\left\{\left|\mathbb{E}_n^{}\Big[\psi_{l}\big(W,\hat{\theta}_{l},\hat{\eta}_{l}\big)\Big]\right|-\inf_{\theta\in\Theta_{l}}\left|\mathbb{E}_n^{}\Big[\psi_{l}\big(W,\theta,\hat{\eta}_{l}\big)\Big]\right|\right\}\le\epsilon_{n},
\end{align}
where $\epsilon_{n}=O\left(\delta_n\varsigma_n^{-1/2}n^{-1/2}\right)$ is the numerical tolerance. Let
\begin{align*}
g(x)=\left(g_1(x),\dots,g_{d_1}(x)\right)^T\in\R^{d_1\times 1}
\end{align*}
and
\begin{align*}
\psi(W,\theta,\eta)=\left(\psi_1(W,\theta,\eta),\dots,\psi_{d_1}(W,\theta,\eta)\right)^T\in\R^{d_1\times 1}.
\end{align*}
Define the Jacobian matrix
\begin{align*}
J_0:=\frac{\partial}{\partial\theta}\E[\psi(W,\theta,\eta_0)]\bigg|_{\theta=\theta_0}=\diag\left(J_{0,1},\dots,J_{0,d_1}\right)\in \mathbb{R}^{d_1\times d_1}
\end{align*}
and the approximate covariance matrix
\begin{align*}
\Sigma_n:&=J_0^{-1}\E\big[\psi(W,\theta_0,\eta_0)\psi(W,\theta_0,\eta_0)^T\big](J_0^{-1})^T\in \mathbb{R}^{d_1\times d_1}.
\end{align*}
Additionally, define
\begin{align*}
\mathcal{S}_n:=\E\left[\sup\limits_{l=1,\dots,d_1}\left|\sqrt{n}\E_n\left[\psi_l(W,\theta_{0,l},\eta_{0,l})\right]\right|\right]
\end{align*}
and
$$t_1:= \sup\limits_{x\in I} \|g(x)\|_0.$$
The definition of $t_1$ is helpful if the functions $g_l$, $l=1,\dots,d_1$ are local in the sense that for any point $x$ in $I$ there are at most $t_1\ll d_1$ non-zero functions. We state the conditions needed for the uniformly valid confidence bands.
\begin{assumptionb}\label{assumption2.0}
It holds
\begin{itemize}
\item[(i)]  $\frac{1}{\sqrt{t_1}}\lesssim \inf\limits_{x\in I} \|g(x)\|_2\le C<\infty\quad$ and  $\quad\sup_{x\in I} \sup_{l=1,\dots,d_1}|g_l(x)|\le C$.
\item[(ii)] Further,
$$ c \varsigma_n\le \inf\limits_{\|\xi\|_2=1} \xi^T\Sigma_n\xi\le \sup\limits_{\|\xi\|_2=1} \xi^T\Sigma_n\xi \le C \varsigma_n.$$
\end{itemize}
\end{assumptionb}
\noindent
It is worth noting that Assumption B.\ref{assumption2.0} is more general as in \cite{belloni2018uniformly} because we allow the eigenvalue to diverge at a given rate $\varsigma_n$. This is needed since we allow the number of function $g_l$, $l=1,\dots,d_1$, to grow with sample since. Typically, it holds $\varsigma_n\approx d_1$. As the proof of our main result in this section relies on the techniques in  \cite{belloni2018uniformly}, we try formulate the following conditions as similar as possible to make the use of their methodology transparent. 
\begin{assumptionb}\label{assumption2.1}
For all $n\ge n_0, P\in \mathcal{P}_n$ and $l\in\{1,\dots,d_1\}$, the following conditions hold:
\begin{itemize}
\item[(i)] The true parameter value $\theta_{0,l}$ obeys (\ref{momentcondition}), and $\Theta_{l}$ contains a ball of radius $C_0n^{-1/2}\mathcal{S}_n\varsigma_n\log(n)$ centered at $\theta_{0,l}$.
\item[(ii)] The map $(\theta_l,\eta_l)\mapsto \E[\psi_l(W,\theta_l,\eta_l)]$ is twice continuously Gateaux-differentiable on $\Theta_l\times\mathcal{T}_l$.
\item[(iii)] The score function $\psi_l$ obeys the near orthogonality condition (\ref{nearorthogonality}) for the set $\mathcal{T}_l\subset T_l$.
\item[(iv)] For all $\theta_l\in\Theta_l$, $|\E[\psi_l(W,\theta_l,\eta_{0,l})]|\ge 2^{-1} |J_{0,l}(\theta_l-\theta_{0,l})|\wedge c_0$, where $J_{0,l}$ satisfies $c_0\varsigma_n^{-1}\le |J_{0,l}|\le C_0\varsigma_n^{-1}$.
\item[(v)] For all $r\in [0,1)$, $\theta_l\in\Theta_l$ and $\eta_l\in \mathcal{T}_l$
\begin{itemize}
\item[(a)] $\E[(\psi_l(W,\theta_l,\eta_l)-\psi_l(W,\theta_{0,l},\eta_{0,l}))^2]\le C_0(|\theta_l-\theta_{0,l}|\varsigma_n^{-1/2}\vee\|\eta_l-\eta_{0,l}\|_e)^2$
\item[(b)]$|\partial_r\E[\psi_l(W,\theta_l,\eta_{0,l}+r(\eta_l-\eta_{0,l}))]|\le \varsigma_n^{-1/2}\|\eta_l-\eta_{0,l}\|_e$
\item[(c)] $|\partial^2_r\E[\psi_l(W,\theta_{0,l}+r(\theta_l-\theta_{0,l}),\eta_{0,l}+r(\eta_l-\eta_{0,l}))]|\le C(|\theta_l-\theta_{0,l}|^2\varsigma_n^{-1}\vee\|\eta_l-\eta_{0,l}\|_e^2).$
\end{itemize}
\end{itemize}
\end{assumptionb}
\noindent
Note that the notation $\E$ abbreviates $\E_P$. For a detailed discussion about the ideas and intuitions of these and the following assumptions, see \cite{belloni2018uniformly}.\\
Let $(\Delta_n)_{n\ge 1}$ and $(\tau_n)_{n\ge 1}$ be some sequences of positive constants converging to zero. Also, let $(a_n)_{n\ge 1}$ and $(\upsilon_n)_{n\ge 1}$ be some sequences of positive constants, possibly growing to infinity where $a_n\ge n$ and $\upsilon \ge 1$ for all $ n \ge 1$. Finally, let $q\ge 2$ be some constant.
\begin{assumptionb}\label{assumption2.2}
For all $n\ge n_0$ and $P\in \mathcal{P}_n$, the following conditions hold:
\begin{itemize}
\item[(i)] With probability at least $1-\Delta_n$, we have $\hat{\eta}_l\in\mathcal{T}_l$ for all $l=1,\dots,d_1$.
\item[(ii)] For all $l=1,\dots,d_1$ and $\eta_l\in \mathcal{T}_l$, it holds $\|\eta_l-\eta_{0,l}\|_e\le \tau_n$.
\item[(iii)] For all $l=1,\dots,d_1$, we have $\eta_{0,l}\in \mathcal{T}_l$.
\item[(iv)] The function class $\mathcal{F}_1=\{\psi_l(\cdot,\theta_l,\eta_l):\l=1,\dots,d_1,\theta_l\in\Theta_l,\eta_l\in\mathcal{T}_l\}$ is suitably measurable and its uniform entropy numbers obey
\begin{align*}
\sup\limits_Q\log N(\epsilon\|F_1\|_{Q,2},\mathcal{F}_1,\|\cdot\|_{Q,2})\le \upsilon_n\log(a_n/\epsilon),\quad \text{for all } 0<\epsilon\le 1,
\end{align*}
where $F_1$ is a measurable envelope for $\mathcal{F}_1$ that satisfies $\|F_1\|_{P,q}\le C$.
\item[(v)] For all $f\in \mathcal{F}_1$, we have $c_0\varsigma_n^{-1/2}\le \|f\|_{P,2}\le C_0\varsigma_n^{-1/2}$.
\item[(vi)] The complexity characteristics $a_n$ and $\upsilon_n$ satisfy
\begin{itemize}
\item[(a)] $(\upsilon_n \log(a_n)/n)^{1/2}\le C_0\tau_n,$
\item[(b)] $(\tau_n+\mathcal{S}_n\log(n)/\sqrt{n})(\upsilon_n\log(a_n))^{1/2}+n^{-1/2+1/q}\upsilon_n \log(a_n)\le \delta_n \varsigma_n^{-1/2},$
\item[(c)] $n^{1/2}\tau_n^2\le \delta_n \varsigma_n^{-1/2}.$
\end{itemize}
\end{itemize}
\end{assumptionb}
\noindent
Whereas the Assumptions B.\ref{assumption2.1} and B.\ref{assumption2.2} are closely related to the Assumptions 2.1 and 2.2 from \cite{belloni2018uniformly}, the analogs to their Assumptions 2.3 and 2.4 need more modifications to fit our setting in order to construct a uniformly valid confidence band for the linear functional $G(x)$. In this context, define 
$$\psi_x(\cdot):=(g(x)^T\Sigma_n g(x))^{-1/2}g(x)^TJ_0^{-1}\psi(\cdot,\theta_{0},\eta_{0})$$
and the corresponding plug-in estimator
$$\hat{\psi}_x(\cdot):=(g(x)^T\hat{\Sigma}_n g(x))^{-1/2}g(x)^T\hat{J}_0^{-1}\psi(\cdot,\hat{\theta}_{0},\hat{\eta}_{0}).$$
Let $(\bar{\delta}_n)_{n\ge 1}$ be a sequence of positive constants converging to zero. Also, let $(\varrho_n)_{n\ge 1}$, $(\bar{\varrho}_n)_{n\ge 1}$, $(A_n)_{n\ge 1}$, $(\bar{A}_n)_{n\ge 1}$, and $(L_n)_{n\ge 1}$ be some sequences of positive constants, possibly growing to infinity where $\varrho \ge 1$, $A_n \ge n$, and $\bar{A}_n\ge n$ for all $n\ge 1$. In addition, assume that $q > 4$.
\begin{assumptionb}\label{assumption2.3}
For all $n\ge n_0$ and $P\in \mathcal{P}_n$, the following conditions hold:
\begin{itemize}
\item[(i)] The function class $\mathcal{F}_0=\{\psi_x(\cdot):x\in I\}$ is suitably measurable and its uniform entropy numbers obey
\begin{align*}
\sup\limits_Q\log N(\varepsilon\|F_0\|_{Q,2},\mathcal{F}_0,\|\cdot\|_{Q,2})\le \varrho_n\log(A_n/\varepsilon),\quad \text{for all } 0<\epsilon\le 1,
\end{align*}
where $F_0$ is a measurable envelope for $\mathcal{F}_0$ that satisfies $\|F_0\|_{P,q}\le L_n$.
\item[(ii)] For all $f\in \mathcal{F}_0$ and $k=3,4$, we have $\E[|f(W)|^k]\le C_0 L_n^{k-2}$.
\item[(iii)] The function class $\hat{\mathcal{F}}_0=\{\psi_x(\cdot)-\hat{\psi}_x(\cdot): x\in I\}$ satisfies with probability $1-\Delta_n$: 
\begin{align*}
\log N(\varepsilon,\hat{\mathcal{F}}_0,\|\cdot\|_{\mathbb{P}_n,2})\le \bar{\varrho}_n\log(\bar{A}_n/\varepsilon),\quad \text{for all } 0<\epsilon\le 1,
\end{align*}
and $\|f\|_{\mathbb{P}_n,2}\le \bar{\delta}_n$ for all $f\in\hat{\mathcal{F}}_0$.
\item[(iv)] $t_1^2\delta_n^2\varrho_n\log(A_n)=o(1)$, $L_n^{2/7}\varrho_n\log(A_n) =o(n^{1/7})$ and $L_n^{2/3}\varrho_n\log(A_n)=o(n^{1/3-2/(3q)})$.
\end{itemize}
\end{assumptionb}
\noindent
Additionally, we need to be able to estimate the variance of the linear functional sufficiently well. Let $\hat{\Sigma}_n$ be an estimator of $\Sigma_n$.
\begin{assumptionb}\label{assumption2.4}
 For all $n\ge n_0$ and $ P\in\mathcal{P}_n$, it holds
\begin{align*}
P\left(\sup\limits_{x\in I} \left|\frac{(g(x)^T\hat{\Sigma}_n g(x))^{1/2}}{(g(x)^T\Sigma_n g(x))^{1/2}}-1\right|>\varepsilon_n\right)\le \Delta_n,
\end{align*}
where $\varepsilon_n\varrho_n\log(A_n)=o(1)$ and $\bar{\delta}_n^2\bar{\varrho}_n\varrho_n\log(\bar{A}_n)\log(A_n)=o(1)$.
\end{assumptionb}
\noindent
As in \cite{chernozhukov2013gaussian}, we employ the Gaussian multiplier bootstrap method to estimate the relevant quantiles. Let
\begin{align*}
\hat{\mathcal{G}}=\left(\hat{\mathcal{G}}_x\right)_{x\in I}=\left(\frac{1}{\sqrt{n}}\sum\limits_{i=1}^n \xi_i\hat{\psi}_x(W_i) \right)_{x\in I},
\end{align*}
where $(\xi_i)_{i=1}^n$ are independent standard normal random variables (especially independent from $(W_i)_{i=1}^n$). Define the multiplier bootstrap critical value $c_\alpha$ as the $(1-\alpha)$ quantile of the conditional distribution of $\sup_{x\in I}|\hat{\mathcal{G}}_x|$ given $(W_i)_{i=1}^n$.
\begin{theorem}\label{uniformappendix}
Define 
\begin{align*}
\hat{u}(x)&:= \hat{G}(x)+\frac{(g(x)^T\hat{\Sigma}_n g(x))^{1/2}c_\alpha}{\sqrt{n}}\\
\hat{l}(x)&:= \hat{G}(x)-\frac{(g(x)^T\hat{\Sigma}_n g(x))^{1/2}c_\alpha}{\sqrt{n}}
\end{align*}
with $\hat{G}(x)=g(x)^T\hat{\theta}_0$. Under the Assumptions B.\ref{assumption2.0} - B.\ref{assumption2.4}, it holds
\begin{align*}
P\left(\hat{l}(x)\le G(x)\le \hat{u}(x), \forall x\in I\right)\to 1-\alpha
\end{align*}
uniformly over $P\in\mathcal{P}_n$.
\end{theorem}
\begin{proof} Since Theorem 2.1 in \cite{belloni2018uniformly} is not directly applicable to our problem, we have to modify the proof to obtain a uniform Bahadur representation. We want to prove that 
\begin{align}\label{bahadur}
\sup\limits_{x\in I}\Big|\sqrt{n}(g(x)^T\Sigma_n g(x))^{-1/2} g(x)^T\big(\hat{\theta}-\theta_0\big)\Big|=\sup\limits_{x\in I}\Big|\mathbb{G}_n(\psi_x)\Big| + O_P(t_1\delta_n),
\end{align}
where $\delta_n$ is a sequence of positive constants converging to zero. It is worth noting that for some kind of approximations functions, e.g. B-Splines, we can get rid of $t_1$ on the right hand side of the equation (due to partition of unity). Assumptions B.\ref{assumption2.1} and B.\ref{assumption2.2} contain Assumptions 2.1 and 2.2 from \cite{belloni2018uniformly} which enables us to adapt their proof to show that
\begin{align}\label{bahadur_priliminary}
\sup\limits_{l=1,\dots,d_1}\left|J_{0,l}^{-1}\sqrt{n}\E_n\left[\psi_l(W,\theta_{0,l},\eta_{0,l})\right]+\sqrt{n}\left(\hat{\theta}_l-\theta_{0,l}\right)\right|=O_P(\delta_n\varsigma^{1/2}).
\end{align} 
\textbf{Step 1:} 
We first show that $\sup_{l=1,\dots,d_1}|\hat{\theta}_l-\theta_{0,l}|\le \varsigma_n^{1/2}\tau_n$ with probability $1-o(1)$. By the definition of $\hat{\theta}$, we have
\begin{align*}
    \left|\E_n\left[\psi_l(W,\hat{\theta}_l,\hat{\eta_{l}})\right]\right|\le \inf_{\theta\in\Theta_l} \left|\E_n\left[\psi_l(W,\theta,\hat{\eta_{l}})\right]\right|+\epsilon_n,
\end{align*}
which implies
\begin{align*}
    \left|\E\left[\psi_l(W,\theta,\eta_{l})\right]\big|_{\theta=\hat{\theta}_l}\right|\le\epsilon_n+2I_1+2I_2\le \varsigma_n^{-1/2}\tau_n,
\end{align*}
where
\begin{align*}
  I_1:=\sup_{l=1,\dots,d_1,\theta\in\Theta_l}\left| \E_n\left[\psi_l(W,\theta,\hat{\eta_{l}})\right]-\E_n\left[\psi_l(W,\theta,\eta_{l})\right]\right|\le \varsigma_n^{-1/2}\tau_n
\end{align*}
and
\begin{align*}
       I_2:=\sup_{l=1,\dots,d_1,\theta\in\Theta_l}\left| \E_n\left[\psi_l(W,\theta,\eta_{l})\right]-\E\left[\psi_l(W,\theta,\eta_{l})\right]\right|\le \varsigma_n^{-1/2}\tau_n.
\end{align*}
The bounds on $I_1$ and $I_2$ are derived in Step 2. Note that by Assumption B.\ref{assumption2.2} (iv), we conclude that
\begin{align*}
    \sup_{l=1,\dots,d_1}|\hat{\theta}_l-\theta_{0,l}|\le  \left(\inf_{l=1,\dots,d_1}|J_l|\right)^{-1}\varsigma_n^{-1/2}\tau_n\le \varsigma_n^{1/2}\tau_n
\end{align*}
with probability $1-o(1)$.\\ \\
\textbf{Step 2:} It holds $I_1\le 2I_{1a}+I_{1b}$ and $I_2\le I_{1a}$ with
\begin{align*}
    I_{1a}:= \sup_{l=1,\dots,d_1,\theta\in\Theta_l,\eta\in\mathcal{T}_l}\left| \E_n\left[\psi_l(W,\theta,\eta)\right]-\E\left[\psi_l(W,\theta,\eta)\right]\right|
\end{align*}
and\begin{align*}
    I_{1b}:= \sup_{l=1,\dots,d_1,\theta\in\Theta_l,\eta\in\mathcal{T}_l}\left| \E\left[\psi_l(W,\theta,\eta)\right]-\E\left[\psi_l(W,\theta,\eta_{0,l})\right]\right|.
\end{align*}
To bound $I_{1b}$, we employ Taylor expansion:
\begin{align*}
    I_{1b}&\le \sup_{l=1,\dots,d_1,\theta\in\Theta_l,\eta\in\mathcal{T}_l,r\in [0,1)} \partial_r\E\left[\psi_l(W,\theta,\eta_{0,l}+r(\eta-\eta_{0,l}))\right]\\
    &\le \varsigma^{-1/2} \sup_{l=1,\dots,d_1,\eta\in\mathcal{T}_l}\|\eta-\eta_{0,l}\|\le \varsigma^{-1/2}\tau_n.
\end{align*}
To bound $I_{1a}$, we apply Lemma P.2 in \cite{belloni2018uniformly} to the class $\mathcal{F}_1$ defined in Assumption B.\ref{assumption2.3}:
\begin{align*}
    I_{1a}\le n^{-1/2}\left(\sqrt{\varsigma^{-1}\upsilon_n\log(a_n)}+n^{-1/2+1/q}\upsilon_n\log(a_n)\right)\le \varsigma^{-1/2}\tau_n
\end{align*}
with probability $1-o(1)$ by the growth conditions in Assumption B.\ref{assumption2.3}. Combining presented bounds gives the claim of this step.\\ \\
\textbf{Step 3:} In this step, we prove the claim in \eqref{bahadur_priliminary}. By the definition of $\hat{\theta}$, we have
\begin{align}\label{eqA4}
    \sqrt{n}\left|\E_n\left[\psi_l(W,\hat{\theta}_l,\hat{\eta_{l}})\right]\right|\le \inf_{\theta\in\Theta_l} \sqrt{n}\left|\E_n\left[\psi_l(W,\theta,\hat{\eta_{l}})\right]\right|+\sqrt{n}\epsilon_n
\end{align}
for fixed $l=1,\dots,d_1$. For any $\theta\in\Theta_l$ and $\eta\in\mathcal{T}_l$, we have
\begin{align}\label{eqA5}
    \sqrt{n}\left|\E_n\left[\psi_l(W,\theta,\eta)\right]\right|&=\sqrt{n}\E_n\left[\psi_l(W,\theta_{0,l},\eta_{0,l})\right]-\mathbb{G}_n\left[\psi_l(W,\theta_{0,l},\eta_{0,l})\right]\\
    &\quad - \sqrt{n}\left(\E\left[\psi_l(W,\theta_{0,l},\eta_{0,l})\right]-\E\left[\psi_l(W,\theta,\eta)\right]\right)+\mathbb{G}_n\left[\psi_l(W,\theta,\eta)\right]\nonumber
\end{align}
and by Taylor expansion
\begin{align*}
   &\quad \E\left[\psi_l(W,\theta,\eta)\right]-\E\left[\psi_l(W,\theta_{0,l},\eta_{0,l})\right]\\
   & =J_{0,l}(\theta-\theta_{0,l})+D_{l,0}[\eta,\eta_{0,l}]\\
   &\quad + \frac{1}{2}\partial_r^2\E\left[\psi_l(W,\theta_{0,l}+r(\theta-\theta_{0,l}),\eta_{0,l}+r(\eta-\eta_{0,l}))\right]\big|_{r=\bar{r}}
\end{align*}
for some $\bar{r}\in (0,1)$. Substituting this equality into \eqref{eqA5} for $\theta=\hat{\theta}$ and $\eta=\hat{\eta}$ using \eqref{eqA4} gives
\begin{align}
    &\quad \sqrt{n}\left|\E_n\left[\psi_l(W,\theta_{0,l},\eta_{0,l})\right]+ J_{0,l}(\hat{\theta}-\theta_{0,l})+D_{l,0}[\hat{\eta},\eta_{0,l}]\right|\le \epsilon_n\sqrt{n}\\
    &+ \inf_{\theta\in\Theta_l}\sqrt{n}\left|\E_n\left[\psi_l(W,\theta,\hat{\eta})\right]\right|+\left|II_1\right|+\left|II_2\right|,\nonumber
\end{align}
where
\begin{align*}
    II_1(l)=\sqrt{n}\sup_{r\in [0,1)}\left|\partial_r^2\E\left[\psi_l(W,\theta_{0,l}+r(\theta-\theta_{0,l}),\eta_{0,l}+r(\eta-\eta_{0,l}))\right]\right|
\end{align*}
and
\begin{align*}
    II_2(l)= \mathbb{G}_n\left[\psi_l(W,\theta,\eta)-\psi_l(W,\theta_{0,l},\eta_{0,l})\right]
\end{align*}
with $\theta=\hat{\theta}$ and $\eta=\hat{\eta}$. It will be shown in Step 4 that
\begin{align}\label{step4II}
    \sup_{l=1,\dots,d_1}|II_1(l)| + |II_2(l)|=O_p\left(\delta_n\varsigma_n^{-1/2}\right)
\end{align}
and in Step 5 that
\begin{align}\label{step5}
   \sup_{l=1,\dots,d_1}\inf_{\theta\in\Theta_l}\sqrt{n}\left|\E_n\left[\psi_l(W,\theta,\hat{\eta})\right]\right|=O_p\left(\delta_n\varsigma_n^{-1/2}\right).
\end{align}
Moreover, it holds $\epsilon_n\sqrt{n}=O\left(\delta_n\varsigma_n^{-1/2}\right)$ by the construction of the estimator. Also $\sup_{l=1,\dots,d_1}D_{l,0}[\hat{\eta},\eta_{0,l}]=O\left(n^{-1/2}\varsigma_n^{-1/2}\delta_n\right)$ by near Neyman orthogonality with probability $1-o(1)$. Therefore, Assumption B.$\ref{assumption2.1}$ gives
\begin{align}
   \sup_{l=1,\dots,d_1}\left|J_{0,l}^{-1} \sqrt{n}\E_n\left[\psi_l(W,\theta_{0,l},\eta_{0,l})\right]+ \sqrt{n}(\hat{\theta}-\theta_{0,l})\right|=O_p(\delta_n\varsigma_n^{1/2}).
\end{align}
\textbf{Step 4:}
Here, we prove that \eqref{step4II} holds. First, with probability $1-o(1)$, we have
\begin{align*}
    \sup_{l=1,\dots,d_1}|II_1(l)|=\sqrt{n}\sup_{l=1,\dots,d_1}\left(|\hat{\theta}-\theta_{0,l}|^2\varsigma_n^{-1}\vee \|\hat{\eta}_l-\eta_{0,l}\|_e^2\right)\lesssim\sqrt{n}\tau_n^2\le \delta_n\varsigma_n^{-1/2}
\end{align*}
using Assumption B.\ref{assumption2.1}, Assumption B.\ref{assumption2.2} (vi) and Step 1. Second, we have $$\sup_{l=1,\dots,d_1}|II_2(l)|\lesssim\sup_{f\in\mathcal{F}_2}\left|\mathbb{G}_n[f]\right|$$ with
\begin{align*}
\mathcal{F}_2:=\{\psi_l(W,\theta,\eta)-\psi_l(W,\theta_{0,l},\eta_{0,l}),l=1,\dots,d_1,\eta\in\mathcal{T}_l,|\theta-\theta_{0,l}|\le C\varsigma^{1/2}\tau_n\}
\end{align*}
for a sufficiently large constant $C$. To bound $\sup_{f\in\mathcal{F}_2}\left|\mathbb{G}_n[f]\right|$, we apply Lemma P.2 in \cite{belloni2018uniformly}. It holds
\begin{align*}
    \sup_{f\in\mathcal{F}_2}\|f\|_{P,2}^2&\le \sup_{l=1,\dots,d_1,|\theta-\theta_{0,l}|\le C\varsigma^{1/2}\tau_n,\eta\in\mathcal{T}_l}\E\left[(\psi_l(W,\theta,\eta)-\psi_l(W,\theta_{0,l},\eta_{0,l}))^2\right]\\
    &\le \sup_{l=1,\dots,d_1,|\theta-\theta_{0,l}|\le C\varsigma^{1/2}\tau_n,\eta\in\mathcal{T}_l} C\left(|\theta-\theta_{0,l}|\varsigma^{-1/2}\vee \|\eta-\eta_{0,l}\|_e\right)^2\\
    &\lesssim\tau_n^2
\end{align*}
due to Assumption B.\ref{assumption2.1}. Therefore, with probability $1-o(1)$, an application of Lemma P.2 in \cite{belloni2018uniformly} with an envelope $F_2=2F_1$ and $\sigma=C\tau_n$ for sufficiently large $C$ gives
\begin{align*}
    \sup_{f\in\mathcal{F}_2}\left|\mathbb{G}_n[f]\right|\lesssim \tau_n\sqrt{\upsilon_n\log (a_n)}+n^{-1/2+1/q}\upsilon_n\log (a_n)\lesssim \delta_n\varsigma^{-1/2}.
\end{align*}
Here, we used $\|F_2\|_{P,q}\le C$ and
\begin{align*}
\sup\limits_Q\log N(\epsilon\|F_2\|_{Q,2},\mathcal{F}_2,\|\cdot\|_{Q,2})\lesssim \upsilon_n\log(a_n/\epsilon),\quad \text{for all } 0<\epsilon\le 1
\end{align*}
by Assumption B.\ref{assumption2.2}.\\ \\
\textbf{Step 5:} Here, we prove \eqref{step5}. For all $l=1,\dots,d_1$, define $\bar{\theta}_l=\theta_l-J_{0,l}^{-1}\E_n[\psi_l(W,\theta_{0,l},\eta_{0,l})]$. Then, we have $\sup_{l=1,\dots,d_1}|\bar{\theta}_l-\theta_l|=O_p(\mathcal{S}_n\varsigma_nn^{-1/2})$ and thus $\bar{\theta}_l\in\Theta_l$ with probability $1-o(1)$ by Assumption B.\ref{assumption2.1} (i) and (iv). Hence,
\begin{align*}
   \inf_{\theta\in\Theta_l}\sqrt{n}\left|\E_n\left[\psi_l(W,\theta,\hat{\eta})\right]\right|\le \sqrt{n}\left|\E_n\left[\psi_l(W,\bar{\theta}_l,\hat{\eta})\right]\right|
\end{align*}
 for all $l=1,\dots,d_1$ with probability $1-o(1)$ and it suffices to show that 
 \begin{align*}
  \sqrt{n}\left|\E_n\left[\psi_l(W,\bar{\theta}_l,\hat{\eta})\right]\right|=O_p\left(\delta_n\varsigma_n^{-1/2}\right).
\end{align*}
We substitute $\theta=\bar{\theta}_l$ and $\eta=\hat{\eta}_l$ in \eqref{eqA5} and use our Taylor expansion which gives
\begin{align*}
    \sqrt{n}\left|\E_n\left[\psi_l(W,\bar{\theta}_l,\hat{\eta})\right]\right|&\le \left|\tilde{II}_1\right|+\left|\tilde{II}_2\right|+ \sqrt{n}\left|\E_n\left[\psi_l(W,\theta_{0,l},\eta_{0,l})\right]\right|\\
    &\quad +\sqrt{n}J_{0,l}(\bar{\theta}_l-\theta_{0,l})+\sqrt{n}D_{l,0}[\hat{\eta}_l,\eta_{0,l}].
\end{align*}
The terms $\tilde{II}_1$ and $\tilde{II}_2$ are essentially the same as $II_1$ and $II_2$ but $\hat{\theta}_l$ replaced by $\bar{\theta}_l$. Then, given $\sup_{l=1,\dots,d_1}|\bar{\theta}_l-\theta_{0,l}|\lesssim \mathcal{S}_n\varsigma_n\log(n)/\sqrt{n}$ with probability $1-o(1)$, the argument in Step 4 shows that
\begin{align*}
    \sup_{l=1,\dots,d_1}|\tilde{II}_1(l)| + |\tilde{II}_2(l)|=O_p\left(\delta_n\varsigma_n^{-1/2}\right).
\end{align*}
In addition $\E_n\left[\psi_l(W,\theta_{0,l},\eta_{0,l})\right]+J_{0,l}(\bar{\theta}_l-\theta_{0,l})=0$ by the definition of $\bar{\theta}_l$ and $D_{l,0}[\hat{\eta}_l,\eta_{0,l}]=O_p\left(\delta_n\varsigma_n^{-1/2}n^{-1/2}\right)$ by near Neyman orthogonality. This completes the proof of \eqref{step5} and therefore of \eqref{bahadur_priliminary}.\\ \\ \textbf{Step 6:} Next, we conclude our main results. Using Assumption B.\ref{assumption2.0}, this implies
\begin{align*}
&\ \sup\limits_{x\in I}\left|\sqrt{n}\E_n\left[g(x)^TJ_0^{-1}\psi(W,\theta_0,\eta_0)\right]+\sqrt{n}g(x)^T\big(\hat{\theta}-\theta_0\big)\right|\\
=&\ \sup\limits_{x\in I}\left|\sum\limits_{l=1}^{d_1}g_l(x)\left(J_{0,l}^{-1}\sqrt{n}\E_n\left[\psi_l(W,\theta_{0,l},\eta_{0,l})\right]+\sqrt{n}(\hat{\theta_l}-\theta_{0,l})\right)\right|\\
\le &\ t_1\underbrace{\sup\limits_{x\in I}\sup\limits_{l=1,\dots,d_1}|g_l(x)|}_{\le C}\sup\limits_{l=1,\dots,d_1}\left|J_{0,l}^{-1}\sqrt{n}\E_n\left[\psi_l(W,\theta_{0,l},\eta_{0,l})\right]+\sqrt{n}\left(\hat{\theta}_l-\theta_{0,l}\right)\right|\\
=&\ O_p\left(t_1\delta_n\varsigma^{1/2}\right).
\end{align*}
Due to Assumption B.\ref{assumption2.0}, it holds $\left(g(x)^T\Sigma_n g(x)\right)^{-1/2}=O_p(\varsigma^{-1/2})$ which implies (\ref{bahadur}).
By Assumption B.\ref{assumption2.4}, we have
$$P\left(\sup\limits_{x\in I} \left|\frac{(g(x)^T\hat{\Sigma}_n g(x))^{1/2}}{(g(x)^T\Sigma_n g(x))^{1/2}}-1\right|>\varepsilon_n\right)\le \Delta_n,$$
with $\Delta_n=o(1)$, which is an analogous version of the Assumption 2.4 from \cite{belloni2018uniformly}. Therefore, given the Assumptions  B.\ref{assumption2.1} - B.\ref{assumption2.4}, the proofs of Corollary 2.1 and 2.2 from \cite{belloni2018uniformly} can be applied implying the stated theorem.
\end{proof}

%% file: Uniform_lasso.tex
\section{Uniform nuisance function estimation}\label{uniformestimation}
\noindent 
To establish uniform estimation properties of the nuisance function, we consider the following linear regression model
$$ Y_r=\sum\limits_{j=1}^p \beta_{r,j} X_{r,j}+\varepsilon_{r}=\beta_r ^TX_r+\varepsilon_r,$$
where the true parameter obeys
$$\beta_r\in \arg\min_{\beta}\E[(Y_r-\beta^T X_r)^2],$$ with $\E[\varepsilon_rX_r]=0$ for each $r =1,\dots,d$.
We show that the lasso and post-lasso lasso estimators have sufficiently fast uniform estimation rates if the vector $\beta_r$ is sparse for all $r=1,\dots,d$. In this setting, $d=d_n$ is explicitly allowed to grow with $n$.
\subsection{Uniform lasso estimation}\label{uniformla}
Define the weighted lasso estimator
\begin{align*}
 \hat{\beta}_r\in \arg\min\limits_{\beta}\left(\frac{1}{2}\E_n\left[\left(Y_r-\beta^T X_r\right)^2\right]+\frac{\lambda}{n}\|\hat{\Psi}_{r,m}\beta\|_1\right)
\end{align*}
with the penalty level
\begin{align*}
\lambda = c_\lambda\sqrt{n}\Phi^{-1}\left(1-\frac{\gamma}{2pd}\right)
\end{align*}
for a suitable $c_\lambda >1$, $\gamma\in[1/n,1/\log(n)]$ and a fix $m\ge 0$. Define the post-regularized weighted least squares estimator as
\begin{align*}
\tilde{\beta}_r\in \arg\min\limits_{\beta}\left(\frac{1}{2}\E_n\left[\left(Y_r-\beta^T X_r\right)^2\right]\right):\quad \text{supp}(\beta)\subseteq \text{supp}(\hat{\beta}_r).
\end{align*}
The penalty loadings $\hat{\Psi}_{r,m}=\diag(\{\hat{l}_{r,j,m},j=1,\dots,p\})$ are defined by
$$\hat{l}_{r,j,0}=\max\limits_{1\le i\le n}||X_r^{(i)}||_\infty$$ 
for $m=0$ and for all $m\ge 1$ by the following algorithm:
\begin{algorithm}
	\caption{Penalty loadings} \label{penalg}
	\begin{enumerate}
	\item Set $\bar{m}=0$. Compute $\hat{\beta}_r$ based on $\hat{\Psi}_{r,\bar{m}}$.
	\item Set $\hat{l}_{r,j,\bar{m}+1}=\E_n\left[\left(\left(Y_r-\hat{\beta}_r^T X_r\right) X_{r,j}\right)^2\right]^{1/2}$.
	\item If $\bar{m}=m$ stop and report the current value of $\hat{\Psi}_{r,m}$, otherwise set $\bar{m}=\bar{m}+1$.
	\end{enumerate}
\end{algorithm}
\ \\
\noindent Let $a_n:=\max(p,n,d,e)$. In order to establish uniform convergence rates, the following assumptions are required to hold uniformly in $n\ge n_0$ and $P\in\mathcal{P}_n$:\\
\begin{assumptionc}\label{condC}
\item[(i)]
We have (almost surely) $$\max\limits_{r=1,\dots,d}\max\limits_{j=1,\dots,p}|X_{r,j}| \le CM_n \text{ and }
\max_{r=1,\dots,d} |\varepsilon_r|\le C M_n$$
with $M_n\ge 1$. Further 
$$c\le\E\left[\varepsilon_r^2 X_{r,j}^2\right]\le C,$$
uniformly for all  $r=1,\dots,d_n$ and
\begin{align*}
\max\limits_{r=1,\dots,d}\max\limits_{j=1,\dots,p}\|\varepsilon_r X_{r,j}\|_{P,3}&\le CK_{n},\\
    \max\limits_{r=1,\dots,d}\max\limits_{j=1,\dots,p} \E\left[\varepsilon_r^4X_{r,j}^4\right]&\le C L_n
\end{align*}
\item[(ii)] For all $r=1,\dots,d_n$, it holds
\begin{align*}
\inf\limits_{\|\xi\|_2=1} \E\left[(\xi X_r)^2\right]\ge c \text{,} \sup\limits_{\|\xi\|_2=1} \E\left[(\xi X_r)^2\right]\le C.
\end{align*}
\item[(iii)] 
The coefficients obey an approximate sparsity condition
$$\beta_r= \beta_r^{(1)} + \beta_r^{(2)}$$
with
$$ \max\limits_{r=1,\dots,d}\|\beta_r^{(1)}\|_0\le s$$
and
$$\max\limits_{r=1,\dots,d}\|\beta_r^{(2)}\|_1^2\le C\sqrt{\frac{s^2\log(a_n)}{nM_n}},\quad \max\limits_{r=1,\dots,d}\|\beta_r^{(2)}\|_2^2\le C\frac{s\log(a_n)}{n},$$
\item[(iv)] There exists a positive number $\tilde{q}>0$ such that the following growth condition is fulfilled:
\begin{align*}
n^{\frac{1}{\tilde{q}}}\frac{s M_n^4\log^{}(a_n)}{n}\vee
\frac{L_n\log^{}(a_n)}{n^{}}\vee
\frac{K_{n}^{6}\log^{3}(a_n)}{n^{}} =o(1).
\end{align*}

\end{assumptionc}
In contrast to Assumption C.\ref{condC}(i), \cite{klaassen2018uniform} assume an uniformly bounded Orlicz norm of the regressors and errors. This assumption excludes growing (with $n$) moments of the regressors, which is, e.g., necessary for some series estimators such as standardized B-splines. The Assumption C.\ref{condC} is similar to Assumption 6.1 in  \cite{belloni2017program}, who also consider bounded support of the regressors, but they rely on bounded empirical sparse eigenvalues. Further, they bound the third moments of $\varepsilon_r X_{r,j}$, which is not possible for standardized B-splines.
It is worth noting that $K_n \le C\M_n$.

\begin{theorem}\label{uniformlasso}
Under the Assumption C.\ref{condC}, the lasso estimator $\hat{\beta}_r$ obeys uniformly over all $P\in \mathcal{P}_n$ with probability $1-o(1)$
\begin{align}
\max\limits_{r=1,\dots,d}\|\hat{\beta}_r-\beta_r\|_2&\le C\sqrt{\frac{s\log(a_n)}{n}},\\ \max\limits_{r=1,\dots,d}\|\hat{\beta}_r-\beta_r\|_1&\le C\sqrt{\frac{s^2\log(a_n)}{n}}
\end{align}
and
\begin{align}
\max\limits_{r=1,\dots,d}\|\hat{\beta}_r\|_0\le C s.
\end{align}
Additionally, the post-lasso estimator $\tilde{\beta}_r$ obeys uniformly over all $P\in \mathcal{P}_n$ with probability $1-o(1)$
\begin{align}
\max\limits_{r=1,\dots,d}\|\tilde{\beta}_r-\beta_r\|_2&\le C\sqrt{\frac{s\log(a_n)}{n}},\\ \max\limits_{r=1,\dots,d}\|\tilde{\beta}_r-\beta_r\|_1&\le C\sqrt{\frac{s^2\log(a_n)}{n}}.
\end{align}
\end{theorem}\ \\
\begin{proof}[Proof of Theorem \ref{uniformlasso}]\ \\ \\
In the following, we use $C$ for a strictly positive constant, independent of $n$, which may have a different value in each appearance. The notation $a_n\lesssim b_n$ stands for $a_n\le Cb_n$ for all $n$ for some fixed $C$. Additionally, $a_n=o(1)$ stands for uniform convergence towards zero meaning that there exists a sequence $(b_n)_{n\ge 1}$ with $|a_n|\le b_n$, where $b_n$ is independent of $P\in\mathcal{P}_n$ for all $n$ and $b_n\to 0$. Finally, the notation $a_n\lesssim_{P}b_n$ means that for any $\epsilon>0$, there exists $C$ such that uniformly over all $n$ we have $P_P(a_n>Cb_n)\le \epsilon$.\\ \\
Due to Assumption C.\ref{condC}$(i)$, it immediately holds
\begin{align*}
 \|\max_{r=1,\dots,d} \max_{j=1,\dots,p} |X_{r,j}|\|_{P,q}&\lesssim M_n \\
  \|\max_{r=1,\dots,d} |\varepsilon_r|\|_{P,q}&\lesssim M_n
\end{align*}

Now, we essentially modify the proof from Theorem $4.2$ from  \cite{belloni2018uniformly} to fit our setting and keep the notation as similar as possible. Let $\mathcal{U}=\{1,\dots,d\}$ and
$$\beta_r^{(1)}\in\arg\min\limits_{\beta\in\mathbb{R}^p} \E\Big[ \underbrace{\frac{1}{2} \left(Y_r-\beta^TX_r - (\beta^{(2)})^TX_r\right)^2}_{:=M_r(Y_r,X_r,\beta, a_r)}\Big]$$
with $a_r= (\beta^{(2)})^TX_r$ for all $r=1,\dots,d$. Since the coefficient $\beta^{(2)}$ captures the approximately sparse part, $a_r$ is estimated with $\hat{a}_r\equiv 0$. Further, define
$$M_r(Y_r,X_r,\beta):=M_r(Y_r,X_r,\beta, \hat{a}_r)= \frac{1}{2}\left(Y_r-\beta^TX_r\right)^2$$
Then, we have
$$\hat{\beta}_r\in \arg\min\limits_{\beta\in\mathbb{R}^p}\left(\E_n\left[M_r(Y_r,X_r,\beta)\right]+\frac{\lambda}{n}\|\hat{\Psi}_r\beta\|_1\right)$$
and 
$$\tilde{\beta}_r\in \arg\min\limits_{\beta\in\mathbb{R}^p}\left(\E_n\left[M_r(Y_r,X_r,\beta)\right]\right):\quad \text{supp}(\beta)\subseteq \text{supp}(\hat{\beta}_r).$$
First, we verify the Condition WL from \cite{belloni2018uniformly}. Since $N_n=d$, we have $N(\varepsilon,\mathcal{U},d_{\mathcal{U}})\le N_n$ for all $\varepsilon\in(0,1)$ with 
$$d_{\mathcal{U}}(i,j)=\begin{cases}0 \quad \text{for }i=j\\1\quad \text{for }i\neq j.\end{cases}$$
To prove WL(i), we note that
\begin{align*}
S_r=\partial_\beta M_r(Y_r,X_r,\beta,a_r)|_{\beta=\beta^{(1)}_r}=-\varepsilon_r X_r.
\end{align*}
Since $\Phi^{-1}(1-t)\lesssim\sqrt{\log(1/t)}$, uniformly over $t\in (0,1/2)$ , it holds
\begin{align*}
\|S_{r,j}\|_{P,3}\Phi^{-1}(1-\gamma/2pd)&=\|\varepsilon_r X_{r,j}\|_{P,3}\Phi^{-1}(1-\gamma/2pd)\\
&\le CK_n\log^{\frac{1}{2}}(a_n)\lesssim \varphi_n n^{\frac{1}{6}}
\end{align*}
with 
$$\varphi_n=O\left(\frac{K_n\log^{\frac{1}{2}}(a_n)}{n^{\frac{1}{6}}}\right)=o(1)$$
uniformly over all $j=1,\dots,p$ and $r=1,\dots,d$ by Assumption C.\ref{condC}$(i)$ and C.\ref{condC}$(iv)$.
Further, it holds
\begin{align*}
c\le \E\left[S_{r,j}^2\right]\le C
\end{align*}
for all $j=1,\dots,p$ and $r=1,\dots,d$ by Assumption C.\ref{condC}$(i)$,
which implies Condition $WL(ii)$.
Note that Condition $WL(iii)$ simplifies to 
\begin{align*}
\max\limits_{r=1,\dots,d}\max\limits_{j=1,\dots,p}|(\E_n-\E)[S_{r,j}^2]|\le \varphi_n
\end{align*}
with probability $1-\Delta_n$. Now, we use a Maximal Inequality, see Lemma $P.2$ in \cite{belloni2018uniformly}. Let $\mathcal{W}=(\mathcal{Y},\mathcal{X})$ with $Y=(Y_1,\dots,Y_d)\in\mathcal{Y}$ and $X=(X_1,\dots,X_d)\in\mathcal{X}$.
Define
\begin{align*}
\mathcal{F}:=\{f_{r,j}^2|r=1,\dots,d, j=1,\dots,p\}
\end{align*}
with
\begin{align*}
f_{r,j}:&\mathcal{W}=(\mathcal{Y},\mathcal{X})\to \mathbb{R}\\
& W=(Y,X) \mapsto -\left(Y_r-\beta_r^TX_r\right)X_{r,j}=-\varepsilon_rX_{r,j}=S_{r,j}.
\end{align*}
Note that
\begin{align*}
\|\sup\limits_{f\in\mathcal{F}}|f|\|_{P,q}&=\|\max\limits_{r=1,\dots,d}\max\limits_{j=1,\dots,p}|f_{r,j}^2|\|_{P,q}\\
&=\E\left[\max\limits_{r=1,\dots,d}\max\limits_{j=1,\dots,p}\varepsilon_r^{2q} X_{r,j}^{2q}\right]^{1/q}\\
&\le\E\left[\max\limits_{r=1,\dots,d}\varepsilon_r^{2q} \max\limits_{r=1,\dots,d}\max\limits_{j=1,\dots,p}X_{r,j}^{2q}\right]^{1/q}\\
&\le CM_n^4.
\end{align*}
Since
\begin{align*}
\sup\limits_{f\in\mathcal{F}}\|f\|_{P,2}^2&=\max\limits_{r=1,\dots,d}\max\limits_{j=1,\dots,p} \E\left[\varepsilon_r^4X_{r,j}^4\right],
\end{align*}
it holds
$$\sup\limits_{f\in\mathcal{F}}\|f\|_{P,2}^2\le L_n\le \|\sup\limits_{f\in\mathcal{F}}|f|\|^2_{P,2}.$$
Additionally, it holds $|\mathcal{F}|=dp$ which implies
$$\log\sup\limits_{Q}N(\epsilon\|F\|_{Q,2},\mathcal{F},\|\cdot\|_{Q,2})\le  \log(dp)\lesssim \log(a_n/\epsilon), \quad 0<\epsilon\le 1.$$
Using Lemma $P.2$ from \cite{belloni2018uniformly}, we obtain with probability not less than $1-o(1)$
\begin{align*}
&\max\limits_{r=1,\dots,d}\max\limits_{j=1,\dots,p}|(\E_n-\E)[S_{r,j}^2]|\\
=\ &n^{-1/2}\sup\limits_{f\in\mathcal{F}}|\mathbb{G}_n(f)|\\
\le\ &n^{-1/2}C\left(\sqrt{L_n\log\left(a_n\right)}+n^{-1/2+1/q}M_n^4\log^{}(a_n)\right)\\
=\ &C\left(\sqrt{\frac{L_n\log\left(a_n\right)}{n}}+n^{\frac{1}{q}}\frac{M_n^4\log^{}(a_n)}{n^{}}\right)\\
\le\ &\varphi_n=o(1)
\end{align*}
by the growth condition in Assumption C.\ref{condC}$(iv)$. We proceed by verifying Assumption $M.1$ in \cite{belloni2018uniformly}. The function $\beta\mapsto M_r\left(Y_r,X_r,\beta\right)$ is convex, which is the first requirement of Assumption $M.1$. We now proceed with a simplified version of proof of $K.1$ from \cite{belloni2018uniformly}.
To show Assumption $M.1$ (a), note that for all $\delta \in \R^p$
\begin{align*}
    &\left|\E_n\left[\partial_\beta M_r(Y_r,X_r,\beta_r^{(1)})-\partial_\beta M_r(Y_r,X_r,\beta_r^{(1)},a_r)\right]^T\delta\right|\\
    =&\left|\E_n\left[X_r\left((\beta_r^{(2)})^TX_r\right)\right]^T\delta\right|\\
    \le& \|(\beta_r^{(2)})^TX_r\|_{\mathbb{P}_n,2}\|\delta^TX_r\|_{\mathbb{P}_n,2}
\end{align*}
due to Assumption C.\ref{condC}$(ii)$ and $(iii)$. Further, define
$$\mathcal{G}:=\{g_r:X\mapsto \left(\beta_r^{(2)})^TX_r\right)^2|r=1,\dots,d\}$$
with envelope
$$G:= \max_{r=1,\dots,d}\|X_r\|^2_\infty\max_{r=1,\dots,d}\|\beta_r^{(2)}\|_1^2\lesssim M_n^2\max_{r=1,\dots,d}\|\beta_r^{(2)}\|_1^2,$$
implying
\begin{align*}
    \|G\|_{P,q}\lesssim M_n^2\max_{r=1,\dots,d}\|\beta_r^{(2)}\|_1^2.
\end{align*}
Note that for all $0< \epsilon\le1$, we have
\begin{align*}
    N(\epsilon\|G\|_{P,2},\mathcal{G},\|\cdot\|_{P,2})\le d\le d/\epsilon
\end{align*}
and
\begin{align*}
    \sup_{g\in\mathcal{G}}\|g\|_{P,2}^2\lesssim M_n^4\max_{r=1,\dots,d}\|\beta_r^{(2)}\|_1^4.
\end{align*}
Using Lemma $P.2$ from \cite{belloni2018uniformly}, we obtain
\begin{align*}
    &\max_{r=1,\dots,d} \left|\E_n\left[\left((\beta_r^{(2)})^TX_r\right)^2\right]-\E\left[\left((\beta_r^{(2)})^TX_r\right)^2\right]\right|\\
    =& n^{-1/2}\sup_{g\in \mathcal{G}} |\mathbb{G}_n(g)|\\
    \lesssim & M_n^2\max_{r=1,\dots,d}\|\beta_r^{(2)}\|_1^2\sqrt{\frac{\log(a_n)}{n}} + n^{-1 + 1/q}M_n^2 \max_{r=1,\dots,d}\|\beta_r^{(2)}\|_1^2\log(a_n)\\
    \lesssim & \frac{s\log(a_n)}{n}
\end{align*}
with probability $1-o(1).$ due to Assumption  C.\ref{condC}$(iii)$. Using $\|beta_r^{(2)}\|_2^2\lesssim \frac{s\log(a_n)}{n}$, we obtain
\begin{align*}
    &\max_{r=1,\dots,d} \|(\beta_r^{(2)})^TX_r\|^2_{\mathbb{P}_n,2}\\ \le\ &\max_{r=1,\dots,d}\big|\|(\beta_r^{(2)})^TX_r\|_{\mathbb{P}_n,2}-\|(\beta_r^{(2)})^TX_r\|_{\mathbb{P},2}\big| + \max_{r=1,\dots,d}\|(\beta_r^{(2)})^TX_r\|^2_{\mathbb{P},2}\\
    \lesssim_P\ & \frac{s\log(a_n)}{n},
\end{align*}
implying Assumption M.1. (a) with $C_n\lesssim \sqrt{\frac{s\log(a_n)}{n}}.$
Further, we have
\begin{align*}
&\quad\E_n\left[\frac{1}{2}\left(Y_r-(\beta_r+\delta)^TX_r\right)^2\right]-\E_n\left[\frac{1}{2}\left(Y_r-\beta_r^TX_r\right)^2\right]\\
&=-\E_n\left[\left(Y_r-\beta_r^TX_r\right)\delta^T X_r\right]+\frac{1}{2}\E_n\left[(\delta^T X_r)^2\right],
\end{align*}
where
$$-\E_n\left[\left(Y_r-\beta_r^TX_r\right)\delta^T X_r\right]=\E_n\left[\partial_\beta M_r(Y_r,X_r,\beta_r)\right]^T\delta$$
and
$$\frac{1}{2}\E_n\left[(\delta^T X_r)^2\right]=||\sqrt{w_r}\delta^T X_r||_{\bP_n,2}^2$$
with $\sqrt{w_r}=1/4$. This gives us Assumption $M.1$ (c) with $\Delta_n=0$ and $\bar{q}_{A_r}=\infty$. Since Condition $WL(ii)$ and $WL(iii)$ hold, we have with probability $1-o(1)$ 
\begin{align*}
1\lesssim l_{r,j}=\left(\E_n[S_{r,j}^2]\right)^{1/2}\lesssim 1
\end{align*}
uniformly over all $r=1,\dots,d$ and $j=1,\dots,p$ which directly implies
\begin{align*}
1\lesssim\|\hat{\Psi}^{(0)}_r\|_{\infty}:=\max\limits_{j=1,\dots,p}|l_{r,j}|\lesssim 1
\end{align*}
and additionally 
\begin{align*}
1\lesssim\|(\hat{\Psi}^{(0)}_r)^{-1}\|_{\infty}:=\max\limits_{j=1,\dots,p}|l_{r,j}^{-1}|\lesssim 1.
\end{align*}
For now, we suppose that $m=0$ in Algorithm \ref{penalg}. Uniformly over $r=1,\dots,d$ and $j=1,\dots,p$, we have 
\begin{align*}
\hat{l}_{r,j,0}&=\left(\E_n\left[\max\limits_{1\le i \le n} \|X_r^{(i)}\|^2_{\infty}\right]\right)^{1/2}\ge \left(\E_n[ \|X_r\|^2_{\infty}]\right)^{1/2}\gtrsim_{P} c,
\end{align*}
where the last inequality holds due to Assumption C.\ref{condC}$(ii)$ and an application of the Maximal Inequality. Also uniformly over $r=1,\dots,d$ and $j=1,\dots,p$, it holds
\begin{align*}
\hat{l}_{r,j,0}&=\max\limits_{1\le i \le n} \|X_r^{(i)}\|_{\infty}  \lesssim_P M_n.
\end{align*}

Therefore, Assumption $M.1(b)$ holds for some $\Delta_n=o(1)$, $L\lesssim M_n$ and $l\gtrsim 1$.
Hence, we can find a $c_l$ with $l>1/c_l$. Setting $c_\lambda>c_l$ and $\gamma=\gamma_n\in[1/n,1/\log(n)]$ in the choice of $\lambda$, we obtain
\begin{align*}
P\left(\frac{\lambda}{n}\ge c_l\max\limits_{r=1,\dots,d}\|(\hat{\Psi}_r^{(0)})^{-1}\E_n[S_r]\|_{\infty}\right)\ge 1-\gamma-o(\gamma)-\Delta_n=1-o(1)
\end{align*}
due to Lemma $M.4$ in \cite{belloni2018uniformly}. Now, we uniformly bound the sparse eigenvalues. Set 
$$l_n=n^{2/\bar{q}}M_n^2$$
for a $\bar{q} > 3\tilde{q}$ with $\tilde{q}$ in  C.\ref{condC}$(iv)$. 
We apply Lemma $Q.1$ in \cite{belloni2018uniformly} with $K\lesssim M_n$ and

\begin{align*}
\delta_n&\lesssim K\sqrt{sl_n}n^{-1/2}\log(sl_n)\log^{\frac{1}{2}}(a_n)\log^{\frac{1}{2}}(n)\\
&\lesssim \sqrt{n^{\frac{2}{\bar{q}}} M_n^4\log(n)\log^2(sl_n)\frac{s\log(a_n)}{n}}\\
&\lesssim \sqrt{n^{\frac{3}{\bar{q}}}\frac{s M_n^4\log(a_n)}{n}}
\end{align*}
for $n$ large enough. Hence, by the growth condition in Assumption C.\ref{condC}$(iv)$, it holds 
$$\delta_n=o(1)$$
which implies
\begin{align*}
1\lesssim \min\limits_{\|\delta\|_0\le l_n s}\frac{\|\delta^T X_r\|^2_{\bP_n,2}}{\|\delta\|_2^2}\le \max\limits_{\|\delta\|_0\le l_n s}\frac{\|\delta^T X_r\|^2_{\bP_n,2}}{\|\delta\|_2^2}\lesssim 1
\end{align*}
with probability $1-o(1)$ uniformly over $r=1,\dots,d$.\\
Define $T_r:=\text{supp}(\beta^{(1)}_r)$ and
$$\tilde{c}:=\frac{Lc_l+1}{lc_l-1}\max\limits_{r=1,\dots,d}\|\hat{\Psi}^{(0)}_r\|_{\infty}\|(\hat{\Psi}^{(0)}_r)^{-1}\|_{\infty}\lesssim L.$$
Let the restricted eigenvalues be defined as 
\begin{align*}
\bar{\kappa}_{2\tilde{c}}:=\min\limits_{r=1,\dots,d}\inf\limits_{\delta\in\Delta_{2\tilde{c},r}}\frac{\|\delta^T X_r\|_{\bP_n,2}}{\|\delta_{T_r}\|_2},
\end{align*}
where $\Delta_{2\tilde{c},r}:=\{\delta:\|\delta^c_{T_r}\|_1\le 2\tilde{c}\|\delta_{T_r}\|_1\}$.
By the argument given in \cite{BickelRitovTsybakov2009}, it holds
\begin{align*}
\bar{\kappa}_{2\tilde{c}}&\ge \left(\min\limits_{\|\delta\|_0\le l_n s}\frac{\|\delta^T X_r\|^2_{\bP_n,2}}{\|\delta\|_2^2}\right)^{1/2}-2\tilde{c}\left(\max\limits_{\|\delta\|_0\le l_n s}\frac{\|\delta^T X_r\|^2_{\bP_n,2}}{\|\delta\|_2^2}\right)^{1/2}\left(\frac{s}{sl_n}\right)^{1/2}\\
&\gtrsim \left(\min\limits_{\|\delta\|_0\le l_n s}\frac{\|\delta^T X_r\|^2_{\bP_n,2}}{\|\delta\|_2^2}\right)^{1/2}-2n^{-\frac{1}{\bar{q}}}\left(\max\limits_{\|\delta\|_0\le l_n s}\frac{\|\delta^T X_r\|^2_{\bP_n,2}}{\|\delta\|_2^2}\right)^{1/2}\\
&\gtrsim 1
\end{align*}
with probability $1-o(1)$.
Since 
\begin{align*}
\frac{\lambda}{n}\lesssim n^{-1/2}\Phi^{-1}\left(1-\gamma/(2dp)\right)\lesssim n^{-1/2}\sqrt{\log(2dp/\gamma)}\lesssim n^{-1/2}\log^{\frac{1}{2}}(a_n)
\end{align*} and the uniformly bounded penalty loading from above and away from zero, we obtain
\begin{align*}
\max\limits_{r=1,\dots,d}\|(\hat{\beta}_r-\beta_r^{(1)})^TX_r\|_{\bP_n,2}\lesssim_P L\sqrt{\frac{s\log(a_n)}{n}}
\end{align*}
by Lemma $M.1$ from  \cite{belloni2018uniformly}. To show Assumption $M.1(b)$ for $m\ge 1$, we proceed by induction. Assume that the assumption holds for $\hat{\Psi}_{r,m-1}$ with some $\Delta_n=o(1)$, $l\gtrsim 1$ and $L\lesssim M_n$.
We have shown that the estimator based on $\hat{\Psi}_{r,m-1}$ obeys
\begin{align*}
\max\limits_{r=1,\dots,d}\|(\hat{\beta}_r-\beta_r^{(1)})^TX_r\|_{\bP_n,2}\lesssim L\sqrt{\frac{s\log(a_n)}{n}}
\end{align*}
with probability $1-o(1)$. Further, we have shown
\begin{align*}
    \max_{r=1,\dots,d} \|(\beta_r^{(2)})^TX_r\|_{\mathbb{P}_n,2}\lesssim_P \sqrt{\frac{s\log(a_n)}{n}}
\end{align*}
implying
\begin{align*}
\max\limits_{r=1,\dots,d}\|(\hat{\beta}_r-\beta_r)^TX_r\|_{\bP_n,2}\lesssim L\sqrt{\frac{s\log(a_n)}{n}}
\end{align*}
with probability $1-o(1)$.
This implies 
\begin{align*}
|\hat{l}_{r,j,m}-l_{r,j}|&=\left|\E_n\left[\left(\left(Y_r-\hat{\beta}_r^T X_r\right) X_{r,j}\right)^2\right]^{1/2}-\E_n\left[\left(\left(Y_r-\beta_r^T X_r\right) X_{r,j}\right)^2\right]^{1/2}\right|\\
&\le \left|\E_n\left[\left(\left((\hat{\beta}_r-\beta_r)^T X_r\right) X_{r,j}\right)^2\right]^{1/2}\right|\\
&\lesssim \|(\hat{\beta}_r-\beta_r)^TX_r\|_{\bP_n,2} \max\limits_{1\le i \le n} \max\limits_{r=1,\dots,d}\|X_r^{(i)}\|_{\infty}\\
&\lesssim_P L\sqrt{\frac{s\log(a_n)}{n}} M_n\\
&\lesssim \sqrt{\frac{s M_n^4\log^{}(a_n)}{n}}=o(1)
\end{align*}
uniformly over $r=1,\dots,d$ and $j=1,\dots,p$. Therefore, Assumption $M.1(b)$ holds for $\hat{\Psi}_{r,m}$ for some $\Delta_n=o(1)$, $l\gtrsim 1$ and $L\lesssim 1$. Consequently, we obtain
\begin{align*}
\max\limits_{r=1,\dots,d}\|(\hat{\beta}_r-\beta_r^{(1)})^TX_r\|_{\bP_n,2}\lesssim \sqrt{\frac{s\log(a_n)}{n}}.
\end{align*}
and
\begin{align*}
\max\limits_{r=1,\dots,d}\|\hat{\beta}_r-\beta_r^{(1)}\|_{1}\lesssim \sqrt{\frac{s^2\log(a_n)}{n}}
\end{align*}
with probability $1-o(1)$ due to Lemma $M.1$ in \cite{belloni2018uniformly}. Uniformly over all $r=1,\dots,d$, it holds
\begin{align*}
&\left|\left(\E_n\left[\partial_{\beta}M_r(Y_r,X_r,\hat{\beta}_r)-\partial_{\beta}M_r(Y_r,X_r,\beta_r)\right]\right)^T\delta\right|\\
=&\left|\left(\E_n\left[(\hat{\beta}_r-\beta_r)^TX_r X_r^T\right]\right)^T\delta\right|\\
\le&\|(\hat{\beta}_r-\beta_r)^TX_r\|_{\bP_n,2}\|\delta^T X_r\|_{\bP_n,2}\le L_n\|\delta^T X_r\|_{\bP_n,2}
\end{align*}
with probability $1-o(1)$ where $L_n\lesssim (s\log(a_n)/n)^{1/2}$. Since the maximal sparse eigenvalues
$$ \phi_{max}(l_ns,r):=\max\limits_{\|\delta\|_0\le l_n s}\frac{\|\delta^T X_r\|^2_{\bP_n,2}}{\|\delta\|_2^2}$$
are uniformly bounded from above, Lemma $M.2$ from  \cite{belloni2018uniformly} implies
\begin{align*}
\max\limits_{r=1,\dots,d}\|\hat{\beta}_r\|_0 \lesssim s
\end{align*}
with probability $1-o(1)$. Combining this result with the uniform restrictions on the sparse eigenvalues from above, we obtain
\begin{align*}
\max\limits_{r=1,\dots,d}\|\hat{\beta}_r-\beta_r^{(1)}\|_2 \lesssim \max\limits_{r=1,\dots,d}\|(\hat{\beta}_r-\beta_r^{(1)})^TX_r\|_{\bP_n,2}\lesssim \sqrt{\frac{s\log(a_n)}{n}} 
\end{align*}
with probability $1-o(1)$. We now proceed by using Lemma $M.3$ in \cite{belloni2018uniformly}. We obtain uniformly over all $r=1,\dots,d$
\begin{align*}
\E_n[M_r(Y_r,X_r,\tilde{\beta}_r)]-\E_n[M_r(Y_r,X_r,\beta_r^{(1)})]&\le \frac{\lambda L}{n}\|\hat{\beta}_r-\beta_r^{(1)}\|_1\max\limits_{r=1,\dots,d}\|\hat{\Psi}^{(0)}_r\|_{\infty}\\
&\lesssim \frac{\lambda}{n}\|\hat{\beta}_r-\beta_r^{(1)}\|_1\\
&\lesssim \frac{s\log(a_n)}{n}
\end{align*}
with probability $1-o(1)$, where we used $L\lesssim 1$ and $\max\limits_{r=1,\dots,d}\|\hat{\Psi}^{(0)}_r\|_{\infty}\lesssim 1$. Since
\begin{align*}
\max\limits_{r=1,\dots,d}\|\E_n[S_r]\|_{\infty}&\le \max\limits_{r=1,\dots,d}\|\hat{\Psi}^{(0)}_r\|_{\infty}\|\big(\hat{\Psi}^{(0)}_r\big)^{-1}\E_n[S_r]\|_{\infty}\lesssim \frac{\lambda}{n}\lesssim n^{-1/2}\log^{\frac{1}{2}}(a_n) 
\end{align*}
with probability $1-o(1)$, we obtain
\begin{align*}
\max\limits_{r=1,\dots,d}\|(\tilde{\beta}_r-\beta_r^{(1)})^TX_r\|_{\bP_n,2}\lesssim \sqrt{\frac{s\log(a_n)}{n}}
\end{align*} 
with probability $1-o(1)$, where we used 
$$\max\limits_{r=1,\dots,d}\|\hat{\beta}_r\|_0 \lesssim s ,\ C_n\lesssim (s\log(a_n)/n)^{1/2}$$
and that the minimum sparse eigenvalues are uniformly bounded away from zero. With the same argument as above, we obtain 
\begin{align*}
\max\limits_{r=1,\dots,d}\|\tilde{\beta}_r-\beta_r^{(1)}\|_2 \lesssim \max\limits_{r=1,\dots,d}\|(\tilde{\beta}_r-\beta_r^{(1)})^TX_r\|_{\bP_n,2}\lesssim \sqrt{\frac{s\log(a_n)}{n}}.
\end{align*}
Combined with the Assumption C.\ref{condC}$(iii)$ on the approximate sparsity this completes the proof.
\end{proof}

%% file: Computational_Details.tex

\section{Computational Details} \label{compdetails}


\subsection{Computation and Infrastructure}

The simulation study has been run on a x86\_64\_redhat\_linux-gnu (64-bit) (CentOS Linux 7 (Core)) cluster using R version \texttt{3.6.1 (2019-07-05)}. All lasso estimations are performed using the \texttt{rlasso} learner as provided by the R package \texttt{hdm}, version \texttt{0.3.1} by \cite{hdm} which can be downloaded from CRAN. Construction of B-splines is based on the R package \texttt{splines}. The R code for the implementation and simulation study will be made available in the future.

\section{Additional Simulation Results} \label{addsimresults}

\begin{figure}[h]
\begin{center}
\includegraphics[scale=0.25]{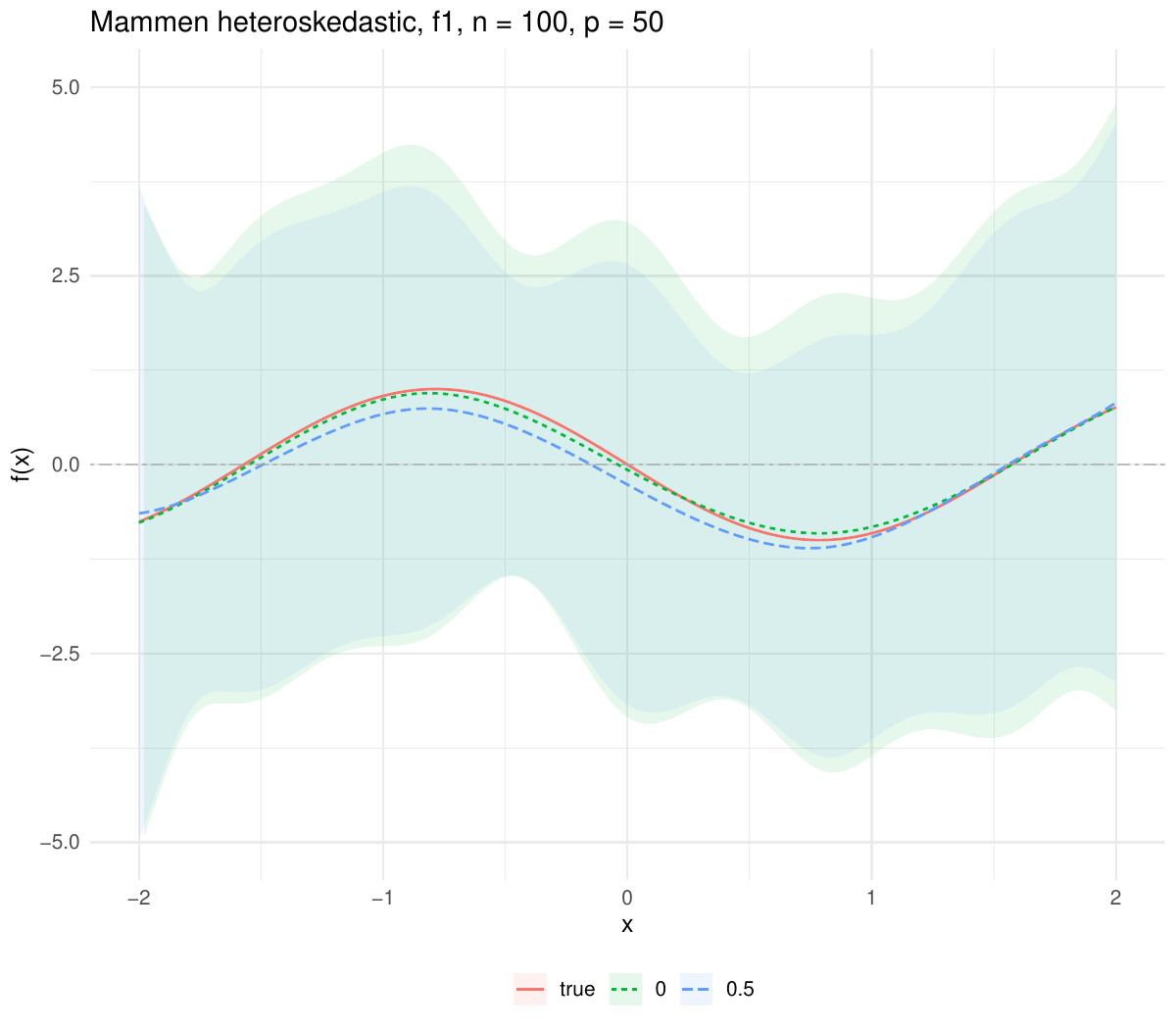} 
\includegraphics[scale=0.25]{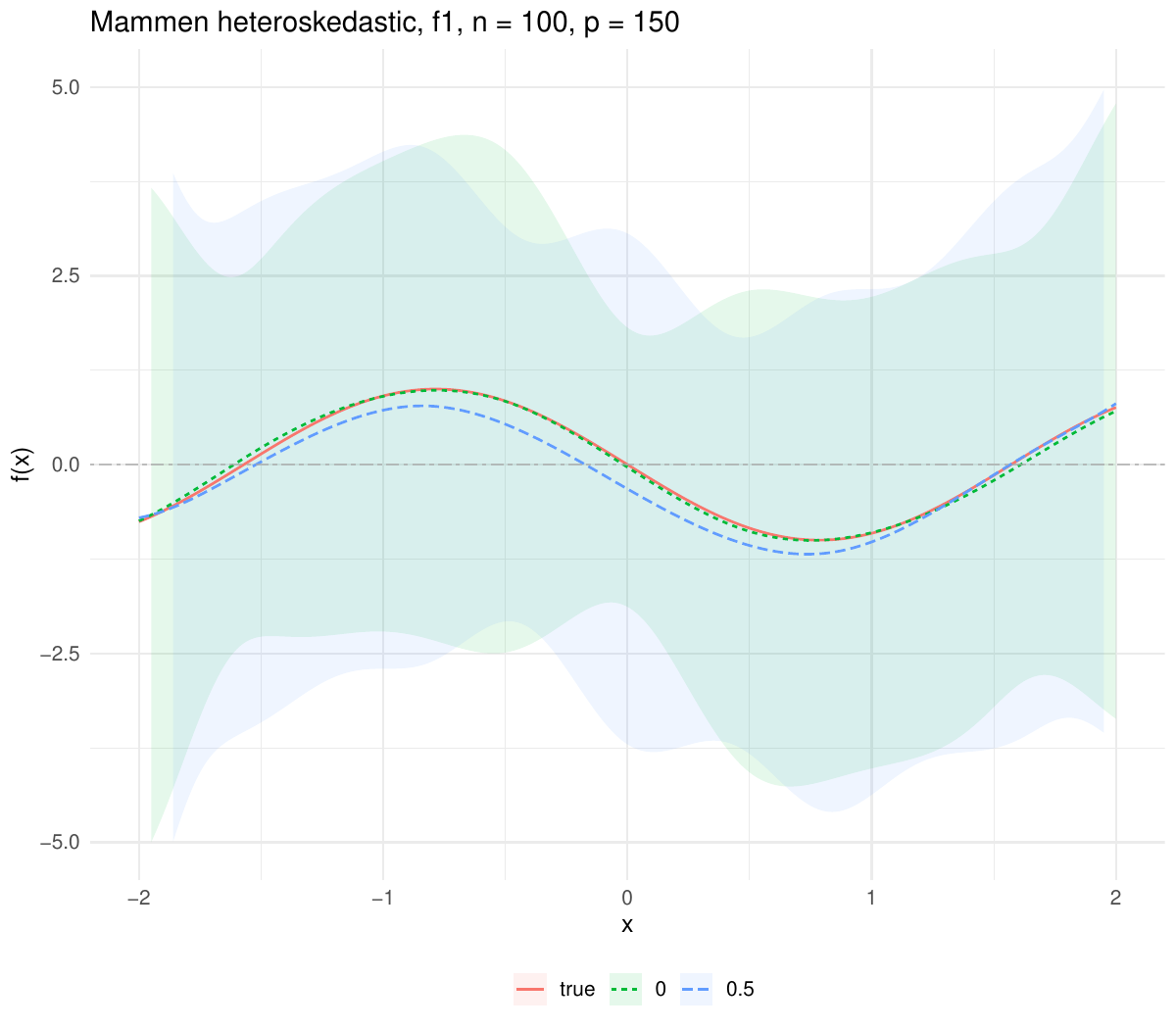} \\
\includegraphics[scale=0.25]{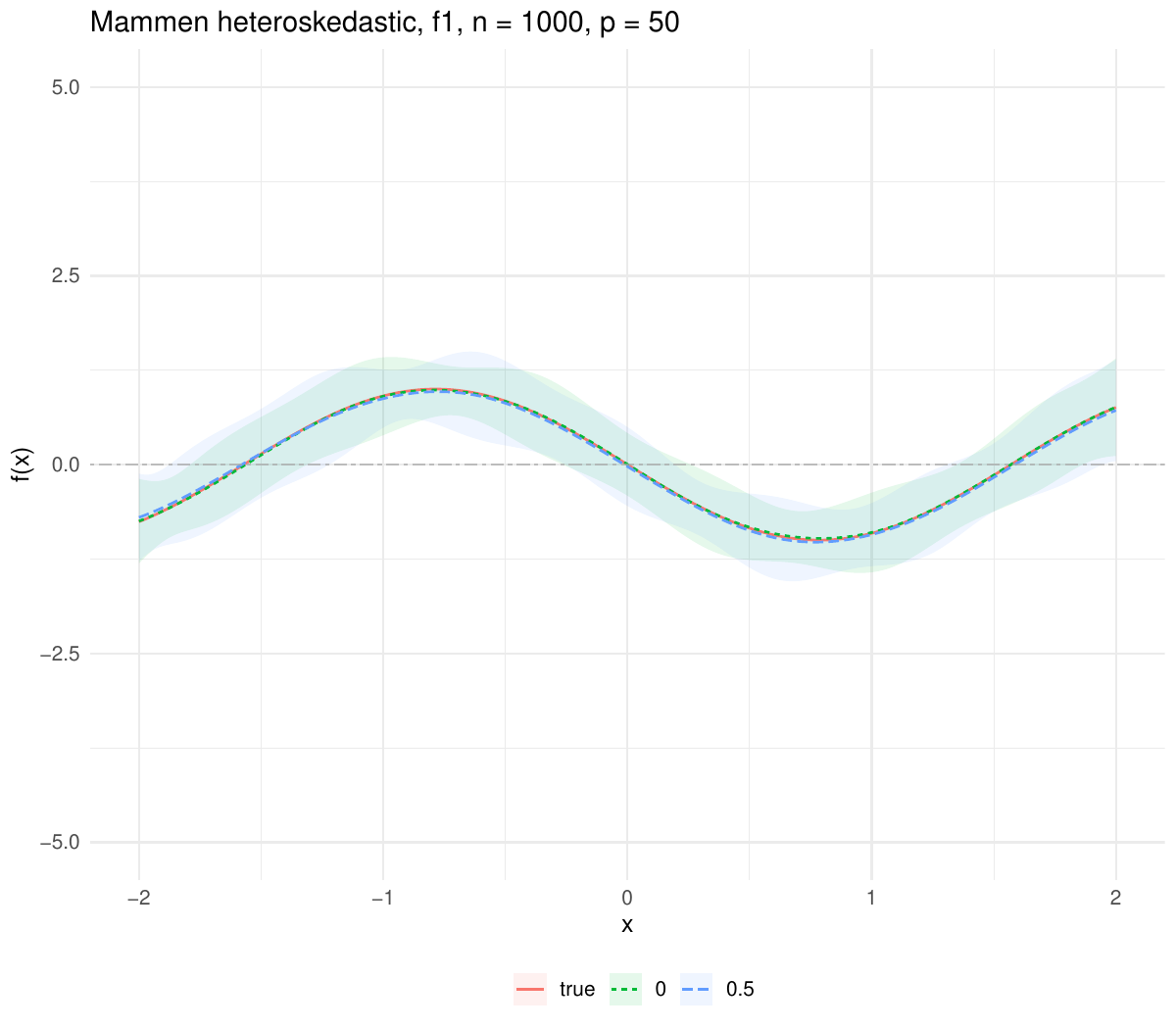}  
\includegraphics[scale=0.25]{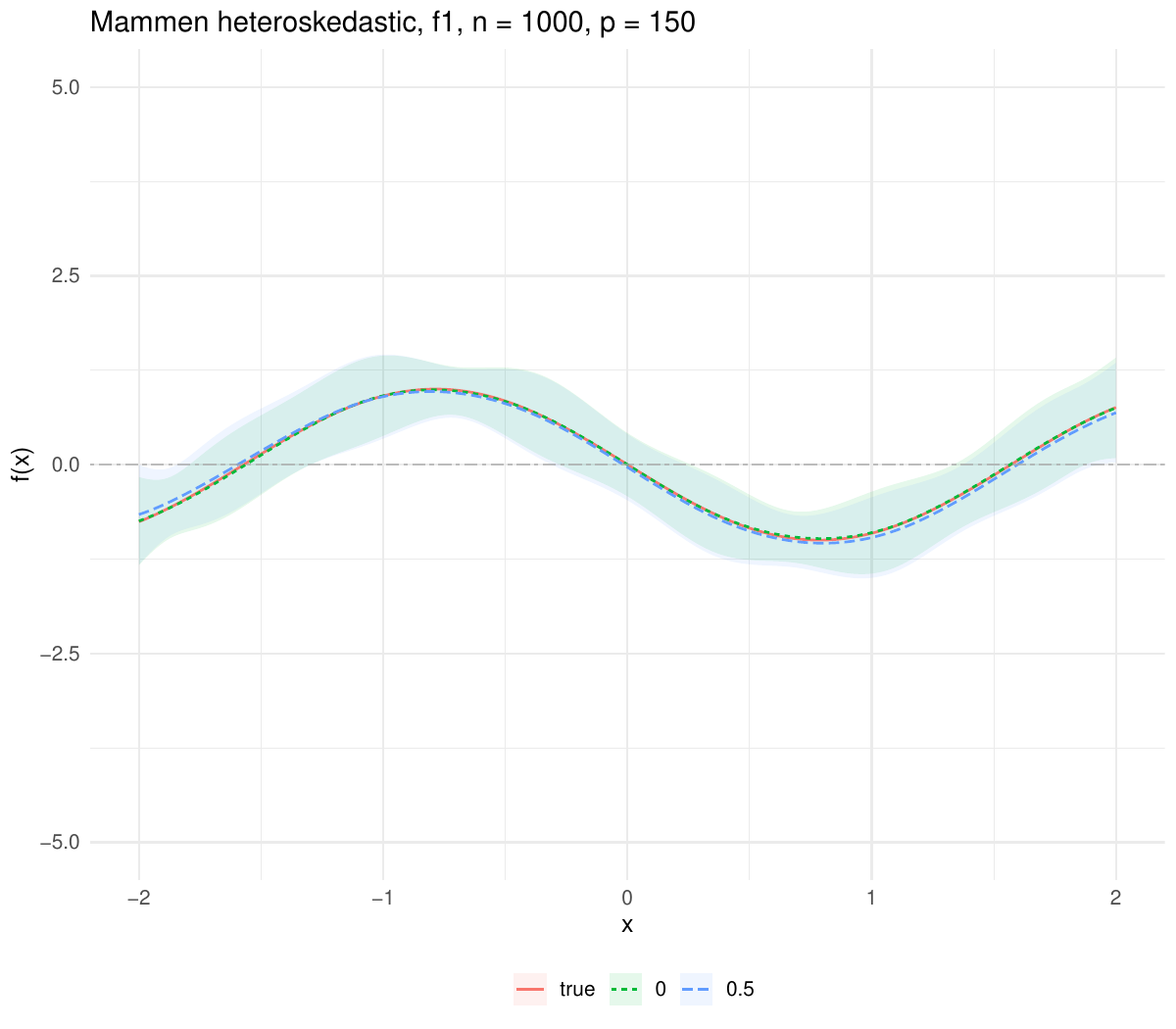} 
\caption{Average confidence bands, $f_1(x_1)$, heteroskedastic setting.}
\end{center}
The green dashed curve illustrates averaged estimated functions $\hat{f}_j(x_j)$  as obtained in $R=500$ repetitions in the scenario with $\rho = 0$. The corresponding averaged $95\%$-confidence bands are shaded green. The blue long-dashed line illustrates the results for the setting with $\rho = 0.5$ together with corresponding averaged confidence bands (shaded blue). The true function $f_j(x_j)$ is illustrated by the red solid curve.
\label{heteroskf1}
\end{figure}

\begin{figure}[t]
\begin{center}
\includegraphics[scale=0.25]{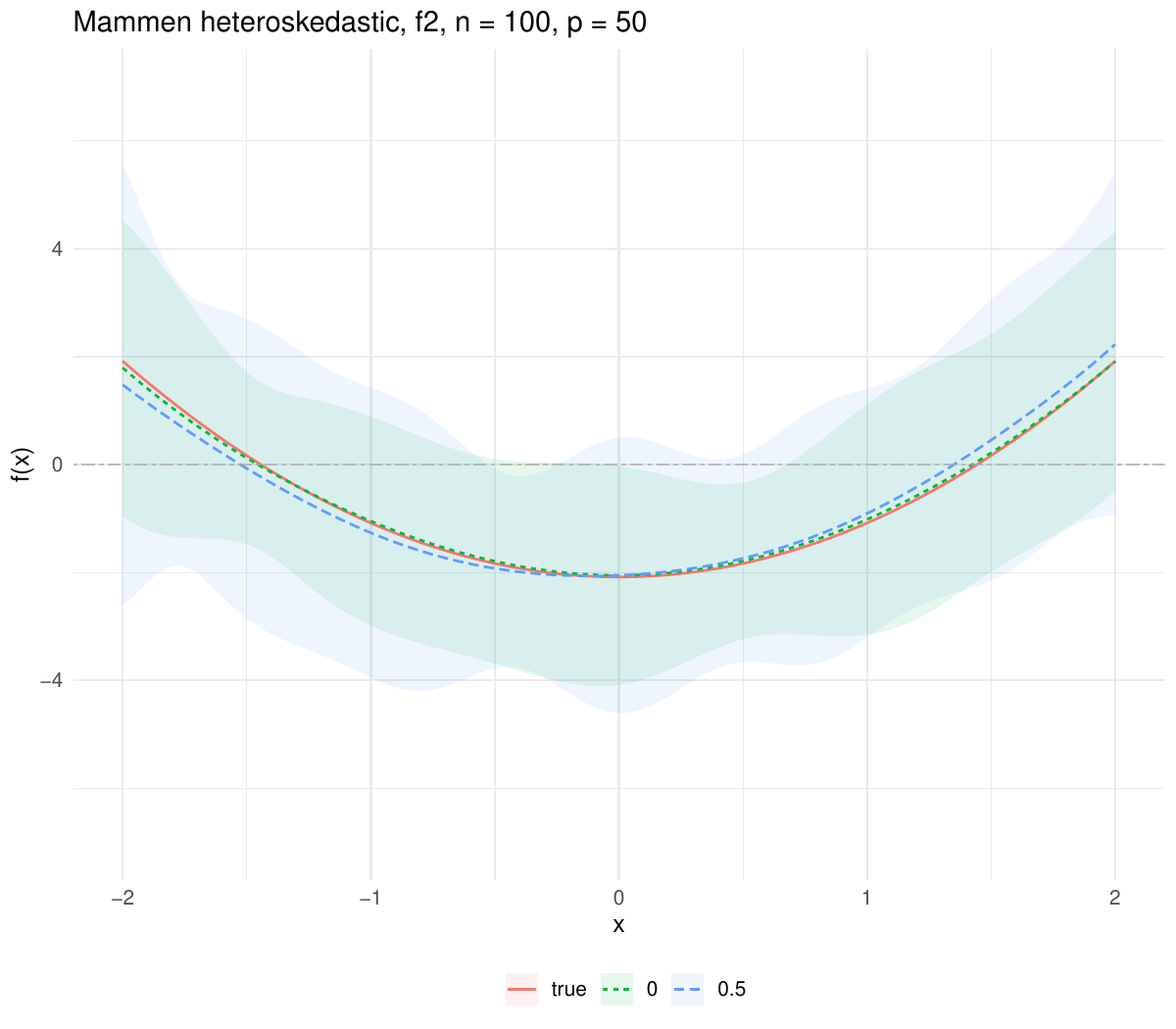} 
\includegraphics[scale=0.25]{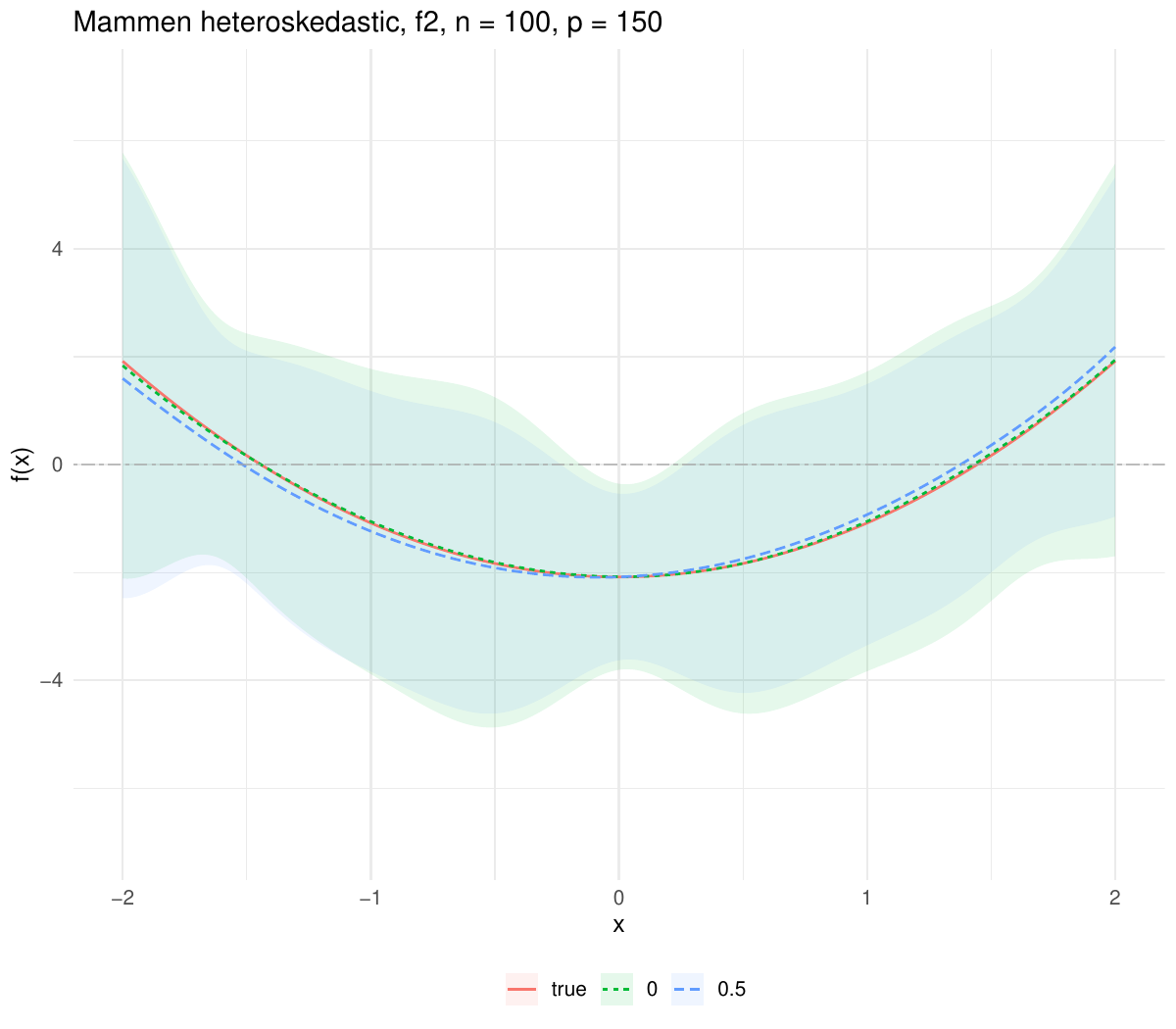} \\
\includegraphics[scale=0.25]{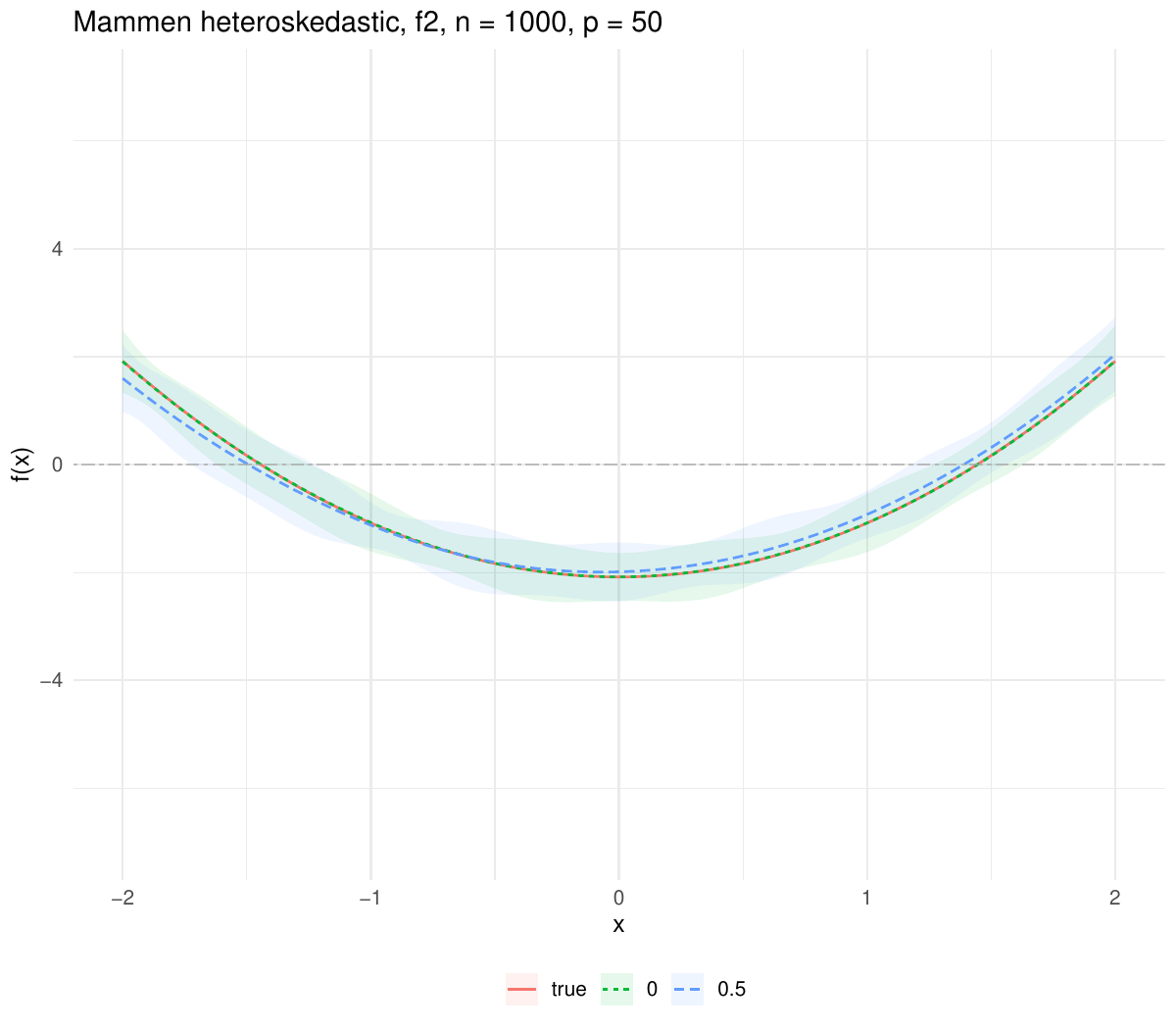}  
\includegraphics[scale=0.25]{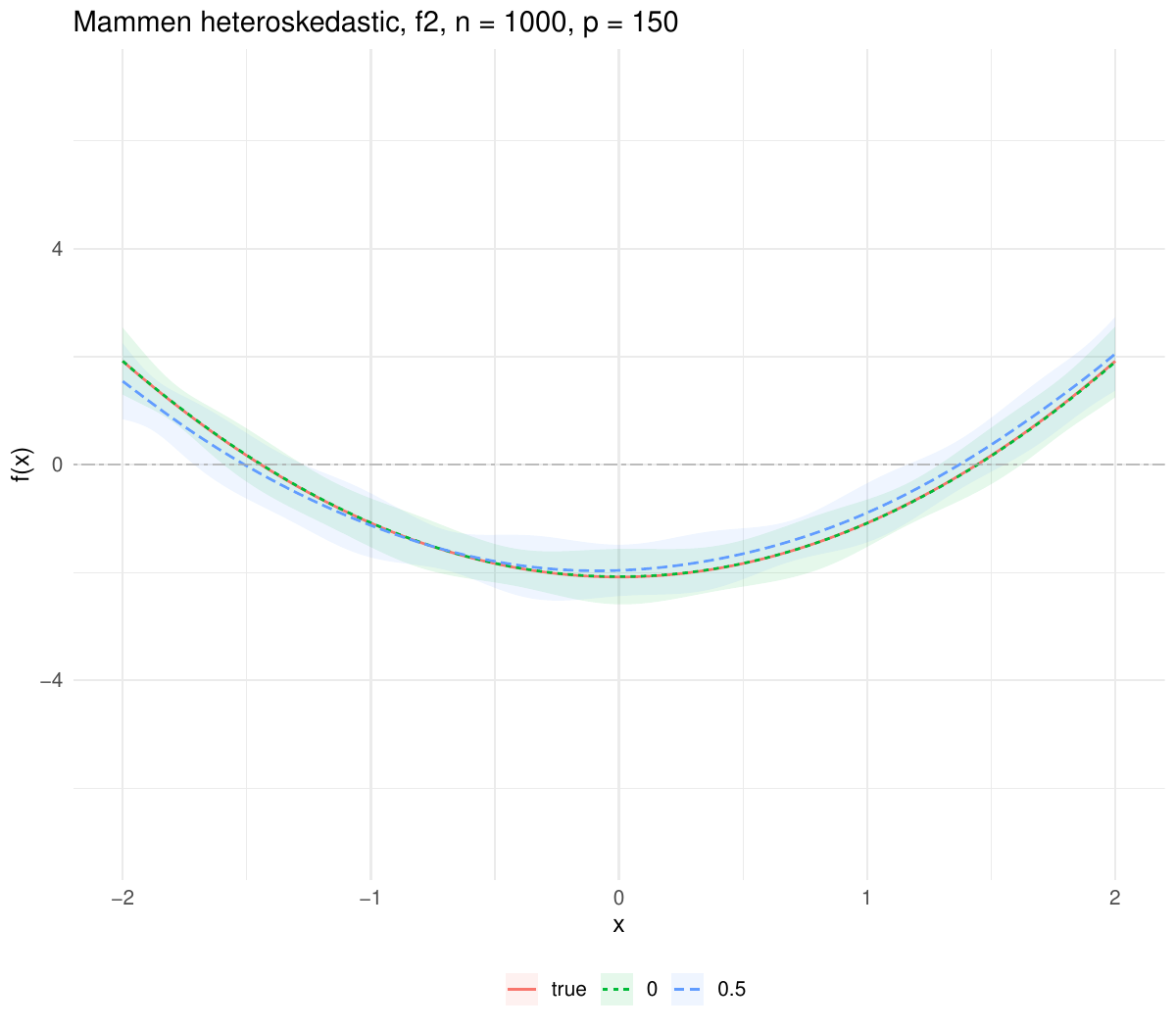} 
\caption{Average confidence bands, $f_2(x_2)$, heteroskedastic setting.}
\end{center}
The green dashed curve illustrates averaged estimated functions $\hat{f}_j(x_j)$  as obtained in $R=500$ repetitions in the scenario with $\rho = 0$. The corresponding averaged $95\%$-confidence bands are shaded green. The blue long-dashed line illustrates the results for the setting with $\rho = 0.5$ together with corresponding averaged confidence bands (shaded blue). The true function $f_j(x_j)$ is illustrated by the red solid curve.
\label{heteroskf2}
\end{figure}

\begin{figure}[t]
\begin{center}
\includegraphics[scale=0.25]{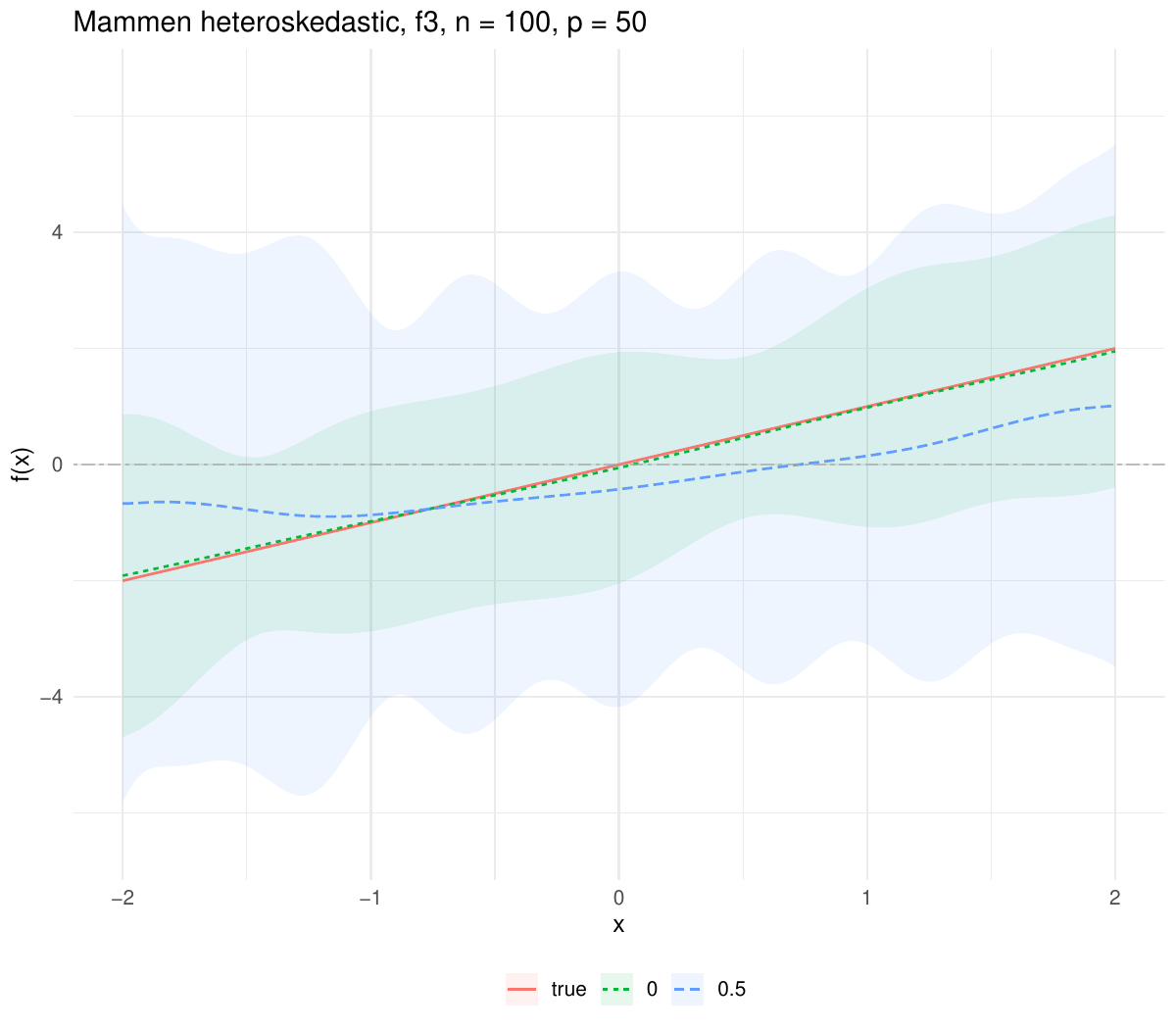} 
\includegraphics[scale=0.25]{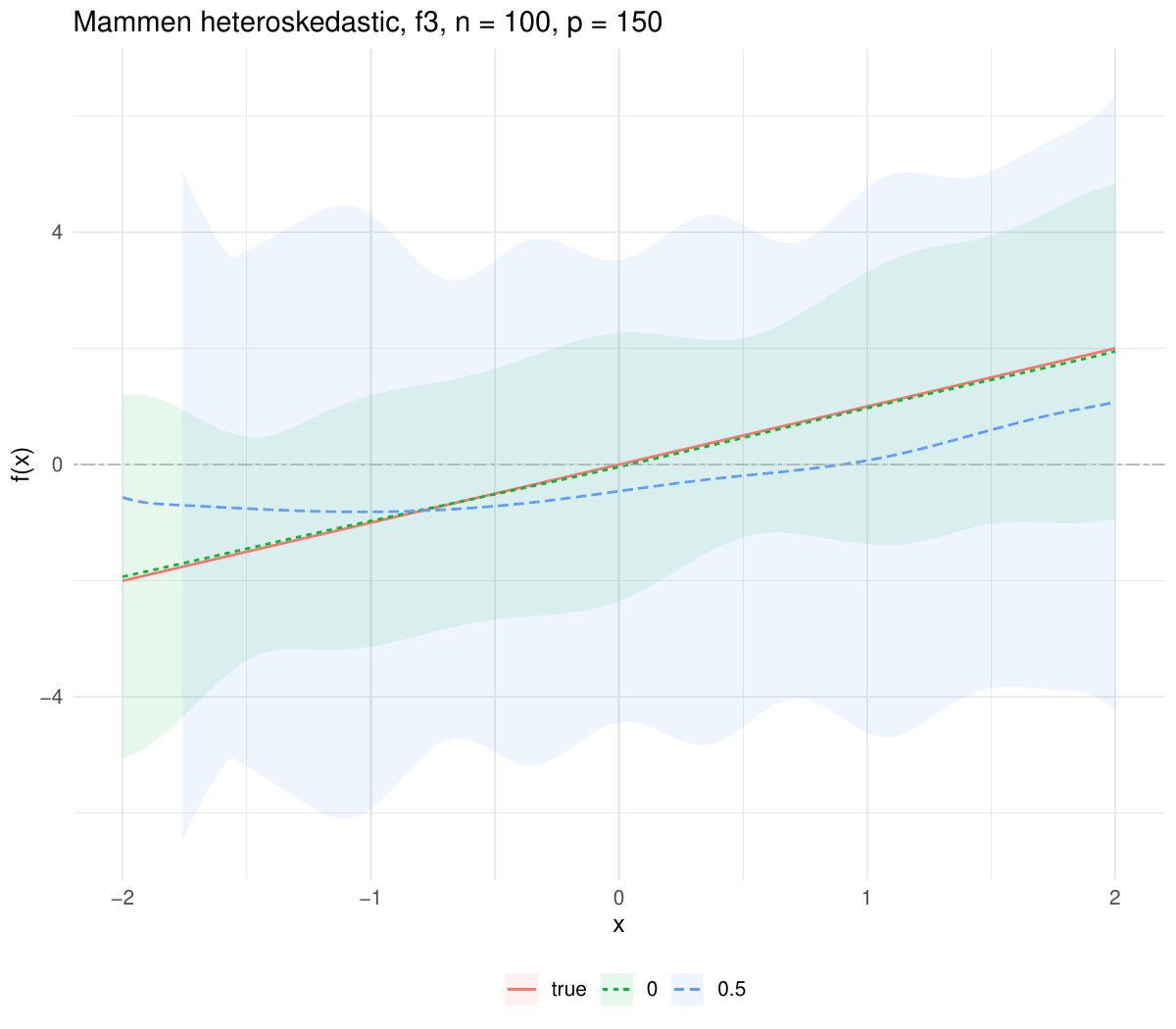} \\
\includegraphics[scale=0.25]{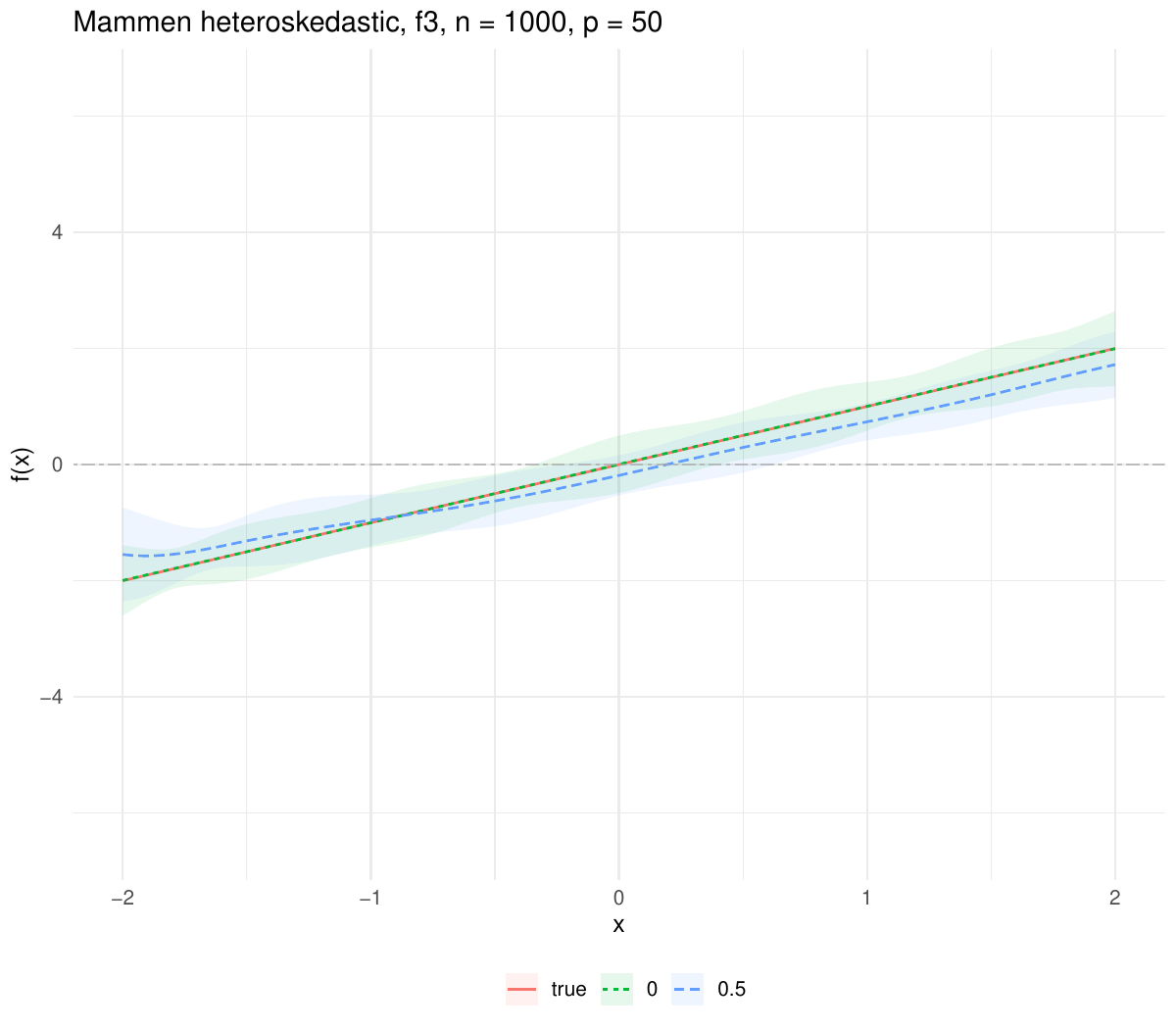}  
\includegraphics[scale=0.25]{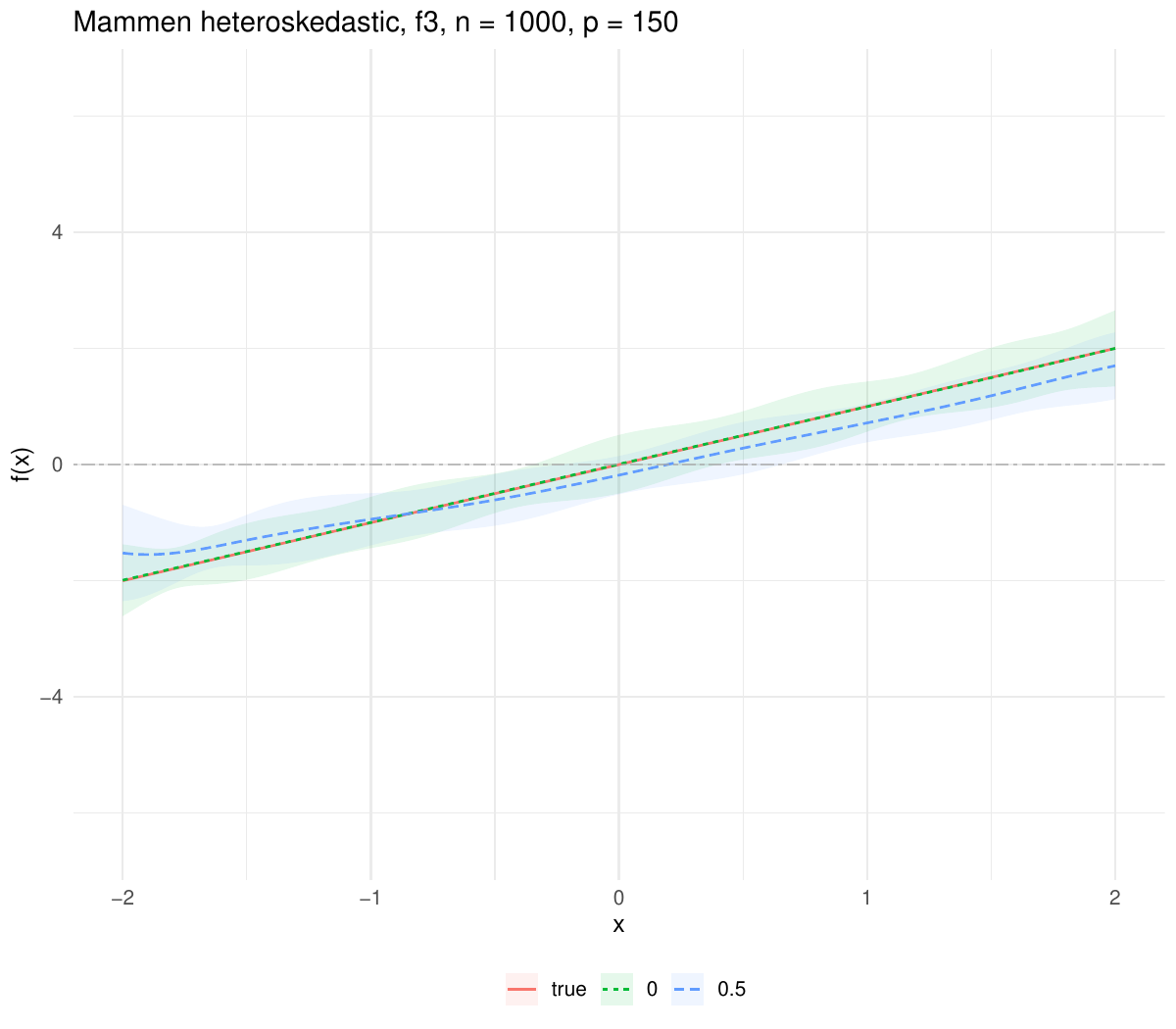} 
\caption{Average confidence bands, $f_3(x_3)$, heteroskedastic setting.}
\end{center}
The green dashed curve illustrates averaged estimated functions $\hat{f}_j(x_j)$  as obtained in $R=500$ repetitions in the scenario with $\rho = 0$. The corresponding averaged $95\%$-confidence bands are shaded green. The blue long-dashed line illustrates the results for the setting with $\rho = 0.5$ together with corresponding averaged confidence bands (shaded blue). The true function $f_j(x_j)$ is illustrated by the red solid curve.
\label{heteroskf3}
\end{figure}

\begin{figure}[t]
\begin{center}
\includegraphics[scale=0.25]{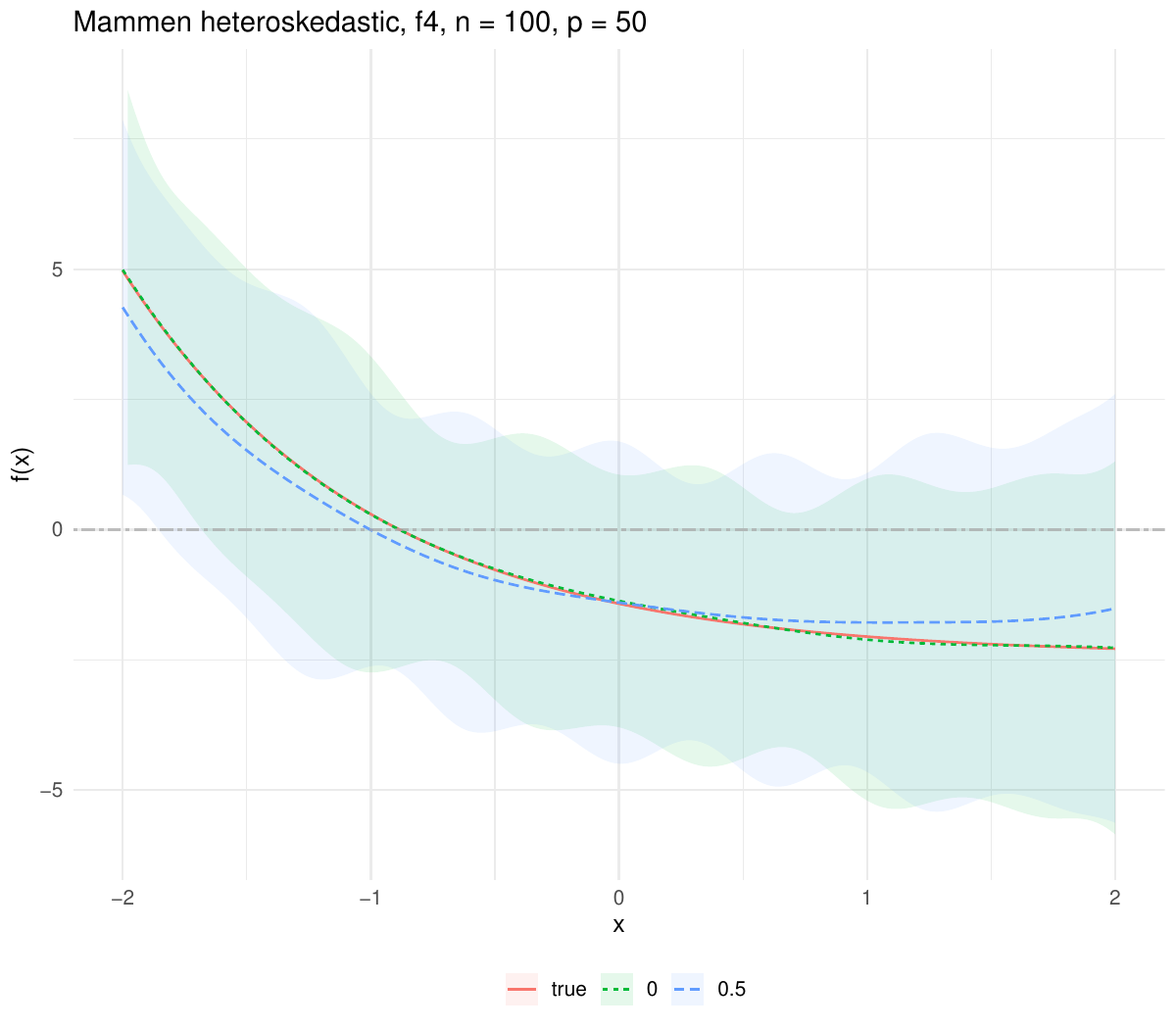} 
\includegraphics[scale=0.25]{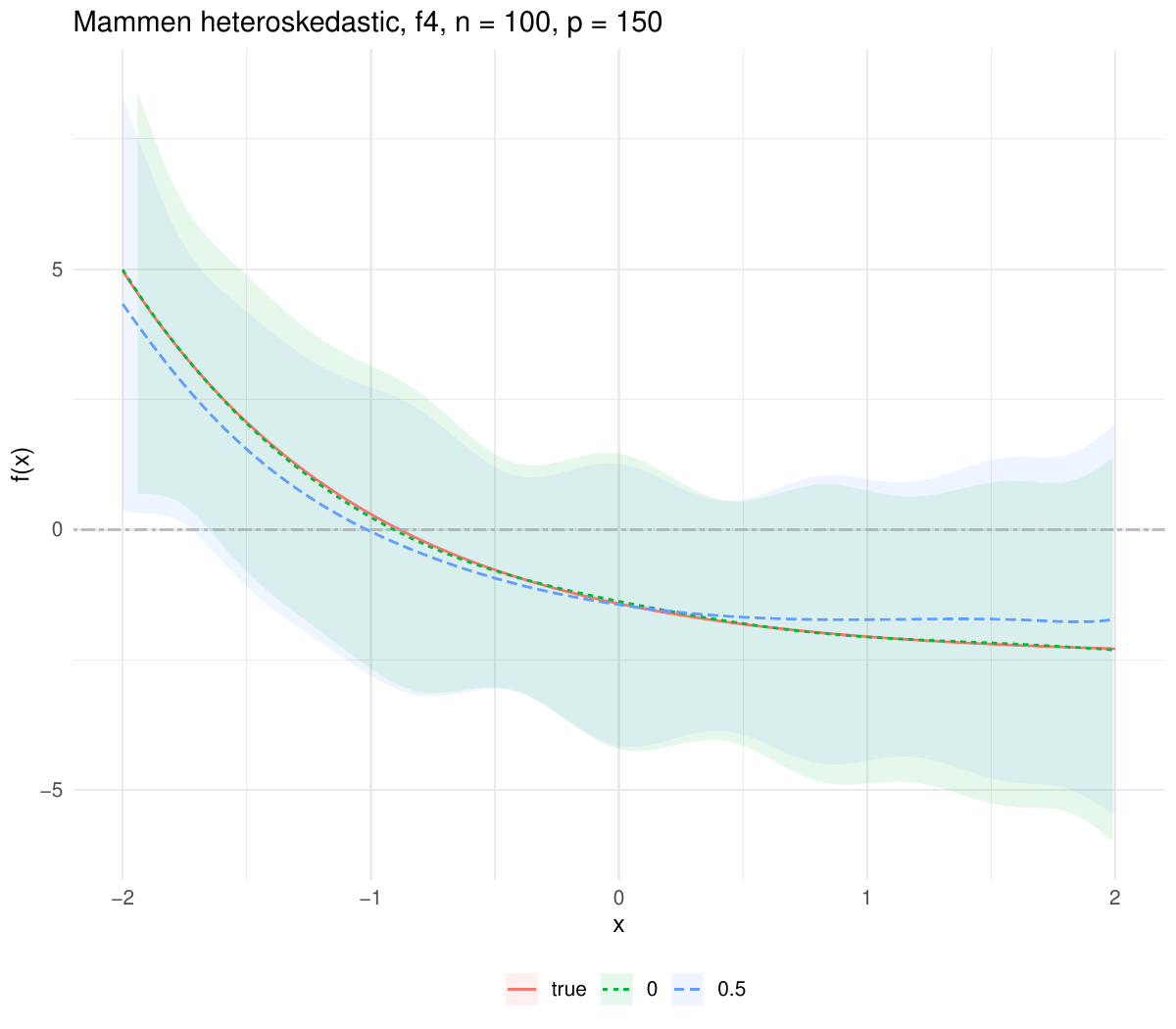} \\
\includegraphics[scale=0.25]{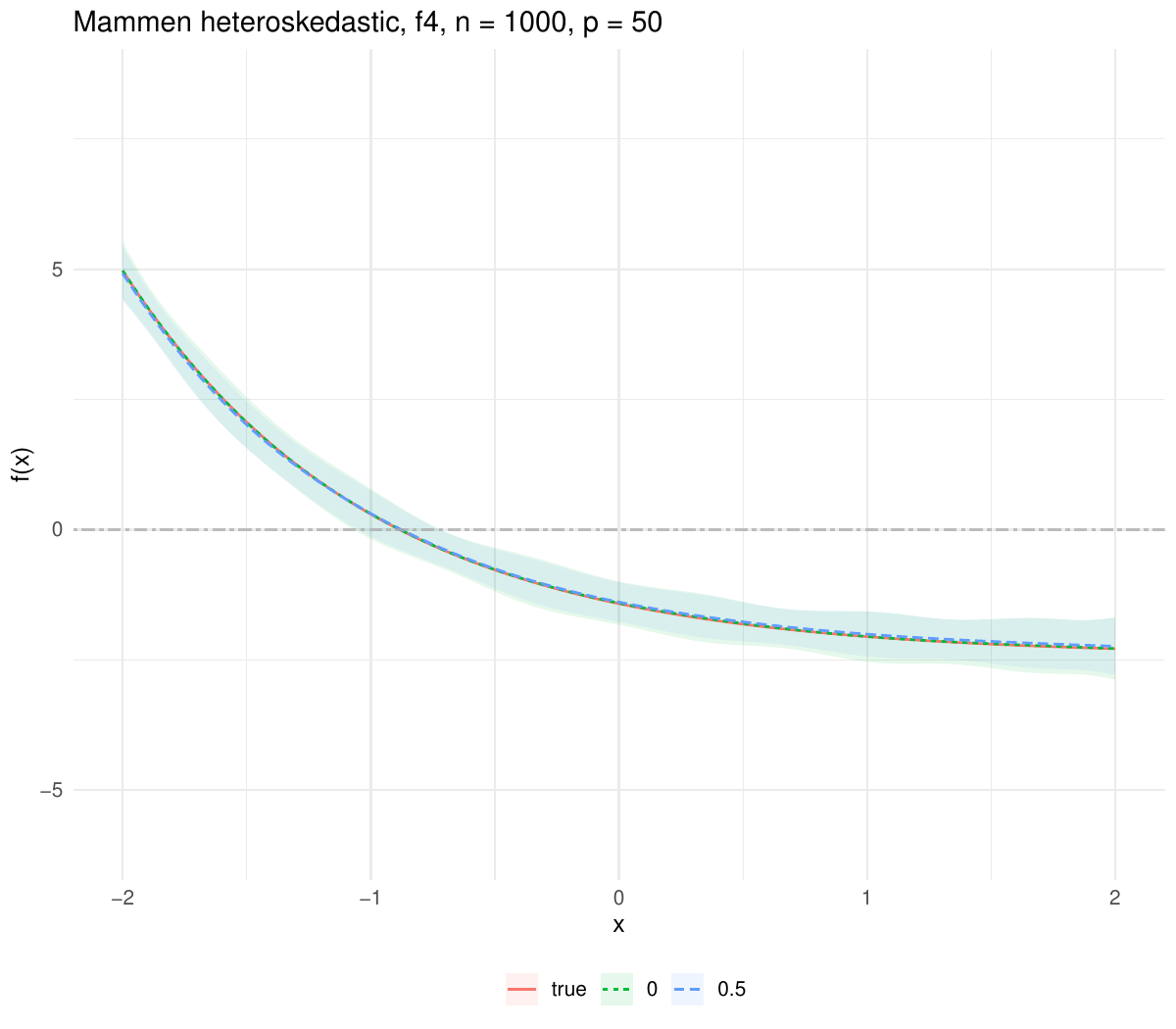}  
\includegraphics[scale=0.25]{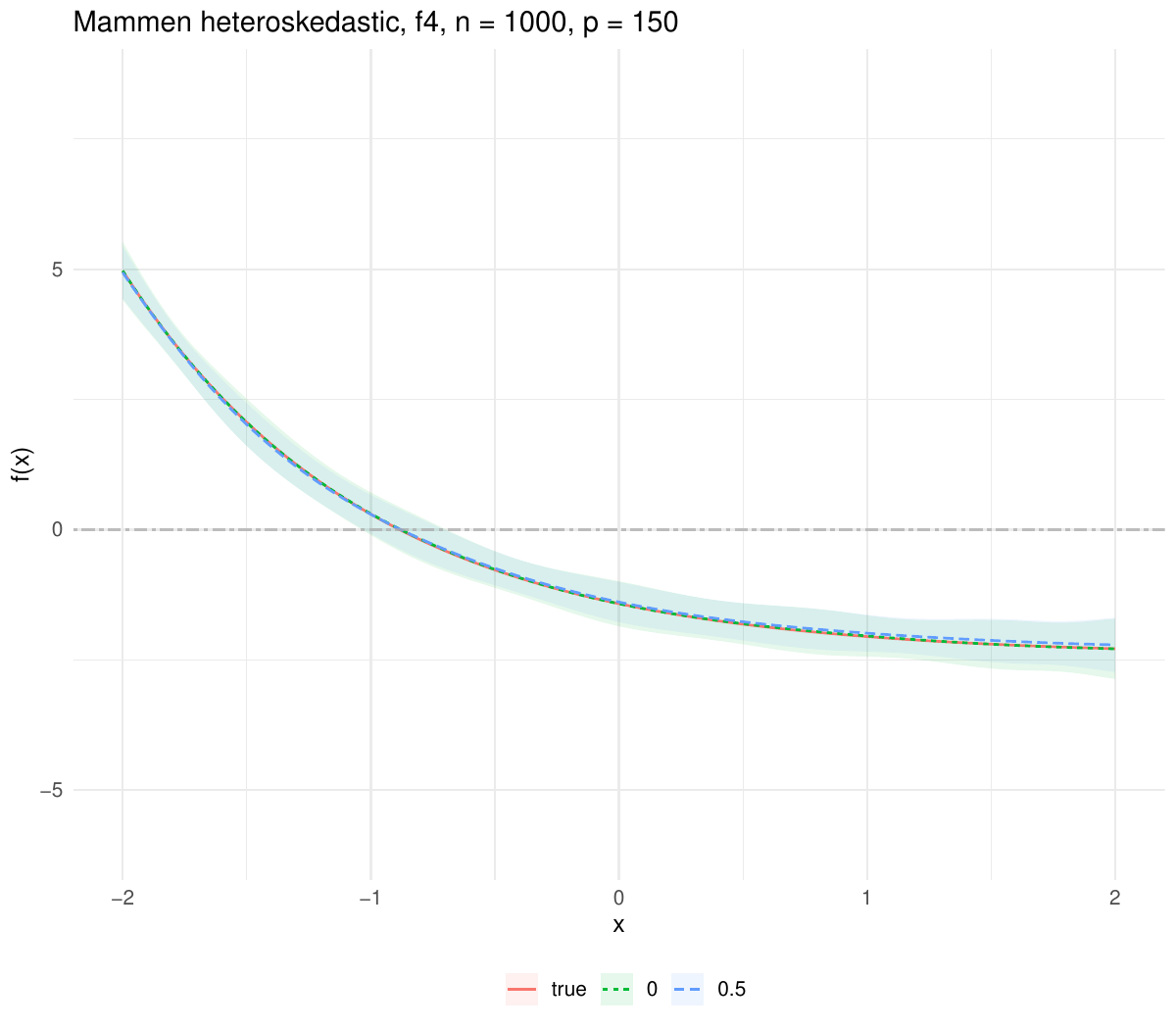} 
\caption{Average confidence bands, $f_4(x_4)$, heteroskedastic setting.}
\end{center}
The green dashed curve illustrates averaged estimated functions $\hat{f}_j(x_j)$  as obtained in $R=500$ repetitions in the scenario with $\rho = 0$. The corresponding averaged $95\%$-confidence bands are shaded green. The blue long-dashed line illustrates the results for the setting with $\rho = 0.5$ together with corresponding averaged confidence bands (shaded blue). The true function $f_j(x_j)$ is illustrated by the red solid curve.
\label{heteroskf4}
\end{figure}

\begin{figure}[t]
\begin{center}
\includegraphics[scale=0.25]{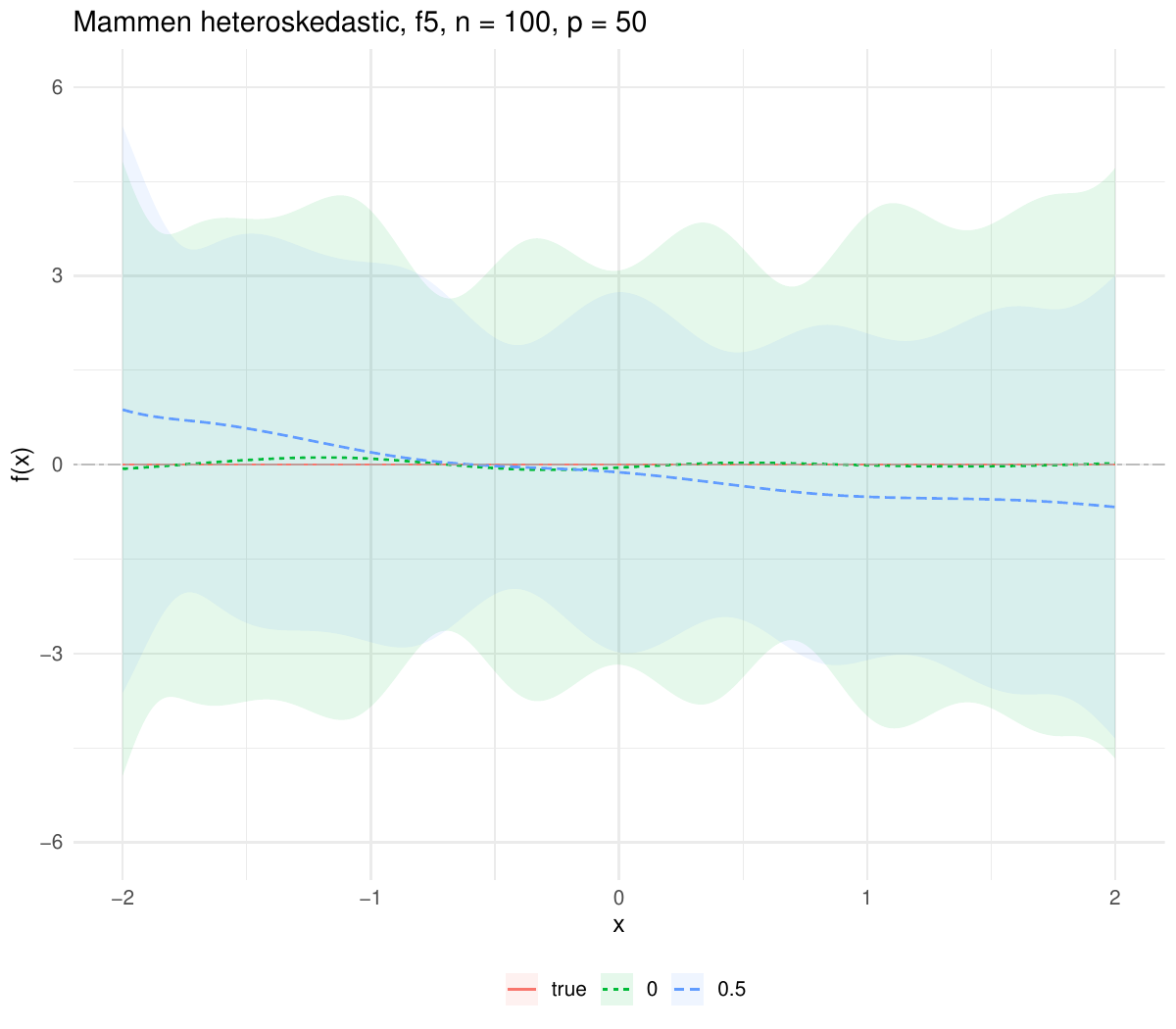} 
\includegraphics[scale=0.25]{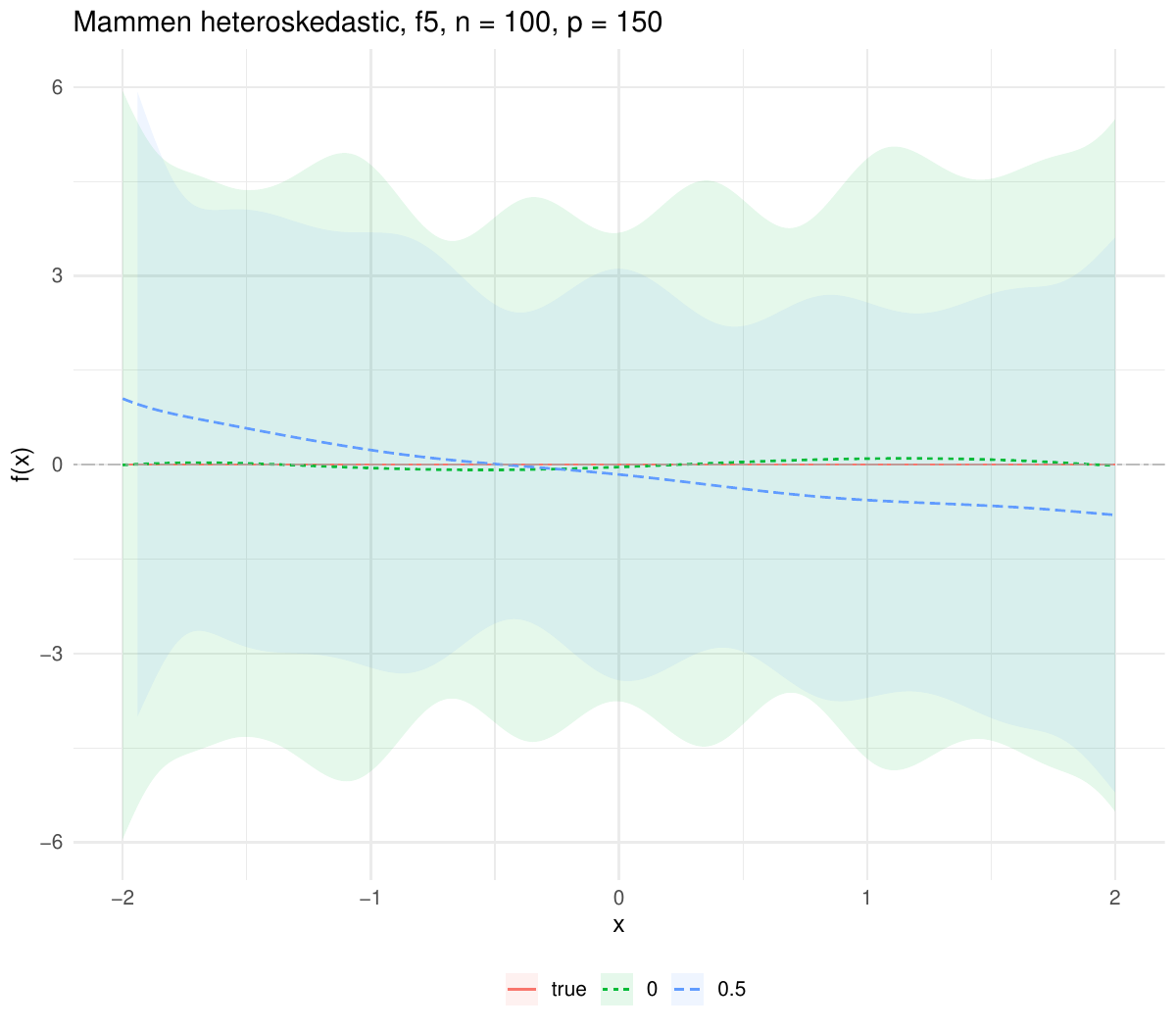} \\
\includegraphics[scale=0.25]{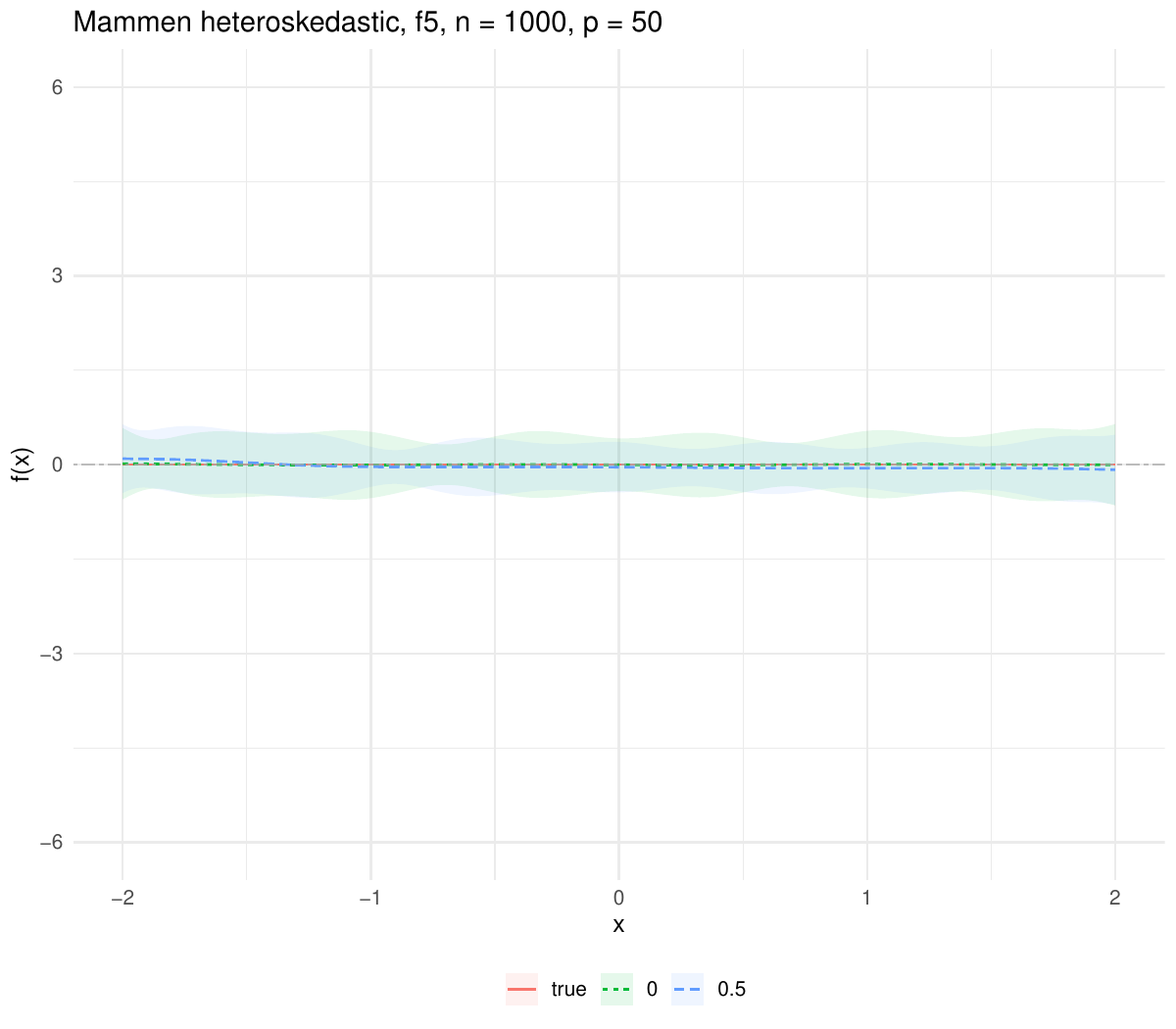}  
\includegraphics[scale=0.25]{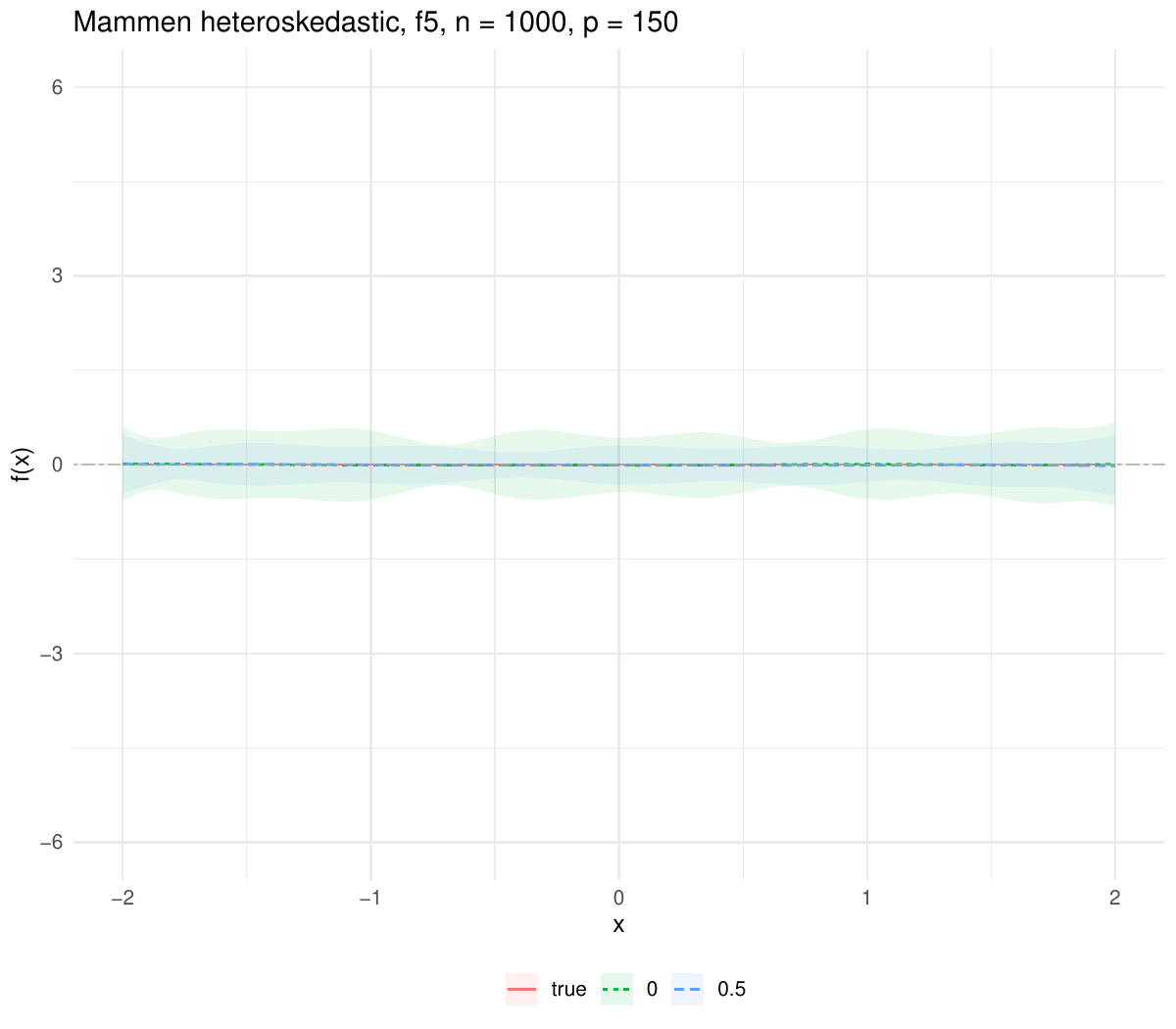} 
\caption{Average confidence bands, $f_5(x_5)$, heteroskedastic setting.}
\end{center}
The green dashed curve illustrates averaged estimated functions $\hat{f}_j(x_j)$  as obtained in $R=500$ repetitions in the scenario with $\rho = 0$. The corresponding averaged $95\%$-confidence bands are shaded green. The blue long-dashed line illustrates the results for the setting with $\rho = 0.5$ together with corresponding averaged confidence bands (shaded blue). The true function $f_j(x_j)$ is illustrated by the red solid curve.
\label{heteroskf5}
\end{figure}